\documentclass[11pt,a4paper]{article}
\usepackage{amssymb}
\usepackage{graphicx}
\usepackage{amsmath}
\usepackage{makeidx}
\usepackage{indentfirst}
\usepackage[T1]{fontenc}
\usepackage[utf8]{inputenc}

\topmargin- 20mm
\newtheorem{assumption}{Assumption}[section]

\newcounter{assumptionnum}[section]
\newcounter{resultnum}[section]
\newtheorem{consequence}{Consequence}[section]

\newcounter{consequencenum}[section]
\newtheorem{conclusion}{Conclusion}[section]

\newcounter{conclusionnum}[section]
\newcounter{conditionnum}[section]
\newcounter{conjecturenum}[section]
\newtheorem{example}{Example}[section]

\newcounter{examplenum}[section]
\newcounter{exercisenum}[section]
\newtheorem{lemma}{Lemma}[section]

\newcounter{lemmanum}[section]
\newcounter{notationnum}[section]
\newtheorem{theorem}{Theorem}[section]

\newcounter{theoremnum}[section]
\newtheorem{definition}{Definition}[section]

\newcounter{definitionnum}[section]
\newtheorem{corollary}{Corollary}[section]

\newcounter{corollarynum}[section]
\newtheorem{remark}{Remark}[section]

\newcounter{remarknum}[section]
\newcounter{propositionnum}[section]
\newcounter{acknowledgementnum}[section]
\newcounter{algorithmnum}[section]
\newcounter{axiomnum}[section]
\newcounter{casenum}[section]
\newcounter{claimnum}[section]
\newcounter{conseqnum}[section]
\newtheorem{convention}{Convention}[section]

\newcounter{conventionnum}[section]
\newcounter{summarynum}[section]
\newtheorem{principle}{Principle}[section]

\newcounter{principlenum}[section]
\newcounter{problemnum}[section]
\newenvironment{proof}[1][]{\textbf{Proof.} }{}

\begin{document}

\title{Quasi--Stationary Solutions in Gravity Theories with Modified
Dispersion Relations and Finsler-Lagrange-Hamilton Geometry\\
{\quad}}
\date{January 30, 2020}
\author{{\Large \vspace{.1 in} {\ \textbf{\ Lauren\c{t}iu Bubuianu}}\thanks{%
email: laurentiu.bubuianu@tvr.ro }} \\
{\textit{TVR Ia\c{s}i, \ 33 Lasc\v{a}r Catargi street; and University
Apollonia, 2 Muzicii street, 700107 Ia\c{s}i, Romania }}\\
${\ \quad }$ \\
and \\
${\ \quad }$ \\
{\Large \vspace{.1 in} \textbf{\ Sergiu I. Vacaru}\thanks{%
emails: sergiu.vacaru@gmail.com and sergiuvacaru@mail.fresnostate.edu ;
 \newline
 \textit{Address for post correspondence in 2019-2020 as a visitor senior researcher at YF CNU Ukraine is:\ } 37 Yu. Gagarin street, ap. 3, Chernivtsi, Ukraine, 58008}}\\
{\small \textit{Physics Department, California State University at Fresno, Fresno, CA 93740, USA;}} \\
 {\small \textit{Yuriy Fedkovych Chernivtsi National University, Institute of Applied-Physics and Computer Sciences,}}\\
 {\small \textit{Depart. Theoretical Physics and Computer Modelling, 101  Storozhynetska street,
   Chernivtsi, 58029, Ukraine}}
}
\maketitle

\begin{abstract}
Modified gravity theories, MGTs, with modified (nonlinear) dispersion
relations, MDRs, encode via indicator functionals possible modifications and
effects of quantum gravity; in string/brane, noncommutative and/or
nonassociative gravity theories etc. MDRs can be with global and/or local
Lorentz invariance violations, LIVs, determined by quantum fluctuations,
random, kinetic, statistical and/or thermodynamical processes etc. Such MGTs
with MDRs and corresponding models of locally anisotropic spacetime and
curved phase spaces can be geometrized in an axiomatic form for theories
constructed on (co) tangent bundles with base spacetime Lorentz manifolds.
In certain canonical nonholonomic variables, the geometric/physical objects
are defined equivalently as in generalized Einstein-Finsler and/or
Lagrange-Hamilton spaces. In such Finsler like MGTs, the coefficients of
metrics and connections depend both on local Lorentz spacetime coordinates
and, additionally, on (co) fiber velocity and/or momentum type variables.
The main goal of this work is to elaborate on a nonholonomic diadic "shell
by shell" formulation of MGTs with MDRs, with a conventional (2+2)+(2+2)
splitting of total phase space dimensions, when the (dual) modified
Einstein-Hamilton equations can be decoupled in general forms. We show how
this geometric formalism allows us to construct various classes of exact and
parametric solutions determined by generating and integration functions and
effective sources depending, in principle, on all phase space coordinates.
There are derived certain most general and important formulas for nonlinear
quadratic elements and studied the main geometric and physical properties of
quasi-stationary generic off-diagonal and diagonalizable phase spaces. This
work provides a self-consistent geometric and analytic method for
constructing in our further partner papers different types of black hole
solutions for theories with MDRs and LIVs and elaborating various
applications in modern cosmology and astrophysics.

\vskip4pt


\vskip4pt


\vskip4pt

\end{abstract}

\tableofcontents


\section{Introduction}

Modified gravity theories, MGTs, with nonlinear (i.e modified) dispersion
relations, MDRs, have played an increasingly important role for research in
modern particle physics, classical and quantum gravity, QG, accelerating
cosmology and astrophysics. Geometric methods, axiomatic approaches, various
applications and important developments are summarized and discussed in a
series of our recent works \cite%
{v18a,gvvepjc14,bubuianucqg17,vbubuianu17,vacaruplb16,vacaruepjc17}. We also
cite \cite{capoz,nojod1,clifton}, for reviews on MGTs and applications, and
\cite%
{vap97,vnp97,amelino98,vapny01,castro07,mavromatos11,mavromatos13a,kostelecky11,kostelecky16,vplb10,basilakos13}%
, for various theoretical and phenomenological works involving possible
(local) Lorentz invariance violations, LIVs.

In \cite{v18a}, we provided a comprehensive geometric study of physical
principles and axiomatic approaches to MGTs with MDRs and LIVs. We concluded
that such theories can be formulated as self-consistent causal extensions of
general relativity (GR, i.e. the Einstein gravity theory) if the geometric
constructions are extended on tangent and cotangent Lorentz bundles, $TV$
and $T^{\ast }V.$ Such relativistic phase space models are elaborated on a
base spacetime, considered as a four dimensional, 4-d, Lorentz manifold $V$
enabled with a pseudo-Riemannian metric structure $g_{ij}(x^{k}),$ with $%
i,j,k,...=1,2,3,4;$ for instance, of local signature $(+,+,+,-).$ A metric
tensor in GR, $g_{ij}(x^{k})$, is defined by a solution of standard Einstein
equations.\footnote{%
In a more general approach, we can consider on $V$ other type MGTs, for
instance, models with Lagrange density determined by a functional $f$ of a
Ricci scalar, $f(R).$} On total (phase) spaces, we suppose that the metric
structures $g_{\alpha \beta }(x^{k},v^{a}),$ or $g_{\alpha \beta
}(x^{k},p_{a}),$ is determined by MDRs in some forms depending additionally
on velocity/momentum like (co) fiber coordinates, when indices $\alpha
=(i,a),\beta =(j,b)$ etc. run values $1,2,...,7,8.$ Such locally anisotropic
metrics (in our geometric constructions, we shall use also certain classes
of (non) linear connections) can be derived naturally for various
semi-classical commutative and/or noncommutative models of MGTs and/or QG
theories, when the MDRs can be written locally in a general form
\begin{equation}
c^{2}\overrightarrow{\mathbf{p}}^{2}-E^{2}+c^{4}m^{2}=\varpi (E,%
\overrightarrow{\mathbf{p}},m;\ell _{P}).  \label{mdrg}
\end{equation}%
In this formula, an \textbf{indicator} of modifications $\varpi (...)$
encodes possible contributions of modified/generalized physical theories
with LIVs etc. In general, an indicator $\varpi (x^{i},E,\overrightarrow{%
\mathbf{p}},m;\ell _{P})$ depends on spacetime coordinates $x^{i}$ on base
spacetime manifold and can be computed, chosen and studied following
theoretical and/or phenomenological arguments and experimental/ observations
data.\footnote{%
For $\varpi =0,$ the formula (\ref{mdrg}) transforms into a standard
quadratic dispersion relation for a relativistic point particle, or
perturbations of a scalar field, with (effective) mass $m$, energy $E$, and
momentum $p_{\acute{\imath}}$ (for $\acute{\imath}=1,2,3$) propagating in a
four dimensional, 4-d, flat Minkowski spacetime. Here we note that certain
modifications of the special relativity theory, SRT, and GR could be
consequences of some (deformed) modified symmetries, for instance, with
(non) commutative deformed Poincar\'{e} transforms, quantum groups and
interactions, LIVs etc. Such values may involve, for instance, effective
energy-momentum variables $p_{a}=(p_{\acute{\imath}},p_{4}=E),%
\overrightarrow{\mathbf{p}}=\{p_{\acute{\imath}}\},$ (for $a=5,6,7,8$), at
the Planck scale $\ell _{p}:=\sqrt{\hbar G/c^{3}}\sim 10^{-33}cm,$ or other
scale parameters. In this work, we fix the light velocity, $c=1,$ for a
respective system of physical units.}

In corresponding nonholonomic variables, MGTs with MDRs can be elaborated
and studied as certain Finsler-Lagrange-Hamilton gravity theories modelled
on $TV$ and/or $T^{\ast}V.$ For well-defined conditions, such generalized
non-Riemannian geometries can be constructed in certain equivalent forms
when the fundamental geometric objects and effective gravitational and
matter field equations are naturally determined by $\varpi
(x^{i},p_{a},...). $ In \cite{v18a}, we proved that modified
Einstein-Finsler/-Lagrange/-Hamilton equations can derived following
geometric and/or variational principles, when the geometric and analytic
calculi are adapted to nonlinear connection structures. For MGTs elaborated
on $TV$, higher order generalizations and extra dimension nonholonomic
manifolds\footnote{%
equivalently, such spaces are called anholonomic, i.e. enabled with
non-integrable distributions of geometric objects and local frames; in this
work, we elaborate on geometric constructions which are adapted to nonlinear
connection, N-connection, structures}, the modified Einstein-Finsler
equations can be solved in very general forms following the anholonomic
frame deformation method, AFDM. We cite \cite%
{gvvepjc14,bubuianucqg17,vbubuianu17,vacaruplb16,vacaruepjc17} and
references therein for details and various examples of locally anisotropic
black hole and cosmological solutions. Our program for research on classical
and quantum MGTs with MDRs involves extensions of the AFDM for gravitational
and matter field theories and geometric evolution equations formulated on
cotangent bundles $T^{\ast }V.$ Such constructions include models with
generalized Einstein-Hamilton spaces and their supersymmetric/
noncommutative/ nonassocitative and other type extensions. Various classes
of exact and parametric solutions which can be constructed by geometric and
analytic methods depend, in general, on all spacetime and phase space
variables $(x^{k},p_{a}).$

The goal of this work is to elaborate on the nonholonomic diadic geometric
formulation\footnote{%
"diad" is from a corresponding Latin word; in our approach, it means a
splitting of coordinates and indices of physical variables into certain
conventional shells of dimension 2 when the total dimension is of type
2+2+2..., or 3+2+2+...; when in certain adapted frames the components of
geometric objects of higher dimension depend only on components of the same
and lower dimension} of MGTs with MDRs and extensions of the GR on (co)
tangent Lorentz bundles. This formalism allow us to develop the AFDM for
constructing exact solutions of physically important systems of nonlinear
partial differential equations, PDEs. We shall consider such conventional
diadic splitting: (2+2) for a base spacetime Lorentz manifold $V$; and
(2+2)+(2+2) decompositions on $T^{\ast }V $, when the total dimension $\dim
(T^{\ast }V)=8$. In correspondingly defined diadic nonholonomic variables,
the modified Einstein-Hamilton equations can be decoupled and integrated in
very general forms. In general, it is possible to construct in explicit form
various classes of physically important generic off-diagonal solutions
determined by generating and integration functions depending, in principle,
on all phase space variables. Such methods and solutions will be applied in
our further works on MGTs and applications in modern cosmology and
astrophysics.

We suppose that readers are familiar with standard results on mathematical
relativity and particle physics (including the geometry of vector bundles
and metric and connection structures) summarized in typical monographs on
the Einstein gravity theory \cite{hawrking73,misner73,wald82} and a
collection of most important exact solutions in GR theory \cite{kramer03}.
One can be consulted also some standard monographs on (generalized) Finsler
geometry and nonholonomic manifolds \cite%
{cartan35,rund59,asanov85,matsumoto86,bejancu90,bao00,bejancu03}. Such works
were written by mathematicians or physicists who devoted their research to
Finsler spaces and applications involving in the bulk local Euclidean
signatures and the Riemann-Finsler geometry. Those works have not discussed
possible generalizations of the Einstein gravity with Finsler like
modifications following modern MGTs approaches for metric--affine gravity,
elaborating classical and quantum (for various noncommutative and
supersymmetric generalizations) models with MDRs, applications in
accelerating cosmology etc. Physicists and mathematicians interested in
research in such recent directions of mathematical physics and geometric
methods in physics may consider the results, methods and critics summarized
in monograph \cite{vmon06} and a recent unconventional review \cite{v18a},
on the axiomatic approaches and historical remarks on relativistic
Finsler-Lagrange-Hamilton spaces and applications.

This article can be considered as a second partner of \cite{v18a} when the
axiomatic approach to MGTs with MDRs is completed with new geometric methods
for constructing exact solutions in such theories with additional momentum
like variables. We prove that there are general decoupling and integration
properties of fundamental field equations on cotangent Lorentz bundles, and
related Einstein-Hamilton gravity theories. Such geometric methods and
applications were elaborated for nonholonomic manifolds in (super) string
and/or noncommutative gravity theories and certain tangent bundle
(generalized Finsler-Lagrange spaces) models, see \cite%
{gvvepjc14,bubuianucqg17,vbubuianu17,vacaruplb16,vacaruepjc17}. In this
work, we concentrate on a comprehensive study of quasi-stationary exact and
parametric solutions when in certain systems of reference the coefficients
of fundamental geometric objects (metrics, connections etc.) do not depend
on time like coordinates (at least for base manifold projections) and
possess a Killing symmetry, for instance, on a momentum-like variable.%
\footnote{%
in principle, using the AFDM, we can generate more general classes of
solutions without Killing symmetries and dependence on all phase space and
spacetime coordinates but such constructions are more cumbersome and it is
not clear what physical importance may have such "very general" nonholonomic
phase spaces} It is shown how exact (non) vacuum off-diagonal solutions with
matter field and effective sources and conventional cosmological constants
can be constructed in general forms. We state certain conditions for
extracting Levi-Civita, LC, configurations with zero torsion, and analyze
corresponding nonlinear symmetries relating different classes of metrics and
connections with diadic structure. It should be noted that the material is
arranged in a form which is typical for "geometry and physics" articles with
definitions, corollaries, theorems etc. Certain proofs are sketched in the
main part or provided in appendices (or with references to previous works
where similar constructions are provided). Explicit applications in modern
cosmology and astrophysics (for instance, for constructing cosmological
scenarios with inflation and acceleration, and black hole solutions,
involving momentum like variables) with be provided in our further works,
see also further references and discussions.

This paper is organized as follows:\ In section \ref{s2}, we elaborate on a
nonholonomic diadic geometric formulation of gravity theories with MDRs and
LIVs modelled on (co) tangent Lorentz bundles. There are defined
corresponding (non) linear connection and metric structures and computed
curvature, torsion and nonmetricity tensors (and related Ricci and Einstein
tensors, scalar curvatures). We prove that there are such nonholonomic
variables when the geometry of such curved phase spaces can be equivalently
modelled as pseudo (with local pseudo-Euclidean signature) Lagrange-Hamilton
(in particular, Finsler and their cotangent dual) spaces. 

In section \ref{s3}, there are defined in adapted form the Lagrange
densities for gravitational and scalar matter fields with locally
anisotropic phase space interactions on cotangent Lorentz bundles. We prove
that modified Einstein equations with MDRs can be formulated as equivalent
generalized Einstein-Lagrange/ -Hamilton equations. Such formulas can be
derived in equivalent N-adapted variational and abstract geometric forms,
see related details in \cite{v18a}. We consider canonical formulations with
nonholonomic 4+4 and diadic (2+2)+(2+2) splitting of locally anisotropic
spacetime total phase space gravitational and scalar field moving equations.

Section \ref{s4} contains further developments of the AFDM: There are
formulated some important theorems on general decoupling and integrability
of generalized Einstein equations with MDR and their equivalent nonlinear
dynamical equations for modified Einstein-Hamilton systems. We prove that
corresponding systems of nonlinear PDEs split into coupled "shell by shell"
two dimensional systems of nonlinear equations on base and conventional
fiber spaces. For such splitting, such nonlinear systems with decoupling can
be solved in explicit general form for generic off-diagonal ansatz with
certain spacetime and phase space Killing symmetries for gravitational and
(effective) matter field sources. We define an important class of nonlinear
symmetries relating generation functions to generating (effective) sources
and shell cosmological constants.

Section \ref{s5} is devoted to a study of nonholonomic deformations of phase
space metrics into exact and parametric solutions. There are studied
diagonal phase space configurations; (off-) diagonal phase space vacuum
solutions, and how such locally anisotropic solutions can model explicit
examples of Einstein-Hamilton spaces and Finsler-Lagrange geometry.

A summary of main results and final remarks are presented in section \ref{s6}%
.

Finally, we also note that the proofs of theorems are sketched in some forms
being accessible both for researchers on mathematical physics and
phenomenology of particle physics and gravity. It is preferred the so-called
abstract geometric method and emphasized the possibility of alternative
adapted variational proofs. In Appendix \ref{appendixa}, there are stated
certain important N-adapted coefficient formulas which are necessary for
proofs of decoupling properties and integrability of physically important
systems of nonlinear PDEs. In some subsections, the geometric and analytic
computations are provided in detailed forms in order to show such methods
can be applied in MGTs. For convenience, we provide in Appendix \ref%
{appendixb} a brief summary on the geometry of relativistic
Finsler-Lagrange-Hamilton spaces formulated in canonical (so-called "tilde")
variables and show how such geometric structures are naturally determined by
indicators of MDRs and LIVs.

\section{Nonholonomic Diadic Geometry of (co) Tangent Lorentz Bundles}

\label{s2} This section contains an introduction to the geometry of
nonholonomic diadic structures with conventional splitting of dimensions
(2+2)+(2+2) on (co) tangent Lorentz bundles. Such a formalism is important
for constructing exact solutions and quantizing various modifications of the
Einstein gravity theory and related Lagrange-Hamilton gravity theories
canonically defined by MDRs, see reviews of relevant former results, methods
and applications in \cite%
{v18a,gvvepjc14,bubuianucqg17,vbubuianu17,vacaruplb16,vacaruepjc17}. We
elaborate both a coordinate free formulation adapted coefficient formalism
and provide necessary coordinate formulas in Appendix \ref{appendixa}. Basic
concepts on relativistic Finsler-Lagrange-Hamilton are outlined in Appendix %
\ref{appendixb}.

\subsection{Geometric preliminaries}

\label{ssassumpts} Standard theories of particle physics based on special
relativity, SR, and general relativity, GR, principles are elaborated for
(correspondingly) a four dimensional, 4-d, Lorentz spacetime manifold $V$
and respective (co) tangent bundles, $TV$ and/or $T^{\ast }V.$ The standard
curved spacetime geometry is stated following this

\begin{assumption}
\label{assumptqelorentz} \textsf{[standard quadratic line elements on total
spaces of (co) tangent Lorentz bundles with flat (co) fiber metrics] } The
corresponding quadratic elements determined by flat (co) fiber metrics but
curved base spacetime and total phase space metrics with signature $%
(+++-;+++-)$, can be written in local coordinate form as
\begin{eqnarray}
ds^{2} &=&g_{\alpha \beta }(x^{k})du^{\alpha }du^{\beta
}=g_{ij}(x^{k})dx^{i}dx^{j}+\eta _{ab}dy^{a}dy^{b},\mbox{ for }y^{a}\sim
dx^{a}/d\tau ;\mbox{ and/ or }  \label{lqe} \\
d\ ^{\shortmid }s^{2} &=&\ ^{\shortmid }g_{\alpha \beta }(x^{k})d\
^{\shortmid }u^{\alpha }d\ ^{\shortmid }u^{\beta
}=g_{ij}(x^{k})dx^{i}dx^{j}+\eta ^{ab}dp_{a}dp_{b},\mbox{ for }p_{a}\sim
dx_{a}/d\tau .  \label{lqed}
\end{eqnarray}
\end{assumption}

In above formulas, we can consider $g=\{g_{ij}(x^{k})\}$ as a solution of
the Einstein equations for the Levi-Civita connection $\nabla .$ We note
that even in the GR theory we work on higher order (co) tangent Lorentz
bundles $TV,T^{\ast }V,TTV,TT^{\ast }V$ etc. and not only on the base
Lorentz manifold $V$. This is because the gravitational and matter field
equations contain at least second derivatives of metrics, connections,
tensor and spinor fields. Nevertheless, the solutions for such fields can be
always parameterized in certain forms depending only on spacetime
coordinates $x=\{x^{k}\}$ (for instance, $g_{ij}(x),\Gamma
_{jk}^{i}(x),A_{k}(x)$) if the extensions on (co) fiber spaces is performed
with the flat Minkowski metric. This holds for a zero MDR indicator $\varpi $
(\ref{mdrg}) and/or in any fixed point on a (co) fiber space, $%
v^{a}=v_{(0)}^{a}$ (or $p_{a}=p_{a}^{(0)}),$ when normal (co) fiber
coordinates are considered in vicinity of such points. That why we consider
that in (\ref{lqe}) and (\ref{lqed}), respectively, the (co) vertical metric
$\eta _{ab}$ and its dual $\eta ^{ab}$ are standard Minkowski metrics, $\eta
_{ab}=diag[1,1,1,-1].$ The metric structures on tangent and cotangent
bundles enabled with above type quadratic elements can be parameterized
equivalently by the same h-components of $g_{\alpha \beta }(x^{k})$ when $\
^{\shortmid }g_{ij}(x^{k})=g_{ij}(x^{k}).$ Relativistic curves $x^{a}(\tau )$
on $V$ are parameterized by a positive parameter $\tau .$ In principle, we
consider negative values of parameters for curves but we follow the
convention that parameterizations are introduced in such ways when at least
locally, in vicinity of a point of a Lorentz manifold, the relativistic
curves can be defined in a causal form with $ds=cd\tau $ for $\tau $ being a
relativistic time like parameter taken to be positive.\footnote{%
Index and coordinate conventions for geometric objects:
\par
\begin{itemize}
\item \textbf{\ 8-d and 4-d indices:}\ For local coordinates on a tangent
bundle $TV,$ we shall write \ $u^{\alpha }=(x^{i},v^{a}),$ (or in brief, $%
u=(x,v)),$ when indices $i,j,k,...=1,2,3,4$ and $a,b,c,...=5,6,7,8;$ when
for cumulative indices $\alpha ,\beta ,...=1,2,...8.$ Similarly, on a
cotangent bundle $T^{\ast }V,$ we write $\ ^{\shortmid }u^{\alpha
}=(x^{i},p_{a}),$ (or in brief, $\ ^{\shortmid }u=(x,p)),$ where $%
x=\{x^{i}\} $ are considered as coordinated for a base Lorentz manifold $V$
(as in GR). The coordinate $x^{4}=t$ is considered as time like one and $%
p_{8}=E$ is an energy type one. If necessary, we shall work with 3+1
decompositions when, for instance, $x^{\grave{\imath}},$ for $\grave{\imath}%
=1,2,3,$ are used for space coordinates; and $p_{\grave{a}},$ for $\grave{a}%
=5,6,7,$ are used for momentum like coordinates.
\par
\item \textbf{Diadic indices:} conventional (2+2)+(2+2) splitting of indices
are labeled following such rules: $\alpha _{1}=(i_{1}),$ $\alpha
_{2}=(i_{1},i_{2}=a_{2}),\beta _{2}=(j_{1},j_{2}=b_{2});\alpha
_{3}=(i_{3},a_{3}),\beta _{3}=(j_{3},b_{3}),...;\alpha
_{4}=(i_{4},a_{4}),\beta _{4}=(j_{4},b_{4}),$ for $i_{1},j_{1}=1,2;$ $%
i_{3},j_{3}=1,2,3,4;$ $i_{4},j_{4}=1,2,3,4,5,6;$ and $%
a_{2},b_{2}=3,4;a=(a_{3},a_{4}),b=(b_{3},b_{4}),$ for $a_{3},b_{3}=5,6$ and $%
a_{4},b_{4}=7,8,$ etc.; such a splitting adapted to the splitting of $h$%
-space into 2-d horizontal and vertical subspaces, $\ \ _{1}h$ and $\ \
_{2}v,$ of (co) vertical spaces $v$ into $\ ^{3}v$ and $\ ^{4}v$ into
conventional four 2-d shells labeled with lef up or low abstract indices
like $\ ^{s}v,$ or $\alpha _{s}=(i_{s},a_{s})$ for $s=1,2,3,4$ referring to
ordered shells. Hereafter, we shall put shall labels on the left up, or left
low, on convenience. Right indices will get additional labels 1,2,3,4 if it
will be necessary. It is convenient to contract indices using $\alpha
_{2}=(i_{1},a_{2})=1,2,3,4;\alpha _{3}=(\alpha
_{2},a_{3})=1,2,3,4,5,6;\alpha =\alpha _{4}=(\alpha _{3},a_{4})=1,2,...,8.$
In diadic form, the coordinates will split as $%
x^{i}=(x^{i_{1}},y^{a_{2}}),v^{a}=(v^{a_{3}},v^{a_{4}});p_{a}=(p_{a_{3}},p_{a_{4}}).
$ We shall write $\ _{s}u=\{$ $u^{\alpha _{s}}=(x^{i_{s}},v^{a_{s}})\}$ and $%
\ _{s}^{\shortmid }u=\{\ ^{\shortmid }u^{\alpha
_{s}}=(x^{i_{s}},p_{a_{s}})\} $ for cumulative indices on corresponding $s$%
-shell.
\end{itemize}
}

The main consequence of Assumption \ref{assumptqelorentz} is that the
geometric and physical objects in standard physical theories (like GR and
particle physics) depend only on base spacetime coordinates and do not
involve, for instance, metrics and connections depending on velocity/
momentum variables. Even for certain constructions (for instance, for
elaborating relativistic kinetic theories) there are considered phase spaces
with effective velocity/momentum variables, one can be defined always
certain canonical lifts to total spaces of standard pseudo-Riemannian
metrics and Levi-Civita connections (metric compatible and with zero
torsion) depending only on base Lorentz manifold spacetime coordinates. For
general frame/coordinate transforms on total Lorentz (co) bundle spaces, the
total and (co) fiber metrics (\ref{lqe}) and (\ref{lqed}) transform
respectively into Sasaky type ones with coefficients depending, in
principle, on all phase space coordinates $(x,v)$ and/or $(x,p),$ see
formulas (\ref{dmt})-(\ref{dmcts}) and discussions therein.

\begin{definition}
\textsf{[Nonlinear connections and nonholonomic h-v, h-cv and diadic
splitting] } \label{defnc}Nonlinear connection, N--connection, structures
for $TV,$ or $T^{\ast }V,$ are defined by a Whitney sum $\oplus $ of
conventional $h$- and $v$--distributions, or $h$ and $cv$--distributions,
when
\begin{equation}
\mathbf{N}:T\mathbf{TV}=hTV\oplus vTV\mbox{ or }\ \ ^{\shortmid }\mathbf{N}:T%
\mathbf{T}^{\ast }\mathbf{V}=hT^{\ast }V\oplus vT^{\ast }V.  \label{ncon}
\end{equation}%
A (2+2)+(2+2) splitting is stated by for respective N--connections when%
\begin{eqnarray}
\ _{s}\mathbf{N}:\ _{s}T\mathbf{TV} &=&\ ^{1}hTV\oplus \ ^{2}hTV\oplus \
^{3}vTV\oplus \ ^{4}vTV\mbox{ and }  \notag \\
\ _{s}^{\shortmid }\mathbf{N}:\ \ _{s}T\mathbf{T}^{\ast }\mathbf{V} &=&\
^{1}hT^{\ast }V\oplus \ ^{2}hT^{\ast }V\oplus \ ^{3}vT^{\ast }V\oplus \
^{4}vT^{\ast }V,  \label{ncon2}
\end{eqnarray}%
where the low left label $s=1,2,3,4$ is for a "shell" splitting and
corresponding shell with rank 2 of  fibers.\footnote{%
A corresponding nonholonomic diadic splitting is important for decoupling
(modified) Einstein equations and generating exact solutions in explicit
form, see next sections. Let us briefly motivate why we need such shell by
shell 2+2+2+2 nonholonomic geometric constructions. MGTs with MDRs extending
GR are modeled on total 8-d $TV$ and $T^{\ast }V.$ Metrics on such total
bundles can be considered as $8\times 8$ dimensional symmetric matrices with
coefficients depending on 8-d phase space coordinates. To elaborate, for
instance, on generalizations of black hole solutions in GR into similar
objects on $TV$ and/or $T^{\ast }V$ (for stationary configurations) we have
to work with certain block off-diagonal symmetric matrices when the
coefficients depend at least on two coordinates, like $(r,v)$ or $(r,p)$ for
a radial base coordinate $r$ and additional (co) fiber coordinates like $v$
and/or $p.$ For more general classes of quasi-stationary solutions, the
constructions are more cumbersome. Nevertheless, in our previous works, we
proved that modified Einstein equations can be decoupled in generic
off-diagonal form for coefficients depending on all spacetime and extra
dimension coordinates using conventional splitting 2+2+2+2+.... The idea was
to introduce two dimensional shells, to construct generic off-diagonal
solutions for 2+2 dimensions, with metrics depending, in principle, on all
spacetime coordinates; then to extend them nonholonomically with mixing of
frames and coordinates to 2+2+2 dimensions, and then to other extra
dimensions. The decoupling of modified Einstein equations could not be
proven for any type shell/ diadic splitting or for general 3+3+..., or
4+4+..., splitting; and for any generalized connections or for the
Levi-Civita connections. The frames, metrics and connections for
nonholonomic shall diadic structures have to be correspondingly defined in
order to obtain a general decoupling of physically important systems of
nonlinear PDEs. Then, constructing some general classes of solutions in
explicit forms, we can re-introduce general covariant transforms, re-define
the connection structure in an another necessary form. We can also generate
and study some important classes of solutions with special symmetries,
interesting asymptotic behavior etc. For momentum like variables, we show
this in the next sections. In this section, we outline a nonholonomic diadic
geometry with adapted nonlinear and linear connections structures which provides possibilities to decouple and solve modified gravitational
equations. Readers should follow our conventions even they are familiar with other type ones for another geometric models. The point is that in our approach we
were able to elaborate the anholonomic frame deformation method, AFDM, which
allows us to construct very general classes of exact solutions in commutative
and noncommutative gravity theories, string gravity and extra dimensions,
Finsler-Lagrange-Hamilton models etc., see \cite%
{v18a,gvvepjc14,bubuianucqg17,vbubuianu17,vacaruplb16,vacaruepjc17,vijtp10,bvepjc18,bvap19,vcqg18}%
. Other geometric methods allowed only generating certain particular classes
of black hole, wormholes, and cosmological solutions which higher order
symmetries and usually parameterized by diagonal metrics.}
\end{definition}

Such decompositions into Whitney sums (mathematicians consider alternatively
certain classes of exact sequences in order to define N-connection
nonholonomic distributions, see \cite{matsumoto86,bao00,vmon06}) define
splitting into conventional 2-dim nonholonomic distributions of $TTV$ and $%
TT^{\ast }V,$ when%
\begin{eqnarray*}
\dim (\ ^{1}hTV) &=&\dim (\ ^{2}hTV)=\dim (\ ^{3}vTV)=(\ ^{4}vTV)=2%
\mbox{
and } \\
\dim (\ ^{1}hT^{\ast }V) &=&\dim (\ ^{2}hT^{\ast }V)=\dim (\ ^{3}vT^{\ast
}V)=(\ ^{4}vT^{\ast }V)=2.
\end{eqnarray*}%
In above formulas, left up labels like $\ ^{1}h,\ ^{3}v$ etc. state that we
using nonholonomic (equivalently, anholonomic and/or non-integrable)
distributions we split respective 8-d total spaces into 2-d shells 1,2,3 and
4. We cite \cite{v18a,gvvepjc14} for details and historical remarks on
N-connection geometry and applications in physics.\footnote{\label%
{fnotnconcoef2} The formulas (\ref{ncon}) can be written respectively in
local form, $\mathbf{N}=N_{i}^{a}\frac{\partial }{\partial v^{a}}\otimes
dx^{i}$ or $\ ^{\shortmid }\mathbf{N}=\ ^{\shortmid }N_{ia}\frac{\partial }{%
\partial p_{a}}\otimes dx^{i},$ using the N-connection coefficients $\mathbf{%
N}=\{N_{i}^{a}\}$ or $\ ^{\shortmid }\mathbf{N}=\{\ ^{\shortmid }N_{ia}\};$
the geometric objects on $\mathbf{V,}T\mathbf{V,}T^{\ast }\mathbf{V}$ will
be labeled by "bold face" symbols if they can be written in N-adapted form;
the up label bar $\ "^{\shortmid }"$ will be used if it will be necessary to
emphasize that certain geometric objects are defined on cotangent bundles.
The nonholonomic dyadic splitting with N-connections (\ref{ncon2}) are
defined locally by such coefficients
\begin{eqnarray*}
\ _{s}\mathbf{N} &=&%
\{N_{i_{1}}^{a_{2}}(x^{i_{1}},y^{a_{2}}),N_{i_{1}}^{a_{3}}(x^{i_{1}},y^{a_{2}},v^{b_{3}}), N_{i_{1}}^{a_{4}}(x^{i_{1}},y^{a_{2}},v^{b_{3}},v^{b_{4}}),\ N_{i_{2}}^{a_{3}}(x^{i_{1}},y^{a_{2}},v^{b_{3}}),N_{i_{2}}^{a_{4}}(x^{i_{1}}, y^{a_{2}},v^{b_{3}},v^{b_{4}}), \\
&&N_{a_{3}}^{a_{4}}(x^{i_{1}},y^{a_{2}},v^{b_{3}},v^{b_{4}})\}, \\
\ _{s}^{\shortmid }\mathbf{N}& =&\{\ ^{\shortmid
}N_{i_{1}}^{i_{2}}(x^{i_{1}},y^{a_{2}}),\ ^{\shortmid
}N_{i_{1}a_{3}}(x^{i_{1}},y^{a_{2}},p_{b_{3}}),\ ^{\shortmid
}N_{i_{1}a_{4}}(x^{i_{1}},y^{a_{2}},p_{b_{3}},p_{b_{4}}),
\ ^{\shortmid }N_{i_{2}a_{3}}(x^{i_{1}},y^{a_{2}},p_{b_{3}}),\ ^{\shortmid
}N_{i_{2}a_{4}}(x^{i_{1}},y^{a_{2}},p_{b_{3}},p_{b_{4}}),\\
&&\ ^{\shortmid }N_{\
a_{4}}^{a_{3}}(x^{i_{1}},y^{a_{2}},p_{b_{3}},p_{b_{4}})\}.
\end{eqnarray*}%
}

\begin{lemma}
\textsf{[N-adapted/dyadic (co) frames ] } \label{lemadap} N--connection
structures (\ref{ncon}) and (\ref{ncon2}) defines respective canonical
systems of N--adapted bases%
\begin{eqnarray*}
&&\mathbf{e}_{\alpha }\mbox{ on }\ T\mathbf{TV}\mbox{ and }\mathbf{e}%
^{\alpha }\mbox{ on }T^{\ast }\mathbf{TV;}\ ^{\shortmid }\mathbf{e}_{\alpha }%
\mbox{ on }\ T\mathbf{T}^{\ast }\mathbf{V}\mbox{ and }\ ^{\shortmid }\mathbf{%
e}^{\alpha }\mbox{ on }T^{\ast }\mathbf{T}^{\ast }\mathbf{V;} \\
\mbox{ and \  } &&\mathbf{e}_{\alpha _{s}}\mbox{ on }\ \ \ _{s}T\mathbf{TV}%
\mbox{ and }\mathbf{e}^{\alpha _{s}}\mbox{ on }\ \ _{s}T^{\ast }\mathbf{TV;}%
\ ^{\shortmid }\mathbf{e}_{\alpha _{s}}\mbox{ on }\ \ \ _{s}T\mathbf{T}%
^{\ast }\mathbf{V}\mbox{ and }\ ^{\shortmid }\mathbf{e}^{\alpha _{s}}%
\mbox{
on }\ \ _{s}T^{\ast }\mathbf{T}^{\ast }\mathbf{V.}
\end{eqnarray*}
\end{lemma}

The proofs follows from N-adapted constructions (considered canonical for
transforms which are linear on N--coefficients) in Appendix \ref{appendixa1}.

\begin{assumption}
\textsf{[d-metrics on (co) tangent Lorentz bundles and dyadic decomposition]}
\label{assumpt3} The total spaces of tangent, $\mathbf{TV},$ and cotangent, $%
\mathbf{T}^{\ast }\mathbf{V},$ Lorentz bundles used for elaborating physical
theories with MDR--generalizations of the Einstein gravity can be enabled,
respectively, with pseudo-Riemannian metric, $\mathbf{g},$ and $\
^{\shortmid }\mathbf{g,}$ structures and respective shell decompositions, $\
_{s}\mathbf{g},$ and $\ _{s}^{\shortmid }\mathbf{g}$. Using frame transforms
in N--adapted form, respective metric structures can be parameterized as
distinguished metrics (d-metrics) in such forms:
\begin{eqnarray}
\mathbf{g} &=&\mathbf{g}_{\alpha \beta }(x,y)\mathbf{\mathbf{e}}^{\alpha }%
\mathbf{\otimes \mathbf{e}}^{\beta }=g_{ij}(x)e^{i}\otimes e^{j}+\mathbf{g}%
_{ab}(x,y)\mathbf{e}^{a}\otimes \mathbf{e}^{a}\mbox{ and/or }  \label{dmt} \\
\ ^{\shortmid }\mathbf{g} &=&\ ^{\shortmid }\mathbf{g}_{\alpha \beta }(x,p)\
^{\shortmid }\mathbf{\mathbf{e}}^{\alpha }\mathbf{\otimes \ ^{\shortmid }%
\mathbf{e}}^{\beta }=g_{ij}(x)e^{i}\otimes e^{j}+\ ^{\shortmid }\mathbf{g}%
^{ab}(x,p)\ ^{\shortmid }\mathbf{e}_{a}\otimes \ ^{\shortmid }\mathbf{e}_{b}.
\label{dmct}
\end{eqnarray}%
and, for shell decompositions, as so-called, s-metrics,
\begin{eqnarray}
\mathbf{g} &=&\ _{s}\mathbf{g=g}_{\alpha _{s}\beta _{s}}(\ _{s}u)\mathbf{%
\mathbf{e}}^{\alpha _{s}}\mathbf{\otimes \mathbf{e}}^{\beta
_{s}}=g_{i_{s}j_{s}}(\ _{s}x)e^{i_{s}}\otimes e^{j_{s}}+\mathbf{g}%
_{a_{s}b_{s}}(\ _{s}x,\ _{s}v)\mathbf{e}^{a_{s}}\otimes \mathbf{e}^{a_{s}}%
\mbox{ and/or }  \label{dmts} \\
\ ^{\shortmid }\mathbf{g} &=&\ _{s}^{\shortmid }\mathbf{g=}\ ^{\shortmid }%
\mathbf{g}_{\alpha _{s}\beta _{s}}(\ _{s}^{\shortmid }u)\ ^{\shortmid }%
\mathbf{\mathbf{e}}^{\alpha _{s}}\mathbf{\otimes \ ^{\shortmid }\mathbf{e}}%
^{\beta _{s}}=\ ^{\shortmid }g_{ij}(\ _{s}^{\shortmid }x)e^{i_{s}}\otimes
e^{j_{s}}+\ ^{\shortmid }\mathbf{g}^{a_{s}b_{s}}(\ _{s}^{\shortmid }x,\
_{s}p)\ ^{\shortmid }\mathbf{e}_{a_{s}}\otimes \ ^{\shortmid }\mathbf{e}%
_{b_{s}}.  \label{dmcts}
\end{eqnarray}
\end{assumption}

The metrics (\ref{dmt})-(\ref{dmcts}) define generalizations for phase space
with curved typical fibers and generalize respectively the classes of
metrics contained in nonlinear quadratic elements of type (\ref{nqe}), or (%
\ref{nqed}).

For geometric constructions on $TV$ and their dual, we can consider general
vierbein transforms $e_{\alpha }=e_{\ \alpha }^{\underline{\alpha }%
}(u)\partial /\partial u^{\underline{\alpha }}$ and $e^{\beta }=e_{\
\underline{\beta }}^{\beta }(u)du^{\underline{\beta }}$, where the local
coordinate indices are underlined in order to distinguish them from
arbitrary abstract ones. In such formulas, we can consider nonsingular frame
transforms when a matrix $e_{\ \underline{\beta }}^{\beta }$ is inverse to $%
e_{\ \alpha }^{\underline{\alpha }}$. For a total space nonsingular metric
structure, we can define orthonormalized bases and dual bases. Similarly, we
can perform various frame transforms $\ ^{\shortmid }e_{\alpha }=\
^{\shortmid }e_{\ \alpha }^{\underline{\alpha }}(\ ^{\shortmid }u)\partial
/\partial \ ^{\shortmid }u^{\underline{\alpha }}$ and $\ ^{\shortmid
}e^{\beta }=\ ^{\shortmid }e_{\ \underline{\beta }}^{\beta }(\ ^{\shortmid
}u)d\ ^{\shortmid }u^{\underline{\beta }}$ for geometric objects on $T^{\ast
}V$ and their dual. It should be noted that there are not used boldface
symbols for such transforms because an arbitrary decomposition (for
instance, one can be considered as particular cases certain coordinate
dyadic 2+2+2+2 splitting) is not adapted to a N--connection structure. Frame
transforms can be also constructed/considered in such adapted forms when
certain N-connection and other type symmetry structures are encoded. For
instance, we can always prescribe a regular Lagrange or Hamilton generating
function on a Lorentz manifold $V$ and define canonical nonholonomic
variables (with "tilde" on geometric objects) as it is briefly reviewed in
Appendix \ref{appendixb} and references therein. We can consider metric and
N-connection structures in general forms (without "tilde") on $TV$ and/or $%
T^{\ast }V$. It is possible always to introduce necessary type atlases of
covering carts on such (co) tangent bundles and construct frame/coordinate
transforms expressing the coefficients of geometric objects (with respect to
general coordinate frame) into respective sets of coefficients with "tilde"
(with respected to N-adapted frames and d-metric structure induced by
"mechanical" generating functions). In a similar form, any geometric data
with "tilde" can be transformed into general ones (or with nonholonomic
shell dyadic structure) considering respective general frame (or frames for
a dyadic splitting). This way we prove

\begin{theorem}
\textsf{[equivalent representations of d- and s-metrics as canonical
d-metrics for Lagrange-Hamilton spaces]} Using frame/coordinate transforms,
we can establish such equivalence of respective geometric data%
\begin{eqnarray}
(\mathbf{N};\mathbf{e}_{\alpha },\mathbf{e}^{\alpha };\mathbf{g}_{\alpha
\beta })\mbox{ and/or }(\ _{s}\mathbf{N};\mathbf{e}_{\alpha _{s}},\mathbf{e}%
^{\alpha _{s}};\mathbf{g}_{\alpha _{s}\beta _{s}}) &\longleftrightarrow &(%
\widetilde{L},\ \ \widetilde{\mathbf{N}};\widetilde{\mathbf{e}}_{\alpha },%
\widetilde{\mathbf{e}}^{\alpha };\widetilde{g}_{jk},\widetilde{g}^{jk})
\label{eqgd} \\
(\ ^{\shortmid }\mathbf{N};\ ^{\shortmid }\mathbf{e}_{\alpha },\ ^{\shortmid
}\mathbf{e}^{\alpha };\ \ ^{\shortmid }\mathbf{g}_{\alpha \beta })%
\mbox{
and/or }(\ \ _{s}^{\shortmid }\mathbf{N};\ \ ^{\shortmid }\mathbf{e}_{\alpha
_{s}},\ \ ^{\shortmid }\mathbf{e}^{\alpha _{s}};\ \ \ ^{\shortmid }\mathbf{g}%
_{\alpha _{s}\beta _{s}}) &\longleftrightarrow &(\widetilde{H},\ ^{\shortmid
}\widetilde{\mathbf{N}};\ ^{\shortmid }\widetilde{\mathbf{e}}_{\alpha },\
^{\shortmid }\widetilde{\mathbf{e}}^{\alpha };\ \ ^{\shortmid }\widetilde{g}%
^{ab},\ \ ^{\shortmid }\widetilde{g}_{ab}),  \notag
\end{eqnarray}
where tilde values are defined by canonical geometric objects on
Lagrange-Hamilton spaces, see Appendix \ref{appendixb} and canonical
d-metrics $\widetilde{\mathbf{g}}_{\alpha \beta }$ (\ref{cdms}) and $\
^{\shortmid }\widetilde{\mathbf{g}}_{\alpha \beta }$ (\ref{cdmds}). We can
establish $L$--dual variables for data (\ref{eqgd}).
\end{theorem}

It is convenient to work with canonical d-metric structures $\widetilde{%
\mathbf{g}}$ and $\ ^{\shortmid }\widetilde{\mathbf{g}}$ if we are going to
provide a geometric mechanic interpretation (in particular, in Finsler like
variables) for propagating of probing particles in phase spaces and when the
fundamental geometric objects are explicitly determined by a MDR (\ref{mdrg}%
). Such data are also important for defining almost symplectic variables and
performing different models, for instance, of deformation quantization, see
\cite{v18a} and references therein. In another turn, s-metrics are necessary
for constructing exact solutions in MGTs with LIVs.

\subsection{D- and s-connections and their curvatures}

In this subsection, we analyze which classes of linear connections and
respective covariant derivative operators can be constructed in N-adapted
forms. Such geometric objects will be applied for formulating classical and
quantum gravity and matter field theories with MDRs, and which allow
decoupling and integrability in nonholonomic dyadic form of generalized
Einstein equations.

\subsubsection{Distinguished connections and N-adapted distortions}

Let us denote by $D$ be a linear connection on $\mathbf{TV.}$ We can define
a linear connection $\ ^{\shortmid }D$ on $\mathbf{T}^{\ast }\mathbf{V}$ as
follows: $\ ^{\shortmid }D_{\ ^{\shortmid }\mathbf{X}}\ ^{\shortmid }\mathbf{%
Y}:=(D_{\mathbf{X}}\mathbf{Y})^{\ast }=\ ^{\shortmid }(D_{\mathbf{X}}\mathbf{%
Y}),$ where $\ ^{\shortmid }\mathbf{X}$ and $\ ^{\shortmid }\mathbf{Y}$ are
d-vector fields on $\mathbf{T}^{\ast }\mathbf{V.}$ Inversely, considering a
linear connection $\ \ ^{\shortmid }D$ on $\mathbf{T}^{\ast }\mathbf{V},$ it
is possible to construct a linear connection $\ ^{\circ }D$ on $\mathbf{TV},$
when $\ ^{\circ }D_{\mathbf{X}}\mathbf{Y}:=(\ ^{\shortmid }D_{\ ^{\shortmid }%
\mathbf{X}}\ ^{\shortmid }\mathbf{Y})^{\circ },$ for any d-vector fields $%
\mathbf{X}$ and $\mathbf{Y}$ on $\mathbf{TV.}$ Such $D$ and $\ ^{\shortmid
}D $ can be defined if respective (co) tangent Lorentz bundles are enabled
with shell by shell N-connection structure (\ref{ncon2}) and/or canonical
Lagrange-Hamilton variables and $\mathcal{L}$ --duality stated by Legendre (%
\ref{legendre}) and inverse Legendre (\ref{invlegendre}) transforms. All
such geometric constructions can be performed/ proved in a rigorous
differential geometric form using the so-called "pullback" bundle formalism
correspondingly adapted to Lagrange-Hamilton and/or nonholonomic dyadic
structures. In this work, we follow the axiomatic approach and "shell by
shell" extension geometric generating solutions techniques elaborated in
\cite{v18a,bvepjc18,gvvepjc14,bubuianucqg17,vacaruepjc17,vijtp10,bvap19}.
This will be enough for decoupling the modified Einstein equations with MDRs
and constructing exact solutions in sections \ref{s4} and \ref{s5}.

On (co) tangent bundles, we can elaborate different geometries of affine
(linear) connections and respective covariant derivatives. Such
constructions can be performed in different forms which are adapted, or not
adapted, to certain prescribed N--connection and/or dyadic structures.

\begin{definition}
\textsf{[d-connections and s-connections as N-adapted linear connections] }
A distinguished connection (d--connection) is a linear connection $\mathbf{D}
$ on $\mathbf{TV}$ (or $\ ^{\shortmid }\mathbf{D}$ on $\mathbf{T}^{\ast }%
\mathbf{V})$ which preserves under parallelism a respective N--connection
splitting (\ref{ncon}). Such a d-connection is a shell type one, i.e. a
s-connection, $\ _{s}\mathbf{D}$ (or $\ _{s}^{\shortmid }\mathbf{D),}$ if it
preserve under parallelism a respective shell N--connection structure (\ref%
{ncon2}).
\end{definition}

The coefficients of d--connections $\mathbf{D}$ and $\ ^{\shortmid }\mathbf{D%
}$ can be defined in corresponding N-adapted forms with respect to
N--adapted frames Lemma (\ref{lemadap}), $\ \mathbf{D}_{\mathbf{e}_{\beta }}%
\mathbf{e}_{\gamma }:=\mathbf{\Gamma }_{\ \beta \gamma }^{\alpha }\mathbf{e}%
_{\alpha }$ and $\mathbf{\ ^{\shortmid }D}_{\mathbf{\ ^{\shortmid }e}_{\beta
}}\ \mathbf{^{\shortmid }e}_{\gamma }:=\mathbf{\ ^{\shortmid }\Gamma }_{\
\beta \gamma }^{\alpha }\mathbf{\ ^{\shortmid }e}_{\alpha },$ where (for a
h-v splitting)
\begin{equation*}
\mathbf{D}_{\mathbf{e}_{k}}\mathbf{e}_{j}:=L_{\ jk}^{i}\mathbf{e}_{i},%
\mathbf{D}_{\mathbf{e}_{k}}e_{b}:=\acute{L}_{\ bk}^{a}e_{a},\mathbf{D}%
_{e_{c}}\mathbf{e}_{j}:=\acute{C}_{\ jc}^{i}\mathbf{e}_{i},\mathbf{D}%
_{e_{c}}e_{b}:=C_{\ bc}^{a}e_{a}
\end{equation*}%
and (for a h-cv splitting)
\begin{equation*}
\ \ ^{\shortmid }\mathbf{D}_{\ ^{\shortmid }\mathbf{e}_{k}}\ ^{\shortmid }%
\mathbf{e}_{j}:=\ ^{\shortmid }L_{\ jk}^{i}\ ^{\shortmid }\mathbf{e}_{i},\
^{\shortmid }\mathbf{D}_{\mathbf{e}_{k}}\ ^{\shortmid }e^{b}:=-\ ^{\shortmid
}\acute{L}_{a\ k}^{\ b}\ ^{\shortmid }e^{a},\ ^{\shortmid }\mathbf{D}_{\
^{\shortmid }e^{c}}\ ^{\shortmid }\mathbf{e}_{j}:=\ ^{\shortmid }\acute{C}%
_{\ j}^{i\ c}\ ^{\shortmid }\mathbf{e}_{i},\ ^{\shortmid }\mathbf{D}_{\
^{\shortmid }e^{c}}\ ^{\shortmid }e^{b}:=-\ ^{\shortmid }C_{a}^{\ bc}\
^{\shortmid }e^{a}.
\end{equation*}%
Summarizing above formulas, we conclude that the N-adapted coefficients of
d-connections on (co) tangent Lorentz bundles are respectively parameterized
\begin{eqnarray*}
\mathbf{D} &=&\left( \mathbf{\ }_{h}\mathbf{D,\ }_{v}\mathbf{D}\right)
\mathbf{=\{\Gamma }_{\ \beta \gamma }^{\alpha }=(L_{\ jk}^{i},\acute{L}_{\
bk}^{a},\acute{C}_{\ jc}^{i},C_{\ bc}^{a})\}\mbox{ and } \\
\ ^{\shortmid }\mathbf{D} &=&\left( \mathbf{\ }_{h}^{\shortmid }\mathbf{D,\ }%
_{v}^{\shortmid }\mathbf{D}\right) \mathbf{=\{}\ ^{\shortmid }\mathbf{\Gamma
}_{\ \beta \gamma }^{\alpha }=(\ ^{\shortmid }L_{\ jk}^{i},\ ^{\shortmid }%
\acute{L}_{a\ k}^{\ b},\ ^{\shortmid }\acute{C}_{\ j}^{i\ c},\ ^{\shortmid
}C_{a}^{\ bc})\},
\end{eqnarray*}%
where $\ _{h}\mathbf{D}=(L_{\ jk}^{i},\acute{L}_{\ bk}^{a}),\ _{v}\mathbf{D}%
=(\acute{C}_{\ jc}^{i},C_{\ bc}^{a})$ and $\ _{h}^{\shortmid }\mathbf{D}=(\
^{\shortmid }L_{\ jk}^{i},\ ^{\shortmid }\acute{L}_{a\ k}^{\ b}),$ $\
_{v}^{\shortmid }\mathbf{D}=(\ ^{\shortmid }\acute{C}_{\ j}^{i\ c},\
^{\shortmid }C_{a}^{\ bc}).$

In similar forms, we can define d-connections and distortion d-tensors and
compute their N-adapted coefficients for every shell decomposition of
tangent bundles. Here we note that the nonholonomic structures and
coefficient formulas of geometric objects for respective 4+4 and dyadic
decompositions (the last variant results in shell by shell splitting of
nonlinear systems of PDEs, see sections \ref{s4} and \ref{s5}) are
different. For instance, we provide explicit formulas for a dyadic
decomposition on $\mathbf{T}^{\ast }\mathbf{V}$ defined by a $\
_{s}^{\shortmid }\mathbf{N}$ (\ref{ncon2}). The coefficients of a
d--connection on a s-shell, $\ _{s}^{\shortmid }\mathbf{D=\{\ ^{\shortmid}%
\mathbf{D}_{\ ^{\shortmid }\mathbf{e}_{\beta _{s}}}=\ ^{\shortmid }\mathbf{D}%
}_{\beta _{s}}\}$ can be defined and computed N-adapted forms with respect
to $\ ^{\shortmid }\mathbf{e}_{\alpha _{s}}$ and $\ ^{\shortmid }\mathbf{e}%
^{\alpha _{s}},$ see (\ref{nadapbds}),
\begin{eqnarray*}
\ ^{\shortmid }\mathbf{D}_{\ ^{\shortmid }\mathbf{e}_{\beta _{s}}}\
^{\shortmid }\mathbf{e}_{\gamma _{s}}:= \ ^{\shortmid }\mathbf{\Gamma }_{\
\beta _{s}\gamma _{s}}^{\alpha _{s}}\mathbf{\ ^{\shortmid }e}_{\alpha
_{s}},\ \mbox{ where } &&\ ^{\shortmid }\mathbf{D}_{\ ^{\shortmid }\mathbf{e}%
_{k_{s}}}\ ^{\shortmid }\mathbf{e}_{j_{s}}:=\ ^{\shortmid }L_{\
j_{s}k_{s}}^{i_{s}}\ ^{\shortmid }\mathbf{e}_{i_{s}},\ ^{\shortmid }\mathbf{D%
}_{\mathbf{e}_{k_{s}}}\ ^{\shortmid }e^{b_{s}}:=-\ ^{\shortmid }\acute{L}%
_{a_{s}\ k_{s}}^{\ b_{s}}\ ^{\shortmid }e^{a_{s}}, \\
&&\ ^{\shortmid }\mathbf{D}_{\ ^{\shortmid }e^{c_{s}}}\ ^{\shortmid }\mathbf{%
e}_{j_{s}}:=\ ^{\shortmid }\acute{C}_{\ j_{s}}^{i_{s}\ c_{s}}\ ^{\shortmid }%
\mathbf{e}_{i_{s}},\ ^{\shortmid }\mathbf{D}_{\ ^{\shortmid }e^{c_{s}}}\
^{\shortmid }e^{b_{s}}:=-\ ^{\shortmid }C_{a_{s}}^{\ b_{s}c_{s}}\
^{\shortmid }e^{a_{s}}.
\end{eqnarray*}%
The N-adapted coefficients of d-connections on the 3d shell of cotangent
Lorentz bundles can be parameterized in the form
\begin{equation*}
\ _{s}^{\shortmid }\mathbf{D}=\left( \ _{^{s}h}^{\shortmid }\mathbf{D,\ }_{\
^{s}v}^{\shortmid }\mathbf{D}\right) =\{\ ^{\shortmid }\mathbf{\Gamma }_{\
\beta _{s}\gamma _{s}}^{\alpha _{s}}=(\ ^{\shortmid }L_{\
j_{s}k_{s}}^{i_{s}},\ ^{\shortmid }\acute{L}_{a_{s}\ k_{s}}^{\ b_{s}},\
^{\shortmid }\acute{C}_{\ j_{s}}^{i_{s}\ c_{s}},\ ^{\shortmid }C_{a_{s}}^{\
b_{s}c_{s}})\},
\end{equation*}
where $\ _{^{s}h}^{\shortmid }\mathbf{D}=(\ ^{\shortmid }L_{\
j_{s}k_{s}}^{i_{s}},\ ^{\shortmid }\acute{L}_{a_{s}\ k_{s}}^{\ b_{s}})$ and $%
\ _{\ ^{s}v}^{\shortmid }\mathbf{D}=(\ ^{\shortmid }\acute{C}_{\
j_{s}}^{i_{s}\ c_{s}},\ ^{\shortmid }C_{a_{s}}^{\ b_{s}c_{s}}).$

Similar N-adapted formulas can be derived for $\ _{s}\mathbf{D=\{\mathbf{\ D}%
_{\mathbf{\ e}_{\beta _{s}}}=\mathbf{\ D}}_{\beta _{s}}\mathbf{\}=\{\mathbf{%
\Gamma }_{\ \beta _{s}\gamma _{s}}^{\alpha _{s}}\}.}$

\begin{lemma}
\label{lemmadist} \textsf{[distortion of linear connection structures ]} %
\label{ldist}Let us consider on $\mathbf{TV}$ a linear connection $%
\underline{D}$ (which is not obligatory a d-connection) and a d-connection $%
\mathbf{D}$. Such values on $\mathbf{T}^{\ast }\mathbf{V,}$ are respectively
denoted $\ ^{\shortmid }\underline{D}$ and $\ ^{\shortmid }\mathbf{D}$ and
can be related by corresponding distortion d-tensors%
\begin{eqnarray}
\mathbf{Z} &=& \mathbf{D}-\underline{D}\mbox{ and/or } \ ^{\shortmid }%
\mathbf{Z}:=\ ^{\shortmid }\mathbf{D}-\ ^{\shortmid }\underline{D};
\label{dist} \\
\ _{s}\mathbf{Z} &:=&\ _{s}\mathbf{D}-\underline{D}\mbox{ and/or }\
_{s}^{\shortmid }\mathbf{Z}:=\ _{s}^{\shortmid }\mathbf{D}-\ ^{\shortmid }%
\underline{D}.  \label{sdist}
\end{eqnarray}
\end{lemma}

\begin{proof}
For simplicity, let us sketch the proofs for tangent and cotangent bundles.
Fixing respective N-adapted frames, such distortion d-tensors{\small
\begin{eqnarray*}
\mathbf{Z}_{\ \beta \gamma }^{\alpha } &=&\mathbf{\Gamma }_{\ \beta \gamma
}^{\alpha }-\underline{\Gamma }_{\ \beta \gamma }^{\alpha }=\{Z_{\
jk}^{i}=L_{\ jk}^{i}-\underline{L}_{\ jk}^{i},\acute{Z}_{\ bk}^{a}=\acute{L}%
_{\ bk}^{a}-\underline{\acute{L}}_{\ bk}^{a},\acute{Z}_{\ jc}^{i}=\acute{C}%
_{\ jc}^{i}-\underline{\acute{C}}_{\ jc}^{i},Z_{\ bc}^{a}=C_{\ bc}^{a}-%
\underline{C}_{\ bc}^{a}\}\mbox{ and } \\
\ ^{\shortmid }\mathbf{Z}_{\ \beta \gamma }^{\alpha } &=&\ ^{\shortmid }%
\mathbf{\Gamma }_{\ \beta \gamma }^{\alpha }-\ ^{\shortmid }\underline{%
\mathbf{\Gamma }}_{\ \beta \gamma }^{\alpha } \\
&=&\{\ ^{\shortmid }Z_{\ jk}^{i}=\ ^{\shortmid }L_{\ jk}^{i}-\ ^{\shortmid }%
\underline{L}_{\ jk}^{i},\ ^{\shortmid }\acute{Z}_{a\ k}^{\ b}=\ ^{\shortmid
}\acute{L}_{a\ k}^{\ b}-\ ^{\shortmid }\underline{\acute{L}}_{a\ k}^{\ b},\
^{\shortmid }\acute{Z}_{\ j}^{i\ c}=\ ^{\shortmid }\acute{C}_{\ j}^{i\ c}-\
^{\shortmid }\underline{\acute{C}}_{\ j}^{i\ c},\ ^{\shortmid }Z_{a}^{\
bc}=\ ^{\shortmid }C_{a}^{\ bc}-\ ^{\shortmid }\underline{C}_{a}^{\ bc}\},
\end{eqnarray*}%
} can be constructed in explicit form by considering corresponding
differences of N-adapted coefficients for linear and d--connections. Similar
formulas can be introduced for the $s$-shell distortion of linear
connections,
\begin{equation*}
\ _{s}^{\shortmid }\mathbf{Z}:=\ _{s}^{\shortmid }\mathbf{D}-\
_{s}^{\shortmid }\underline{D}\ \mbox{ for }\ ^{\shortmid }\mathbf{Z}_{\
\beta _{s}\gamma _{s}}^{\alpha _{s}}=\ ^{\shortmid }\mathbf{\Gamma }_{\
\beta _{s}\gamma _{s}}^{\alpha _{s}}-\ ^{\shortmid }\underline{\mathbf{%
\Gamma }}_{\ \beta _{s}\gamma _{s}}^{\alpha _{s}},
\end{equation*}%
where we omit formulas for h,v and cv-coefficients, which are similar to (%
\ref{dist}) and (\ref{sdist}). $\square $
\end{proof}

\vskip5pt

Let us provide certain physical motivations for above Definition, Lemma and
related formulas. It is not possible to decouple and integrate in general
form the (modified) Einstein equations in MGTs and GR if we work only with
the Levi-Civita or Lagrange-Hamilton type linear connections. General
decoupling properties can be proven for certain classes of more special
d-connections which are adapted to shell dyadic decompositions (see sections %
\ref{s4} and \ref{s5}; and, for other type higher dimension gravity
theories, \cite{gvvepjc14,bubuianucqg17,vacaruepjc17,vijtp10,bvap19}). MGTs
with MDRs are formulated, in general, in not  s-adapted dyadic variables
but using arbitrary frame, metric and (non) linear connection structures.
Lemma \ref{lemmadist} states that we can work equivalently with different
classes of connections determined by the same metric structure but adapted
to different nonholonomic frame structures. Such linear connections are
related by respective d-tensors. For instance, a class of d-connections is
convenient for decoupling physically important systems of nonlinear PDEs and
another one for formulating analogous Lagrange--Hamilton gravity theories
and/or almost symplectic and respective classical, thermodynamic, or quantum
models. The main conclusion of this subsection is that we can use any
s-connection adapted to a nonholonomic dyadic structure when general
off-diagonal solutions can be found in explicit form. Then the constructions
can be generalized in covariant form (for corresponding distortion
d-tensors) for other types of linear connection structures.

\subsubsection{Distinguished and shell curvature, torsion and nonmetricity
tensors}

\begin{definition}
\textbf{-Theorem}\footnote{%
In mathematical physics, there are used terms like Definition-Theorem / -
Lemma / - Corollary etc. for such definitions (new ideas, concepts, or
conventions) which are motivated by certain explicit geometric constructions
and/or requesting formulation of some theorems and respective mathematical
proofs.}\ \textsf{[curvature, torsion and nonmetricity of d- and
s-connections ] } Any d--connection $\mathbf{D,}$ or $\ ^{\shortmid }\mathbf{%
D,}$ is characterized by respective curvature $(\mathcal{R},$ or $\
^{\shortmid }\mathcal{R}),$ torsion $(\mathcal{T},$ or $\ ^{\shortmid }%
\mathcal{T}),$ and nonmetricity, $(\mathcal{Q},$ or $\ ^{\shortmid }\mathcal{%
Q})$, d-tensors defined and computed in standard forms:
\begin{eqnarray}
\mathcal{R}(\mathbf{X,Y})&:= &\mathbf{D}_{\mathbf{X}}\mathbf{D}_{\mathbf{Y}}-%
\mathbf{D}_{\mathbf{Y}}\mathbf{D}_{\mathbf{X}}-\mathbf{D}_{\mathbf{[X,Y]}},
\label{dcurvabstr} \\
\mathcal{T}(\mathbf{X,Y})&:=&\mathbf{D}_{\mathbf{X}}\mathbf{Y}-\mathbf{D}_{%
\mathbf{Y}}\mathbf{X}-[\mathbf{X,Y}]\mbox{ and }\mathcal{Q}(\mathbf{X}):=%
\mathbf{D}_{\mathbf{X}}\mathbf{g},  \notag \\
\mbox{ or }\ ^{\shortmid }\mathcal{R}(\ ^{\shortmid }\mathbf{X,\ ^{\shortmid
}Y})&:= &\ ^{\shortmid }\mathbf{D}_{\ ^{\shortmid }\mathbf{X}}\ ^{\shortmid }%
\mathbf{D}_{\ ^{\shortmid }\mathbf{Y}}-\ ^{\shortmid }\mathbf{D}_{\
^{\shortmid }\mathbf{Y}}\ ^{\shortmid }\mathbf{D}_{\ ^{\shortmid }\mathbf{X}%
}-\ ^{\shortmid }\mathbf{D}_{\mathbf{[\ ^{\shortmid }X,\ ^{\shortmid }Y]}},
\notag \\
\ ^{\shortmid }\mathcal{T}(\ ^{\shortmid }\mathbf{X,\ ^{\shortmid }Y})&:= &\
^{\shortmid }\mathbf{D}_{\ ^{\shortmid }\mathbf{X}}\ ^{\shortmid }\mathbf{Y}%
-\ ^{\shortmid }\mathbf{D}_{\ ^{\shortmid }\mathbf{Y}}\ ^{\shortmid }\mathbf{%
X}-[\ ^{\shortmid }\mathbf{X,\ ^{\shortmid }Y}]\mbox{ and }\ ^{\shortmid }%
\mathcal{Q}(\ ^{\shortmid }\mathbf{X}):=\ ^{\shortmid }\mathbf{D}_{\
^{\shortmid }\mathbf{X}}\ ^{\shortmid }\mathbf{g}.  \notag
\end{eqnarray}%
For s--connections $\ _{s}\mathbf{D}$ and/or $\ _{s}^{\shortmid }\mathbf{D,}$
we can consider similar values defined by following formulas curvature, $\
^{\shortmid }\mathcal{R}$ and/or $\ _{s}^{\shortmid }\mathcal{R},$ torsion, $%
\ ^{\shortmid }\mathcal{T}$ and/or$\ _{s}^{\shortmid }\mathcal{T},$ and
nonmetricity, $\ ^{\shortmid }\mathcal{Q},$ or $\ _{s}^{\shortmid }\mathcal{Q%
},$ d-tensors defined following such formulas
\begin{eqnarray}
\ _{s}\mathcal{R}(\ _{s}\mathbf{X,}\ _{s}\mathbf{Y})&:= &\ _{s}\mathbf{D}_{\
_{s}\mathbf{X}}\mathbf{D}_{\ _{s}\mathbf{Y}}-\ _{s}\mathbf{D}_{\ _{s}\mathbf{%
Y}}\ _{s}\mathbf{D}_{\ _{s}\mathbf{X}}-\ _{s}\mathbf{D}_{\mathbf{[\ _{s}X,\
_{s}Y]}},  \label{scurvabstr} \\
\ _{s}\mathcal{T}(\ _{s}\mathbf{X,\ _{s}Y})&:= &\ _{s}\mathbf{D}\ _{s\mathbf{%
X}}\ _{s}\mathbf{Y}-\ _{s}\mathbf{D}_{\ _{s}\mathbf{Y}}\ _{s}\mathbf{X}-[\
_{s}\mathbf{X,\ _{s}Y}]\mbox{ and }\ _{s}\mathcal{Q}(\ _{s}\mathbf{X}):=\
_{s}\mathbf{D}_{\ _{s}\mathbf{X}}\ _{s}\mathbf{g},  \notag \\
\mbox{ or }\ _{s}^{\shortmid }\mathcal{R}(\ _{s}^{\shortmid }\mathbf{X,\
_{s}^{\shortmid }Y})&:= &\ _{s}^{\shortmid }\mathbf{D}_{\ _{s}^{\shortmid }%
\mathbf{X}}\ _{s}^{\shortmid }\mathbf{D}_{\ _{s}^{\shortmid }\mathbf{Y}}-\
_{s}^{\shortmid }\mathbf{D}_{\ _{s}^{\shortmid }\mathbf{Y}}\ _{s}^{\shortmid
}\mathbf{D}_{\ _{s}^{\shortmid }\mathbf{X}}-\ _{s}^{\shortmid }\mathbf{D}_{%
\mathbf{[}\ _{s}^{\shortmid }\mathbf{X,}\ _{s}^{\shortmid }\mathbf{Y]}},
\notag \\
\ _{s}^{\shortmid }\mathcal{T}(\ _{s}^{\shortmid }\mathbf{X,}\
_{s}^{\shortmid }\mathbf{Y})&:= &\ _{s}^{\shortmid }\mathbf{D}_{\
_{s}^{\shortmid }\mathbf{X}}\ _{s}^{\shortmid }\mathbf{Y}-\ _{s}^{\shortmid }%
\mathbf{D}_{\ _{s}^{\shortmid }\mathbf{Y}}\ _{s}^{\shortmid }\mathbf{X}-[\
_{s}^{\shortmid }\mathbf{X,}\ _{s}^{\shortmid }\mathbf{Y}]\mbox{ and }\
_{s}^{\shortmid }\mathcal{Q}(\ _{s}^{\shortmid }\mathbf{X}):=\
_{s}^{\shortmid }\mathbf{D}_{\ _{s}^{\shortmid }\mathbf{X}}\ _{s}^{\shortmid
}\mathbf{g}.  \notag
\end{eqnarray}
\end{definition}

We put labels "s" in geometric formulas (\ref{scurvabstr}) in order to
emphasize that such abstract formulas encode geometrically a nonholonomic
dyadic structure. This means that on any shall $s=1,2,3,4$ the respective
shell curvature, torsion and nonmetricity d-tensors are defined and computed
following (\ref{dcurvabstr}) but using extensions of respective d-metrics
and d-connections on previous shells. This prescription performed in
explicit coefficient form with respect nonholonomic dyadic frames will allow
to prove an important decoupling property of modified Einstein equations
with MDRs, see bellow the Theorem \ref{theordecoupl} (and to prove, in
general, the main results in sections \ref{s4} and \ref{s5}). We are not
able to find such properties if we work with general d-connections for 4+4
splitting as in abstract formulas (\ref{dcurvabstr}) and/or their general/
particular coefficient forms. So, it is important to distinguish if certain
geometric/physical objects encode (or there are adapted) to certain dyadic
shell structure even the formulas are introduced in abstract form. In our
works, we use left labels "s".

The N--adapted coefficients for the curvature, torsion and nonmetricity
d-tensors (\ref{dcurvabstr}) and cotangent part of (\ref{scurvabstr}) are
provided in Corollary \ref{acorolcurv} in Appendix \ref{appendixa}.
Hereafter, we shall provide explicit \ geometric or coordinate formulas only
for cotangent bundle objects with a shell decomposition of equations if that
will not result in ambiguities. This is motivated also for $\mathbf{TV}$ and
$\mathbf{T}^{\ast }\mathbf{V}$ with established $L$-duality structure. In
such cases, for instance, we can transform formulas for $\ _{s}^{\shortmid }%
\mathcal{R}(\ _{s}^{\shortmid }\mathbf{X,\ _{s}^{\shortmid }Y})$ into those
for $\ _{s}\mathcal{R}(\ _{s}\mathbf{X,}\ _{s}\mathbf{Y})$ by omitting
labels "$\ ^{\shortmid }$". We eliminate shell splitting and get respective
formulas (\ref{dcurvabstr}) if omit $\ "_{s}".$ Details on coefficients
formulas for $\mathbf{TV,}$ $\mathbf{T}^{\ast }\mathbf{V}$ and dyadic
structures on$\ _{s}\mathbf{TV}$ can be found in \cite{v18a,gvvepjc14}. In
this article, we shall emphasize new geometric constructions for $\mathbf{T}%
^{\ast }\mathbf{V}$ and $_{s}\mathbf{T}^{\ast }\mathbf{V.}$

\subsubsection{The Ricci and Einstein d--tensors}

The Ricci tensor for a d--connection and/or a s-connection on a (co) tangent
bundle can be constructed in standard form by contracting, for instance, the
first and forth indices of respective curvature d-tensors $\mathcal{R}$ and$%
\ ^{\shortmid }\mathcal{R}$ in (\ref{dcurvabstr}) and, for instance, in $\
_{s}^{\shortmid }\mathcal{R}$ (\ref{scurvabstr}).

\begin{definition}
\textbf{-Theorem\ } \textsf{[Ricci tensors for d-- and s-connections ] } The
Ricci d--tensors are defined and computed as $Ric=\{\mathbf{R}_{\alpha \beta
}:=\mathbf{R}_{\ \alpha \beta \tau }^{\tau }\},$ for a d-connection $\mathbf{%
D}$, and $\ ^{\shortmid }Ric=\{\ ^{\shortmid }\mathbf{R}_{\alpha \beta }:=\
^{\shortmid }\mathbf{R}_{\ \alpha \beta \tau }^{\tau }\},$ for a
d-connection $\ ^{\shortmid }\mathbf{D.}$ For a shell decomposition, we can
define (for instance, on $_{s}\mathbf{T}^{\ast }\mathbf{V)}$ $\
_{s}^{\shortmid }Ric=\{\ ^{\shortmid }\mathbf{R}_{\alpha _{s}\beta _{s}}:=\
^{\shortmid }\mathbf{R}_{\ \alpha _{s}\beta _{s}\tau _{s}}^{\tau _{s}}\}.$
\end{definition}

In N-adapted form using respective formulas (\ref{dcurv}), one proves

\begin{corollary}
\textsf{[computation of Ricci d-tensors] } The N-adapted coefficients of the
Ricci d--tensors of a d-connection in a (co) tangent Lorentz bundle are
parameterized in $h$- and/or $v$-, or $cv$-form, and $h_{s}$-$cv_{s}$-forms
by formulas
\begin{eqnarray}
\mathbf{R}_{\alpha \beta } &=&\{R_{hj}:=R_{\ hji}^{i},\ \ R_{ja}:=-P_{\
jia}^{i},\ R_{bk}:=P_{\ bka}^{a},R_{\ bc}=S_{\ bca}^{a}\},\mbox{and/ or }
\label{dricci} \\
\ ^{\shortmid }\mathbf{R}_{\alpha \beta } &=&\{\ ^{\shortmid }R_{hj}:=\
^{\shortmid }R_{\ hji}^{i},\ ^{\shortmid }R_{j}^{\ a}:=-\ ^{\shortmid }P_{\
ji}^{i\ \ \ a},\ \ ^{\shortmid }R_{\ k}^{b}:=\ ^{\shortmid }P_{a\ k}^{\ b\ \
a},\ ^{\shortmid }R_{\ }^{bc}=\ ^{\shortmid }S_{a\ }^{\ bca}\},
\label{driccid} \\
\ ^{\shortmid }\mathbf{R}_{\alpha _{s}\beta _{s}} &=&\{\ ^{\shortmid
}R_{h_{s}j_{s}}:=\ ^{\shortmid }R_{\ h_{s}j_{s}i_{s}}^{i_{s}},\ ^{\shortmid
}R_{j_{s}}^{\ a_{s}}:=-\ ^{\shortmid }P_{\ j_{s}i_{s}}^{i_{s}\ \ \ a_{s}},\
\ ^{\shortmid }R_{\ k_{s}}^{b_{s}}:=\ ^{\shortmid }P_{a_{s}\ k_{s}}^{\
b_{s}\ \ a_{s}},\ ^{\shortmid }R_{\ }^{b_{s}c_{s}}=\ ^{\shortmid }S_{a_{s}\
}^{\ b_{s}c_{s}a_{s}}\},  \label{driccisd}
\end{eqnarray}
\end{corollary}

We note that in above Definition-Theorem for the Ricci d-tensors we use
total space indices $\alpha, \beta ...$ but the formulas in this Corollary
split into irreducible N-adapted coefficients with base manifold indices, $%
i,j,...$; (co) fiber indices, $a,b, ...$; or respective shell by shell
indices (with label s). This can be proved by explicit computations of the
coefficients of corresponding Ricci d-tensors with respect to N-adapted
bases, for instance, formula (\ref{dricci}). For tangent bundles,
Lagrange-Finsler spaces, higher dimension nonholonomic manifolds etc., such
cumbersome formulas can be found in \cite%
{gvvepjc14,bubuianucqg17,vacaruepjc17,vijtp10,bvap19} and references
therein. The formulas (\ref{driccid}) and (\ref{driccisd}) are introduced in
this work following general and abstract index definitions and similarities
for cotangent Lorentz bundles and their nonholonomic dyadic splitting. We
omit cumbersome coefficient formulas and their proofs which are similar
(they are dual but with possible Legendre symmetries of Lagrange-Hamilton
type) to those considered in our previous/ cited works.

If a (co) tangent bundle is enabled both with a d-connection, $\mathbf{D}$
(or$\ ^{\shortmid }\mathbf{D),}$ and d-metric, $\mathbf{g}$ (\ref{dmt}) (or $%
\ ^{\shortmid }\mathbf{g}$ (\ref{dmct})), we can introduce nonholonomic
Ricci scalars:

\begin{definition}
\textbf{-Theorem\ } \textsf{[scalar curvature of d-connections and
s-connections] } The scalar curvature of a d-connection $\mathbf{D,}$ or $\
^{\shortmid }\mathbf{D,}$ or (for instance) $\ _{s}^{\shortmid }\mathbf{D,} $
can be defined and computed respectively for the inverse d-metric $\mathbf{g}%
^{\alpha \beta },$ or $\ ^{\shortmid }\mathbf{g}^{\alpha \beta },$ or $\
^{\shortmid }\mathbf{g}^{\alpha _{s}\beta _{s}},$
\begin{eqnarray*}
\ _{d}R &:&=\mathbf{g}^{\alpha \beta }\mathbf{R}_{\alpha \beta
}=g^{ij}R_{ij}+g^{ab}R_{ab}=R+S;\ \mbox{ or }\ _{d}^{\shortmid }R:=\
^{\shortmid }\mathbf{g}^{\alpha \beta }\ ^{\shortmid }\mathbf{R}_{\alpha
\beta }=\ ^{\shortmid }g^{ij}\ ^{\shortmid }R_{ij}+\ ^{\shortmid }g^{ab}\
^{\shortmid }R_{ab}=\ ^{\shortmid }R+\ ^{\shortmid }S; \\
\mbox{ or }\ _{s}^{\shortmid }\overline{R} &:&=\ ^{\shortmid }\mathbf{g}%
^{\alpha _{s}\beta _{s}}\ ^{\shortmid }\mathbf{R}_{\alpha _{s}\beta _{s}}=\
^{\shortmid }g^{i_{s}j_{s}}\ ^{\shortmid }R_{i_{s}j_{s}}+\ ^{\shortmid
}g^{a_{s}b_{s}}\ ^{\shortmid }R_{a_{s}b_{s}}=\ _{1}^{\shortmid }R+\
_{2}^{\shortmid }S+\ _{3}^{\shortmid }S+\ _{4}^{\shortmid }S,
\end{eqnarray*}%
with respective h-- and v--components, and s-components
\begin{eqnarray}
R &=&g^{ij}R_{ij},S=g^{ab}S_{ab},\mbox{ or }\ ^{\shortmid }R=\ ^{\shortmid
}g^{ij}\ ^{\shortmid }R_{ij},\ ^{\shortmid }S=\ ^{\shortmid }g_{ab}\
^{\shortmid }S^{ab},\mbox{ or }  \notag \\
\ _{1}^{\shortmid }R &=&\ ^{\shortmid }g^{i_{1}j_{1}}\ ^{\shortmid
}R_{i_{1}j_{1}},\ _{2}^{\shortmid }S=\ ^{\shortmid }g^{a_{2}b_{2}}\
^{\shortmid }R_{a_{2}b_{2}},\ _{3}^{\shortmid }S=\ ^{\shortmid
}g_{a_{3}b_{3}}\ ^{\shortmid }S^{a_{3}b_{3}},\ _{4}^{\shortmid }S=\
^{\shortmid }g_{a_{4}b_{4}}\ ^{\shortmid }S^{a_{4}b_{4}}.  \label{riccidscal}
\end{eqnarray}
\end{definition}

We note that in above formulas the indices on the 2d shell are labeled as in
tangent bundles of total dimension 4 even the constructions are for a 8-d
total dimension of $_{s}\mathbf{T}^{\ast }\mathbf{V.}$ The indices on shells
$s=3$ and 4 are dual ones as for 2-d typical cofibers. On $\mathbf{TV}$
and/or $\mathbf{T}^{\ast }\mathbf{V,}$ we can consider canonical d-metrics $%
\widetilde{\mathbf{g}}$ (\ref{cdms}) and/or $\ ^{\shortmid }\widetilde{%
\mathbf{g}}$ (\ref{cdmds}) determined directly by an indicator of MDRs.

As dual shall extensions of constructions from \cite{v18a,gvvepjc14} (and
using above Definitions-Theorems and Corollaries), we formulate

\begin{definition}
\textbf{-Theorem\ }\label{dteinstdt} \textsf{[the Einstein tensors for
d-connections and s-connection] } By construction, the respective Einstein
d-tensors on $\mathbf{TV,}$ $\mathbf{T}^{\ast }\mathbf{V,}$\ $_{s}\mathbf{T}%
^{\ast }\mathbf{V,}$ are defined:
\begin{equation*}
En=\{\mathbf{E}_{\alpha \beta }:=\mathbf{R}_{\alpha \beta }-\frac{1}{2}%
\mathbf{g}_{\alpha \beta }\ _{d}R\},\ ^{\shortmid }En=\{\ ^{\shortmid }%
\mathbf{E}_{\alpha \beta }:=\ ^{\shortmid }\mathbf{R}_{\alpha \beta }-\frac{1%
}{2}\ ^{\shortmid }\mathbf{g}_{\alpha \beta }\ _{d}^{\shortmid }R\},\
_{s}^{\shortmid }En=\{\ ^{\shortmid }\mathbf{E}_{\alpha _{s}\beta _{s}}:=\
^{\shortmid }\mathbf{R}_{\alpha _{s}\beta _{s}}-\frac{1}{2}\ ^{\shortmid }%
\mathbf{g}_{\alpha _{s}\beta _{s}}\ _{s}^{\shortmid }\overline{R}\}.
\end{equation*}
\end{definition}

\begin{proof}
Such proofs follow from explicit constructions on regions of some atlases
covering respectively $\mathbf{TV}$ and/or $\mathbf{T}^{\ast }\mathbf{V,}$
and \ $_{s}\mathbf{T}^{\ast }\mathbf{V,}$ using N-adapted coefficients (\ref%
{dcurvabstr}) and cotangent part of (\ref{scurvabstr}); and (\ref{dricci})
and/or (\ref{driccid}), and (\ref{driccisd}). $\square $\vskip5pt
\end{proof}

\subsubsection{Physically important d- and s-connections and their
distortions}

Such linear connections are similar to those in Finsler-Lagrange-Hamilton
geometry \cite{v18a,gvvepjc14,vmon06} with "tilde" variables encoding
directly the contributions of certain MDRs indicators and which can be used
for elaborating respective geometric mechanics models, or as almost K\"{a}%
hler geometries which can be used deformation quantization of such theories.
The canonical d--connections and s-connections, in their turn, are
convenient for constructing exact solutions in various MGTs and nonholonomic
geometric evolution models.

\begin{definition}
\textbf{-Theorem\ } \label{phidc} \textsf{[physically important
d-connections and s-connections for decoupling of (generalized) Einstein
equations] } There are such metric compatible d- and s-connections on (co)
tangent Lorentz bundles (which encode in direct or indirect forms MDRs (\ref%
{mdrg}), with and a possible $\mathcal{L}$--duality) being characterized
respectively by such geometric conditions:
\begin{eqnarray}
\lbrack \mathbf{g,N]} &\mathbf{\simeq }&\mathbf{[}\widetilde{\mathbf{g}},%
\widetilde{\mathbf{N}}]\mathbf{\simeq \lbrack }\ _{s}\mathbf{g,}\ _{s}%
\mathbf{N]}  \label{canondcl} \\
&\Longrightarrow &\left\{
\begin{array}{ccccc}
\nabla : &  & \nabla \mathbf{g}=0;\ \mathbf{T[\nabla ]}=0, &  & %
\mbox{Lagrange LC--connection}; \\
\widehat{\mathbf{D}}: &  & \widehat{\mathbf{D}}\ \mathbf{g}=0;\ hv\widehat{%
\mathbf{T}}=0, &  & \mbox{canonical Lagrange
d-connection}; \\
\widetilde{\mathbf{D}}: &  & \widetilde{\mathbf{D}}\widetilde{g}=0,hv%
\widetilde{\mathbf{T}}=0 &  & \mbox{Cartan-Lagrange d-connection}; \\
\ _{s}\widehat{\mathbf{D}}: &  & \ _{s}\widehat{\mathbf{D}}\ \ _{s}\mathbf{g}%
=0;\ h_{s}v_{s^{\prime }}\widehat{\mathbf{T}}=0,s\neq s^{\prime } &  & %
\mbox{ canonical dyadic s-connection},%
\end{array}%
\right.  \notag
\end{eqnarray}%
and/or
\begin{eqnarray}
\lbrack \ ^{\shortmid }\mathbf{g,\ ^{\shortmid }N]} &\mathbf{\simeq }&%
\mathbf{[}\ ^{\shortmid }\widetilde{\mathbf{g}},\ ^{\shortmid }\widetilde{%
\mathbf{N}}]\mathbf{\simeq \mathbf{[}}\ _{s}^{\shortmid }\mathbf{\mathbf{g,}}%
\ _{s}^{\shortmid }\mathbf{\mathbf{N]}}  \label{canondch} \\
&\Longrightarrow &\left\{
\begin{array}{ccccc}
\ ^{\shortmid }\nabla : &  & \ ^{\shortmid }\nabla \ ^{\shortmid }\mathbf{g}%
=0;\ \ ^{\shortmid }\mathbf{T[\ ^{\shortmid }\nabla ]}=0, &  & %
\mbox{Hamilton LC-connection}; \\
\ ^{\shortmid }\widehat{\mathbf{D}}: &  & \ ^{\shortmid }\widehat{\mathbf{D}}%
\ \ ^{\shortmid }\mathbf{g}=0;\ hcv\ ^{\shortmid }\widehat{\mathbf{T}}=0, &
& \mbox{canonical Lagrange
d-connection}; \\
\ ^{\shortmid }\widetilde{\mathbf{D}}: &  & \ ^{\shortmid }\widetilde{%
\mathbf{D}}\ ^{\shortmid }\widetilde{g}=0,hcv\ ^{\shortmid }\widetilde{%
\mathbf{T}}=0, &  & \mbox{Cartan-Hamilton d-connection}; \\
\ _{s}^{\shortmid }\widehat{\mathbf{D}}: &  & \ _{s}^{\shortmid }\widehat{%
\mathbf{D}}\ \ _{s}^{\shortmid }\mathbf{g}=0;\ h_{s}cv_{s^{\prime }}\
^{\shortmid }\widehat{\mathbf{T}}=0,s\neq s^{\prime } &  &
\mbox{ dual canonical dyadic s-connection}.%
\end{array}%
\right.  \notag
\end{eqnarray}
\end{definition}

\begin{proof}
It is similar to that sketched for the Definition 2.10 and Theorem in ref. \cite%
{v18a} being omitted the constructions for almost symplectic structures but
extended "shell by shell" for dyadic decompositions and respective d- and
s-metrics/-connections. $\square $\vskip5pt
\end{proof}

MGTs with MDRs on (co) tangent bundles are characterized by multi-connection
structures. Each such linear connection, in principle, can be derived by a
metric structure following certain geometrically/ physically motivated
principles or technical goals. For instance, to construct exact solutions
following the AFDM is most convenient to work with the canonical
d-connections $\widehat{\mathbf{D}}$ and $^{\shortmid }\widehat{\mathbf{D}}$
but with additional dyadic shell by shell splitting. Following the
conditions of Lemma \ref{ldist}, we prove

\begin{theorem}
\textsf{[existence of physically important distortions of connections for
respective MGTs with MDRs] } \label{thdistr} \newline
There are unique distortions relations
\begin{eqnarray}
\widehat{\mathbf{D}} &=&\nabla +\widehat{\mathbf{Z}},\widetilde{\mathbf{D}}%
=\nabla +\widetilde{\mathbf{Z}},\mbox{ and }\widehat{\mathbf{D}}=\widetilde{%
\mathbf{D}}+\mathbf{Z,}\mbox{  determined by }(\mathbf{g\simeq }\widetilde{%
\mathbf{g}}\mathbf{,N\mathbf{\simeq }\widetilde{\mathbf{N}})};
\label{candistr} \\
\ ^{\shortmid }\widehat{\mathbf{D}} &=&\ ^{\shortmid }\nabla +\ ^{\shortmid }%
\widehat{\mathbf{Z}},\ ^{\shortmid }\widetilde{\mathbf{D}}=\ ^{\shortmid
}\nabla +\ ^{\shortmid }\widetilde{\mathbf{Z}},\mbox{ and }\ ^{\shortmid }%
\widehat{\mathbf{D}}=\ ^{\shortmid }\widetilde{\mathbf{D}}+\ ^{\shortmid }%
\mathbf{Z,}\mbox{ determined by }(\ ^{\shortmid }\mathbf{g\simeq }\
^{\shortmid }\widetilde{\mathbf{g}}\mathbf{,\ ^{\shortmid }N\mathbf{\simeq }%
\ ^{\shortmid }\widetilde{\mathbf{N}})};  \notag \\
\ _{s}^{\shortmid }\widehat{\mathbf{D}} &=&\ ^{\shortmid }\nabla +\
_{s}^{\shortmid }\widehat{\mathbf{Z}},\ \mbox{ and }\ _{s}^{\shortmid }%
\widehat{\mathbf{D}}=\ _{s}^{\shortmid }\widehat{\mathbf{D}}+\
_{s}^{\shortmid }\widehat{\mathbf{Z}}\mathbf{,}\mbox{ determined by }(\
^{\shortmid }\mathbf{g\simeq }\ ^{\shortmid }\widetilde{\mathbf{g}}\mathbf{%
\simeq \ _{s}^{\shortmid }\mathbf{g},\ ^{\shortmid }N\mathbf{\simeq }\
^{\shortmid }\widetilde{\mathbf{N}}\mathbf{\simeq }\ _{s}^{\shortmid }N)};
\notag
\end{eqnarray}%
for distortion d-tensors $\widehat{\mathbf{Z}},\widetilde{\mathbf{Z}},$ and $%
\mathbf{Z,}$ on $T\mathbf{TV;}$ dual distortion d-tensors $\ ^{\shortmid }%
\widehat{\mathbf{Z}},\ ^{\shortmid }\widetilde{\mathbf{Z}},$ and $\
^{\shortmid }\mathbf{Z,}$ on $T\mathbf{T}^{\ast }\mathbf{V;}$ and dual
distortion s-tensor $\ _{s}^{\shortmid }\widehat{\mathbf{Z}},\
_{s}^{\shortmid }\widetilde{\mathbf{Z}},$ and $\ _{s}^{\shortmid }\mathbf{Z,}
$ on $\ _{s}T\mathbf{T}^{\ast }\mathbf{V.}$
\end{theorem}

We emphasize that above Theorem (and all further results involving this
theorem and related distortion formulas) contains nontrivial statements
because there are involved canonical d-connections and their distortions
canonically determined by nontrivial MDRs (\ref{mdrg}) and respective
canonical d-metric and N-connection structures (all labeled by symbols with
"tilde"). This proves that MDRs can be characterized both geometrically and
physically by certain unique distortion d-tensors (from respective Levi -
Civita connections, which is the standard one in GR). The distortion
formulas (\ref{candistr}) allow us to work with different types of
d-connections and linear connections which are convenient for constructing
exact solutions or elaborating analogous mechanical models. Here we note
that the d--tensor $\widehat{\mathbf{Z}}$ in above formulas is an algebraic
combination of coefficients $\widehat{\mathbf{T}}_{\ \alpha \beta }^{\gamma
}[\mathbf{g,N}]$ computed by introducing formulas (\ref{canlc}) into (\ref%
{dtors}). Similar nontrivial torsion components can be computed for
cotangent bundles and respective dyadic decompositions. We formulate this

\begin{consequence}
MDRs (\ref{mdrg}) are characterized by respective canonical d-tensors $%
\widehat{\mathbf{Z}}[\widetilde{\mathbf{g}},\widetilde{\mathbf{N}}],%
\widetilde{\mathbf{Z}}[\widetilde{\mathbf{g}},\widetilde{\mathbf{N}}],$ and $%
\mathbf{Z}[\widetilde{\mathbf{g}},\widetilde{\mathbf{N}}],$ for Lagrange
models; $\ ^{\shortmid }\widehat{\mathbf{Z}}[\ ^{\shortmid }\widetilde{%
\mathbf{g}},\ ^{\shortmid }\widetilde{\mathbf{N}}],\ ^{\shortmid }\widetilde{%
\mathbf{Z}}[\ ^{\shortmid }\widetilde{\mathbf{g}},\ ^{\shortmid }\widetilde{%
\mathbf{N}}],$ and $\ ^{\shortmid }\mathbf{Z}[\ ^{\shortmid }\widetilde{%
\mathbf{g}},\ ^{\shortmid }\widetilde{\mathbf{N}}],$ for Hamilton models, as
in Appendix \ref{appendixb}.
\end{consequence}

In result, we can draw this

\begin{conclusion}
\textsf{[equivalent geometric and physically important data for modeling
phase spaces with MDRs] } \newline
The phase space geometry can be described in equivalent forms (up to
respective nonholonomic deformations of the linear connection structures and
nonholonomic frame transforms) by such data{\small
\begin{equation}
\begin{array}{ccccc}
\mbox{MDRs} & \nearrow & (\mathbf{g,N,}\widehat{\mathbf{D}})\leftrightarrows
(L:\widetilde{\mathbf{g}}\mathbf{,}\widetilde{\mathbf{N}},\widetilde{\mathbf{%
D}}) & \leftrightarrow (\ _{s}\mathbf{g,\ _{s}N,}\ _{s}\widehat{\mathbf{D}})
& \leftrightarrow \lbrack (\mathbf{g[}N],\nabla )],\mbox{ on }\mathbf{TV} \\
\mbox{indicator }\varpi &  & \updownarrow \mbox{ possible }\mathcal{L}%
\mbox{-duality }\& & \mbox{dyadic decomposition} & \updownarrow
\mbox{ not
N-adapted } \\
\mbox{ see (\ref{mdrg})} & \searrow & (\ ^{\shortmid }\mathbf{g,\
^{\shortmid }N,}\ ^{\shortmid }\widehat{\mathbf{D}})\leftrightarrows (H:\
^{\shortmid }\widetilde{\mathbf{g}},\ ^{\shortmid }\widetilde{\mathbf{N}},\
^{\shortmid }\widetilde{\mathbf{D}}) & \leftrightarrow (\ _{s}^{\shortmid }%
\mathbf{g,\ _{s}^{\shortmid }N,}\ _{s}^{\shortmid }\widehat{\mathbf{D}}) &
\leftrightarrow \lbrack (\ ^{\shortmid }\mathbf{g}[\ ^{\shortmid }N],\
^{\shortmid }\nabla )],\mbox{on}\mathbf{T}^{\ast }\mathbf{V}.%
\end{array}
\label{phspgd}
\end{equation}%
}
\end{conclusion}

This Conclusion complete for dyadic splitting the Conclusion 2.4 in \cite%
{v18a}. Following the conventions of that work, we say that certain
geometric constructions are canonical (i.e. formulated in canonical
nonholonomic variables) if they are performed for "hat", or "tilde",
d-connections and related geometric objects with possible respective dyadic
decompositions uniquely derived for certain Lagrange-Hamilton fundamental
generating functions (in particular, for a Finsler metric $F$).

We can always to impose certain (in general, nonholonomic) constraints of
type $\widehat{\mathbf{Z}}=0,$ when $\widehat{\mathbf{D}}_{|\widehat{\mathbf{%
Z}}=0}\simeq \nabla $, i.e. both connections are given by same coefficients
in some adapted frames, even $\widehat{\mathbf{D}}\neq \nabla .$ Similar
LC-conditions can imposed on cotangent bundles and for dyadic splitting.
This is possible because the frame/coordinate transformation laws, and
possible dyadic splitting of nonlinear and distinguished / linear
connections are different from that of tensors. For instance, we can chose a
frame structure when different connections may be determined by the same set
of coefficients with respect to such a special frame and by different sets
of adapted coefficients in other systems of reference. Imposing such
conditions, we can always extract LC-configurations from more (general)
nonholonomic metric-affine and/or dyadic structures.

\begin{corollary}
\textsf{[extracting LC-configurations by additional (non) holonomic
constraints and dyadic conditions]} \newline
There are extracted LC-configurations from $\widehat{\mathbf{D}},$ $\
^{\shortmid }\widehat{\mathbf{D}},$ or $\ _{s}^{\shortmid }\widehat{\mathbf{D%
}},$\ for respective zero distortions ($\widehat{\mathbf{Z}},$ $\
^{\shortmid }\widehat{\mathbf{Z}},$ or $\ _{s}^{\shortmid }\widehat{\mathbf{Z%
}})$ if there are imposed zero torsion conditions, correspondingly, for $%
\widehat{\mathcal{T}}$ $=\{\widehat{\mathbf{T}}_{\ \alpha \beta }^{\gamma
}\}=0,$ $\ ^{\shortmid }\widehat{\mathcal{T}}$ $=\{\ ^{\shortmid }\widehat{%
\mathbf{T}}_{\ \alpha \beta }^{\gamma }\}=0,$ or $\ _{s}^{\shortmid }%
\widehat{\mathcal{T}}$ $=\{\ ^{\shortmid }\widehat{\mathbf{T}}_{\ \alpha
_{s}\beta _{s}}^{\gamma _{s}}\}=0,$ see $\ $(\ref{dtors}). Such conditions
are satisfied if
\begin{eqnarray}
\widehat{C}_{jb}^{i} &=&0,\Omega _{\ ji}^{a}=0\mbox{ and }\widehat{L}%
_{aj}^{c}=e_{a}(N_{j}^{c});\ ^{\shortmid }\widehat{C}_{j}^{i\ b}=0,\
^{\shortmid }\Omega _{\ aji}=0\mbox{ and }\ ^{\shortmid }\widehat{L}_{c\
j}^{\ a}=\ ^{\shortmid }e^{a}(\ ^{\shortmid }N_{cj});  \label{lccondh} \\
\ ^{\shortmid }\widehat{C}_{j_{s}}^{i_{s}\ b_{s}} &=&0,\ ^{\shortmid }\Omega
_{\ a_{s}j_{s}i_{s}}=0\mbox{ and }\ ^{\shortmid }\widehat{L}_{c_{s}\
j_{s}}^{\ a_{s}}=\ ^{\shortmid }e^{a_{s}}(\ ^{\shortmid }N_{c_{s}j_{s}}).
\label{lccondsd}
\end{eqnarray}
\end{corollary}

\begin{proof}
The proofs should be considered for respective (co) tangent bundles and
shell by shell. For simplicity, we provide a sketch of such a proof on $%
\mathbf{TV.}$ Introducing (\ref{canondcl}) in (\ref{dtors}), we obtain zero
values for
\begin{equation*}
\widehat{T}_{\ jk}^{i}=\widehat{L}_{jk}^{i}-\widehat{L}_{kj}^{i},\widehat{T}%
_{\ ja}^{i}=\widehat{C}_{jb}^{i},\widehat{T}_{\ ji}^{a}=-\Omega _{\
ji}^{a},\ \widehat{T}_{aj}^{c}=\widehat{L}_{aj}^{c}-e_{a}(N_{j}^{c}),%
\widehat{T}_{\ bc}^{a}=\ \widehat{C}_{bc}^{a}-\ \widehat{C}_{cb}^{a},
\end{equation*}%
if the conditions (\ref{lccondh}) are satisfied. $\square $ \vskip5pt
\end{proof}

The equations (\ref{lccondh})-(\ref{lccondsd}) can be solved in explicit
form for manifolds/ bundle spaces of total dimensions 4-10, see references
in \cite{v18a} and next Section.

Introducing distortions from Theorem \ref{thdistr} into formulas (\ref%
{dcurvabstr}), we can prove in abstract and N-adapted forms:

\begin{theorem}
\textsf{[existence of canonical distortions of Riemannian and Ricci d-
and/or s-tensors determined by MDRs] } \label{thcandist}There are canonical
distortion relations encoding MDRs for respective canonical Lagrange-Finsler
and/or dyadic nonholonomic variables:
\begin{eqnarray*}
\mbox{For the curvature d-tensors,} &&\widehat{\mathcal{R}}[\mathbf{g},%
\widehat{\mathbf{D}} =\nabla +\widehat{\mathbf{Z}}]=\mathcal{R}[\mathbf{g}%
,\nabla ]+\widehat{\mathcal{Z}}[\mathbf{g},\widehat{\mathbf{Z}}], \\
&&\ ^{\shortmid }\widehat{\mathcal{R}}[\ ^{\shortmid }\mathbf{g},\
^{\shortmid }\widehat{\mathbf{D}}=\ ^{\shortmid }\nabla +\ ^{\shortmid }%
\widehat{\mathbf{Z}}]=\ ^{\shortmid }\mathcal{R}[\ ^{\shortmid }\mathbf{g},\
^{\shortmid }\nabla ]+\ ^{\shortmid }\widehat{\mathcal{Z}}[\ ^{\shortmid }%
\mathbf{g},\ ^{\shortmid }\widehat{\mathbf{Z}}], \\
&&\ _{s}^{\shortmid }\widehat{\mathcal{R}}[\ _{s}^{\shortmid }\mathbf{g},\
_{s}^{\shortmid }\widehat{\mathbf{D}} =\ ^{\shortmid }\nabla +\
_{s}^{\shortmid }\widehat{\mathbf{Z}}]=\ _{s}^{\shortmid }\mathcal{R}[\
_{s}^{\shortmid }\mathbf{g},\ ^{\shortmid }\nabla ]+\ _{s}^{\shortmid }%
\widehat{\mathcal{Z}}[\ _{s}^{\shortmid }\mathbf{g},\ _{s}^{\shortmid }%
\widehat{\mathbf{Z}}],
\end{eqnarray*}%
with respective distortion d-tensors $\ \widehat{\mathcal{Z}},$ on $\mathbf{%
TV;}$ $\ ^{\shortmid }\widehat{\mathcal{Z}},$ on $\mathbf{T}^{\ast }\mathbf{V%
};\ _{s}^{\shortmid }\widehat{\mathcal{Z}},$ on $\mathbf{T}^{\ast }\mathbf{V}%
;$
\begin{eqnarray*}
\mbox{For the Ricci d-tensors,} &&\widehat{R}ic[\mathbf{g},\widehat{\mathbf{D%
}} =\nabla +\widehat{\mathbf{Z}}]=Ric[\mathbf{g},\nabla ]+\widehat{Z}ic[%
\mathbf{g},\widehat{\mathbf{Z}}], \\
&&\ ^{\shortmid }\widehat{R}ic[\ ^{\shortmid }\mathbf{g},\ ^{\shortmid }%
\widehat{\mathbf{D}}=\ ^{\shortmid }\nabla +\ ^{\shortmid }\widehat{\mathbf{Z%
}}]=\ ^{\shortmid }Ric[\ ^{\shortmid }\mathbf{g},\ ^{\shortmid }\nabla ]+\
^{\shortmid }\widehat{Z}ic[\ ^{\shortmid }\mathbf{g},\ ^{\shortmid }\widehat{%
\mathbf{Z}}], \\
&&\ _{s}^{\shortmid }\widehat{R}ic[\ _{s}^{\shortmid }\mathbf{g},\
_{s}^{\shortmid }\widehat{\mathbf{D}} =\ ^{\shortmid }\nabla +\
_{s}^{\shortmid }\widehat{\mathbf{Z}}]=\ _{s}^{\shortmid }Ric[\
_{s}^{\shortmid }\mathbf{g},\ ^{\shortmid }\nabla ]+\ _{s}^{\shortmid }%
\widehat{Z}ic[\ _{s}^{\shortmid }\mathbf{g},\ _{s}^{\shortmid }\widehat{%
\mathbf{Z}}],
\end{eqnarray*}%
with respective distortion d-tensors $\ \widehat{Z}ic,$ on $\mathbf{TV;}$ $\
\ ^{\shortmid }\widehat{Z}ic,$ on $\mathbf{T}^{\ast }\mathbf{V};\
_{s}^{\shortmid }\widehat{Z}ic,$ on $\mathbf{T}^{\ast }\mathbf{V};$

For the scalar curvature of canonical d-connection $\widehat{\mathbf{D}},\
^{\shortmid }\widehat{\mathbf{D}},\ _{s}^{\shortmid }\widehat{\mathbf{D}},$%
\begin{eqnarray*}
&&\ _{d}^{\shortmid }\widehat{R}[\mathbf{g},\widehat{\mathbf{D}} =\nabla +%
\widehat{\mathbf{Z}}]=\mathcal{R}[\mathbf{g},\nabla ]+\ _{d}\widehat{Z}[%
\mathbf{g},\widehat{\mathbf{Z}}], \\
&&\ _{d}^{\shortmid }\widehat{R}[\ ^{\shortmid }\mathbf{g},\ ^{\shortmid }%
\widehat{\mathbf{D}}=\ ^{\shortmid }\nabla +\ ^{\shortmid }\widehat{\mathbf{Z%
}}]=\ _{d}^{\shortmid }R[\ ^{\shortmid }\mathbf{g},\ ^{\shortmid }\nabla ]+\
_{d}^{\shortmid }\widehat{Z}[\ ^{\shortmid }\mathbf{g},\ ^{\shortmid }%
\widehat{\mathbf{Z}}], \\
&&\ _{d}^{\shortmid }\widehat{R}[\ _{s}^{\shortmid }\mathbf{g},\
_{s}^{\shortmid }\widehat{\mathbf{D}} =\ ^{\shortmid }\nabla +\
_{s}^{\shortmid }\widehat{\mathbf{Z}}]=\ _{d}^{\shortmid }R[\ ^{\shortmid }%
\mathbf{g},\ ^{\shortmid }\nabla ]+\ _{d}^{\shortmid }\widehat{Z}[\
_{s}^{\shortmid }\mathbf{g},\ _{s}^{\shortmid }\widehat{\mathbf{Z}}],
\end{eqnarray*}%
with respective distortion scalar functionals $\ \ _{d}\widehat{Z},$ on $%
\mathbf{TV,}$ and $\ _{d}^{\shortmid }\widehat{Z},$ on $\mathbf{T}^{\ast }%
\mathbf{V.}$
\end{theorem}

\begin{proof}
In \cite{v18a} (for Theorem 2.6), we explain and give references for
detailed proofs on $\mathbf{TV}$ with possible nonholonomic dyadic structure
and extra dimensions. We note that for constructions on $\mathbf{T}^{\ast }%
\mathbf{V}$ the formulas can be derived following our abstract geometric
approach with label "$\ ^{\shortmid }$". In this work, we shall construct
exact solutions for MGTs encoding MDRs working with the canonical
d-connection $\ ^{\shortmid }\widehat{\mathbf{D}}$ and/or $\ _{s}^{\shortmid
}\widehat{\mathbf{D}}$ which can be restricted to LC-configurations by
solving, respectively, the equations (\ref{lccondh}) and/or (\ref{lccondsd}).
\end{proof}

$\square $

\section{Modified Einstein-Hamilton Equations in N- and S-adapted Variables}

\label{s3} The goal of this section is to consider explicit examples of
generalized Lagrange densities for gravitational and matter fields on (co)
tangent Lorentz bundles with possible dyadic splitting. We formulate the
gravitational and matter field equations generalizing the Einstein equations
for nontrivial MDRs and LIVs. Such systems of nonlinear PDEs can be derived
both in abstract and/or N-adapted (in general, coordinate free) forms for
various types MGTs, see details in \cite{v18a,gvvepjc14,vmon06}. For
simplicity, we restrict our constructions only to modifications of GR on $%
\mathbf{T}^{\ast }\mathbf{V}$ for $\ _{s}^{\shortmid }\widehat{\mathbf{D}}$
which is necessary for proofs of general decoupling and integrability of
locally anisotropic gravitational field equations in Einstein-Hamilton MGTs.

\subsection{Lagrange densities and energy-momentum d- and s-tensors on (co)
tangent bundles}

\label{sslagrd}In this subsection, we analyse explicit examples of phase
space generalized Lagrange densities and derive respective energy-momentum
tensors. We consider arbitrary metric compatible d-connections $\mathbf{D},\
^{\shortmid }\mathbf{D}$ or $\ _{s}^{\shortmid }\mathbf{D.}$

\subsubsection{Scalar fields on (co) tangent Lorentz bundles and d-
/s-metrics}

Let us speculate on Lagrange densities for matter fields with locally
anisotropic interactions:

\begin{convention}
\label{convscfields}\textsf{[scalar fields and dyadic splitting on (co)
tangent bundles] } \newline
Scalar field locally anisotropic phase interactions can be modeled
respectively by Lagrange densities
\begin{equation}
\ ^{m}\mathcal{L}=\ ^{\phi }\mathcal{L}(\mathbf{g;}\phi )\mbox{ on }T\mathbf{%
V;}\ _{\shortmid }^{m}\mathcal{L}=\ _{\shortmid }^{\phi }\mathcal{L}(\
^{\shortmid }\mathbf{g};\ _{\shortmid }\phi )\mbox{ on }T^{\ast }\mathbf{V;}%
\ _{\shortmid s}^{m}\mathcal{L}=\ _{\shortmid s}^{\phi }\mathcal{L}(\
_{s}^{\shortmid }\mathbf{g};\ _{\shortmid }^{s}\phi )\mbox{ on }\
_{s}T^{\ast }\mathbf{V,}  \label{lagscf}
\end{equation}%
depending, for simplicity, only on respective d-/s-metrics ($\mathbf{g}_{\mu
\nu },$ $^{\shortmid }\mathbf{g}^{\alpha \beta },$ $_{s}^{\shortmid }\mathbf{%
g}^{\alpha \beta }$ - this allows to construct exact solutions in explicit
form in various gravity theories) on scalar fields $\ \phi =\phi (u),$ $\
_{\shortmid }\phi =\ _{\shortmid }\phi (\ ^{\shortmid }u),\ _{\shortmid
}^{s}\phi =\ _{\shortmid }^{s}\phi (\ _{s}^{\shortmid }u).$
\end{convention}

Performing a N-adapted variational calculus, we prove

\begin{consequence}
\textsf{[energy-momentum d-tensors for locally anisotropic interacting
scalar fields ] } \newline
The symmetric energy-momentum d-tensors for scalar fields on (co) tangent
bundles derived for respective Lagrange densities (\ref{lagscf}) are
computed for possible h- and v-, or cv-splitting. For instance, on cotangent
Lorentz bundle with nonholonomic dyadic decomposition%
\begin{eqnarray}
\ _{\shortmid }^{\phi }\mathbf{T}_{\alpha _{s}\beta _{s}} &=&-\frac{2}{\sqrt{%
|\ _{s}^{\shortmid }\mathbf{g}|}}\frac{\delta (\sqrt{|\ _{s}^{\shortmid }%
\mathbf{g}|}\ \ \ _{\shortmid s}^{\phi }\mathcal{L})}{\delta \ ^{\shortmid }%
\mathbf{g}^{\alpha _{s}\beta _{s}}}=\ \ _{\shortmid s}^{\phi }\mathcal{L}\
_{s}^{\shortmid }\mathbf{g}_{\alpha _{s}\beta _{s}}+2\frac{\delta (\ \ \
_{\shortmid s}^{\phi }\mathcal{L})}{\delta \ ^{\shortmid }\mathbf{g}^{\alpha
_{s}\beta _{s}}}  \label{emscdt} \\
&=&\{\ _{\shortmid }^{\phi }\mathbf{T}_{i_{s}j_{s}}=-\frac{2}{\sqrt{|\
_{s}^{\shortmid }\mathbf{g}|}}\frac{\delta (\sqrt{|\ _{s}^{\shortmid }%
\mathbf{g}|}\ \ \ _{\shortmid s}^{\phi }\mathcal{L})}{\delta \ ^{\shortmid }%
\mathbf{g}^{i_{s}j_{s}}}=\ldots ,\ _{\shortmid }^{\phi }\mathbf{T}%
^{a_{s}b_{s}}=-\frac{2}{\sqrt{|\ _{s}^{\shortmid }\mathbf{g}|}}\frac{\delta (%
\sqrt{|\ _{s}^{\shortmid }\mathbf{g}|}\ \ \ _{\shortmid s}^{\phi }\mathcal{L}%
)}{\delta \ ^{\shortmid }\mathbf{g}_{a_{s}b_{s}}}=\ldots \}.  \notag
\end{eqnarray}
\end{consequence}

Scalar field equations can be re-defined by an additional 3+1 splitting on
base and fiber spaces enabled with additional dyadic splitting, as certain
moving equations for ideal fluids. For instance, we can consider a velocity
d-vector $\mathbf{v}_{\alpha }$ subjected to the conditions $\mathbf{v}%
_{\alpha }\mathbf{v}^{\alpha }=1$ and $\mathbf{v}^{\alpha }\mathbf{D}_{\beta}%
\mathbf{v}_{\alpha }=0,$ for $\ ^{\phi }\mathcal{L}:=-p$ in a corresponding
local N--adapted frame, which in s-variables is extended on $\ _{s}T^{\ast }%
\mathbf{V}$.

\begin{remark}
\textsf{[energy-momentum s-tensors for locally anisotropic perfect liquids] }
\newline
Conventionally, the sources (\ref{emscdt}) can be approximated as a perfect
liquid matter (we use a left abstract index "l" for liquid) with respective
density and pressure
\begin{equation}
\ ^{l}\mathbf{T}_{\alpha _{s}\beta _{s}}=(\rho +p)\mathbf{v}_{\alpha _{s}}%
\mathbf{v}_{\beta _{s}}-p\mathbf{g}_{\alpha _{s}\beta _{s}}\mbox{ and/or }\
_{\shortmid }^{l}\mathbf{T}_{\alpha _{s}\beta _{s}}=(\ ^{\shortmid }\rho +\
^{\shortmid }p)\ ^{\shortmid }\mathbf{v}_{\alpha _{s}}\ ^{\shortmid }\mathbf{%
v}_{\beta _{s}}-\ ^{\shortmid }p\ ^{\shortmid }\mathbf{g}_{\alpha _{s}\beta
_{s}}.  \label{lemd}
\end{equation}
On Lorentz manifolds, such approximations are considered in standard
textbooks on GR and cosmology. We generalized them by analogy for locally
anisotropic higher order fluid configurations defined on (co) tangent
Lorentz bundles.
\end{remark}

We can generalize (\ref{emscdt}) for scalar fields subjected to the
condition to be solutions on certain types of equations depending in
explicit form on d-connections (Klein-Gordon, or other types, for instance,
certain hydrodynamical equations). Using respective nonholonomic variables,
we can re-write respective matter field equations and related formulas in
terms of "tilde/hat/dyadic" variables.

\subsubsection{Effective sources determined by distortions of d- and
s-connections}

It is convenient to work with the canonical d- and s-connections ($\widehat{%
\mathbf{D}},\ ^{\shortmid }\widehat{\mathbf{D}}$ and/or $\ _{s}^{\shortmid }%
\widehat{\mathbf{D}}\mathbf{)}$ which allow decoupling and integration in
very general forms of modified Einstein equations. Following the conditions
of Theorem \ref{thcandist}, we can work with distortions $\ _{s}^{\shortmid }%
\widehat{\mathbf{D}}=\nabla +\ _{s}^{\shortmid }\widehat{\mathbf{Z}}$ $=\
^{\shortmid }\widetilde{\mathbf{D}}+\ _{s}^{\shortmid }\widetilde{\mathbf{Z}}%
,$ computed for the same d- and s-metric structure $\ ^{\shortmid }\mathbf{g=%
}\ _{s}^{\shortmid }\mathbf{g}=\ ^{\shortmid }\mathbf{\tilde{g},}$ we can
express the distortions of the Ricci d-tensor $\ ^{\shortmid }\mathbf{R}%
_{\alpha _{s}\beta _{s}}$ (\ref{driccisd})
\begin{equation*}
\ ^{\shortmid }\widehat{\mathbf{R}}_{\alpha _{s}\beta _{s}}=\ ^{\shortmid
}R_{\alpha _{s}\beta _{s}}[\ _{s}^{\shortmid }\mathbf{g},\ ^{\shortmid
}\nabla ]+\ ^{\shortmid }\widehat{\mathbf{Z}}_{\alpha _{s}\beta _{s}}[\
_{s}^{\shortmid }\mathbf{g},\ ^{\shortmid }\nabla ]=\ ^{\shortmid }%
\widetilde{\mathbf{R}}_{\alpha _{s}\beta _{s}}[\ \ ^{\shortmid }\mathbf{%
\tilde{g}},\ \ ^{\shortmid }\widetilde{\mathbf{D}}]+\ ^{\shortmid }%
\widetilde{\mathbf{Z}}_{\alpha _{s}\beta _{s}}[\ \ ^{\shortmid }\mathbf{%
\tilde{g}},\ \ ^{\shortmid }\widetilde{\mathbf{D}}]
\end{equation*}%
and extract LC-configurations imposing additionally the condition (\ref%
{lccondsd}), $\ _{s}^{\shortmid }\mathbf{D}_{\mid \ _{s}^{\shortmid }%
\mathcal{T}=0}=\nabla .$

The effective energy-momenta determined by distortion s-tensors are defined
and computed
\begin{equation}
\varkappa \ _{e}^{\shortmid }\widehat{\mathbf{T}}_{\alpha _{s}\beta _{s}}=-\
^{\shortmid }\widehat{\mathbf{Z}}_{\alpha _{s}\beta _{s}}[\ _{s}^{\shortmid }%
\mathbf{g},\ ^{\shortmid }\nabla ]\mbox{ and }\varkappa \ _{e}^{\shortmid }%
\widetilde{\mathbf{T}}_{\alpha _{s}\beta _{s}}=-\ ^{\shortmid }\widetilde{%
\mathbf{Z}}_{\alpha _{s}\beta _{s}}[\ \ ^{\shortmid }\mathbf{\tilde{g}},\ \
^{\shortmid }\widetilde{\mathbf{D}}].  \label{effdtsourc}
\end{equation}%
We can compute the trace of such values, $\ _{e}^{\shortmid }\widehat{%
\mathbf{T}}$ and $\ _{e}^{\shortmid }\widetilde{\mathbf{T}}$ by contracting
respective indices.

Finally, we note that $\ _{e}^{\shortmid }\widehat{\mathbf{T}}_{\alpha
_{s}\beta _{s}}$ can be transformed into effective shell cosmological
constants $\ _{e}^{\shortmid s}\Lambda $ by redefinition of generating
functions, see next section. In another turn, we can use $\ _{e}^{\shortmid }%
\widetilde{\mathbf{T}}_{\alpha _{s}\beta _{s}}$ and respective
Lagrange-Hamilton variables if we wont to compute directly some MDR-effects
or model, for instance, polarizations of physical constants (BH horizons
etc.) in effective Lagrange-Hamilton variables.

\subsubsection{Actions for minimal nonholonomic dyadic extensions of GR}

Summarizing Lagrange densities introduced in previous subsections, we
formulate:

\begin{principle}
-\textbf{Convention} \label{pactminmodact}\textbf{\ [actions for minimally
MDR-modified dyadic systems]}:\newline
Locally anisotropic interactions of gravitational and scalar fields with
distortions on shells of (co)tangent Lorentz bundles (endowed with metric
compatible s-connections uniquely determined by respective s-metric
structures and admitting nonholonomic variables for distinguishing Hamilton
phase spaces) are described by respective actions%
\begin{equation*}
\mathbf{\ }\ _{\shortmid }^{s}\mathcal{S}=\ _{\shortmid s}^{\mathbf{g}}%
\mathcal{S}+\ _{\shortmid s}^{\phi }\mathcal{S}+\ \ _{\shortmid }^{e}%
\mathcal{S=}\frac{1}{16\pi }\int \delta \mathbf{\ ^{\shortmid }}u\sqrt{|%
\mathbf{\ _{s}^{\shortmid }g}|}(\ _{\shortmid s}^{\mathbf{g}}\mathcal{L+}\
_{\shortmid s}^{\phi }\mathcal{L}+\ \ _{\shortmid }^{e}\mathcal{L}),
\end{equation*}%
where$\ _{\ \shortmid s}^{\mathbf{g}}\mathcal{L}:=\ _{\shortmid }^{s}\mathbf{%
R}=\mathbf{\ ^{\shortmid }g}^{\alpha _{s}\beta _{s}}\mathbf{\ ^{\shortmid }R}%
_{\alpha _{s}\beta _{s}}[\mathbf{\ _{s}^{\shortmid }D}]$ and the Lagrange
densities for scalar and distortion fields are given by corresponding
formulas (\ref{lagscf}) and (\ref{effdtsourc}).
\end{principle}

Performing a N- and s-adapted variational calculus and summarizing previous
Consequences on energy-momentum d-tensors, we obtain this

\begin{consequence}
\label{conseymhs}\textsf{[sources for locally anisotropic scalar and
distortion fields] } \newline
On cotangent Lorentz bundles with dyadic decomposition, minimal
modifications of scalar and distortion systems encoding MDRs can be
characterized by respective sources
\begin{eqnarray}
\ _{\shortmid }^{\phi }\Upsilon _{\alpha _{s}\beta _{s}}&:= &\varkappa (\
_{\shortmid }^{\phi }\mathbf{T}_{\alpha _{s}\beta _{s}}-\frac{1}{2}\mathbf{\
^{\shortmid }g}_{\alpha _{s}\beta _{s}}\ _{\shortmid }^{\phi }T)\mbox{ and }%
\ \ _{\shortmid }^{e}\Upsilon _{\alpha _{s}\beta _{s}}:=\varkappa (\ \
_{\shortmid }^{e}\mathbf{T}_{\alpha _{s}\beta _{s}}-\frac{1}{2}\ ^{\shortmid
}\mathbf{g}_{\alpha _{s}\beta _{s}}\ \ \ _{\shortmid }^{e}\mathbf{T})\ ,
\notag \\
\mbox{ i.e. total source }\ _{\shortmid }\Upsilon _{\alpha _{s}\beta _{s}}
&=&\ _{\shortmid }^{\phi }\Upsilon _{\alpha _{s}\beta _{s}}+\ \ _{\shortmid
}^{e}\Upsilon _{\alpha _{s}\beta _{s}},  \label{totaldiadsourcd}
\end{eqnarray}%
where the energy-momentum tensors can be taken respectively by canonical
values (\ref{emscdt}) and (\ref{effdtsourc}).
\end{consequence}

We shall consider cosmological constants on every shall of dyadic
decomposition or on the total Lorentz cobundle (see section \ref{ssnonlsym}):

\begin{assumption}
\textsf{[effective cosmological constants on locally anisotropic phase
spaces] } \label{assumpt4} \newline
Generalized sources for matter fields (including possible effective sources
defined by distortion tensors of d-connections) can be approximated by
respective cosmological constants $\ _{\shortmid }\Lambda ,$ a general one
for the total space, or by a set of cosmological constants labelled shell by
shell, $\ _{\shortmid s}\Lambda ,$
\begin{equation*}
\ _{\shortmid }\Upsilon _{\ \gamma _{s}}^{\beta _{s}}=\ _{\shortmid }\Lambda
\delta _{\ \gamma _{s}}^{\beta _{s}}\mbox{ and/or }\ _{\shortmid }\Upsilon
_{\ \gamma _{s}}^{\beta _{s}}=\ _{\shortmid s}\Lambda \delta _{\ \gamma
_{s}}^{\beta _{s}},
\end{equation*}%
when, correspondingly, $\ $ $\ _{\shortmid }\Upsilon _{\alpha _{s}\beta
_{s}}=\ _{\shortmid }^{\phi }\Upsilon _{\alpha _{s}\beta _{s}}+\ \
_{\shortmid }^{e}\Upsilon _{\alpha _{s}\beta _{s}},$ or considering that
such sources are subjected to relations of additivity,
\begin{equation}
\ \mathbf{\ }\ _{\shortmid }\Lambda =\ _{\shortmid }^{\phi }\Lambda +\ \ \
_{\shortmid }^{e}\Lambda \mbox{ and/or }\ \ \ _{\shortmid s}\Lambda =\ \
_{\shortmid s}^{\phi }\Lambda +\ \ _{\shortmid s}^{e}\Lambda .
\label{additcconst}
\end{equation}
\end{assumption}

The cosmological constants considered above can be zero, positive, or
negative. They may compensate each other and result in (fictive) vacuum
configurations on a base space, or on some shells etc.

\subsection{Generalized Einstein equations in cotangent nonholonomic dyadic
variables}

\label{ssmeeq} We define and analyse properties of modified gravitational
field equations written in terms of different types of geometric variables
for s-metrics of type $\ \ _{s}^{\shortmid }\mathbf{g}$ (\ref{dmcts}) and
sources $\ _{\shortmid }\Upsilon _{\alpha _{s}\beta _{s}}$ (\ref%
{totaldiadsourcd}).

\subsubsection{Gravitational equations in canonical dyadic variables}

For any metric compatible geometric data $(\ _{s}\mathbf{T}^{\ast }\mathbf{%
V,\ }\ _{s}^{\shortmid }\mathbf{N,}\ \ _{s}^{\shortmid }\mathbf{g,\ }\
_{s}^{\shortmid }\mathbf{D})$ and prescribed source $\ _{\shortmid }\Upsilon
_{\alpha _{s}\beta _{s}}=\ _{\shortmid }^{\phi }\Upsilon _{\alpha _{s}\beta
_{s}}+\ _{e}^{\shortmid }\Upsilon _{\alpha _{s}\beta _{s}},$ we derive in
abstract geometric form \cite{misner73} (or following an N-adapted dyadic
calculus similar to that on Lorentz manifolds \cite{wald82}).\footnote{\label%
{fvariat} In GR, the Einstein equations with the Ricci tensor for $\nabla ,$
$R_{ij}=\Upsilon _{ij}$, can be derived by a variational calculus on a
Lorentz manifold $V$ using the action
\begin{equation*}
\mathcal{S}=\ ^{g}\mathcal{S}+\ ^{m}\mathcal{S}=\frac{1}{16\pi }\int d^{4}x%
\sqrt{|g_{ij}|}(\ ^{g}\mathcal{L+}\ ^{m}\mathcal{L}).
\end{equation*}%
The Lagrange density for gravitational fields is postulated in the form $\
^{g}\mathcal{L}(g_{ij},\nabla )=\frac{^{Pl}M^{2}}{2}R$, where $R$ is the
Ricci scalar of $\nabla .$ The Planck mass $^{Pl}M$ is determined by the
Newton constant $^{New}G$ (in this work, we can consider the units $%
^{New}G=1/16\pi $ with $^{Pl}M=(8\pi ^{New}G)^{-1/2}=\sqrt{2})$ which states
a constant $\varkappa $ for the matter source $\Upsilon _{ij}:=\varkappa
(T_{ij}-\frac{1}{2}g_{ij}T),$ where $T:=g^{ij}T_{ij},$ for $\ T_{kl}:=-\frac{%
2}{\sqrt{|\mathbf{g}_{ij}|}}\frac{\delta (\sqrt{|\mathbf{g}_{ij}|}\ \ ^{m}%
\mathcal{L})}{\delta \mathbf{g}^{kl}}$, with Lagrange density of matter
fields $\ ^{m}\mathcal{L}.$}

\begin{principle}
\textbf{-Theorem} \textsf{[nonholonomic dyadic modifications of Einstein
equations on cotangent Lorentz bundles] } \label{princtheinstctb} On $%
\mathbf{T}^{\ast }\mathbf{V,}$ the modified Einstein equations for $(\ \
_{s}^{\shortmid }\mathbf{g,\ \ }\ _{s}^{\shortmid }\mathbf{D}\
_{s}^{\shortmid }\mathbf{g}=0)$ and Ricci d-tensor (\ref{driccisd}) are
\begin{equation}
\mathbf{\ ^{\shortmid }R}_{\alpha _{s}\beta _{s}}[\mathbf{\ _{s}^{\shortmid
}D}]=\ _{\shortmid }\Upsilon _{\alpha _{s}\beta _{s}}.  \label{meinsteqctb}
\end{equation}
\end{principle}

\begin{proof}
Let us sketch possible two variants of proofs performed in geometric and
N-adapted variational forms in Refs. \cite{v18a,bvepjc18} for d-connections $%
\mathbf{\ D}$ and $\ ^{\shortmid }\mathbf{D,}$ respectively, for $\mathbf{TV}
$ and $\mathbf{T}^{\ast }\mathbf{V}$, see references therein and \cite%
{gvvepjc14,svvijmpd14,
rajpoot15,vacaruplb16,gheorghiuap16,bubuianucqg17,vbubuianu17,vacaruepjc17,vmon06,vijtp10}%
,  on similar results for (non) commutative/ supersymmetric/ higher order
nonholonomic and Lagrange spaces, string gravity models, and massive
gravity. Extending on cotangent Lorentz bundles enabled with dyadic shell
frame structure and arbitrary metric compatible frame connections the
variational calculus on Lorentz manifold $\mathbf{V}$ for the LC-connection $%
\nabla $ provided in footnote \ref{fvariat}, we can prove (\ref{meinsteqctb}%
) using $\ _{s}^{\shortmid }\mathbf{D}$ for $\mathbf{T}^{\ast }\mathbf{V}$.
These modified Einstein equations can be derived following geometric
principles for data $(\mathbf{V,}g,\nabla )$ for a Lorentz manifold
generalized for geometric data $(\ _{s}\mathbf{T}^{\ast }\mathbf{V,\ }\
_{s}^{\shortmid }\mathbf{N,}\ \ _{s}^{\shortmid }\mathbf{g,\ }\
_{s}^{\shortmid }\mathbf{D}).$ Such a theory encodes MDRs (\ref{mdrg}) if$\
\ _{s}^{\shortmid }\mathbf{g}$ is equivalent up to frame transforms to a $\
^{\shortmid }\widetilde{\mathbf{g}}$ (\ref{cdmds}).
\end{proof}

$\square $

Considering distortion d-tensors, $\mathbf{\ ^{\shortmid }D\rightarrow \
^{\shortmid }\check{D}}=\mathbf{\ ^{\shortmid }D+\ ^{\shortmid }Z}$ with a
respective re-definition of effective sources in (\ref{meinsteqctb}), we
prove

\begin{corollary}
\label{meccanvar} \textsf{[canonical form of locally anisotropic Einstein
equations in dyadic variables on cotangent Lorentz bundles ] } The modified
Einstein equations (\ref{meinsteqctb}) can be derived for a canonical
s-connection $\ _{s}^{\shortmid }\widehat{\mathbf{D}},$%
\begin{equation}
\ ^{\shortmid }\widehat{\mathbf{R}}_{\alpha _{s}\beta _{s}}[\
_{s}^{\shortmid }\widehat{\mathbf{D}}]= \ ^{\shortmid }\widehat{\Upsilon }%
_{\alpha _{s}\beta _{s}}.  \label{meinsteqtbcand}
\end{equation}
\end{corollary}

The gravitational field equations in MGTs written for the canonical
s-connection in the form $\ _{s}^{\shortmid }\widehat{\mathbf{D}}$ can be
decoupled and integrated in very general forms (see next section). This
motivates the Corollary \ref{meccanvar} stating the possibility to derives
and write equivalently the modified Einstein equations in MGTs with MDRs in
the form (\ref{meinsteqtbcand}).

It should be noted that in the bulk of MGTs (see reviews \cite{capoz,nojod1,clifton}), there are elaborated variational principles for metric-affine spaces with modified actions for gravitational and matter field interactons (in general, encoding metric compatible or noncompatible linear connections, massive gravity and exotic terms for modeling dark energy and dark matter effects). Corresponding gravitational and matter field equations for MGTs and nonstandard particle physics usually are represented in modified Einstein forms and as generalized Dirac, Yang-Mills-Higgs etc. equations which allow to study new classes of solutions and possible physical implications. They can be also derived from respective modifications of the Einstein-Hilber action. With respect to MGTs with MDRs and possible LIVs, see typical examples in
\cite{amelino98,vapny01,castro07,mavromatos11,mavromatos13a,kostelecky11,kostelecky16,basilakos13}, to elaborate on unified principles for variational formulation of such theories in a general form, and deriving respective field and motion equations, was not possible. Similar problems existed in (generalized) Finsler like gravity theories \cite{cartan35,rund59,asanov85,matsumoto86,bejancu90,bao00,bejancu03}, when modifications of the Einstein equations were proposed in certain heuristic or geometric forms but without an unified axiomatic formulation with generalized relativistic and causal principles and respective generalized variational principles. There were not elaborated  methods for constructing exact solutions in MGTs with MDRs  and not discussed possible implications, for instance, if there are black hole solutions in theories with MDRs. In a series of works  \cite{gvvepjc14,svvijmpd14,rajpoot15,vacaruplb16,gheorghiuap16,
bubuianucqg17,vbubuianu17,vacaruepjc17,vmon06,vijtp10}, we elaborated an unified approach to various classes of (non) commutative, (super) string, Clifford structures and MGTs, involving MDRs and LIVs and various types generalized Finsler-Lagrange structures. The approach is based on geometric methods for nonholonomic manifolds and bundle spaces and resulted in a new analytic formalism for constructing exact solutions in MGTs (see relevant details and explanations on the AFDM in the beginning of next section \ref{s4}).

In our partner works  \cite{v18a,bvepjc18}, we provide an axiomatic formulation of MGTs with MDRs which allows equivalent re-formulations in Finsler-Lagrange-Hamilton variables (which is important for quantization and elaborating thermodynamic like analogous theories) and application of the AFDM for constructing exact solutions. The main goal of this work is to construct such general classes of solutions for quasi-stationary configurations (there were published already certain explicit examples of black hole and cosmological spacetime quasicrystal solutions in \cite{bvap19,vcqg18}). Surprisingly, all such constructions can be derived from respective modifications of the Einstein-Hilbert action but on  respective nonholonomic (co) tangent Lorentz bundles enabled with nonlinear connections, d-metrics and/or s-metrics, d-connections and necessary type s-connections.  Such constructions contain not "just trivial statements that the Einstein equations retrain their form up to an arbitrary additional contribution ...". They involve, for instance, MDRs (\ref{mdrg}) and possible LIV effects which can not modelled in a general axiomatic form on Lorentz manifolds, or their metric-affine generalizations. In this article and our partner and cited works, the main idea is to consider MGTs on (co) tangent Lorentz bundles, with new concepts of nonholonomic relativistic phase spaces, new axiomatic / geometric/ variational formulations with causality and based on N-adapted and nonholonomic shell dyadic calculus.

\subsubsection{Modified Einstein Equations in Hamilton and dyadic variables}

We can consider other different variables and different types of effective
sources in above modified Einstein equations on $\ _{s}\mathbf{T}^{\ast }%
\mathbf{V}$, see (\ref{meinsteqctb}) or (\ref{meinsteqtbcand}). For
instance, considering respective nonholonomic frame transforms and
deformations of d-connections (\ref{candistr}) with distortions of the Ricci
s-tensors as in Theorem \ref{thcandist} and re-definition of effective
sources as in (\ref{effdtsourc}),
\begin{equation}
\ ^{\shortmid }\widetilde{\mathbf{R}}_{\alpha _{s}\beta _{s}}[\
_{s}^{\shortmid }\widetilde{\mathbf{D}}]=\ ^{\shortmid }\widetilde{\Upsilon }%
_{\alpha _{s}\beta _{s}}.  \label{modifeinsthameqdiad}
\end{equation}%
These equations are written in "tilde" variables as for the canonical
Hamilton spaces but with indices labeled shell by shell when canonical hat
values have been used for decoupling. Such equations provide models of phase
space gravity with vacuum (aether) modelled as a relativistic Hamilton
space. The Einstein-Hamilton dyadic equations (\ref{modifeinsthameqdiad})
can be used for describing, for instance, phase space black hole solutions
with Hamilton like degrees of freedom on horizons (we shall provide examples
in next sections).

\begin{remark}
\textsf{[ nonholonomic Hamilton-Cartan variables for locally anisotropic
Einstein equations with dyadic decomposition] } The system of nonlinear PDEs
(\ref{meinsteqctb}) or (\ref{meinsteqtbcand}) can be re-written equivalently
in terms of nonholonomic canonical variables with data $(\ ^{\shortmid }%
\mathbf{g},\ ^{\shortmid }\widehat{\mathbf{D}}),$ or $(\ ^{\shortmid }%
\widetilde{\mathbf{g}},\ ^{\shortmid }\widetilde{\mathbf{D}}).$ Such
equations can be used for constructing generic off-diagonal solutions, or
for performing deformation quantization after additional definition of
almost symplectic Hamilton variables.
\end{remark}

In general, MGTs with MDRs are different for different types of
d-connection, or s-connection, structures because the corresponding phase
spaces and effective sources and Ricci d- and s-tensors are completely
different.

\begin{remark}
\textsf{[(non) equivalence of Einstein-Lagrange and Einstein-Hamilton
theories ]} \newline
The models of locally anisotropic gravity defined by equations (\ref%
{meinsteqctb}), (\ref{meinsteqtbcand}), or (\ref{modifeinsthameqdiad}) are
different because they are derived on different phase spaces and for
different types of d-metric and d-connection structures. Nevertheless, the
geometric and physical data can be transformed equivalently from a tangent
bundle to a cotangent bundle, and inversely, if well-defined $L$-duality
maps (\ref{legendre}) and/or (\ref{invlegendre}) is considered.
\end{remark}

Let us motivate why in the last remark we consider the possibility of equivalent transforms of geometric and physical data for different type of MGTs with MDRs for theories on tangent and cotangent Lorentz bundles
formulated in canonical Lagrange and/or Hamilton variables. This is similar to the fact that the Lagrange and Hamilton formulation of mechanics are, in general, not equivalent. Nevertheless, such geometric mechanics models are equivalent for mechanical systems formulated in certain phase space variables for which well-defined Legendre transforms can be defined. Such results hold true for MGTs formulated in Lagrange-Hamilton variables as we
explain in Appendix \ref{appendixb} (for formulas (\ref{legendre}) and/or (\ref{invlegendre})), see also details in Refs. \cite{v18a,bvepjc18}.

In this work, we shall construct exact solutions for (\ref{meinsteqtbcand})
which posses a general decoupling property. We shall extract
LC-configurations from solutions of such equations if there will be imposed
additional zero-torsion constraints (\ref{lccondsd}).

\section{Decoupling and Integrability of Modified Einstein-Hamilton Equations%
}

\label{s4}In this section we develop the anholonomic frame deformation
method, AFDM, for constructing exact and parametric solutions in (modified)
gravity theories when the coefficients of generic off-diagonal metrics and
generalized connections may depend on all spacetime and/or phase space
coordinates. Our goal is to formulate certain general conditions for
decoupling and integrability of modified Einstein equations for spacetime
and phase spaces with one and two fiber space Killing symmetries. We shall
generate various classes of exact solutions in explicit form for MGTs on
cotangent Lorentz bundles and related models of Einstein-Hamilton gravity.
This geometric method was developed during last 20 years in a series of
works for constructing exact solutions in various theories of commutative,
noncommutative, spinor, supersymmetric, twistor, string, brane,
Finsler-Lagrange and higher order generalizations, of (super) gravity. We
cite \cite%
{gvvepjc14,svvijmpd14,rajpoot15,vacaruplb16,gheorghiuap16,bubuianucqg17,vbubuianu17,vacaruepjc17}
for reviews and new results; then \cite{vmon06,vijtp10}, for reviews of
works published during 1998-2008; and Appendix B in \cite{v18a} for a review
of 20 directions of research in such theories with a number of applications
in modern accelerating cosmology and astrophysics (including results
beginning 2009).\footnote{%
The AFDM was formulated as a geometric method of constructing exact
solutions of physically important systems of nonlinear PDEs, to study
nonholonomic geometric flows and elaborating new methods of geometric and
deformation quantization. It requests definition of respective classes of
nonholonomic distributions of geometric objects and nonholonomic frames,
with respective types 2+2+2+..., 3+1 and (3+1)+(2+2), 3+2+2+... Our approach
should not be confused with the well known Cartan moving frame method and
alternative constructions involving various types of tetradic, dyadic, and
Arnowit-Deser-Wheeler, ADM, formalisms. The main difference of the AFDM from
alternative ones is that in the first case there are considered deformations
both of the frame and nonlinear connection structures to certain
configurations when necessary type PDEs (for instance, modified Einstein
equations) can be decoupled and integrated in certain general forms. When
certain general classes of solutions have been constructed in explicit form,
we can impose necessary types of nonholonomic constraints on the nonlinear
and linear connection structures, respective generating and integration
functions, which allow to extract, for instance, LC-configurations,
elaborated on noncommutative models, and define equivalent configuration in
MGTs etc.}

\subsection{Off-diagonal ansatz for phase space metrics}

Similarly to 4-d and extra dimension geometric constructions in \cite%
{gvvepjc14,svvijmpd14,rajpoot15,vacaruplb16,gheorghiuap16,bubuianucqg17,
vbubuianu17,vacaruepjc17,vmon06,vijtp10}, we shall prove that it is possible
to integrate the system of generalized Einstein equations (\ref%
{meinsteqtbcand}) in MGT with nonholonomic dyadic variables. We can consider
general d-metric and s-metric parameterizations (\ref{dm2and2}) (with
equivalent off-diagonal forms (\ref{offds})) for nonholonomic spacetimes of
dimensions 3-10 and N-adapted coefficients depending, in principle, on all
spacetime and phase space coordinates. Such geometric calculi are cumbersome
and it is not clear what physical importance may have certain "very general"
classes of exact and parametric solutions. In this work, we develop the AFDM
for cotangent nonholonomic dyadic decompositions when the total phase space
s-metric structure may describe extensions of the BH stationary solutions in
GR to s-metrics with Killing symmetry on $\partial _{4}=\partial /\partial
y^{4}=\partial _{t}=\partial /\partial t$ on the shell $s=2;$ Killing
symmetry $\ ^{\shortmid }\partial ^{5}=\partial /\partial p_{5}$ on $s=3,$
and posses a general and/or shell Killing symmetry on $\ ^{\shortmid
}\partial ^{7}=\partial /\partial p_{7}$ on $\ _{s}\mathbf{T}^{\ast }\mathbf{%
V.}$ The last condition is crucial for a simplified proof of decoupling
property of generalized Einstein equations. In particular, such solutions
allow us to construct extensions of the Schwarzschild and Kerr metrics in GR
(with a prime diagonal metric $%
g_{ij}(r)=diag[...,...,(1-r/r_{g})^{-1},(1-r/r_{g})],$ in standard spacetime
spherical coordinates with $\ ^{\shortmid }u^{1}=x^{1}=r, \ ^{\shortmid
}u^{2}=x^{2}=\theta ,\ ^{\shortmid }u^{3}=y^{3}=\varphi ,\ ^{\shortmid
}u^{4}=y^{4}=t)$ to phase space solutions depending in explicit form on the
energy type variable $p_{8}=E$ and $p_{6},$ (also on $p_{5}$ via
N--connection coefficients) but not on $p_{7}.$ We shall study explicit
examples of new classes of black hole solutions in Einstein-Hamilton gravity
in our partner works.

\subsubsection{Quasi-stationary ansatz for s-metrics}

We consider generalizations of the linear quadratic element (\ref{lqed}) when

\begin{definition}
A quasi-stationary s-metric is defined by an ansatz
\begin{eqnarray}
&&ds^{2}=g_{1}(r,\theta )dr^{2}+g_{2}(r,\theta )d\theta ^{2}+g_{3}(r,\theta
,\varphi )\delta \varphi ^{2}+g_{4}(r,\theta ,\varphi )\delta t^{2}+
\label{ansatz1} \\
&&\ ^{\shortmid }g^{5}(r,\theta ,\varphi ,t,p_{6})(\ ^{\shortmid }\mathbf{e}%
_{5})^{2}+\ ^{\shortmid }g^{6}(r,\theta ,\varphi ,t,p_{6})(\ ^{\shortmid }%
\mathbf{e}_{6})+\ ^{\shortmid }g^{7}(r,\theta ,\varphi ,t,p_{6},E)(\
^{\shortmid }\mathbf{e}_{7})^{2}+\ ^{\shortmid }g^{8}(r,\theta ,\varphi
,t,p_{6},E)(\ ^{\shortmid }\mathbf{e}_{8})^{2}.  \notag
\end{eqnarray}
\end{definition}

The N-adapted frames in (\ref{ansatz1}) are parameterized in the form:
\begin{eqnarray}
\ ^{\shortmid }\mathbf{e}^{3} &=&\delta \varphi =d\varphi +\ _{2}^{\shortmid
}w_{1}(r,\theta ,\varphi )dr+\ _{2}^{\shortmid }w_{2}(r,\theta ,\varphi
)d\theta ,  \label{ansatz1nc} \\
\ ^{\shortmid }\mathbf{e}^{4} &=&\delta t=dt+\ _{2}^{\shortmid
}n_{1}(r,\theta ,\varphi )dr+\ _{2}^{\shortmid }n_{2}(r,\theta ,\varphi
)d\theta ,  \notag
\end{eqnarray}
\begin{eqnarray*}
\ ^{\shortmid }\mathbf{e}_{5} &=&\delta p_{5}=dp_{5}+\ _{3}^{\shortmid
}n_{1}(r,\theta ,\varphi ,t,p_{5})dr+\ _{3}^{\shortmid }n_{2}(r,\theta
,\varphi ,t,p_{5})d\theta +\ _{3}^{\shortmid }n_{3}(r,\theta ,\varphi
,t,p_{5})d\varphi +\ _{3}^{\shortmid }n_{4}(r,\theta ,\varphi ,t,p_{5})dt,
\notag \\
\ ^{\shortmid }\mathbf{e}_{6} &=&\delta p_{6}=dp_{6}+\ _{3}^{\shortmid
}w_{1}(r,\theta ,\varphi ,t,p_{5})dr+\ _{3}^{\shortmid }w_{2}(r,\theta
,\varphi ,t,p_{5})d\theta +\ _{3}^{\shortmid }w_{3}(r,\theta ,\varphi
,t,p_{5})d\varphi +\ _{3}^{\shortmid }w_{4}(r,\theta ,\varphi ,t,p_{5})dt,
\notag
\end{eqnarray*}
\begin{eqnarray*}
\ ^{\shortmid }\mathbf{e}_{7} &=&\delta p_{7}=dp_{7}+\ _{4}^{\shortmid
}n_{1}(r,\theta ,\varphi ,t,p_{5},p_{6},E)dr+\ _{4}^{\shortmid
}n_{2}(r,\theta ,\varphi ,t,p_{5},p_{6},E)d\theta +\ _{4}^{\shortmid
}n_{3}(r,\theta ,\varphi ,t,p_{5},p_{6},E)d\varphi  \notag \\
&&+\ _{4}^{\shortmid }n_{4}(r,\theta ,\varphi ,t,p_{5},p_{6},E)dt+\
_{4}^{\shortmid }n_{5}(r,\theta ,\varphi ,t,p_{5},p_{6},E)dp_{5}+\
_{4}^{\shortmid }n_{6}(r,\theta ,\varphi ,t,p_{5},p_{6},E)dp_{6},  \notag \\
\ ^{\shortmid }\mathbf{e}_{8} &=&\delta p_{8}=dp_{8}+\ _{4}^{\shortmid
}w_{1}(r,\theta ,\varphi ,t,p_{5},p_{6},E)dr+\ _{4}^{\shortmid
}w_{2}(r,\theta ,\varphi ,t,p_{5},p_{6},E)d\theta +\ _{4}^{\shortmid
}w_{3}(r,\theta ,\varphi ,t,p_{5},p_{6},E)d\varphi  \notag \\
&&+\ _{4}^{\shortmid }w_{4}(r,\theta ,\varphi ,t,p_{5},p_{6},E)dt+\
_{4}^{\shortmid }w_{5}(r,\theta ,\varphi ,t,p_{5},p_{6},E)dp_{5}+\
_{4}^{\shortmid }w_{6}(r,\theta ,\varphi ,t,p_{5},p_{6},E)dp_{6}.  \notag
\end{eqnarray*}%
The coefficients of s-metric (\ref{ansatz1}) and related N-connection (\ref%
{ansatz1nc}) coefficients consist certain particular examples of s-metrics (%
\ref{dmcts}) and/or (\ref{dm2and2}) with such parameterizations:
\begin{eqnarray*}
g_{i_{1}j_{i}} &=&diag[g_{i_{1}}(x^{k_{1}})],\mbox{ for }i_{1},j_{1}=1,2%
\mbox{ and }x^{k_{1}}=(x^{1}=r,x^{2}=\theta ); \\
g_{a_{2}b_{2}} &=&diag[g_{a_{2}}(x^{k_{1}},y^{3})],\mbox{ for }%
a_{2},b_{2}=3,4\mbox{ and }y^{3}=x^{3}=\varphi ,y^{4}=x^{4}=t; \\
\ ^{\shortmid }g^{a_{3}b_{3}} &=&diag[\ ^{\shortmid
}g^{a_{3}}(x^{k_{1}},y^{a_{2}},p_{6})],\mbox{ for }a_{3},b_{3}=5,6%
\mbox{ and
}\ ^{\shortmid }u^{5}=p_{5},\ ^{\shortmid }u^{6}=p_{6}; \\
\ ^{\shortmid }g^{a_{4}b_{4}} &=&diag[\ ^{\shortmid
}g^{a_{4}}(x^{k_{1}},y^{a_{2}},p_{a_{3}},E)],\mbox{ for }a_{4},b_{4}=7,8%
\mbox{ and }\ ^{\shortmid }u^{7}=p_{7},\ ^{\shortmid }u^{8}=p_{8}=E; \\
\ ^{\shortmid }N_{j_{1}}^{3} &=&\ _{2}^{\shortmid }w_{j_{1}}=\ ^{\shortmid
}w_{j_{1}}(r,\theta ,\varphi ),\ ^{\shortmid }N_{j_{1}}^{4}=\
_{2}^{\shortmid }n_{j_{1}}=\ ^{\shortmid }n_{j_{1}}(r,\theta ,\varphi ); \\
\ ^{\shortmid }N_{j_{2}\ 5} &=&\ _{3}^{\shortmid
}n_{j_{2}}(x^{k_{2}},p_{6})=\ ^{\shortmid }n_{j_{2}}(r,\theta ,\varphi
,t,p_{6}),\ ^{\shortmid }N_{j_{2}\ 6}=\ _{3}^{\shortmid
}w_{j_{2}}(x^{k_{2}},p_{6})=\ ^{\shortmid }w_{j_{2}}(r,\theta ,\varphi
,t,p_{6})\mbox{ for }j_{2}=1,2,3,4; \\
\ ^{\shortmid }N_{j_{3}7} &=&\ _{4}^{\shortmid }n_{j_{3}}(x^{k_{3}},p_{8})=\
^{\shortmid }n_{j_{3}}(r,\theta ,\varphi ,t,E),\ ^{\shortmid }N_{j_{3}8}=\
_{4}^{\shortmid }w_{j_{3}}(x^{k_{3}},p_{8})=\ ^{\shortmid
}w_{j_{3}}(r,\theta ,\varphi ,t,E)\mbox{ for }j_{3}=1,2,3,4,5,6.
\end{eqnarray*}
In GR \cite{hawrking73,misner73,wald82,kramer03}, the term "stationary" is used for metrics which in certain special coordinates do not depend on respective time like coordinates but contain, in principle, certain off-diagonal terms describing, for instance, rotation of Kerr black holes. The horizontal part (i.e. the first 4 terms for a Lorentz manifold base) in (\ref{ansatz1}) is of stationary type. Nevertheless, the cofiber part (next 5-8 terms) of that ansatz  may depend on a time like coordinate "t". This is because MDRs (\ref{mdrg}) may induce certain cofiber dynamics with dependence on "t", and on "E" (for instance, for so-called rainbow metrics). This motivates the term "quasi-stationary" for d-metric ansatz considered above. Fortunately, the AFDM with nonholonomic shell dyadic decompositions allows us to construct exact/parameteric quasi-stationary solutions as we shall demonstrate in Section \ref{s5}.
\begin{remark}
An ansatz (\ref{ansatz1}) is stationary if it s-metric the N-connection
coefficients are prescribed in such an adapted form that such coefficients
do not depend explicitly on the time like variables $t.$
\end{remark}

For stationary configurations both the 4-d base spacetime metric and typical
fiber/shells extensions do not change under $t$-relativistic
dynamics/ evolution. To study extensions of BH solutions (stationary
configurations) in GR for nontrivial stationary MDRs we can consider, for
simplicity, that additional momentum coordinates does not change the
stationary character of solutions. In a more general context, we can
consider that the BH configurations are extended to quasi-stationary
off-diagonal phase space metrics with possible dependencies on $t.$ Such
generalized (co) tangent bundle solutions can be rewritten in equivalent
Lagrange on Hamilton like variables. Here we note that parameterizations of
shell coordinates, and ansatz for s-metrics and N-connections are different
from those considered in our former works \cite%
{gvvepjc14,bubuianucqg17,vbubuianu17,vacaruplb16,vacaruepjc17}.

\subsubsection{Quasi-stationary ansatz for (effective) sources}

We shall be able to generate in explicit form exact off-diagonal solutions
of modified Einstein equations on cotangent Lorentz bundles (\ref%
{meinsteqtbcand}) following such a

\begin{convention}
\label{sourceparam}\textsf{[quasi-stationary (effective) sources for dyadic
splitting on cotangent bundles] } For nonholonomic dyadic decompositions,
the (effective) sources are parameterized in an N- and s-adapted diagonal
form%
\begin{equation}
\ ^{\shortmid }\widehat{\Upsilon }_{\ \beta _{s}}^{\alpha _{s}}=[\
_{1}^{\shortmid }\widehat{\Upsilon }(x^{k_{1}})\delta _{j_{1}}^{i_{1}},\
_{2}^{\shortmid }\widehat{\Upsilon }(x^{k_{1}},y^{3})\delta
_{b_{2}}^{a_{2}},\ _{3}^{\shortmid }\widehat{\Upsilon }(x^{k_{2}},p_{6})%
\delta _{b_{3}}^{a_{3}},\ _{4}^{\shortmid }\widehat{\Upsilon }%
(x^{k_{3}},p_{8})\delta _{b_{4}}^{a_{4}}].  \label{ansatzsourc}
\end{equation}
\end{convention}

We are not able to decouple and integrate in certain general/ explicit form modified Einstein equations for general types of sources. There are necessary additional assumptions on certain general parameterizations on effective and real matter sources which, in this work, are adapted to a nonholonomic dyadic shall frame structure, which will allow a shell by shell integration of the nonlinear system of PDEs (\ref{meinsteqtbcand}) for quasi-stationary ansatz for d-metrics.  Such an assumption is stated by above Convention when the values $\ _{s}^{\shortmid }\widehat{\Upsilon }$ in (\ref{ansatzsourc}) impose certain nonholonomic constraints on the base and cofiber dynamics of (effective) matter sources $\ _{\shortmid }\Upsilon _{\alpha _{s}\beta _{s}}$  (\ref{totaldiadsourcd}). Such sources can be related to a general non-diagonal source $\ ^{\shortmid }\widehat{\Upsilon }_{\alpha _{s}^{\prime }\beta _{s}^{\prime }}$ via frame transforms $\ ^{\shortmid }\widehat{\Upsilon }_{\alpha _{s}\beta
_{s}}=e_{\ \alpha _{s}}^{\alpha _{s}^{\prime }}e_{\ \beta _{s}}^{\beta
_{s}^{\prime }}\ _{\shortmid }\Upsilon _{\alpha _{s}^{\prime }\beta
_{s}^{\prime }}.$ For further constructions, we shall consider that we can
prescribe any $\ _{s}^{\shortmid }\widehat{\Upsilon }$ as \textbf{generating
sources} in order to express certain exact solutions in a general form. We
shall prove that such values and generating functions are related to
conventional cosmological constants via nonlinear symmetries.

In explicit form, $\ _{s}^{\shortmid }\widehat{\Upsilon }$ should be
considered as values determined by certain phenomenological/
experimental/observational data.

\subsubsection{ Ricci d--tensors with dyadic shell decomposition}

Let us introduce such short notations for partial derivatives: $\partial
_{1}q=q^{\bullet },\partial _{2}q=q^{\prime },\partial _{3}q=\partial
_{\varphi }q=q^{\diamond },\partial ^{6}q=\partial q/\partial p^{6},$ and $%
\partial ^{8}q=\partial q/\partial p^{8}=\partial q/\partial E=\partial
_{E}q=q^{\ast }.$ For quasi-stationary configurations, it is always possible%
\footnote{%
we can construct special classes of exact and parametric solutions if the
conditions of this section are not satisfied but the formulas are more
cumbersome and may not allow explicit integration of motion equations} to
define frame and coordinate transforms resulting in N-adapted
parameterizations of metrics in the form (\ref{ansatz1}) with $%
g_{4}^{\diamond }\neq 0,\partial ^{6}g^{5}\neq 0$ and $(g^{7})^{\ast }\neq
0. $ A tedious computation of the coefficients of the canonical
d--connection $\ _{s}^{\shortmid }\widehat{\mathbf{D}}=\{$ $\ ^{\shortmid }%
\widehat{\mathbf{\Gamma }}_{\ \alpha _{s}\beta _{s}}^{\gamma _{s}}\}$ (\ref%
{canondch}) using formulas (\ref{candcons34}) for the ansatz (\ref{ansatz1})
and (at the second step) the corresponding non-trivial coefficients of the
Ricci d--tensor $\ ^{\shortmid }\widehat{\mathbf{R}}_{\alpha _{s}\beta _{s}}$
(\ref{driccisd}), provides a proof of

\begin{lemma}
\label{lemmaricci}The nontrivial shell by shell N-adapted coefficients of $%
\widehat{\mathbf{R}}_{\alpha _{s}\beta _{s}}$ in modified Einstein equations
(\ref{meinsteqtbcand}) for ansatz (\ref{ansatz1}) are given by formulas:
\begin{eqnarray}
\ ^{\shortmid }\widehat{R}_{1}^{1} &=&\ ^{\shortmid }\widehat{R}_{2}^{2}=%
\frac{1}{2g_{1}g_{2}}[\frac{g_{1}^{\bullet }g_{2}^{\bullet }}{2g_{1}}+\frac{%
(g_{2}^{\bullet })^{2}}{2g_{2}}-g_{2}^{\bullet \bullet }+\frac{g_{1}^{\prime
}g_{2}^{\prime }}{2g_{2}}+\frac{\left( g_{1}^{\prime }\right) ^{2}}{2g_{1}}%
-g_{1}^{\prime \prime }]=-\ _{1}^{\shortmid }\widehat{\Upsilon },  \notag \\
\ ^{\shortmid }\widehat{R}_{3}^{3} &=&\ ^{\shortmid }\widehat{R}_{4}^{4}=%
\frac{1}{2g_{3}g_{4}}[\frac{\left( g_{4}^{\diamond }\right) ^{2}}{2g_{4}}+%
\frac{g_{3}^{\diamond }g_{4}^{\diamond }}{2g_{3}}-g_{4}^{\diamond \diamond
}]=-\ _{2}^{\shortmid }\widehat{\Upsilon },  \label{riccist2} \\
\ ^{\shortmid }\widehat{R}_{3k_{1}} &=&\frac{\ w_{k_{1}}}{2g_{4}}%
[g_{4}^{\diamond \diamond }-\frac{\left( g_{4}^{\diamond }\right) ^{2}}{%
2g_{4}}-\frac{(g_{3}^{\diamond })(g_{4}^{\diamond })}{2g_{3}}]+\frac{%
g_{4}^{\diamond }}{4g_{4}}(\frac{\partial _{k_{1}}g_{3}}{g_{3}}+\frac{%
\partial _{k_{1}}g_{4}}{g_{4}})-\frac{\partial _{k_{1}}(g_{3}^{\diamond })}{%
2g_{3}}=0;  \notag \\
\ ^{\shortmid }\widehat{R}_{4k_{1}} &=&\frac{g_{4}}{2g_{3}}%
n_{k_{1}}^{\diamond \diamond }+\left( \frac{3}{2}g_{4}^{\diamond }-\frac{%
g_{4}}{g_{3}}g_{3}^{\diamond }\right) \frac{\ n_{k_{1}}^{\diamond }}{2g_{3}}%
=0,  \notag
\end{eqnarray}%
on shells $s=1$ and $s=2,$with $i_{1},k_{1}...=1,2;$
\begin{eqnarray}
\ ^{\shortmid }\widehat{R}_{5}^{5} &=&\ ^{\shortmid }\widehat{R}_{6}^{6}=-%
\frac{1}{2g^{5}g^{6}}[\frac{(\partial ^{6}(g^{5}))^{2}}{2g^{5}}+\frac{%
(\partial ^{6}(g^{5}))(\partial ^{6}(g^{6}))}{2g^{6}}-\partial ^{6}\partial
^{6}(g^{5})]=-\ _{3}^{\shortmid }\widehat{\Upsilon },  \notag \\
\ ^{\shortmid }\widehat{R}_{5k_{2}} &=&\frac{g^{5}}{2g^{6}}\partial
^{6}\partial ^{6}n_{k_{2}}+\left( \frac{3}{2}\partial ^{6}g^{5}-\frac{g^{5}}{%
g^{6}}\partial ^{6}g^{6}\right) \frac{\ \partial ^{6}n_{k_{2}}}{2g^{6}}=0,
\label{riccist3} \\
\ ^{\shortmid }\widehat{R}_{6k_{2}} &=&\frac{\ w_{k_{2}}}{2g^{5}}[\partial
^{6}\partial ^{6}(g^{5})-\frac{\left( \partial ^{6}g^{5}\right) ^{2}}{2g^{5}}%
-\frac{(\partial ^{6}(g^{5}))(\partial ^{6}(g^{6}))}{2g^{6}}]+\frac{\partial
^{6}(g^{5})}{4g^{6}}(\frac{\partial _{k_{2}}g^{5}}{g^{5}}+\frac{\partial
_{k_{2}}g^{6}}{g^{6}})-\frac{\partial _{k_{2}}\partial ^{6}(g^{6})}{2g^{6}}%
=0;  \notag
\end{eqnarray}%
on shell $s=3$ with $i_{2},k_{2}...=1,2,3,4;$
\begin{eqnarray}
\ ^{\shortmid }\widehat{R}_{7}^{7} &=&\ ^{\shortmid }\widehat{R}_{8}^{8}=-%
\frac{1}{2g^{7}g^{8}}[\frac{((g^{7})^{\ast })^{2}}{2g^{7}}+\frac{%
(g^{7})^{\ast }\ (g^{8})^{\ast }}{2g^{6}}-(g^{7})^{\ast \ast }]=-\
_{4}^{\shortmid }\widehat{\Upsilon },  \notag \\
\ ^{\shortmid }\widehat{R}_{7k_{3}} &=&\frac{g^{7}}{2g^{8}}n_{k_{3}}^{\ast
\ast }+\left( \frac{3}{2}(g^{7})^{\ast }-\frac{g^{7}}{g^{8}}(g^{8})^{\ast
}\right) \frac{\ n_{k_{3}}^{\ast }}{2g^{8}}=0,  \label{riccist4a} \\
\ ^{\shortmid }\widehat{R}_{8k_{2}} &=&\frac{\ w_{k_{3}}}{2g^{7}}%
[(g^{7})^{\ast \ast }-\frac{\left( (g^{7})^{\ast }\right) ^{2}}{2g^{7}}-%
\frac{(g^{7})^{\ast }\ (g^{8})^{\ast }}{2g^{6}}]+\frac{(g^{7})^{\ast }}{%
4g^{8}}(\frac{\partial _{k_{3}}g^{7}}{g^{7}}+\frac{\partial _{k_{3}}g^{8}}{%
g^{8}})-\frac{\partial _{k_{3}}(g^{8})^{\ast }}{2g^{8}}=0  \notag
\end{eqnarray}%
on shell $s=4$ with $i_{3},k_{3}...=1,2,3,4,5,6.$
\end{lemma}

Similar details for nonholonomic extra dimension manifolds and tangent
bundles can be found in \cite%
{gvvepjc14,svvijmpd14,rajpoot15,gheorghiuap16,bubuianucqg17,vbubuianu17,vacaruepjc17,vmon06,vijtp10}
references therein.

Using the above formulas, we can compute the Ricci scalar (\ref{riccidscal})
for $\ _{s}^{\shortmid }\widehat{\mathbf{D}}$ (\ref{candcons34}) when for
the ansatz (\ref{ansatz1}) and prove the

\begin{corollary}
\textsf{[canonical Ricci s-scalar] } The Ricci scalar for the canonical
d-connection of quasi-stationary phase spaces is
\begin{equation*}
\ \ \ _{2}^{\shortmid }\widehat{R}sc=2(\ ^{\shortmid }\widehat{R}_{1}^{1}),\
_{2}^{\shortmid }\widehat{R}sc=2(\ ^{\shortmid }\widehat{R}_{1}^{1}+\
^{\shortmid }\widehat{R}_{3}^{3}),\ _{3}^{\shortmid }\widehat{R}sc=2(\
^{\shortmid }\widehat{R}_{1}^{1}+\ ^{\shortmid }\widehat{R}_{3}^{3}+\
^{\shortmid }\widehat{R}_{5}^{5}),\ \ _{4}^{\shortmid }\widehat{R}sc=2(\
^{\shortmid }\widehat{R}_{1}^{1}+\ ^{\shortmid }\widehat{R}_{3}^{3}+\
^{\shortmid }\widehat{R}_{5}^{5}+\ ^{\shortmid }\widehat{R}_{7}^{7}).
\end{equation*}
\end{corollary}

This Corollary imposes additional N-adapted symmetries on the Einstein
d--tensor constructed following the Definition-Theorem \ref{dteinstdt}. In
result, we prove the

\begin{consequence}
\label{meinstcanon}\textsf{[Modified Einstein equations for the canonical
d-connection with dyadic decomposition] } \newline
On cotangent Lorentz bundles, the nontrivial components with possible
sources of the canonical d-equations \ (\ref{meinsteqtbcand}) \ can be
written in the form
\begin{equation}
\ ^{\shortmid }\widehat{R}_{1}^{1}=\ ^{\shortmid }\widehat{R}_{2}^{2}=-\
_{1}^{\shortmid }\widehat{\Upsilon ;}\ ^{\shortmid }\widehat{R}_{3}^{3}=\
^{\shortmid }\widehat{R}_{4}^{4}=-\ _{2}^{\shortmid }\widehat{\Upsilon };\
^{\shortmid }\widehat{R}_{5}^{5}=\ ^{\shortmid }\widehat{R}_{6}^{6}=-\
_{3}^{\shortmid }\widehat{\Upsilon };\ ^{\shortmid }\widehat{R}_{7}^{7}=\
^{\shortmid }\widehat{R}_{8}^{8}=-\ _{4}^{\shortmid }\widehat{\Upsilon };
\label{sourc1}
\end{equation}%
with (effective) source parameterized in the form $\ _{s}^{\shortmid }%
\widehat{\Upsilon }$ in (\ref{ansatzsourc}).
\end{consequence}

Finally, we note that similar formulas can be derived by other type of
d-connections (Finsler like ones or metric-affine etc.) but the canonical
(with hat ones) will allow us to decouple modified Einstein equations in
certain general forms.

\subsection{Decoupling of modified Einstein equations on cotangent Lorentz
bundles}

The first main result of this work is stated by

\begin{theorem}
\textsf{[decoupling property] } \label{theordecoupl}The gravitational field
equations (\ref{meinsteqtbcand}) with source $\ _{s}^{\shortmid }\widehat{%
\Upsilon }$ (\ref{ansatzsourc}), with Ricci s-tensors (\ref{driccisd})
computed for the canonical s-connection $\ _{s}^{\shortmid }\widehat{\mathbf{%
D}}$ with N-adapted coefficients defined by formulas (\ref{candcons12}) and (%
\ref{candcons34}) for respective s-metric and N-connection ansatz (\ref%
{ansatz1}) and (\ref{ansatz1nc}) decouple respectively: :
\begin{eqnarray}
s &=&1\mbox{ with }g_{i_{1}}=e^{\psi (r,\theta )},i_{1}=1,2  \notag \\
&&\psi ^{\bullet \bullet }+\psi ^{\prime \prime }=2\ _{1}^{\shortmid }%
\widehat{\Upsilon };  \label{eq1} \\
s &=&2\mbox{ with }\left\{
\begin{array}{c}
\alpha _{i_{1}}=g_{4}^{\diamond }\partial _{i_{1}}(\ _{2}^{\shortmid }\varpi
),\ _{2}^{\shortmid }\beta =g_{4}^{\diamond }(\ _{2}^{\shortmid }\varpi
)^{\diamond },\ _{2}^{\shortmid }\gamma =(\ln \frac{|g_{4}|^{3/2}}{|g_{3}|}%
)^{\diamond } \\
\mbox{ for }\ _{2}^{\shortmid }\Psi =\exp (\ _{2}^{\shortmid }\varpi ),\
_{2}^{\shortmid }\varpi =\ln |g_{4}^{\diamond }/\sqrt{|g_{3}g_{4}}|,%
\end{array}%
\right.  \notag \\
&&(\ _{2}^{\shortmid }\varpi )^{\diamond }g_{4}^{\diamond }=2g_{3}g_{4}\
_{2}^{\shortmid }\widehat{\Upsilon },  \label{e2a} \\
&&\ _{2}^{\shortmid }\beta \ ^{\shortmid }w_{j_{1}}-\alpha _{j_{1}}=0,
\label{e2b} \\
&&\ ^{\shortmid }n_{k_{1}}^{\diamond \diamond }+\ _{2}^{\shortmid }\gamma \
^{\shortmid }n_{k_{1}}^{\diamond }=0;  \label{e2c}
\end{eqnarray}%
\begin{eqnarray}
s &=&3\mbox{ with }\left\{
\begin{array}{c}
\alpha _{i_{2}}=(\partial ^{6}g^{5})\partial _{i_{2}}(\ _{3}^{\shortmid
}\varpi ),\ _{3}^{\shortmid }\beta =(\partial ^{6}g^{5})\partial ^{6}(\
_{3}^{\shortmid }\varpi ),\ _{3}^{\shortmid }\gamma =\partial ^{6}(\ln \frac{%
|g^{5}|^{3/2}}{|g^{6}|}) \\
\mbox{ for }\ _{3}^{\shortmid }\Psi =\exp (\ _{3}^{\shortmid }\varpi ),\
_{3}^{\shortmid }\varpi =\ln |\partial ^{6}g^{5}/\sqrt{|g^{5}g^{6}}%
|,(i_{2},j_{2},k_{2}=1,2,3,4)%
\end{array}%
\right.  \notag \\
&&\partial ^{6}(\ _{3}^{\shortmid }\varpi )\partial ^{6}g^{5}=2g^{5}g^{6}\
_{3}^{\shortmid }\widehat{\Upsilon },  \label{e3a} \\
&&\partial ^{66}(\ ^{\shortmid }n_{k_{2}})+\ _{3}^{\shortmid }\gamma
\partial ^{6}(\ ^{\shortmid }n_{k_{2}})=0,  \label{e3b} \\
&&\ _{3}^{\shortmid }\beta \ ^{\shortmid }w_{j_{2}}-\alpha _{j_{2}}=0;
\label{e3c} \\
s &=&4\mbox{ with }\left\{
\begin{array}{c}
\alpha _{i_{3}}=(g^{7})^{\ast }\partial _{i_{3}}(\ _{3}^{\shortmid }\varpi
),\ _{4}^{\shortmid }\beta =(g^{7})^{\ast }(\ _{4}^{\shortmid }\varpi
)^{\ast },\ _{4}^{\shortmid }\gamma =(\ln \frac{|g^{7}|^{3/2}}{|g^{8}|}%
)^{\ast } \\
\mbox{ for }\ _{4}^{\shortmid }\Psi =\exp (\ _{4}^{\shortmid }\varpi ),\
_{4}^{\shortmid }\varpi =\ln |(g^{7})^{\ast }/\sqrt{|g^{7}g^{8}}%
|,(i_{3},j_{3},k_{3}=1,2,...,6)%
\end{array}%
\right.  \notag \\
&&\partial ^{6}(\ _{4}^{\shortmid }\varpi )(g^{7})^{\ast }=2g^{7}g^{8}\
_{4}^{\shortmid }\widehat{\Upsilon },  \label{eq4a} \\
&&(\ ^{\shortmid }n_{k_{3}})^{\ast \ast }+\ _{4}^{\shortmid }\gamma (\
^{\shortmid }n_{k_{3}})^{\ast }=0,  \label{eq4b} \\
&&\ _{4}^{\shortmid }\beta \ ^{\shortmid }w_{j_{3}}-\alpha _{j_{3}}=0.
\label{eq4c}
\end{eqnarray}
\end{theorem}

\begin{proof}
It follows from Lemma \ref{lemmaricci} with a corresponding definition of
coefficients. The main difference is that on $\ _{s}\mathbf{T}^{\ast }%
\mathbf{V}$ we work with mixed conventional tangent bundle variables, $%
y^{a_{2}},$ and cotangent shell variables, $p_{a_{3}}$ and $p_{a_{4}}.$ So,
the equations (\ref{e2a})-(\ref{e2c}), for $s=2,$ are similar to those
presented in Table 3 in \cite{vbubuianu17} but the equations (\ref{e3a})-(%
\ref{e3c}), for $s=3,$ and (\ref{eq4a})-(\ref{eq4c}), for $s=4,$ are similar
to those presented in Table 4 in the same paper (the difference for our work
is that variables for shells $s=3,4$ are momentum type). $\square $\vskip5pt
\end{proof}

We note that above system of nonlinear PDEs possess an explicit decoupling
property: The equation (\ref{eq1}) is a 2-d Poisson equation which can be
integrated in general form for any prescribed source. Other groups of
equations corresponding to shells $s=2,3,4$ contain only partial derivatives
on $y^{3},p_{6},p_{8}$ for certain decoupled functions for the s-metric and
N-connection coefficients which can be integrated step by step. Similar
systems of PDEs have been considered in Theorem 3.1 \cite{vijtp10}, for
extra dimension gravity; the computations for Ricci d-tensors (38)-(45) and
the systems of shell PDEs (47)-(54) for higher dimension MGTs in \cite%
{gvvepjc14}; and for the nonholonomic motion equations and effective
gravitational equations in string gravity; similar proofs can be found in
\cite{vacaruepjc17} (see there section 2.3.2, with formulas (25)-(32), and
section 2.3.5 with formulas (39)-(48)).

\subsection{Generating exact solutions and integral varieties for phase
spaces}

Let us show how the systems of shell PDEs (\ref{eq1})-(\ref{eq4c}) can be
integrated in explicit form (for nontrivial nonholonomic induced torsion or
for LC-configurations). We shall write down respective quadratic line
elements for such integral varieties on cotangent Lorentz bundles.

\subsubsection{Off-diagonal solutions with nontrivial sources and
nonholonomic induced torsion}

For instance, we integrate in explicit form the system (\ref{eq4a})-(\ref%
{eq4c}) following the methods outlined in Section 2.3.6 of \cite%
{vacaruepjc17} applying respective re-parameterizations of locally
anisotropic cosmological solutions from Section 5 \cite{vbubuianu17}, in our
work modified for momentum type variables on typical fiber spaces of
cotangent bundle spaces. In appendix \ref{assgensolt}, we sketch the proof
of this

\begin{theorem}
\textsf{[general solutions with nonholonomic torsion] } \label{gensoltorsion}%
Generic off-diagonal quasi-stationary solutions with Killing symmetry on $%
p_{7}$ of the generalized Einstein equations decoupled in Theorem \ref%
{theordecoupl} are defined by such quadratic line elements on cotangent
Lorentz bundle:{\small
\begin{eqnarray}
ds^{2} &=&e^{\psi (x^{k_{1}})}[(dx^{1})^{2}+(dx^{2})^{2}]+\frac{[(\
_{2}^{\shortmid }\Psi )^{\diamond }]^{2}}{4(\ _{2}^{\shortmid }\widehat{%
\Upsilon })^{2}\{g_{4}^{[0]}-\int d\varphi \lbrack (\ _{2}^{\shortmid }\Psi
)^{2}]^{\diamond }/4(\ _{2}^{\shortmid }\widehat{\Upsilon })\}}\{d\varphi +%
\frac{\partial _{i_{1}}(\ _{2}^{\shortmid }\Psi )}{(\ _{2}^{\shortmid }\Psi
)^{\diamond }}dx^{i_{1}}\}^{2}+  \label{qeltors} \\
&&\{g_{4}^{[0]}-\int d\varphi \frac{\lbrack (\ _{2}^{\shortmid }\Psi
)^{2}]^{\diamond }}{4(\ _{2}^{\shortmid }\widehat{\Upsilon })}\}\{dt+[\
_{1}n_{k_{1}}+\ _{2}n_{k_{1}}\int d\varphi \frac{\lbrack (\ _{2}^{\shortmid
}\Psi )^{2}]^{\diamond }}{4(\ _{2}^{\shortmid }\widehat{\Upsilon }%
)^{2}|g_{4}^{[0]}-\int d\varphi \lbrack (\ _{2}^{\shortmid }\Psi
)^{2}]^{\diamond }/4(\ _{2}^{\shortmid }\widehat{\Upsilon })|^{5/2}}%
]dx^{k_{1}}\}+  \notag
\end{eqnarray}%
\begin{eqnarray*}
&&\{g_{5}^{[0]}-\int dp_{6}\frac{\partial ^{6}[(\ _{3}^{\shortmid }\Psi
)^{2}]}{4(\ _{3}^{\shortmid }\widehat{\Upsilon })}\}\{dp_{5}+[\
_{1}n_{k_{2}}+\ _{2}n_{k_{2}}\int dp_{6}\frac{\partial ^{6}[(\
_{3}^{\shortmid }\Psi )^{2}]}{4(\ _{3}^{\shortmid }\widehat{\Upsilon }%
)^{2}|g_{5}^{[0]}-\int dp_{6}\partial ^{6}[(\ _{3}^{\shortmid }\Psi
)^{2}]/4(\ _{3}^{\shortmid }\widehat{\Upsilon })|^{5/2}}]dx^{k_{2}}\}+ \\
&&\frac{[\partial ^{6}(\ _{3}^{\shortmid }\Psi )]^{2}}{4(\ _{3}^{\shortmid }%
\widehat{\Upsilon })^{2}\{g_{5}^{[0]}-\int dp_{6}\partial ^{6}[(\
_{3}^{\shortmid }\Psi )^{2}]/4(\ _{3}^{\shortmid }\widehat{\Upsilon })\}}%
\{dp_{6}+\frac{\partial _{i_{2}}(\ _{3}^{\shortmid }\Psi )}{\partial ^{6}(\
_{3}^{\shortmid }\Psi )}dx^{i_{2}}\}^{2}+ \\
&&\{g_{7}^{[0]}-\int dE\frac{[(\ _{4}^{\shortmid }\Psi )^{2}]^{\ast }}{4(\
_{4}^{\shortmid }\widehat{\Upsilon })}\}\{dp_{7}+[\ _{1}n_{k_{3}}+\
_{2}n_{k_{3}}\int dE\frac{[(\ _{4}^{\shortmid }\Psi )^{2}]^{\ast }}{4(\
_{4}^{\shortmid }\widehat{\Upsilon })^{2}|g_{7}^{[0]}-\int dE[(\
_{4}^{\shortmid }\Psi )^{2}]^{\ast }/4(\ _{4}^{\shortmid }\widehat{\Upsilon }%
)|^{5/2}}]dx^{k_{3}}\}- \\
&&\frac{[(\ _{4}^{\shortmid }\Psi )^{\ast }]^{2}}{4(\ _{4}^{\shortmid }%
\widehat{\Upsilon })^{2}\{g_{7}^{[0]}-\int dE[(\ _{4}^{\shortmid }\Psi
)^{2}]^{\ast }/4(\ _{4}^{\shortmid }\widehat{\Upsilon })\}}\{dE+\frac{%
\partial _{i_{3}}(\ _{4}^{\shortmid }\Psi )}{(\ _{4}^{\shortmid }\Psi
)^{\ast }}dx^{i_{3}}\}^{2},
\end{eqnarray*}%
for indices $%
i_{1},j_{1},k_{1},...=1,2;i_{2},j_{2},k_{2},...=1,2,3,4;i_{3},j_{3},k_{3},...=1,2,...6;y^{3}=\varphi ,y^{4}=t,p_{8}=E;
$ and
\begin{eqnarray*}
&&\mbox{generating functions: }\psi (x^{k_{1}});\ _{2}^{\shortmid }\Psi
(x^{k_{1}}y^{3});\ _{3}^{\shortmid }\Psi (x^{k_{2}},p_{6});\ _{4}^{\shortmid
}\Psi (x^{k_{3}},E); \\
&&\mbox{generating sources:}\ _{1}^{\shortmid }\widehat{\Upsilon }%
(x^{k_{1}});\ _{2}^{\shortmid }\widehat{\Upsilon }(x^{k_{1}}y^{3});\
_{3}^{\shortmid }\widehat{\Upsilon }(x^{k_{2}},p_{6});\ _{4}^{\shortmid }%
\widehat{\Upsilon }(x^{k_{3}},E); \\
&&\mbox{integr. functions: }g_{4}^{[0]}(x^{k_{1}}),\
_{1}n_{k_{1}}(x^{j_{1}}),\ _{2}n_{k_{1}}(x^{j_{1}});g_{5}^{[0]}(x^{k_{2}}),\
_{1}n_{k_{2}}(x^{j_{2}}),\ _{2}n_{k_{2}}(x^{j_{2}});g_{7}^{[0]}(x^{j_{3}}),\
_{1}n_{k_{3}}(x^{j_{3}}),\ _{2}n_{k_{3}}(x^{j_{3}}).
\end{eqnarray*}%
}
\end{theorem}

We note that above generating functions and sources and integration
functions should be prescribed from additional considerations (compatibility
with experimental/observational data, boundary or asymptotic conditions,
quantum corrections) in order to elaborate on viable physical models. This
is a typical property of general integrals of systems of nonlinear PDEs if
such configurations are constructed in some general forms and not only for
special ansatz transforming PDEs into ODEs.

Finally, we emphasize that solutions (\ref{qeltors}) are with Killing
symmetry on $\partial ^{7}$ and can be used for nontrivial $\
_{s}^{\shortmid }\widehat{\Upsilon }$--sources which can be redefined via
nonlinear transforms of generating functions into effective cosmological
constants, see  subsection \ref{ssnonlsym}. For vacuum spacetime
and phase space configurations with zero $\ _{s}^{\shortmid }\widehat{%
\Upsilon }$--sources, we have to apply more special geometric methods for
generating exact solutions, see similar constructions for 10-d string
gravity in section 2.3.4 of \cite{vacaruepjc17} and (below) section \ref%
{ssvacuumfc}.

\subsubsection{Levi-Civita configurations for phase spaces}

\label{ssslconf}A generic off--diagonal solution (\ref{qeltors}) is
constructed for a canonical d--connections $\ _{s}^{\shortmid }\widehat{%
\mathbf{D}}.$ Such a solution is characterized by nonholonomically induced
d--torsion coefficients $\ _{s}^{\shortmid }\widehat{\mathbf{T}}_{\ \alpha
_{s}\beta _{s}}^{\gamma _{s}}\ $(\ref{dtors}) completely defined by the
N--connection and s--metric structures. The conditions (\ref{lccondsd}) for
a zero torsion can be satisfied by a subclass of nonholonomic distributions
on a cotangent Lorentz bundle with dyadic structure, $\ _{s}\mathbf{T}^{\ast
}\mathbf{V.}$ Resulting LC-configurations are determined by corresponding
parameterizations of the generating and integration functions and sources.
By straightforward computations on cotangent bundles (such details for
nonholonomic manifolds and tangent bundles are presented in Refs. \cite%
{gvvepjc14,svvijmpd14,rajpoot15,gheorghiuap16,bubuianucqg17,vbubuianu17,vacaruepjc17,vmon06,vijtp10}%
), we can verify that all d-torsion coefficients vanish if the coefficients
of N--adapted frames and $^{s}v$--components of d--metrics are subjected to
respective conditions,
\begin{eqnarray}
s=2: &&\ w_{i_{1}}^{\diamond }=\mathbf{e}_{i_{1}}\ln \sqrt{|\ g_{3}|},%
\mathbf{e}_{i_{1}}\ln \sqrt{|\ g_{4}|}=0,\partial _{i_{1}}w_{j_{1}}=\partial
_{j_{1}}w_{i_{1}}\mbox{ and }n_{i_{1}}^{\diamond }=0;  \notag \\
s=3: &&\ \partial ^{6}n_{i_{2}}=\ \mathbf{e}_{i_{2}}\ln \sqrt{|\ g^{6}|},\
\mathbf{e}_{i_{2}}\ln \sqrt{|\ g^{5}|}=0,\partial _{i_{2}}n_{j_{2}}=\partial
_{j_{2}}n_{i_{2}}\mbox{ and }\ \partial ^{6}w_{i_{2}}=0;  \label{zerot} \\
s=4: &&\ n_{i_{3}}^{\ast }=\mathbf{e}_{i_{3}}\ln \sqrt{|\ g^{8}|},\ \mathbf{e%
}_{i_{3}}\ln \sqrt{|\ g^{7}|}=0,\partial _{i_{3}}n_{j_{3}}=\partial
_{j_{3}}n_{i_{3}}\mbox{ and }\ n_{i_{3}}^{\ast }=0;  \notag
\end{eqnarray}%
\begin{eqnarray}
\mbox{and }\ s=2: &&\ n_{k_{1}}(x^{i_{1}})=0\mbox{ and }\partial
_{i_{1}}n_{j_{1}}(x^{k_{1}})=\partial _{j_{1}}n_{i_{1}}(x^{k_{1}});
\label{expcondn} \\
s=3: &&\ w_{k_{2}}(x^{i_{2}})=0\mbox{ and }\partial _{i_{2}}\
w_{j_{2}}(x^{k_{2}})=\partial _{j_{2}}\ w_{i_{2}}(x^{k_{2}});  \notag \\
s=4: &&\ w_{k_{3}}(x^{i_{3}})=0\mbox{ and }\partial _{i_{3}}\
w_{j_{3}}(x^{k_{3}})=\partial _{j_{3}}\ w_{i_{3}}(x^{k_{3}}).  \notag
\end{eqnarray}%
The solutions for on $w$- and $n$-functions derived from (\ref{zerot})
depend on the class of vacuum or non--vacuum metrics we try to construct,
see details in \cite{gvvepjc14}. Here we sketch such steps for constraining
respective classes of generating functions and generating sources:

Prescribing a generating function $\ _{2}^{\shortmid }\Psi =\
_{2}^{\shortmid }\check{\Psi}(x^{i_{1}},y^{3}),$ for which $[\partial
_{i_{1}}(\ _{2}^{\shortmid }\check{\Psi})]^{\diamond }=\partial _{i_{1}}(\
_{2}^{\shortmid }\check{\Psi})^{\diamond },$we solve the conditions for $%
w_{j_{1}}$ in (\ref{zerot}) in explicit form if $\ _{2}^{\shortmid }\widehat{%
\Upsilon }=const,$ or if $\ $such an effective source can be expressed as a
functional $\ \ _{2}^{\shortmid }\widehat{\Upsilon }(x^{i},y^{3})=\
_{2}^{\shortmid }\widehat{\Upsilon }[\ _{2}^{\shortmid }\check{\Psi}].$ On
the shell $s=3,$ we prescribe $\ _{3}^{\shortmid }\Psi =\ _{3}^{\shortmid }%
\check{\Psi}(x^{i_{2}},p^{6}),$ for which $\partial ^{6}[\partial _{i_{2}}(\
_{3}^{\shortmid }\check{\Psi})]=\partial _{i_{2}}\partial ^{6}(\
_{3}^{\shortmid }\check{\Psi}).$ Similarly, we can extend these formulas for
$s=4.$

The third conditions in (\ref{zerot}), for $s=2$, i.e. $\partial
_{i_{1}}w_{j_{1}}=\partial _{j_{1}}w_{i_{1}},$ are solved for any generating
function $\ _{2}^{\shortmid }\check{A}=\ _{2}^{\shortmid }\check{A}%
(x^{k},y^{3})$ for which $\ $%
\begin{equation*}
w_{i_{1}}=\check{w}_{i_{1}}=\partial _{i_{1}}\ _{2}^{\shortmid }\Psi /(\
_{2}^{\shortmid }\Psi )^{\diamond }=\partial _{i_{1}}\ _{2}^{\shortmid }%
\check{A}.
\end{equation*}%
Such formulas on the second shell are written $\ $%
\begin{equation*}
n_{i_{2}}=\check{n}_{i_{2}}=\partial _{i_{2}}\ _{3}^{\shortmid }\Psi
/\partial ^{6}(\ _{3}^{\shortmid }\Psi )=\partial _{i_{2}}\ _{3}^{\shortmid }%
\check{A}.
\end{equation*}%
Summarizing above constructions for shells $s=2,3,4,$ the LC-conditions for
generating functions correlated to functional, or constant values, for
generating sources, we obtain {\small
\begin{eqnarray}
s=2: &&\ _{2}^{\shortmid }\Psi =\ _{2}^{\shortmid }\check{\Psi}%
(x^{i_{1}},y^{3}),(\partial _{i_{1}}\ _{2}^{\shortmid }\check{\Psi}%
)^{\diamond }=\partial _{i_{1}}(\ _{2}^{\shortmid }\check{\Psi})^{\diamond },%
\check{w}_{i_{1}}=\partial _{i_{1}}(\ _{2}^{\shortmid }\check{\Psi})/(\
_{2}^{\shortmid }\check{\Psi})^{\diamond }=\partial _{i_{1}}(\
_{2}^{\shortmid }\check{A});n_{i_{1}}=\partial _{i_{1}}[\ ^{2}n(x^{k_{1}})];
\label{expconda} \\
&&\ _{2}^{\shortmid }\widehat{\Upsilon }(x^{i},y^{3})=\ _{2}^{\shortmid }%
\widehat{\Upsilon }[\ _{2}^{\shortmid }\check{\Psi}], \mbox{ or
}\ _{2}^{\shortmid }\widehat{\Upsilon }= const;  \notag \\
s=3: &&\ \ _{3}^{\shortmid }\Psi =\ _{3}^{\shortmid }\check{\Psi}%
(x^{i_{2}},p^{6}),\partial ^{6}[\partial _{i_{2}}(\ _{3}^{\shortmid }\check{%
\Psi})]=\partial _{i_{2}}\partial ^{6}(\ _{3}^{\shortmid }\check{\Psi});%
\check{w}_{i_{2}}=\partial _{i_{2}}(\ _{3}^{\shortmid }\Psi )/\partial
^{6}(\ _{3}^{\shortmid }\Psi )=\partial _{i_{2}}(\ _{3}^{\shortmid }\check{A}%
);\ n_{i_{2}}=\partial _{i_{2}}[\ ^{3}n(x^{k_{2}})];  \notag \\
&&\ _{3}^{\shortmid}\widehat{\Upsilon }(x^{i_{2}},p^{6})=\ _{3}^{\shortmid }%
\widehat{\Upsilon }[\ _{3}^{\shortmid }\check{\Psi}], \mbox{ or } \
_{3}^{\shortmid }\widehat{\Upsilon }=const;  \notag \\
s=4: &&\ \ _{4}^{\shortmid }\Psi =\ _{4}^{\shortmid }\check{\Psi}%
(x^{i_{3}},E),[\partial _{i_{3}}(\ _{4}^{\shortmid }\check{\Psi})]^{\ast
}=\partial _{i_{3}}(\ _{4}^{\shortmid }\check{\Psi})^{\ast };\check{w}%
_{i_{3}}=\partial _{i_{3}}(\ _{4}^{\shortmid }\Psi )/(\ _{4}^{\shortmid
}\Psi )^{\ast }=\partial _{i_{3}}(\ _{4}^{\shortmid }\check{A});\
n_{i_{3}}=\partial _{i_{3}}[\ ^{4}n(x^{k_{3}})];  \notag \\
&&\ _{4}^{\shortmid}\widehat{\Upsilon }(x^{i_{3}},E)=\ _{4}^{\shortmid }%
\widehat{\Upsilon }[\ _{4}^{\shortmid }\check{\Psi}],\mbox{ or } \
_{4}^{\shortmid }\widehat{\Upsilon }=const.  \notag
\end{eqnarray}
}

In result, we obtain a proof of this

\begin{consequence}
\label{ofdiagztsol}\textsf{[off-diagonal solutions with zero torsion] } A
quadratic nonlinear element (\ref{qeltors}) from Theorem \ref{gensoltorsion}
define an exact solution with zero torsion, i.e. LC-configurations which can
be modelled in Einstein-Hamilton variables, if the generating functions,
generating sources and certain integration functions for the coefficients of
s-metrics are prescribed to satisfy the conditions (\ref{expconda}).
\end{consequence}

The quadratic nonlinear line elements for LC-configurations stated by
Consequence \ref{ofdiagztsol} can be parameterized in such a form%
\begin{eqnarray}
ds_{LCst}^{2} &=&\check{g}_{\alpha \beta }(x^{k_{3}},E)du^{\alpha }du^{\beta
}=e^{\psi (x^{k_{1}})}[(dx^{1})^{2}+(dx^{2})^{2}]+  \label{qellc} \\
&&\frac{[(\ _{2}^{\shortmid }\check{\Psi})^{\diamond }]^{2}}{4(\
_{2}^{\shortmid }\widehat{\Upsilon }[\ _{2}^{\shortmid }\check{\Psi}%
])^{2}\{g_{4}^{[0]}-\int d\varphi \lbrack (\ _{2}^{\shortmid }\check{\Psi}%
)^{2}]^{\diamond }/4(\ _{2}^{\shortmid }\widehat{\Upsilon })\}}\{d\varphi
+[\partial _{i_{1}}(\ _{2}^{\shortmid }\check{A})]dx^{i_{1}}\}+  \notag \\
&&\{g_{4}^{[0]}-\int d\varphi \frac{\lbrack (\ _{2}^{\shortmid }\check{\Psi}%
)^{2}]^{\diamond }}{4(\ _{2}^{\shortmid }\widehat{\Upsilon }[\
_{2}^{\shortmid }\check{\Psi}])}\}\{dt+\partial _{i_{1}}[\
^{2}n(x^{k_{1}})]dx^{i_{1}}\}+  \notag
\end{eqnarray}%
\begin{eqnarray*}
&&\{g_{5}^{[0]}-\int dp_{6}\frac{\partial ^{6}[(\ _{3}^{\shortmid }\Psi
)^{2}]}{4(\ _{3}^{\shortmid }\widehat{\Upsilon }[\ _{3}^{\shortmid }\check{%
\Psi}])}\}\{dp_{5}+\partial _{i_{2}}[\ ^{3}n(x^{k_{2}})]dx^{i_{2}}\}+ \\
&&\frac{[\partial ^{6}(\ _{3}^{\shortmid }\Psi )]^{2}}{4(\ _{3}^{\shortmid }%
\widehat{\Upsilon }[\ _{3}^{\shortmid }\check{\Psi}])^{2}\{g_{5}^{[0]}-\int
dp_{6}\partial ^{6}[(\ _{3}^{\shortmid }\Psi )^{2}]/4(\ _{3}^{\shortmid }%
\widehat{\Upsilon }[\ _{3}^{\shortmid }\check{\Psi}])\}}\{dp_{6}+[\partial
_{i_{2}}(\ _{3}^{\shortmid }\check{A})]dx^{i_{2}}\}^{2}+ \\
&&\{g_{7}^{[0]}-\int dE\frac{[(\ _{4}^{\shortmid }\Psi )^{2}]^{\ast }}{4(\
_{4}^{\shortmid }\widehat{\Upsilon }[\ _{4}^{\shortmid }\check{\Psi}])}%
\}\{dp_{7}+\partial _{i_{3}}[\ ^{4}n(x^{k_{3}})]dx^{i_{3}}\}- \\
&&\frac{[(\ _{4}^{\shortmid }\Psi )^{\ast }]^{2}}{4(\ _{4}^{\shortmid }%
\widehat{\Upsilon }[\ _{4}^{\shortmid }\check{\Psi}])^{2}\{g_{7}^{[0]}-\int
dE[(\ _{4}^{\shortmid }\Psi )^{2}]^{\ast }/4(\ _{4}^{\shortmid }\widehat{%
\Upsilon }[\ _{4}^{\shortmid }\check{\Psi}])\}}\{dE+[\partial _{i_{3}}(\
_{4}^{\shortmid }\check{A})]dx^{i_{3}}\}^{2}.
\end{eqnarray*}

Here we emphasize that quasi-stationary s-metrics (\ref{qellc}) are generic
off-diagonal if at least on one shell there are nontrivial anholonomy
relations of type (\ref{anhrel}). Such exact solutions possess a Killing
symmetry on $\partial ^{7},$ i.e. there are such N-adapted coordinate
systems when the s-metrics do not depend on coordinate $p_{7}.$ There are
generated stationary configurations if the generating and integration
functions and generating sources on respective shells are chosen to not
depend on the $t$-coordinate.

\subsection{Nonlinear symmetries, generation functions and generation sources%
}

\label{ssnonlsym}The classes of solutions (\ref{qeltors}) (and (\ref{qellc})
for LC-configurations) posses important nonlinear shell symmetries relating
nontrivial generating functions and effective sources, see similar details
in \cite{gvvepjc14,vbubuianu17,vacaruepjc17} and references therein. Such
symmetries allow us to re-define the generating functions and introduce on
dyadic shells some nontrivial effective cosmological constants instead of
effective sources. In this subsection, we define and study nonlinear
symmetries of gravity theories encoding MDRs and equivalent
Einstein-Hamilton models, see appendix \ref{assgensolt}).

Let us consider the shell $s=2.$ We can change the generating data, $(\
_{2}^{\shortmid }\Psi ,\ \ _{2}^{\shortmid }\widehat{\Upsilon }%
)\leftrightarrow (\ _{2}^{\shortmid }\Phi ,\ _{2}^{\shortmid }\Lambda
=const),$ following formulas%
\begin{eqnarray}
\frac{\lbrack (\ _{2}^{\shortmid }\Psi )^{2}]^{\diamond }}{\ _{2}^{\shortmid
}\widehat{\Upsilon }} &=&\frac{[(\ _{2}^{\shortmid }\Phi )^{2}]^{\diamond }}{%
\ _{2}^{\shortmid }\Lambda },\mbox{ which can be
integrated as  }  \label{ntransf1} \\
(\ _{2}^{\shortmid }\Phi )^{2} &=&\ _{2}^{\shortmid }\Lambda \int dy^{3}(\
_{2}^{\shortmid }\widehat{\Upsilon })^{-1}[(\ _{2}^{\shortmid }\Psi
)^{2}]^{\diamond }\mbox{ and/or }(\ _{2}^{\shortmid }\Psi )^{2}=(\
_{2}^{\shortmid }\Lambda )^{-1}\int dy^{3}(\ _{2}^{\shortmid }\widehat{%
\Upsilon })[(\ _{2}^{\shortmid }\Phi )^{2}]^{\diamond }.  \label{ntransf2}
\end{eqnarray}%
The equation (\ref{ntransf1}) allows us to simplify the formula for solution
(\ref{g4}) and write $g_{4}=g_{4}^{[0]}-\frac{(\ _{2}^{\shortmid }\Phi )^{2}%
}{4\ _{2}^{\shortmid }\Lambda }$ in terms of a new generating function $\
_{2}^{\shortmid }\Phi $ and a nontrivial (effective) cosmological constant $%
\ _{2}^{\shortmid }\Lambda .$ In order to express the formulas (\ref{g3})
and (\ref{gn}), we have to express $(\ _{2}^{\shortmid }\Psi )^{\diamond }/\
_{2}^{\shortmid }\widehat{\Upsilon }$ in terms of data $(\ _{2}^{\shortmid
}\Phi ,\ _{2}^{\shortmid }\Lambda ).$ We write (\ref{ntransf1}) and the
second equation in (\ref{ntransf2}) in such forms%
\begin{equation*}
\frac{(\ _{2}^{\shortmid }\Psi )[(\ _{2}^{\shortmid }\Psi )]^{\diamond }}{\
_{2}^{\shortmid }\widehat{\Upsilon }}=\frac{[(\ _{2}^{\shortmid }\Phi
)^{2}]^{\diamond }}{2(\ _{2}^{\shortmid }\Lambda )}\mbox{ and }\
_{2}^{\shortmid }\Psi =|\ _{2}^{\shortmid }\Lambda |^{-1/2}\sqrt{|\int
dy^{3}(\ _{2}^{\shortmid }\widehat{\Upsilon })[(\ _{2}^{\shortmid }\Phi
)^{2}]^{\diamond }|}.
\end{equation*}%
Introducing the value of $\ _{2}^{\shortmid }\Psi $ in the first equation,
we get an important formula for redefinition of $(\ _{2}^{\shortmid }\Psi
)^{\diamond }$ in terms of $(\ _{2}^{\shortmid }\widehat{\Upsilon },\
_{2}^{\shortmid }\Phi ,\ _{2}^{\shortmid }\Lambda ),$ when
\begin{equation}
\frac{(\ _{2}^{\shortmid }\Psi )^{\diamond }}{\ _{2}^{\shortmid }\widehat{%
\Upsilon }}=\frac{[(\ _{2}^{\shortmid }\Phi )^{2}]^{\diamond }}{2\sqrt{|\
_{2}^{\shortmid }\Lambda \int dy^{3}(\ _{2}^{\shortmid }\widehat{\Upsilon }%
)[(\ _{2}^{\shortmid }\Phi )^{2}]^{\diamond }|}}.  \label{ntransf3}
\end{equation}%
In result, we can express the coefficients $g_{3}$ and $g_{4}$ in (\ref%
{qellc}) (up to certain classes of integration functions depending on
coordinates $x^{i_{1}})$ in some forms with explicit dependence on the shell
cosmological constant $\ _{2}^{\shortmid }\Lambda ,$ when the (effective)
generating source $\ _{2}^{\shortmid }\widehat{\Upsilon }$ is related a new
generating function $\ \ _{2}^{\shortmid }\Phi .$ So, the solutions of (\ref%
{e2a}), see (\ref{g3}) and (\ref{g4}), can be written in two equivalent
forms,
\begin{eqnarray*}
g_{3}[\ _{2}^{\shortmid }\Psi ] &=&-\frac{[(\ _{2}^{\shortmid }\Psi
)^{\diamond }]^{2}}{4(\ _{2}^{\shortmid }\widehat{\Upsilon })^{2}g_{4}[\
_{2}^{\shortmid }\Psi ]}=g_{3}[\ _{2}^{\shortmid }\Phi ]=-\frac{1}{g_{4}[\
_{2}^{\shortmid }\Phi ]}\frac{(\ _{2}^{\shortmid }\Phi )^{2}[(\
_{2}^{\shortmid }\Phi )^{\diamond }]^{2}}{|\ _{2}^{\shortmid }\Lambda \int
dy^{3}(\ _{2}^{\shortmid }\widehat{\Upsilon })[(\ _{2}^{\shortmid }\Phi
)^{2}]^{\diamond }|},\mbox{ where } \\
g_{4}[\ _{2}^{\shortmid }\Psi ] &=&g_{4}^{[0]}-\int d\varphi \frac{\lbrack
(\ _{2}^{\shortmid }\Psi )^{2}]^{\diamond }}{4(\ _{2}^{\shortmid }\widehat{%
\Upsilon })}=g_{4}[\ _{2}^{\shortmid }\Phi ]=g_{4}^{[0]}-\frac{(\
_{2}^{\shortmid }\Phi )^{2}}{4\ _{2}^{\shortmid }\Lambda }.
\end{eqnarray*}%
In their turn, two equivalent nonlinear formulas for $\ _{2}^{\shortmid
}\Psi \lbrack \ _{2}^{\shortmid }\widehat{\Upsilon },\ _{2}^{\shortmid }\Phi
,\ _{2}^{\shortmid }\Lambda ]$ and $g_{a_{2}}[\ _{2}^{\shortmid }\Psi
]=g_{a_{2}}[\ _{2}^{\shortmid }\Phi ],$ can be used for expressing in two
equivalent formulas for the solutions of the equations (\ref{e2b}) and (\ref%
{e2c}) (see also (\ref{gw}) and (\ref{gn})) for the N--adapted coefficients
of the s--metric and N-connections in (\ref{qellc}). In the second case, the
N--connection coefficients can be computed in more explicit forms using
integrals on $d\varphi $ for certain values determined by $(\
_{2}^{\shortmid }\Phi ,\ _{2}^{\shortmid }\Lambda )$ and respective
nonlinear transforms (\ref{ntransf1}), (\ref{ntransf2}) and (\ref{ntransf3}),%
\begin{eqnarray*}
w_{i_{1}}(x^{k_{1}},y^{3}) &=&\frac{\partial _{i_{1}}(\ _{2}^{\shortmid
}\Psi )}{(\ _{2}^{\shortmid }\Psi )^{\diamond }}=\frac{\partial _{i_{1}}[(\
_{2}^{\shortmid }\Psi )^{2}]}{[(\ _{2}^{\shortmid }\Psi )^{2}]^{\diamond }}=%
\frac{\partial _{i_{1}}\ \int dy^{3}(\ _{2}^{\shortmid }\widehat{\Upsilon }%
)\ [(\ _{2}^{\shortmid }\Phi )^{2}]^{\diamond }}{(\ _{2}^{\shortmid }%
\widehat{\Upsilon })\ [(\ _{2}^{\shortmid }\Phi )^{2}]^{\diamond }};%
\mbox{
and } \\
n_{k_{1}}(x^{k_{1}},y^{3}) &=&\ _{1}n_{k_{1}}+\ _{2}n_{k_{1}}\int dy^{3}\
\frac{g_{3}[\ _{2}^{\shortmid }\Phi ]}{|\ g_{4}[\ _{2}^{\shortmid }\Phi
]|^{3/2}} \\
&=&\ _{1}n_{k_{1}}+\ _{2}n_{k_{1}}\int dy^{3}\left( \frac{(\ _{2}^{\shortmid
}\Psi )^{\diamond }}{2\ _{2}^{\shortmid }\widehat{\Upsilon }}\right)
^{2}\left\vert g_{4}^{[0]}(x^{k_{1}})-\int dy^{3}\frac{[(\ _{2}^{\shortmid
}\Psi )^{2}]^{\diamond }}{4(\ _{2}^{\shortmid }\widehat{\Upsilon })}%
\right\vert ^{-5/2} \\
&=&\ _{1}n_{k_{1}}+\ _{2}n_{k_{1}}\int dy^{3}\frac{(\ _{2}^{\shortmid }\Phi
)^{2}[(\ _{2}^{\shortmid }\Phi )^{\diamond }]^{2}}{|\ _{2}^{\shortmid
}\Lambda \int dy^{3}(\ _{2}^{\shortmid }\widehat{\Upsilon })[(\
_{2}^{\shortmid }\Phi )^{2}]^{\diamond }|}\left\vert g_{4}^{[0]}-\frac{(\
_{2}^{\shortmid }\Phi )^{2}}{4\ _{2}^{\shortmid }\Lambda }\right\vert
^{-5/2}.
\end{eqnarray*}%
On proof of such formulas, see (\ref{g4}).

For the shell $s=4$ (we omit a similar proof for $s=3)$ $(\ _{4}^{\shortmid
}\Psi ,\ _{4}^{\shortmid }\widehat{\Upsilon })\leftrightarrow (\
_{4}^{\shortmid }\Phi ,\ _{4}^{\shortmid }\Lambda =const)$ following
\begin{eqnarray}
\frac{\lbrack (\ _{4}^{\shortmid }\Psi )^{2}]^{\ast }}{\ _{4}^{\shortmid }%
\widehat{\Upsilon }} &=&\frac{[(\ _{4}^{\shortmid }\Phi )^{2}]^{\ast }}{\
_{4}^{\shortmid }\Lambda },\mbox{ which can be inteegrated
as }  \label{ntransf14} \\
(\ _{4}^{\shortmid }\Phi )^{2} &=&\ _{4}^{\shortmid }\widetilde{\Lambda }%
\int dE(\ _{4}^{\shortmid }\widehat{\Upsilon })^{-1}[(\ _{4}^{\shortmid
}\Psi )^{2}]^{\ast }\mbox{ and/or }(\ _{4}^{\shortmid }\Psi )^{2}=(\
_{4}^{\shortmid }\Lambda )^{-1}\int dE(\ _{4}^{\shortmid }\widehat{\Upsilon }%
)[(\ _{4}^{\shortmid }\Phi )^{2}]^{\ast }.  \notag
\end{eqnarray}%
These formulas allow us to express the coefficients $g^{7}$ and $g^{8}$ and
related shell N-connection coefficients in (\ref{qellc}) (up to certain
classes of integration functions depending on coordinates $x^{i_{3}}),$ in
some forms with explicit dependence on the (effective) cosmological constant
$\ _{4}^{\shortmid }\Lambda ,$ when the (effective) generating source $\
_{4}^{\shortmid }\widehat{\Upsilon }$ is encoded into a new generating
function $\ _{4}^{\shortmid }\widetilde{\Psi }.$ Such solutions of the
system (\ref{eq4a})-(\ref{eq4c}) can be written in two equivalent forms by
introducing formulas (\ref{ntransf14}) into (\ref{gnw78}).

Summarising for shells $s=2,3,4$ above formulas, we prove two important
results:

\begin{theorem}
\textsf{[nonlinear symmetries of generating functions and sources] } \label%
{nsymgfs}Any quasi-stationary solution (\ref{qeltors}) of the generalized
Einstein equations for the canonical s-connection on the cotangent bundle
defined by the conditions of Theorem \ref{gensoltorsion} possess such
important nonlinear symmetries{\small
\begin{eqnarray}
s=2: &&\frac{[(\ _{2}^{\shortmid }\Psi )^{2}]^{\diamond }}{\ _{2}^{\shortmid
}\widehat{\Upsilon }}=\frac{[(\ _{2}^{\shortmid }\Phi )^{2}]^{\diamond }}{\
_{2}^{\shortmid }\Lambda },  \label{nonltransf} \\
&&\mbox{ i.e. }\ (\ _{2}^{\shortmid }\Phi )^{2}=\ _{2}^{\shortmid }\Lambda
\int dy^{3}(\ _{2}^{\shortmid }\widehat{\Upsilon })^{-1}[(\ _{2}^{\shortmid
}\Psi )^{2}]^{\diamond }\mbox{ and/or }(\ _{2}^{\shortmid }\Psi )^{2}=(\
_{2}^{\shortmid }\Lambda )^{-1}\int d^{3}(\ _{2}^{\shortmid }\widehat{%
\Upsilon })[(\ _{2}^{\shortmid }\Phi )^{2}]^{\diamond }.  \notag
\end{eqnarray}%
\begin{eqnarray}
s=3: &&\frac{\partial ^{6}[(\ _{3}^{\shortmid }\Psi )^{2}]}{\
_{3}^{\shortmid }\widehat{\Upsilon }}=\frac{\partial ^{6}[(\ _{3}^{\shortmid
}\Phi )^{2}]}{\ _{3}^{\shortmid }\Lambda },  \notag \\
&&\mbox{ i.e. }\ (\ _{3}^{\shortmid }\Phi )^{2}=\ _{3}^{\shortmid }\Lambda
\int dp_{6}(\ _{3}^{\shortmid }\widehat{\Upsilon })^{-1}[(\ _{3}^{\shortmid
}\Psi )^{2}]\mbox{ and/or }(\ _{3}^{\shortmid }\Psi )^{2}=(\ _{3}^{\shortmid
}\Lambda )^{-1}\int dp_{6}(\ _{3}^{\shortmid }\widehat{\Upsilon })[(\
_{3}^{\shortmid }\Phi )^{2}].  \notag
\end{eqnarray}%
} {\small
\begin{eqnarray}
s=4: &&\frac{[(\ _{4}^{\shortmid }\Psi )^{2}]^{\ast }}{\ _{4}^{\shortmid }%
\widehat{\Upsilon }}=\frac{[(\ _{4}^{\shortmid }\Phi )^{2}]^{\ast }}{\
_{4}^{\shortmid }\Lambda },  \notag \\
&&\mbox{ i.e. }\ (\ _{4}^{\shortmid }\Phi )^{2}=\ _{4}^{\shortmid }\Lambda
\int dE(\ _{4}^{\shortmid }\widehat{\Upsilon })^{-1}[(\ _{4}^{\shortmid
}\Psi )^{2}]^{\ast }\mbox{ and/or }(\ _{4}^{\shortmid }\Psi )^{2}=(\
_{4}^{\shortmid }\Lambda )^{-1}\int dE(\ _{4}^{\shortmid }\widehat{\Upsilon }%
)[(\ _{4}^{\shortmid }\Phi )^{2}]^{\ast }.  \notag
\end{eqnarray}%
}
\end{theorem}

We can subject the \ sources and cosmological constants in (\ref{qeltors})
and (\ref{nonltransf}) to additive splitting conditions following formulas (%
\ref{totaldiadsourcd}) and Assumption \ref{assumpt4}, $\ $
\begin{equation}
\ _{\shortmid }\widehat{\Upsilon }_{\alpha _{s}\beta _{s}}=\ _{\shortmid
}^{\phi }\Upsilon _{\alpha _{s}\beta _{s}}+\ \ _{\shortmid }^{e}\Upsilon
_{\alpha _{s}\beta _{s}};\ \mathbf{\ }\ _{\shortmid }\Lambda =\ _{\shortmid
}^{\phi }\Lambda +\ \ \ _{\shortmid }^{e}\Lambda \mbox{ and/or }\ \ \
_{\shortmid s}\Lambda =\ \ _{\shortmid s}^{\phi }\Lambda +\ \ _{\shortmid
s}^{e}\Lambda .  \label{sheelsplitcan}
\end{equation}

\begin{consequence}
\label{offdiagcosmc}\textsf{[off-diagonal solutions with effective shell
cosmological constants] } The quadratic element for off-diagonal solutions
with Killing symmetry on $p_{7}$ determined by formula (\ref{qeltors}) under
conditions of Theorem \ref{gensoltorsion} can be written in equivalent form
encoding the nonlinear symmetries (\ref{nonltransf}) and respective
effective shell cosmological constants,{\small
\begin{eqnarray}
ds^{2} &=&g_{\alpha _{s}\beta _{s}}(x^{k},y^{3},p_{a_{3}},p_{a_{4}},\ \
_{s}^{\shortmid }\Phi ,_{\shortmid s}\Lambda )du^{\alpha _{s}}du^{\beta
_{s}}=e^{\psi (x^{k_{1}})}[(dx^{1})^{2}+(dx^{2})^{2}]  \label{offdiagcosmcsh}
\\
&&-\frac{(\ _{2}^{\shortmid }\Phi )^{2}[(\ _{2}^{\shortmid }\Phi )^{\diamond
}]^{2}}{|\ _{2}^{\shortmid }\Lambda \int dy^{3}(\ _{2}^{\shortmid }\widehat{%
\Upsilon })[(\ _{2}^{\shortmid }\Phi )^{2}]^{\diamond }|[g_{4}^{[0]}-(\
_{2}^{\shortmid }\Phi )^{2}/4\ _{2}^{\shortmid }\Lambda ]}\{dy^{3}+\frac{%
\partial _{i_{1}}\ \int dy^{3}(\ _{2}^{\shortmid }\widehat{\Upsilon })\ [(\
_{2}^{\shortmid }\Phi )^{2}]^{\diamond }}{(\ _{2}^{\shortmid }\widehat{%
\Upsilon })\ [(\ _{2}^{\shortmid }\Phi )^{2}]^{\diamond }}dx^{i_{1}}\}^{2}-
\notag \\
&&\{g_{4}^{[0]}-\frac{(\ _{2}^{\shortmid }\Phi )^{2}}{4\ _{2}^{\shortmid
}\Lambda }\}\{dt+[\ _{1}n_{k_{1}}+\ _{2}n_{k_{1}}\int dy^{3}\frac{(\
_{2}^{\shortmid }\Phi )^{2}[(\ _{2}^{\shortmid }\Phi )^{\diamond }]^{2}}{|\
_{2}^{\shortmid }\Lambda \int dy^{3}(\ _{2}^{\shortmid }\widehat{\Upsilon }%
)[(\ _{2}^{\shortmid }\Phi )^{2}]^{\diamond }|[g_{4}^{[0]}-(\
_{2}^{\shortmid }\Phi )^{2}/4\ _{2}^{\shortmid }\Lambda ]^{5/2}}]\}+  \notag
\end{eqnarray}%
\begin{eqnarray*}
&&\{g_{5}^{[0]}-\frac{(\ _{3}^{\shortmid }\Phi )^{2}}{4\ _{3}^{\shortmid
}\Lambda }\}\{dp_{5}+[\ _{1}n_{k_{2}}+\ _{2}n_{k_{2}}\int dp_{6}\frac{(\
_{3}^{\shortmid }\Phi )^{2}[\partial ^{6}(\ _{3}^{\shortmid }\Phi )]^{2}}{|\
_{3}^{\shortmid }\Lambda \int dp_{6}\ (\ _{3}^{\shortmid }\widehat{\Upsilon }%
)\partial ^{6}[(\ _{3}^{\shortmid }\Phi )^{2}]|[g_{5}^{[0]}-(\
_{3}^{\shortmid }\Phi )^{2}/4\ _{3}^{\shortmid }\Lambda ]^{5/2}}%
]dx^{k_{2}}\}- \\
&&\frac{(\ _{3}^{\shortmid }\Phi )^{2}[\partial ^{6}(\ _{3}^{\shortmid }\Phi
)]^{2}}{|\ _{2}^{\shortmid }\Lambda \int dp_{6}(\ _{3}^{\shortmid }\widehat{%
\Upsilon })\ \partial ^{6}[(\ _{3}^{\shortmid }\Phi )^{2}]|\ [g_{5}^{[0]}-(\
_{3}^{\shortmid }\Phi )^{2}/4\ _{3}^{\shortmid }\Lambda ]}\{dp_{6}+\frac{%
\partial _{i_{2}}\ \int dp_{6}(\ _{3}^{\shortmid }\widehat{\Upsilon })\
\partial ^{6}[(\ _{3}^{\shortmid }\Phi )^{2}]}{(\ _{3}^{\shortmid }\widehat{%
\Upsilon })\ \partial ^{6}[(\ _{3}^{\shortmid }\Phi )^{2}]}dx^{i_{2}}\}^{2}+
\\
&&\{g_{7}^{[0]}-\frac{(\ _{4}^{\shortmid }\Phi )^{2}}{4\ _{4}^{\shortmid
}\Lambda }\}\{dp_{7}+[\ _{1}n_{k_{3}}+\ _{2}n_{k_{3}}\int dE\frac{(\
_{4}^{\shortmid }\Phi )^{2}[(\ _{4}^{\shortmid }\Phi )^{\ast }]^{2}}{|\
_{2}^{\shortmid }\Lambda \int dE(\ _{4}^{\shortmid }\widehat{\Upsilon })[(\
_{4}^{\shortmid }\Phi )^{2}]^{\ast }|\ [g_{7}^{[0]}-(\ _{4}^{\shortmid }\Phi
)^{2}/4\ _{4}^{\shortmid }\Lambda ]^{5/2}}]dx^{k_{3}}\}- \\
&&\frac{(\ _{4}^{\shortmid }\Phi )^{2}[(\ _{4}^{\shortmid }\Phi )^{\ast
}]^{2}}{|\ _{4}^{\shortmid }\Lambda \int dE(\ _{4}^{\shortmid }\widehat{%
\Upsilon })[(\ _{4}^{\shortmid }\Phi )^{2}]^{\diamond }|[g_{7}^{[0]}-(\
_{4}^{\shortmid }\Phi )^{2}/4\ _{4}^{\shortmid }\Lambda ]}\{dE+\frac{%
\partial _{i_{3}}\ \int dE(\ _{4}^{\shortmid }\widehat{\Upsilon })\ [(\
_{4}^{\shortmid }\Phi )^{2}]^{\ast }}{(\ _{4}^{\shortmid }\widehat{\Upsilon }%
)[(\ _{4}^{\shortmid }\Phi )^{2}]^{\ast }}dx^{i_{3}}\}^{2},
\end{eqnarray*}%
for indices: $%
i_{1},j_{1},k_{1},...=1,2;i_{2},j_{2},k_{2},...=1,2,3,4;i_{3},j_{3},k_{3},...=1,2,...6;y^{3}=\varphi ,y^{4}=t,p_{8}=E;
$ and
\begin{eqnarray*}
&&\mbox{generating functions: }\psi (x^{k_{1}});\ _{2}^{\shortmid }\Phi
(x^{k_{1}}y^{3});\ _{3}^{\shortmid }\Phi (x^{k_{2}},p_{6});\ _{4}^{\shortmid
}\Phi (x^{k_{3}},E); \\
&&\mbox{generating sources:}\ _{1}^{\shortmid }\widehat{\Upsilon }%
(x^{k_{1}})=\ _{\shortmid 1}^{\phi }\Upsilon (x^{k_{1}})+\ \ _{\shortmid
1}^{e}\Upsilon (x^{k_{1}});\ _{2}^{\shortmid }\widehat{\Upsilon }%
(x^{k_{1}},y^{3})=\ _{\shortmid 2}^{\phi }\Upsilon (x^{k_{1}},y^{3})+\ \
_{\shortmid 2}^{e}\Upsilon (x^{k_{1}},y^{3}); \\
&&\ _{3}^{\shortmid }\widehat{\Upsilon }(x^{k_{2}},p_{6})=\ _{\shortmid
3}^{\phi }\Upsilon (x^{k_{2}},p_{6})+\ \ _{\shortmid 3}^{e}\Upsilon
(x^{k_{2}},p_{6});\ \ _{4}^{\shortmid }\widehat{\Upsilon }(x^{k_{3}},E)=\
_{\shortmid 4}^{\phi }\Upsilon (x^{k_{3}},E)+\ \ _{\shortmid 4}^{e}\Upsilon
(x^{k_{3}},E); \\
&&\mbox{integr. functions:}g_{4}^{[0]}(x^{k_{1}}),\
_{1}n_{k_{1}}(x^{j_{1}}),\ _{2}n_{k_{1}}(x^{j_{1}});g_{5}^{[0]}(x^{k_{2}}),\
_{1}n_{k_{2}}(x^{j_{2}}),\ _{2}n_{k_{2}}(x^{j_{2}});g_{7}^{[0]}(x^{j_{3}}),\
_{1}n_{k_{3}}(x^{j_{3}}),\ _{2}n_{k_{3}}(x^{j_{3}}).
\end{eqnarray*}%
}
\end{consequence}

It should be emphasized that, in general, sources $\ _{s}^{\shortmid }%
\widehat{\Upsilon }=\ _{\shortmid s}^{\phi }\Upsilon +\ \ _{\shortmid
s}^{e}\Upsilon $ are not eliminated completely from (\ref{offdiagcosmcsh})
even we redefined the generating functions to $(\ _{s}^{\shortmid }\Phi ,\
_{s}^{\shortmid }\Lambda ).$ Nevertheless, we can "absorb" completely such $%
\ _{s}^{\shortmid }\widehat{\Upsilon }$ into respective $(\ _{s}^{\shortmid
}\Phi ,\ _{s}^{\shortmid }\Lambda )$ for nonholonomic configurations with
respective shell sources $\ _{1}^{\shortmid }\widehat{\Upsilon }%
(x^{k_{1}}),\ _{2}^{\shortmid }\widehat{\Upsilon }(x^{k_{1}}),\
_{3}^{\shortmid }\widehat{\Upsilon }(x^{k_{2}})$ and $\ _{4}^{\shortmid }%
\widehat{\Upsilon }(x^{k_{3}})$ which depend only on respective shell
horizontal coordinates but not on respective (co) fiber "vertical"
coordinates. In such cases can integrate in explicit form on $dy^{3},dp_{7}$
and $dE$ and reduce multiples with $\ _{s}^{\shortmid }\widehat{\Upsilon }$
or to absorb such values into integration constants.

Finally, we note that effective cosmological constants $\ _{s}^{\shortmid
}\Lambda $ can be introduced for redefined generating functions for
LC-configurations and quasi-stationary s-metrics (\ref{qellc}).

\section{Nonholonomic Deformations into Exact and Parametric Solutions}

\label{s5}

There are 5 goals in this section:

\begin{enumerate}
\item We show that certain coefficients of s-metrics can be considered as shell generating functions which is important for constructing new classes of exact and parametric solutions. In explicit form, the solutions will be found for MGTs with MDRs with modified Einstein equations (\ref{meinsteqtbcand}) for quasi-stationary ansatz (\ref{ansatz1}) for  d-metrics.

\item It is provided a geometric procedure with conventional gravitational
polarization functions for nonholonomic deformations of arbitrary s-metrics
into certain classes of exact solutions.

\item We formulate a small parameter scheme for off-diagonal deformations
and diagonalization procedures for generating new classes of exact solutions.

\item There are studied three examples of vacuum phase space generic
off-diagonal metrics for canonical and Levi-Civita configurations.

\item There are provided main formulas for transforming exact and parametric
solutions on cotangent Lorentz bundles into formulas for Finsler like and
Hamilton variables and, inversely, when nonholonomic deformations with
certain classes of generating functions transform Einstein-Hamilton
configurations into quasi-stationary phase ones.
\end{enumerate}

\subsection{Coefficients of s-metrics as generation functions}

We can consider certain N-adapted coefficients of s-metrics as shell
generating functions following this

\begin{corollary}
\label{corolgenfg}\textsf{[some s-metric coefficients as generating
functions] } For an exact solution (\ref{qeltors}) and/or (\ref%
{offdiagcosmcsh}), we can consider as respective shell generating functions
the s-connection coefficients $%
g_{4}(x^{i_{1}},y^{3}),g^{5}(x^{i_{1}},y^{a_{1}},p_{6}),$\newline
$g^{7}(x^{i_{1}},y^{a_{1}},p_{a_{2}},E)$ and redefine correspondingly the
generating functions following formulas:%
\begin{equation*}
\begin{array}{ccccc}
\mbox{ shell }s=2: &  & [(\ _{2}^{\shortmid }\Psi )^{2}]^{\diamond }=-\int
dy^{3}(\ _{2}^{\shortmid }\widehat{\Upsilon })g_{4}^{\diamond }, &  & (\
_{2}^{\shortmid }\Phi )^{2}=-4\ _{2}^{\shortmid }\Lambda g_{4}; \\
&  &  &  &  \\
\mbox{ shell }s=3: &  & \partial ^{6}[(\ _{3}^{\shortmid }\Psi )^{2}]=-\int
dp_{6}(\ _{3}^{\shortmid }\widehat{\Upsilon })\ \partial ^{6}g^{5}, &  & (\
_{3}^{\shortmid }\Phi )^{2}=-4\ _{3}^{\shortmid }\Lambda g^{6}; \\
&  &  &  &  \\
\mbox{ shell }s=4: &  & [(\ _{4}^{\shortmid }\Psi )^{2}]^{\ast }=-\int dE(\
_{4}^{\shortmid }\widehat{\Upsilon })\ (g^{7})^{\ast }, &  & (\
_{8}^{\shortmid }\Phi )^{2}=-4\ _{3}^{\shortmid }\Lambda g^{8}.%
\end{array}%
\end{equation*}
\end{corollary}

\begin{proof}
For $s=2,$ we take the partial derivative on $y^{3}$ of formula (\ref{g4})
and obtain $g_{4}^{\diamond }=-[(\ _{2}^{\shortmid }\Psi )^{2}]^{\diamond
}/4(\ _{2}^{\shortmid }\widehat{\Upsilon }).$ Prescribing $g_{4}$ and $\
_{2}^{\shortmid }\widehat{\Upsilon },$ we can compute up to certain
integration functions a $\ _{2}^{\shortmid }\Psi $ from $[(\ _{2}^{\shortmid
}\Psi )^{2}]^{\diamond }=\int dy^{3}(\ _{2}^{\shortmid }\widehat{\Upsilon }%
)g_{4}^{\diamond }.$ Using nonlinear symmetries (\ref{nonltransf}), we can
repeat similar considerations and obtain $(\ _{2}^{\shortmid }\Phi )^{2}=-4\
_{2}^{\shortmid }\Lambda g_{4}.$ We omit similar considerations for the
shells $s=3,4$ which should be repeated for corresponding sets of
coordinates and shell generating functions, generating sources, and
effective shell cosmological constants. $\square $\vskip5pt
\end{proof}

In result, the exact solutions (\ref{qeltors}) and/or (\ref{offdiagcosmcsh})
can be rewritten equivalently in terms of generating data $%
(g_{4},g^{5},g^{7};\ _{s}^{\shortmid }\Lambda )$ encoding via nonlinear
symmetries (\ref{nonltransf}) information on $\ _{s}^{\shortmid }\widehat{%
\Upsilon }.$

\begin{consequence}
\label{offsolgenerfg}\textsf{[off-diagonal solutions with s-coefficient
generating functions and effective cosmological constants] } The quadratic
element (\ref{offdiagcosmcsh}) for off-diagonal solutions with Killing
symmetry on $p_{7}$ determined by formula (\ref{qeltors}) under conditions
of Theorem \ref{gensoltorsion}, Consequence \ref{offdiagcosmc} and Corollary %
\ref{corolgenfg} can be written in an equivalent form generated by data $%
(g_{4},g^{5},g^{7};\ _{s}^{\shortmid }\Lambda )$ and respective nonlinear
symmetries involving generating sources $\ _{s}^{\shortmid }\widehat{%
\Upsilon },$
\begin{eqnarray}
&&ds^{2}=g_{\alpha _{s}\beta
_{s}}(x^{k},y^{3},p_{a_{3}},p_{a_{4}};g_{4},g^{5},g^{7},_{\shortmid
s}\Lambda ;\ _{s}^{\shortmid }\widehat{\Upsilon })du^{\alpha _{s}}du^{\beta
_{s}}=e^{\psi (x^{k_{1}})}[(dx^{1})^{2}+(dx^{2})^{2}]
\label{offdsolgenfgcosmc} \\
&&-\frac{(g_{4}^{\diamond })^{2}}{|\int dy^{3}[(\ _{2}^{\shortmid }\widehat{%
\Upsilon })g_{4}]^{\diamond }|\ g_{4}}\{dy^{3}+\frac{\partial _{i_{1}}[\int
dy^{3}(\ _{2}^{\shortmid }\widehat{\Upsilon })\ g_{4}^{\diamond }]}{(\
_{2}^{\shortmid }\widehat{\Upsilon })\ g_{4}^{\diamond }}dx^{i_{1}}\}^{2}+
\notag \\
&&g_{4}\{dt+[\ _{1}n_{k_{1}}+\ _{2}n_{k_{1}}\int dy^{3}\frac{%
(g_{4}^{\diamond })^{2}}{|\int dy^{3}[(\ _{2}^{\shortmid }\widehat{\Upsilon }%
)g_{4}]^{\diamond }|\ [g_{4}]^{5/2}}]dx^{\acute{k}_{1}}\}+  \notag
\end{eqnarray}%
\begin{eqnarray*}
&&g^{5}\{dp_{5}+[\ _{1}n_{k_{2}}+\ _{2}n_{k_{2}}\int dp_{6}\frac{[\partial
^{6}(g^{5})]^{2}}{|\int dp_{6}\ \partial ^{6}[(\ _{3}^{\shortmid }\widehat{%
\Upsilon })g^{5}]|\ [g^{5}]^{5/2}}]dx^{k_{2}}\}- \\
&&\frac{[\partial ^{6}(g^{5})]^{2}}{|\int dp_{6}\ \partial ^{6}[(\
_{3}^{\shortmid }\widehat{\Upsilon })g^{5}]\ |\ g^{5}}\{dp_{6}+\frac{%
\partial _{i_{2}}[\int dp_{6}(\ _{3}^{\shortmid }\widehat{\Upsilon })\
\partial ^{6}(g^{5})]}{(\ _{3}^{\shortmid }\widehat{\Upsilon })\partial
^{6}(g^{5})}dx^{i_{2}}\}^{2}+ \\
&&g^{7}\{dp_{7}+[\ _{1}n_{k_{3}}+\ _{2}n_{k_{3}}\int dE\frac{[(g^{7})^{\ast
}]^{2}}{|\int dE\ [(\ _{4}^{\shortmid }\widehat{\Upsilon })g^{7}]^{\ast }|\
[g^{7}]^{5/2}}]dx^{k_{3}}\}- \\
&&\frac{[(g^{7})^{\ast }]^{2}}{|\int dE\ [(\ _{4}^{\shortmid }\widehat{%
\Upsilon })g^{7}]^{\ast }\ |\ g^{7}}\{dE+\frac{\partial _{i_{3}}[\int dE(\
_{4}^{\shortmid }\widehat{\Upsilon })\ (g^{7})^{\ast }]}{(\ _{4}^{\shortmid }%
\widehat{\Upsilon })(g^{7})^{\ast }}dx^{i_{3}}\}^{2},
\end{eqnarray*}%
where the signs and integration functions/constants have to be chosen in
order to get compatibility with certain experimental/ observational data or
corrections from generalized theories.
\end{consequence}

It should be noted that we can consider in explicit form nonsingular
solutions of type (\ref{offdsolgenfgcosmc}) if the values $\ _{s}^{\shortmid
}\widehat{\Upsilon }$ and/or $\ _{s}^{\shortmid }\Lambda $ are not zero. We
can apply the AFDM for vacuum phase space or Lorentz manifold configurations
but the geometric constructions are different (see section \ref%
{ssvacuumfc}\ and, for instance, the constructions related to formulas
(60)-(63) in section 2.3.6 of \cite{vacaruepjc17}).

\begin{remark}
\textsf{[off-diagonal solutions with effective shell cosmological constants
and zero torsion] } We can more special conditions stated by Consequence \ref%
{ofdiagztsol} which allow to generate LC-configurations using some data
\begin{equation*}
\ \check{g}_{4}(x^{i_{1}},y^{3}),\check{g}^{5}(x^{i_{1}},y^{a_{1}},p_{6}),%
\check{g}^{7}(x^{i_{1}},y^{a_{1}},p_{a_{2}},E)
\end{equation*}%
when the generating functions are of type (\ref{expconda}) (that why we put
inverse hat labels emphasizing that certain additional integrability
conditions are imposed on generating functions) and redefined following
formulas:%
\begin{equation*}
\begin{array}{ccccc}
\mbox{ shell }s=2: &  & [(\ _{2}^{\shortmid }\check{\Psi})^{2}]^{\diamond
}=\int dy^{3}(\ _{2}^{\shortmid }\widehat{\Upsilon })\check{g}_{4}^{\diamond
}, &  & (\ _{2}^{\shortmid }\check{\Phi})^{2}=-4\ _{2}^{\shortmid }\Lambda
\check{g}_{4}; \\
&  &  &  &  \\
\mbox{ shell }s=3: &  & \partial ^{6}[(\ _{3}^{\shortmid }\check{\Psi}%
)^{2}]=\int dp_{6}(\ _{3}^{\shortmid }\widehat{\Upsilon })\ \partial ^{6}%
\check{g}^{5}, &  & (\ _{3}^{\shortmid }\check{\Phi})^{2}=-4\
_{3}^{\shortmid }\Lambda \check{g}^{6}; \\
&  &  &  &  \\
\mbox{ shell }s=4: &  & [(\ _{4}^{\shortmid }\check{\Psi})^{2}]^{\ast }=\int
dE(\ _{4}^{\shortmid }\widehat{\Upsilon })\ (\check{g}^{7})^{\ast }, &  & (\
_{8}^{\shortmid }\check{\Phi})^{2}=-4\ _{3}^{\shortmid }\Lambda \check{g}%
^{8}.%
\end{array}%
\end{equation*}%
Corresponding s-metrics (\ref{offdsolgenfgcosmc}) transform into
\begin{eqnarray}
ds_{LC}^{2} &=&g_{\alpha _{s}\beta
_{s}}(x^{k},y^{3},p_{a_{3}},p_{a_{4}};g_{4},g^{5},g^{7},_{\shortmid
s}\Lambda ;\ _{s}^{\shortmid }\widehat{\Upsilon })du^{\alpha _{s}}du^{\beta
_{s}}=e^{\psi (x^{k_{1}})}[(dx^{1})^{2}+(dx^{2})^{2}]  \notag \\
&&-\frac{(g_{4}^{\diamond })^{2}}{|\int dy^{3}(\ _{2}^{\shortmid }\widehat{%
\Upsilon })[g_{4}]^{\diamond }|\ g_{4}}\{dy^{3}+\partial _{i_{1}}(\
_{2}^{\shortmid }\check{A})dx^{i_{1}}\}^{2}+g_{4}\{dt+\partial _{i_{1}}[\
^{2}n(x^{k_{1}})]dx^{i_{1}}\}+  \label{offidiagcosmconstlc} \\
&&g^{5}\{dp_{5}+\partial _{i_{2}}[\ ^{2}n(x^{k_{2}})]dx^{i_{2}}\}-\frac{%
[\partial ^{6}(g^{5})]^{2}}{|\int dp_{6}\ (\ _{3}^{\shortmid }\widehat{%
\Upsilon })\partial ^{6}[g^{5}]\ |\ g^{5}}\{dp_{6}+\partial _{i_{2}}(\
_{3}^{\shortmid }\check{A})dx^{i_{2}}\}^{2}+  \notag \\
&&g^{7}\{dp_{7}+\partial _{i_{3}}[\ ^{3}n(x^{k_{3}})]dx^{i_{3}}\}-\frac{%
[(g^{7})^{\ast }]^{2}}{|\int dE\ (\ _{4}^{\shortmid }\widehat{\Upsilon }%
)[g^{7}]^{\ast }\ |\ g^{7}}\{dE+\partial _{i_{2}}(\ _{4}^{\shortmid }\check{A%
})dx^{i_{3}}\}^{2}.  \notag
\end{eqnarray}
\end{remark}

Finally, we note that effective shell sources $\ _{s}^{\shortmid }\widehat{%
\Upsilon }$ can be absorbed into certain integration functions of (\ref%
{offdsolgenfgcosmc}) if such values do not depend on vertical shell
coordinates on respective shells. The sources $\ _{s}^{\shortmid }\widehat{%
\Upsilon }$ are encoded correspondingly in $\ _{s}^{\shortmid }\check{A}$
for (\ref{offidiagcosmconstlc}).

\subsection{Polarization functions for off-diagonal prime and target
s-metrics}

Let us consider a dual Lorentz bundle enabled with nonholonomic dyadic
structure, $_{s}\mathbf{T}^{\ast }\mathbf{V,}$ and a \textbf{prime } metric $%
\ ^{\shortmid }\mathbf{\mathring{g}}$ structure which can be written
equivalently in an off-diagonal form and/or (\ref{offds}) as a s-metric (\ref%
{dmcts}), see also the parametrization (\ref{dm2and2}),
\begin{eqnarray}
\ \ ^{\shortmid }\mathbf{\mathring{g}} &=&\ _{s}^{\shortmid }\mathbf{%
\mathring{g}}=\ ^{\shortmid }\mathring{g}_{\alpha _{s}\beta
_{s}}(x^{i_{s}},p_{a_{s}})d\ ^{\shortmid }u^{\alpha _{s}}\otimes d\
^{\shortmid }u^{\beta _{s}}=\ ^{\shortmid }\mathbf{\mathring{g}}_{\alpha
_{s}\beta _{s}}(\ _{s}^{\shortmid }u)\ ^{\shortmid }\mathbf{\mathbf{%
\mathring{e}}}^{\alpha _{s}}\mathbf{\otimes \ ^{\shortmid }\mathbf{\mathring{%
e}}}^{\beta _{s}}  \label{primedm} \\
&=&\ ^{\shortmid }\mathring{g}_{i_{s}j_{s}}(x^{k_{s}})e^{i_{s}}\otimes
e^{j_{s}}+\ ^{\shortmid }\mathbf{\mathring{g}}%
^{a_{s}b_{s}}(x^{i_{s}},p_{a_{s}})\ ^{\shortmid }\mathbf{\mathring{e}}%
_{a_{s}}\otimes \ ^{\shortmid }\mathbf{\mathring{e}}_{b_{s}},\mbox{ for }
\notag \\
\ ^{\shortmid }\mathbf{\mathring{e}}_{\alpha _{s}} &=&(\ ^{\shortmid }%
\mathbf{\mathring{e}}_{i_{s}}=\partial _{i_{s}}-\ ^{\shortmid }\mathring{N}%
_{i_{s}}^{b_{s}}(\ ^{\shortmid }u)\partial _{b_{s}},\ \ ^{\shortmid }{e}%
_{a_{s}}=\partial _{a_{s}})\mbox{ and }\ ^{\shortmid }\mathbf{\mathring{e}}%
^{\alpha _{s}}=(dx^{i_{s}},\mathbf{\mathring{e}}^{a_{s}}=dy^{a_{s}}+\
^{\shortmid }\mathring{N}_{i_{s}}^{a_{s}}(\ ^{\shortmid }u)dx^{i_{s}}).
\notag
\end{eqnarray}
Hereafter we shall label prime metrics and related geometric objects like
connections, frames etc. with a small circle on the left/right/up of
corresponding symbols. In general, a prime s-metric $\ ^{\shortmid }\mathbf{%
\mathring{g}}$ (\ref{primedm}) my not be a solution of certain gravitational
field equations in a MGT or GR with phase space extension. To study, for
instance, generalizations of physically important solutions in GR to $_{s}%
\mathbf{T}^{\ast }\mathbf{V}$ we can consider trivial embedding, for
instance, of black hole metrics into linear quadratic elements \ (\ref{lqed}%
). For nontrivial MDRs with an indicator $\varpi (x^{i},E,\overrightarrow{%
\mathbf{p}},m;\ell _{P})$ (\ref{mdrg}), we obtain nonholonomic deformations
to nonlinear quadratic elements determined by \textbf{target} s-metrics of
type $\ \ ^{\shortmid }\mathbf{g}=\ _{s}^{\shortmid }\mathbf{g}$ (\ref{dmcts}%
).

In this section, we consider general nonholonomic transforms of a prime
s-metric, $\ _{s}^{\shortmid }\mathbf{\mathring{g},}$ into a target one, $\
_{s}^{\shortmid }\mathbf{g}$ which is positively a solution of modified
Einstein equations (\ref{meinsteqtbcand}) for canonical s-connections (in
particular, defining quasi-stationary phase configurations).

\begin{definition}
$\eta $-polarization (or gravitational polarization) functions,
\begin{equation*}
\ _{s}^{\shortmid }\mathbf{\mathring{g}}\rightarrow \ _{s}^{\shortmid }%
\mathbf{g}=[\ ^{\shortmid }g_{\alpha _{s}}=\ ^{\shortmid }\eta _{\alpha
_{s}}\ ^{\shortmid }\mathring{g}_{\alpha _{s}},\ ^{\shortmid
}N_{i_{s-1}}^{a_{s}}=\ \ ^{\shortmid }\eta _{i_{s-1}}^{a_{s}}\ ^{\shortmid }%
\mathring{N}_{i_{s-1}}^{a_{s}}],
\end{equation*}%
are defined as transforms of a prime s-metric (\ref{primedm}) into certain
classes of exact/ parametric solutions of generalized Einstein equations
with MDRs.
\end{definition}

We note that in above formulas it is not considered a summation on repeating
indices because they are not 'up-low' type.

\begin{convention}
\label{gravitpolarizations}\textsf{[gravitational polarizations] } For
nonholonomic transforms of prime to target s-metrics, the $\eta $%
-polarizations are defined by formulas
\begin{eqnarray}
\ \ _{s}^{\shortmid }\mathbf{\mathring{g}} &\rightarrow &\ _{s}^{\shortmid }%
\mathbf{g}=\ ^{\shortmid }g_{i_{s}}(x^{k_{s}})dx^{i_{s}}\otimes dx^{i_{s}}+\
^{\shortmid }g_{a_{s}}(x^{i_{s}},p_{b_{s}})\ ^{\shortmid }\mathbf{e}%
^{a_{s}}\otimes \ ^{\shortmid }\mathbf{e}^{a_{s}}  \label{dmpolariz} \\
&=&\ ^{\shortmid }\eta _{i_{k}}(x^{i_{1}},y^{a_{2}},p_{a_{3}},p_{a_{4}})\
^{\shortmid }\mathring{g}%
_{i_{s}}(x^{i_{1}},y^{a_{2}},p_{a_{3}},p_{a_{4}})dx^{i_{s}}\otimes dx^{i_{s}}
\notag \\
&&+\ ^{\shortmid }\eta _{b_{s}}(x^{i_{1}},y^{a_{2}},p_{a_{3}},p_{a_{4}})\
^{\shortmid }\mathring{g}_{b_{s}}(x^{i_{1}},y^{a_{2}},p_{a_{3}},p_{a_{4}})\
^{\shortmid }\mathbf{e}^{b_{s}}[\eta ]\otimes \ ^{\shortmid }\mathbf{e}%
^{b_{s}}[\eta ],  \notag \\
\ ^{\shortmid }\mathbf{e}^{\alpha _{s}}[\eta ] &=&(dx^{i_{s}},\ ^{\shortmid }%
\mathbf{e}^{a_{s}}=dy^{a_{s}}+\ ^{\shortmid }\eta
_{i_{s}}^{a_{s}}(x^{i_{1}},y^{a_{2}},p_{a_{3}},p_{a_{4}})\ ^{\shortmid }%
\mathring{N}%
_{i_{s}}^{a_{s}}(x^{i_{1}},y^{a_{2}},p_{a_{3}},p_{a_{4}})dx^{i_{s}}),  \notag
\end{eqnarray}%
where the s-coefficients and N-elongated dual basis $\ ^{\shortmid }\mathbf{e%
}^{\alpha _{s}}[\eta ]$ is defined by formulas
\begin{eqnarray}
\ ^{\shortmid }g_{i_{1}}(x^{k_{1}}) &=&\ ^{\shortmid }\eta
_{k_{1}}(x^{i_{1}},y^{a_{2}},p_{a_{3}},p_{a_{4}})\ ^{\shortmid }\mathring{g}%
_{k_{1}}(x^{i_{1}},y^{a_{2}},p_{a_{3}},p_{a_{4}}),  \notag \\
\ ^{\shortmid }g_{b_{2}}(x^{i_{1}},y^{3}) &=&\ ^{\shortmid }\eta
_{b_{2}}(x^{i_{1}},y^{a_{2}},p_{a_{3}},p_{a_{4}})\ ^{\shortmid }\mathring{g}%
_{b_{1}}(x^{i_{1}},y^{a_{2}},p_{a_{3}},p_{a_{4}}),  \notag \\
\ ^{\shortmid }g_{a_{3}}(x^{i_{2}},p_{6}) &=&\ ^{\shortmid }\eta
^{a_{3}}(x^{i_{1}},y^{b_{2}},p_{b_{3}},p_{b_{4}})\ ^{\shortmid }\mathring{g}%
^{a_{3}}(x^{i_{1}},y^{b_{2}},p_{b_{3}},p_{b_{4}}),  \notag \\
\ ^{\shortmid }g_{a_{4}}(x^{i_{3}},E) &=&\ ^{\shortmid }\eta
^{a_{4}}(x^{i_{1}},y^{b_{2}},p_{b_{3}},p_{b_{4}})\ ^{\shortmid }\mathring{g}%
^{a_{4}}(x^{i_{1}},y^{a_{2}},p_{a_{3}},p_{a_{4}})  \notag \\
\ ^{\shortmid }N_{i_{1}}^{a_{2}}(x^{k_{1}},y^{3}) &=&\eta
_{i_{1}}^{a_{2}}(x^{i_{1}},y^{b_{2}},p_{b_{3}},p_{b_{4}})\ ^{\shortmid }%
\mathring{N}_{i_{1}}^{a_{2}}(x^{i_{1}},y^{b_{2}},p_{b_{3}},p_{b_{4}}),
\notag \\
\ ^{\shortmid }N_{i_{2}}^{a_{3}}(x^{k_{1}},y^{b_{2}},p_{6}) &=&\ ^{\shortmid
}\eta _{i_{2}a_{3}}(x^{i_{1}},y^{b_{2}},p_{b_{3}},p_{b_{4}})\ ^{\shortmid }%
\mathring{N}_{i_{2}a_{3}}(x^{i_{1}},y^{b_{2}},p_{b_{3}},p_{b_{4}}),  \notag
\\
\ ^{\shortmid }N_{i_{3}}^{a_{4}}(x^{k_{1}},y^{b_{2}},p_{a_{3}},E) &=&\eta
_{i_{3}a_{4}}(x^{i_{1}},y^{b_{2}},p_{b_{3}},p_{b_{4}})\mathring{N}%
_{i_{3}a_{4}}(x^{i_{1}},y^{b_{2}},p_{b_{3}},p_{b_{4}}).  \label{coeftargpol}
\end{eqnarray}
\end{convention}

We note that any of multiples of type in a $\ ^{\shortmid }\eta \ ^{\shortmid }\mathring{g}$ from right sides of (\ref{coeftargpol}) may depend, in principle, on extra shell coordinates but their products are subjected to
the condition that the target s-metrics (with the coefficients in the left sides) are adapted to the shell coordinates ordered form $s=1,2,3,4$. Let us discuss the geometric and physical meaning of $\eta$-coefficients introduced above. They describe nonholonomic deformations of certain prime d-metrics into certain target ones and can be defined in such forms that prime metrics are solutions of certain modified Einstein equations. We use the term "gravitational polarizations" because for $\eta$-deformations with a small parameter, we can generate solutions, for instance, of black hole/ ellipsoid type but with effective polarization of fundamental physical constants. There were considered many examples of such solutions with polarizations determined by noncommutative/ string / massive gravity / geometric flow etc. corrections  \cite{gvvepjc14,svvijmpd14,rajpoot15,gheorghiuap16,bubuianucqg17,vbubuianu17,vacaruepjc17,vmon06,vijtp10}.  The results of this section have been used recently in article \cite{bvap19} on black hole solutions in MGTs with MDRs and Finsler-Hamilton variables.

Following the Convention \ref{gravitpolarizations} for any prescribed prime
s-metric $\ _{s}^{\shortmid }\mathbf{\mathring{g},}$ we can consider as
generating functions a subclass of $\eta $-polarizations $\ ^{\shortmid
}\eta _{4}(x^{i_{1}},y^{3}),\ ^{\shortmid }\eta
^{5}(x^{i_{1}},y^{a_{1}},p_{6}),\ ^{\shortmid }\eta
^{7}(x^{i_{1}},y^{a_{1}},p_{a_{2}},E)$ which should be defined from the
condition that the target s-metric $\ _{s}^{\shortmid }\mathbf{g}$ is a
quasi-stationary solution (in this work) of the canonically modified
Einstein equations (\ref{meinsteqtbcand}). By straightforward computations,
we prove

\begin{corollary}
\label{nqelgravpol}\textsf{[nonlinear quadratic elements in terms of
gravitational polarization generating functions]} Quasi-stationary phase
configurations on cotangent Lorentz bundles are described in terms of $\eta $%
-polarizations by
\begin{eqnarray}
&&ds^{2}=g_{\alpha _{s}\beta
_{s}}(x^{k},y^{3},p_{a_{3}},p_{a_{4}};g_{4},g^{5},g^{7},_{\shortmid
s}\Lambda ;\ _{s}^{\shortmid }\widehat{\Upsilon })du^{\alpha _{s}}du^{\beta
_{s}}=e^{\psi (x^{k_{1}})}[(dx^{1})^{2}+(dx^{2})^{2}]  \label{offdiagpolf} \\
&&-\frac{[(\ ^{\shortmid }\eta _{4}\ ^{\shortmid }\mathring{g}%
_{4})^{\diamond }]^{2}}{|\int dy^{3}(\ _{2}^{\shortmid }\widehat{\Upsilon }%
)(\ ^{\shortmid }\eta _{4}\ ^{\shortmid }\mathring{g}_{4})^{\diamond }|\ (\
^{\shortmid }\eta _{4}\ ^{\shortmid }\mathring{g}_{4})}\{dy^{3}+\frac{%
\partial _{i_{1}}[\int dy^{3}(\ _{2}^{\shortmid }\widehat{\Upsilon })\ (\
^{\shortmid }\eta _{4}\ ^{\shortmid }\mathring{g}_{4})^{\diamond }]}{(\
_{2}^{\shortmid }\widehat{\Upsilon })(\ ^{\shortmid }\eta _{4}\ ^{\shortmid }%
\mathring{g}_{4})^{\diamond }}dx^{i_{1}}\}^{2}+  \notag \\
&&(\ ^{\shortmid }\eta _{4}\ ^{\shortmid }\mathring{g}_{4})\{dt+[\
_{1}n_{k_{1}}+\ _{2}n_{k_{1}}\int dy^{3}\frac{[(\ ^{\shortmid }\eta _{4}\
^{\shortmid }\mathring{g}_{4})^{\diamond }]^{2}}{|\int dy^{3}(\
_{2}^{\shortmid }\widehat{\Upsilon })(\ ^{\shortmid }\eta _{4}\ ^{\shortmid }%
\mathring{g}_{4})^{\diamond }|\ (\ ^{\shortmid }\eta _{4}\ ^{\shortmid }%
\mathring{g}_{4})^{5/2}}]dx^{\acute{k}_{1}}\}+  \notag
\end{eqnarray}%
\begin{eqnarray*}
&&(\ ^{\shortmid }\eta ^{5}\ ^{\shortmid }\mathring{g}^{5})\{dp_{5}+[\
_{1}n_{k_{2}}+\ _{2}n_{k_{2}}\int dp_{6}\frac{[\partial ^{6}(\ ^{\shortmid
}\eta ^{5}\ ^{\shortmid }\mathring{g}^{5})]^{2}}{|\int dp_{6}(\
_{3}^{\shortmid }\widehat{\Upsilon })\ \partial ^{6}(\ ^{\shortmid }\eta
^{5}\ ^{\shortmid }\mathring{g}^{5})|\ (\ ^{\shortmid }\eta ^{5}\
^{\shortmid }\mathring{g}^{5})^{5/2}}]dx^{k_{2}}\} \\
&&-\frac{[\partial ^{6}(\ ^{\shortmid }\eta ^{5}\ ^{\shortmid }\mathring{g}%
^{5})]^{2}}{|\int dp_{6}\ (\ _{3}^{\shortmid }\widehat{\Upsilon })\partial
^{6}(\ ^{\shortmid }\eta ^{5}\ ^{\shortmid }\mathring{g}^{5})\ |\ (\
^{\shortmid }\eta ^{5}\ ^{\shortmid }\mathring{g}^{5})}\{dp_{6}+\frac{%
\partial _{i_{2}}[\int dp_{6}(\ _{3}^{\shortmid }\widehat{\Upsilon })\
\partial ^{6}(\ ^{\shortmid }\eta ^{5}\ ^{\shortmid }\mathring{g}^{5})]}{(\
_{3}^{\shortmid }\widehat{\Upsilon })\partial ^{6}(\ ^{\shortmid }\eta ^{5}\
^{\shortmid }\mathring{g}^{5})}dx^{i_{2}}\}^{2}+
\end{eqnarray*}%
\begin{eqnarray*}
&&(\ ^{\shortmid }\eta ^{7}\ ^{\shortmid }\mathring{g}^{7})\{dp_{7}+[\
_{1}n_{k_{3}}+\ _{2}n_{k_{3}}\int dE\frac{[(\ ^{\shortmid }\eta ^{7}\
^{\shortmid }\mathring{g}^{7})^{\ast }]^{2}}{|\int dE\ (\ _{4}^{\shortmid }%
\widehat{\Upsilon })[(\ ^{\shortmid }\eta ^{7}\ ^{\shortmid }\mathring{g}%
^{7})]^{\ast }|\ [(\ ^{\shortmid }\eta ^{7}\ ^{\shortmid }\mathring{g}%
^{7})]^{5/2}}]dx^{k_{3}}\} \\
&&-\frac{[(\ ^{\shortmid }\eta ^{7}\ ^{\shortmid }\mathring{g}^{7})^{\ast
}]^{2}}{|\int dE\ (\ _{4}^{\shortmid }\widehat{\Upsilon })(\ ^{\shortmid
}\eta ^{7}\ ^{\shortmid }\mathring{g}^{7})^{\ast }\ |\ (\ ^{\shortmid }\eta
^{7}\ ^{\shortmid }\mathring{g}^{7})}\{dE+\frac{\partial _{i_{3}}[\int dE(\
_{4}^{\shortmid }\widehat{\Upsilon })\ (\ ^{\shortmid }\eta ^{7}\
^{\shortmid }\mathring{g}^{7})^{\ast }]}{(\ _{4}^{\shortmid }\widehat{%
\Upsilon })(\ ^{\shortmid }\eta ^{7}\ ^{\shortmid }\mathring{g}^{7})^{\ast }}%
dx^{i_{3}}\}^{2},
\end{eqnarray*}%
where the polarization functions are determined by generating data $[\psi
(x^{i_{1}}),^{\shortmid }\eta _{4}(x^{i_{1}},y^{3}),\ ^{\shortmid }\eta
^{5}(x^{i_{1}},y^{a_{1}},p_{6}),$ \newline
$\ ^{\shortmid }\eta ^{7}(x^{i_{1}},y^{a_{1}},p_{a_{2}},E)]$ and prime
s-metric coefficients $\ ^{\shortmid }\mathbf{\mathring{g}}_{\alpha
_{s}\beta _{s}}$ \ following such formulas%
\begin{eqnarray*}
\ ^{\shortmid }\eta _{1}\ ^{\shortmid }\mathring{g}_{1} &=&\ ^{\shortmid
}\eta _{2}\ ^{\shortmid }\mathring{g}_{2}=e^{\psi (x^{k_{1}})},\ ^{\shortmid
}\eta _{3}\ ^{\shortmid }\mathring{g}_{3}=-\frac{[(\ ^{\shortmid }\eta _{4}\
^{\shortmid }\mathring{g}_{4})^{\diamond }]^{2}}{|\int dy^{3}[(\
_{2}^{\shortmid }\widehat{\Upsilon })(\ ^{\shortmid }\eta _{4}\ ^{\shortmid }%
\mathring{g}_{4})]^{\diamond }|\ (\ ^{\shortmid }\eta _{4}\ ^{\shortmid }%
\mathring{g}_{4})}, \\
\ ^{\shortmid }\eta ^{6}\ ^{\shortmid }\mathring{g}^{6} &=&-\frac{[\partial
^{6}(\ ^{\shortmid }\eta ^{5}\ ^{\shortmid }\mathring{g}^{5})]^{2}}{|\int
dp_{6}(\ _{3}^{\shortmid }\widehat{\Upsilon })\ \partial ^{6}(\ ^{\shortmid
}\eta ^{5}\ ^{\shortmid }\mathring{g}^{5})|\ (\ ^{\shortmid }\eta ^{5}\
^{\shortmid }\mathring{g}^{5})},\ ^{\shortmid }\eta ^{8}\ ^{\shortmid }%
\mathring{g}^{8}=-\frac{[(\ ^{\shortmid }\eta ^{7}\ ^{\shortmid }\mathring{g}%
^{7})^{\ast }]^{2}}{|\int dE[(\ _{4}^{\shortmid }\widehat{\Upsilon })(\
^{\shortmid }\eta ^{7}\ ^{\shortmid }\mathring{g}^{7})^{\ast }\ |\ (\
^{\shortmid }\eta ^{7}\ ^{\shortmid }\mathring{g}^{7})};
\end{eqnarray*}%
\begin{eqnarray}
\ ^{\shortmid }\eta _{i_{1}}^{3}\ ^{\shortmid }\mathring{N}_{i_{1}}^{3} &=&%
\frac{\partial _{i_{1}}\ \int dy^{3}(\ _{2}^{\shortmid }\widehat{\Upsilon }%
)\ (\ ^{\shortmid }\eta _{4}\ ^{\shortmid }\mathring{g}_{4})^{\diamond }}{(\
_{2}^{\shortmid }\widehat{\Upsilon })\ (\ ^{\shortmid }\eta _{4}\
^{\shortmid }\mathring{g}_{4})^{\diamond }},  \label{noffdiagpolf} \\
\ \ ^{\shortmid }\eta _{k_{1}}^{4}\ ^{\shortmid }\mathring{N}_{k_{1}}^{4}
&=&\ _{1}n_{k_{1}}+\ _{2}n_{k_{1}}\int dy^{3}\frac{[(\ ^{\shortmid }\eta
_{4}\ ^{\shortmid }\mathring{g}_{4})^{\diamond }]^{2}}{|\int dy^{3}(\
_{2}^{\shortmid }\widehat{\Upsilon })(\ ^{\shortmid }\eta _{4}\ ^{\shortmid }%
\mathring{g}_{4})^{\diamond }|\ (\ ^{\shortmid }\eta _{4}\ ^{\shortmid }%
\mathring{g}_{4})^{5/2}},  \notag
\end{eqnarray}%
\begin{eqnarray*}
\ ^{\shortmid }\eta _{k_{2}5}\ ^{\shortmid }\mathring{N}_{k_{2}5} &=&\
_{1}n_{k_{2}}+\ _{2}n_{k_{2}}\int dp_{6}\frac{[\partial ^{6}(\ ^{\shortmid
}\eta ^{5}\ ^{\shortmid }\mathring{g}^{5})]^{2}}{|\int dp_{6}\ (\
_{3}^{\shortmid }\widehat{\Upsilon })\partial ^{6}(\ ^{\shortmid }\eta ^{5}\
^{\shortmid }\mathring{g}^{5})|\ (\ ^{\shortmid }\eta ^{5}\ ^{\shortmid }%
\mathring{g}^{5})^{5/2}}, \\
\ ^{\shortmid }\eta _{i_{2}6}\ ^{\shortmid }\mathring{N}_{i_{2}6} &=&\frac{%
\partial _{i_{2}}\ \int dp_{6}(\ _{3}^{\shortmid }\widehat{\Upsilon })\
\partial ^{6}(\ ^{\shortmid }\eta ^{5}\ ^{\shortmid }\mathring{g}^{5})}{\
_{3}^{\shortmid }\widehat{\Upsilon }\ \partial ^{6}(\ ^{\shortmid }\eta
^{5}\ ^{\shortmid }\mathring{g}^{5})},
\end{eqnarray*}%
\begin{eqnarray*}
\ ^{\shortmid }\eta _{k_{3}7}\ ^{\shortmid }\mathring{N}_{k_{3}7} &=&\
_{1}n_{k_{3}}+\ _{2}n_{k_{3}}\int dE\frac{[(\ ^{\shortmid }\eta ^{7}\
^{\shortmid }\mathring{g}^{7})^{\ast }]^{2}}{|\int dE\ (\ _{4}^{\shortmid }%
\widehat{\Upsilon })(\ ^{\shortmid }\eta ^{7}\ ^{\shortmid }\mathring{g}%
^{7})^{\ast }|\ (\ ^{\shortmid }\eta ^{7}\ ^{\shortmid }\mathring{g}%
^{7})^{5/2}}, \\
\ ^{\shortmid }\eta _{i_{3}8}\ ^{\shortmid }\mathring{N}_{i_{3}8} &=&\frac{%
\partial _{i_{3}}\ \int dE(\ _{4}^{\shortmid }\widehat{\Upsilon })\ (\
^{\shortmid }\eta ^{7}\ ^{\shortmid }\mathring{g}^{7})^{\ast }}{\
_{4}^{\shortmid }\widehat{\Upsilon }\ (\ ^{\shortmid }\eta ^{7}\ ^{\shortmid
}\mathring{g}^{7})^{\ast }}.
\end{eqnarray*}
\end{corollary}

Similar $\eta $-polarization redefinitions are possible for generating
functions $\ _{s}^{\shortmid }\Psi $ and $\ _{s}^{\shortmid }\widetilde{\Psi
}$ and respective solutions (\ref{qeltors}) and/or (\ref{offdiagcosmcsh})
(for LC-configurations, see (\ref{qellc}) and/or (\ref{offidiagcosmconstlc}%
)). We omit such constructions in this work.

\subsection{Nonholonomic transforms with a small parameter to
quasi-stationary s-metrics}

To compute possible nonholonomic deformations by MDR of physically important
solutions (for instance, to construct ellipsoid like configurations as in
various MGTs \cite{gvvepjc14,rajpoot15,vacaruepjc17,vmon06,vijtp10}) we can
consider certain decompositions linearised on a small parameter $\varepsilon
,0\leq \varepsilon <1.$

\begin{definition}
\label{defsmalpnd} \textsf{[gravitational polarizations with a small
parameter] } Small parametric $\varepsilon $--decompositions of the $\eta $%
-polarization functions (\ref{coeftargpol}) resulting in quasi-stationary
configurations defined by parameterizations
\begin{eqnarray}
\ ^{\shortmid }g_{i_{1}}(x^{k_{1}}) &=&\ ^{\shortmid }\eta _{i_{k}}\
^{\shortmid }\mathring{g}_{i_{1}}=\zeta _{i_{1}}(1+\varepsilon \chi
_{i_{1}})\ ^{\shortmid }\mathring{g}_{i_{1}}=  \label{edecomsm} \\
&=&\{\zeta _{i_{1}}(x^{i_{1}},y^{a_{2}},p_{a_{3}},p_{a_{4}})[1+\varepsilon
\chi _{i_{1}}(x^{i_{1}},y^{a_{2}},p_{a_{3}},p_{a_{4}})]\}\ ^{\shortmid }%
\mathring{g}_{i_{1}}(x^{i_{1}},y^{a_{2}},p_{a_{3}},p_{a_{4}}),  \notag \\
\ ^{\shortmid }g_{b_{2}}(x^{i_{1}},y^{3}) &=&\ ^{\shortmid }\eta _{b_{2}}\
^{\shortmid }\mathring{g}_{b_{1}}=\zeta _{b_{2}}(1+\varepsilon \ \chi
_{b_{2}})\ \ ^{\shortmid }\mathring{g}_{b_{1}}=  \notag \\
&=&\{\zeta _{b_{2}}(x^{i_{1}},y^{a_{2}},p_{a_{3}},p_{a_{4}})[1+\varepsilon \
\chi _{b_{2}}(x^{i_{1}},y^{a_{2}},p_{a_{3}},p_{a_{4}})]\}\ ^{\shortmid }%
\mathring{g}_{b_{1}}(x^{i_{1}},y^{a_{2}},p_{a_{3}},p_{a_{4}}),  \notag \\
\ ^{\shortmid }g_{a_{3}}(x^{i_{2}},p_{6}) &=&\ ^{\shortmid }\eta ^{a_{3}}\
^{\shortmid }\mathring{g}^{a_{3}}=\zeta ^{a_{3}}(1+\varepsilon \ \chi
^{a_{3}})\ ^{\shortmid }\mathring{g}^{a_{3}}=  \notag \\
&=&\{\zeta ^{a_{3}}(x^{i_{1}},y^{b_{2}},p_{b_{3}},p_{b_{4}})\ [1+\varepsilon
\ \chi ^{a_{3}}(x^{i_{1}},y^{b_{2}},p_{b_{3}},p_{b_{4}})]\}\ ^{\shortmid }%
\mathring{g}^{a_{3}}(x^{i_{1}},y^{b_{2}},p_{b_{3}},p_{b_{4}}),  \notag \\
\ ^{\shortmid }g_{a_{4}}(x^{i_{3}},E) &=&\ ^{\shortmid }\eta ^{a_{4}}\
^{\shortmid }\mathring{g}^{a_{4}}=\zeta ^{a_{4}}(1+\varepsilon \ \chi
^{a_{4}})\ ^{\shortmid }\mathring{g}^{a_{4}}=  \notag \\
&=&\{\zeta ^{a_{4}}(x^{i_{1}},y^{b_{2}},p_{b_{3}},p_{b_{4}})[1+\varepsilon \
\chi ^{a_{4}}(x^{i_{1}},y^{b_{2}},p_{b_{3}},p_{b_{4}})]\}\ ^{\shortmid }%
\mathring{g}^{a_{4}}(x^{i_{1}},y^{a_{2}},p_{a_{3}},p_{a_{4}})  \notag
\end{eqnarray}%
and (for N-connection coefficients)
\begin{eqnarray}
\ ^{\shortmid }N_{i_{1}}^{a_{2}}(x^{k_{1}},y^{3}) &=&\ ^{\shortmid }\eta
_{i_{1}}^{a_{2}}\ ^{\shortmid }\mathring{N}_{i_{1}}^{a_{2}}=\zeta
_{i_{1}}^{a_{2}}(1+\varepsilon \ \chi _{i_{1}}^{a_{2}})\ ^{\shortmid }%
\mathring{N}_{i_{1}}^{a_{2}}=  \label{edecompncs} \\
&=&\{\zeta
_{i_{1}}^{a_{2}}(x^{i_{1}},y^{b_{2}},p_{b_{3}},p_{b_{4}})[1+\varepsilon \
\chi _{i_{1}}^{a_{2}}(x^{i_{1}},y^{b_{2}},p_{b_{3}},p_{b_{4}})]\}\
^{\shortmid }\mathring{N}%
_{i_{1}}^{a_{2}}(x^{i_{1}},y^{b_{2}},p_{b_{3}},p_{b_{4}}),  \notag \\
\ ^{\shortmid }N_{i_{2}}^{a_{3}}(x^{k_{1}},y^{b_{2}},p_{6}) &=&\ ^{\shortmid
}\eta _{i_{2}a_{3}}\ ^{\shortmid }\mathring{N}_{i_{2}a_{3}}=\zeta
_{i_{2}a_{3}}(1+\varepsilon \chi _{i_{2}a_{3}})\ ^{\shortmid }\mathring{N}%
_{i_{2}a_{3}}=  \notag \\
&=&\{\ \zeta
_{i_{2}a_{3}}(x^{i_{1}},y^{b_{2}},p_{b_{3}},p_{b_{4}})[1+\varepsilon \chi
_{i_{2}a_{3}}(x^{i_{1}},y^{b_{2}},p_{b_{3}},p_{b_{4}})]\}\ ^{\shortmid }%
\mathring{N}_{i_{2}a_{3}}(x^{i_{1}},y^{b_{2}},p_{b_{3}},p_{b_{4}}),  \notag
\\
\ ^{\shortmid }N_{i_{3}}^{a_{4}}(x^{k_{1}},y^{b_{2}},p_{a_{3}},E) &=&\
^{\shortmid }\eta _{i_{3}a_{4}}\ ^{\shortmid }\mathring{N}%
_{i_{3}a_{4}}=\zeta _{i_{3}a_{4}}(1+\varepsilon \chi _{i_{3}a_{4}})\
^{\shortmid }\mathring{N}_{i_{3}a_{4}}=  \notag \\
&=&\{\zeta
_{i_{3}a_{4}}(x^{i_{1}},y^{b_{2}},p_{b_{3}},p_{b_{4}})[1+\varepsilon \chi
_{i_{3}a_{4}}(x^{i_{1}},y^{b_{2}},p_{b_{3}},p_{b_{4}})]\}\ ^{\shortmid }%
\mathring{N}_{i_{3}a_{4}}(x^{i_{1}},y^{b_{2}},p_{b_{3}},p_{b_{4}}),  \notag
\end{eqnarray}%
result in nonholonomic deformations of s-metrics on $_{s}\mathbf{T}^{\ast }%
\mathbf{V,}$
\begin{equation}
\ _{s}^{\shortmid }\mathbf{\mathring{g}}\rightarrow \ _{s}^{\shortmid
\varepsilon }\mathbf{g}=[\ ^{\shortmid }g_{\alpha _{s}}=\zeta _{\alpha
_{s}}(1+\varepsilon \chi _{\alpha _{s}})\ ^{\shortmid }\mathring{g}_{\alpha
_{s}},\ ^{\shortmid }N_{i_{s}}^{a_{s}}=\zeta
_{i_{s-1}}^{a_{s}}(1+\varepsilon \ \chi _{i_{s-1}}^{a_{s}})\ ^{\shortmid }%
\mathring{N}_{i_{s-1}}^{a_{s}}].  \label{epstargsm}
\end{equation}
\end{definition}

We can formulate and prove two important results on nonholonomic $%
\varepsilon $-deformations of prime metrics into target metrics, $\
_{s}^{\shortmid }\mathbf{\mathring{g}}\rightarrow \ _{s}^{\shortmid
\varepsilon }\mathbf{g.}$

\begin{theorem}
\label{epsilongeneration}\textsf{[generating exact solutions with small
parameters for generating functions and sources] } Nonholonomic small
parametric $\varepsilon $--decompositions stated by Definition \ref%
{defsmalpnd} transform a prime metric $\ _{s}^{\shortmid }\mathbf{\mathring{g%
}}$ into an target exact solution $\ _{s}^{\shortmid \varepsilon }\mathbf{g}$
similar to (\ref{offdiagpolf}) with linear parametric dependense on $%
\varepsilon $ if the $\zeta $- and $\chi $-coefficients for deformations of
s-metrics and N-connections in respective formulas (\ref{edecomsm}) \ and (%
\ref{edecompncs}) are generated by shell data
\begin{equation}
\ \ ^{\shortmid }\eta _{2}=\zeta _{2}(1+\varepsilon \chi _{2}),\ \
^{\shortmid }\eta _{4}=\zeta _{4}(1+\varepsilon \chi _{4}),\ ^{\shortmid
}\eta ^{5}=\zeta ^{5}(1+\varepsilon \ \chi ^{5}),\ ^{\shortmid }\eta
^{7}=\zeta ^{7}(1+\varepsilon \ \chi ^{7}),  \label{epsilongenfdecomp}
\end{equation}%
following formulas for $s=1:\ \zeta _{i_{1}}=(\mathring{g}%
_{i_{1}})^{-1}e^{\psi _{0}(x^{k_{1}})}$ and $\chi _{i_{1}}=(\mathring{g}%
_{i_{1}})^{-1}\ ^{\psi }\chi (x^{k_{1}})$, \newline
where $\zeta _{i_{1}}(1+\varepsilon \chi _{i_{1}})\ ^{\shortmid }\mathring{g}%
_{i_{1}}=e^{\psi (x^{k_{1}})}\approx e^{\psi _{0}(x^{k_{1}})(1+\varepsilon \
^{\psi }\chi (x^{k_{1}}))}\approx e^{\psi _{0}(x^{k_{1}})}(1+\varepsilon \
^{\psi }\chi )$ for $\psi _{0}(x^{k_{1}})$ and $\chi (x^{k_{1}})$ defined by
a solution of a 2-d Poisson equation (\ref{eq1});\newline
$s=2$ (with generating functions, $\zeta _{4},\chi _{4};$ generating source
and cosmological constant$,\ _{2}^{\shortmid }\widehat{\Upsilon },\
_{2}^{\shortmid }\Lambda ;$ integration functions$,\ _{1}n_{k_{1}},\
_{2}n_{k_{1}};$ prescribed data for a prime s-metric, $(\ ^{\shortmid }%
\mathring{g}_{3},\ ^{\shortmid }\mathring{g}_{4};\ ^{\shortmid }\mathring{N}%
_{i_{1}}^{3},\ ^{\shortmid }\mathring{N}_{k_{1}}^{4})):$%
\begin{eqnarray*}
\zeta _{3} &=&-\frac{4}{\ ^{\shortmid }\mathring{g}_{3}}\frac{[(|\ \zeta
_{4}\ ^{\shortmid }\mathring{g}_{4}|^{1/2})^{\diamond }]^{2}}{|\int
dy^{3}\{(\ _{2}^{\shortmid }\widehat{\Upsilon })(\ \zeta _{4}\ ^{\shortmid }%
\mathring{g}_{4})^{\diamond }\}|}\mbox{ and }\chi _{3}=\frac{(\chi _{4}|\
\zeta _{4}\ ^{\shortmid }\mathring{g}_{4}|^{1/2})^{\diamond }}{4(|\ \zeta
_{4}\ ^{\shortmid }\mathring{g}_{4}|^{1/2})^{\diamond }}-\frac{\int
dy^{3}\{[(\ _{2}^{\shortmid }\widehat{\Upsilon })\ (\zeta _{4}\ ^{\shortmid }%
\mathring{g}_{4})\chi _{4}]^{\diamond }\}}{\int dy^{3}\{(\ _{2}^{\shortmid }%
\widehat{\Upsilon })(\ \zeta _{4}\ ^{\shortmid }\mathring{g}_{4})^{\diamond
}\}}, \\
\zeta _{i_{1}}^{3} &=&\frac{\partial _{i_{1}}\ \int dy^{3}(\ _{2}^{\shortmid
}\widehat{\Upsilon })\ (\zeta _{4})^{\diamond }}{(\ ^{\shortmid }\mathring{N}%
_{i_{1}}^{3})(\ _{2}^{\shortmid }\widehat{\Upsilon })(\zeta _{4})^{\diamond }%
}\mbox{ and }\chi _{i_{1}}^{3}=\frac{\partial _{i_{1}}[\int dy^{3}(\
_{2}^{\shortmid }\widehat{\Upsilon })(\zeta _{4}\chi _{4})^{\diamond }]}{%
\partial _{i_{1}}\ [\int dy^{3}(\ _{2}^{\shortmid }\widehat{\Upsilon }%
)(\zeta _{4})^{\diamond }]}-\frac{(\zeta _{4}\chi _{4})^{\diamond }}{(\zeta
_{4})^{\diamond }}, \\
\zeta _{k_{1}}^{4} &=&\ (\ ^{\shortmid }\mathring{N}_{k_{1}}^{4})^{-1}[\
_{1}n_{k_{1}}+16\ _{2}n_{k_{1}}[\int dy^{3}\{\frac{\left( [(\ \zeta _{4}\
^{\shortmid }\mathring{g}_{4})^{-1/4}]^{\diamond }\right) ^{2}}{|\int
dy^{3}(\ _{2}^{\shortmid }\widehat{\Upsilon })(\zeta _{4}\ ^{\shortmid }%
\mathring{g}_{4})^{\diamond }|}]\mbox{ and } \\
\chi _{k_{1}}^{4} &=&\ -\frac{16\ _{2}n_{k_{1}}\int dy^{3}\frac{\left( [(\
\zeta _{4}\ ^{\shortmid }\mathring{g}_{4})^{-1/4}]^{\diamond }\right) ^{2}}{%
|\int dy^{3}(\ _{2}^{\shortmid }\widehat{\Upsilon })[(\zeta _{4}\
^{\shortmid }\mathring{g}_{4})]^{\diamond }|}\left( \frac{[(\ \zeta _{4}\
^{\shortmid }\mathring{g}_{4})^{-1/4}\chi _{4})]^{\diamond }}{2[(\ \zeta
_{4}\ ^{\shortmid }\mathring{g}_{4})^{-1/4}]^{\diamond }}+\frac{\int
dy^{3}[(\ _{2}^{\shortmid }\widehat{\Upsilon })(\zeta _{4}\chi _{4}\
^{\shortmid }\mathring{g}_{4})]^{\diamond }}{\int dy^{3}(\ _{2}^{\shortmid }%
\widehat{\Upsilon })(\zeta _{4}\ ^{\shortmid }\mathring{g}_{4})^{\diamond }}%
\right) }{\ _{1}n_{k_{1}}+16\ _{2}n_{k_{1}}[\int dy^{3}\frac{\left( [(\
\zeta _{4}\ ^{\shortmid }\mathring{g}_{4})^{-1/4}]^{\diamond }\right) ^{2}}{%
|\int dy^{3}(\ _{2}^{\shortmid }\widehat{\Upsilon })[(\zeta _{4}\
^{\shortmid }\mathring{g}_{4})]^{\diamond }|}]};
\end{eqnarray*}%
$s=3$ (with generating functions, $\zeta ^{5},\chi ^{5};$ generating source
and cosmological constant$,\ _{3}^{\shortmid }\widehat{\Upsilon },\
_{3}^{\shortmid }\Lambda ;$ integration functions$,\ _{1}n_{k_{3}},\
_{2}n_{k_{3}};$ prescribed data for a prime s-metric, $(\ ^{\shortmid }%
\mathring{g}^{5},\ ^{\shortmid }\mathring{g}^{6};\ ^{\shortmid }\mathring{N}%
_{k_{2}5},\ ^{\shortmid }\mathring{N}_{i_{2}6}):$%
\begin{eqnarray*}
\zeta ^{6} &=&-\frac{4}{\ ^{\shortmid }\mathring{g}^{6}}\frac{[\partial
^{6}(|\ \zeta ^{5}\ ^{\shortmid }\mathring{g}^{5}|^{1/2})]^{2}}{|\int
dp_{6}\{(\ _{3}^{\shortmid }\widehat{\Upsilon })\partial ^{6}(\zeta ^{5}\
^{\shortmid }\mathring{g}^{5})\}|}\mbox{ and }\chi ^{6}=\frac{\partial
^{6}(\chi ^{5}|\ \zeta ^{5}\ ^{\shortmid }\mathring{g}^{5}|^{1/2})}{%
4\partial ^{6}(|\ \zeta ^{5}\ ^{\shortmid }\mathring{g}^{5}|^{1/2})}-\frac{%
\int dp_{6}\{\partial ^{6}[(\ _{3}^{\shortmid }\widehat{\Upsilon })\ (\zeta
^{5}\ ^{\shortmid }\mathring{g}^{5})\chi ^{5}]\}}{\int dp_{6}\{(\
_{3}^{\shortmid }\widehat{\Upsilon })\partial ^{6}(\zeta ^{5}\ ^{\shortmid }%
\mathring{g}^{5})\}}, \\
\zeta _{i_{2}5} &=&\ (\ ^{\shortmid }\mathring{N}_{i_{2}5})^{-1}[\
_{1}n_{i_{2}}+16\ _{2}n_{i_{2}}[\int dp_{6}\{\frac{\left( \partial
^{6}[(\zeta ^{5}\ ^{\shortmid }\mathring{g}^{5})^{-1/4}]\right) ^{2}}{|\int
dp_{6}\ (\ _{3}^{\shortmid }\widehat{\Upsilon })\partial ^{6}(\zeta ^{5}\
^{\shortmid }\mathring{g}^{5})|}]\mbox{ and } \\
\chi _{i_{2}5} &=&\ -\frac{16\ _{2}n_{i_{2}}\int dp_{6}\frac{\left( \partial
^{6}[(\ \zeta ^{5}\ ^{\shortmid }\mathring{g}^{5})^{-1/4}]\right) ^{2}}{%
|\int dp_{6}\ (\ _{3}^{\shortmid }\widehat{\Upsilon })\partial ^{6}(\ \zeta
^{5}\ ^{\shortmid }\mathring{g}^{5})|}\left( \frac{\ \partial ^{6}[(\ \zeta
^{5}\ ^{\shortmid }\mathring{g}^{5})^{-1/4}\chi ^{5})]}{2\ \partial ^{6}[(\
\zeta ^{5}\ ^{\shortmid }\mathring{g}^{5})^{-1/4}]}+\frac{\int dp_{6}\
\partial ^{6}[(\ _{3}^{\shortmid }\widehat{\Upsilon })(\zeta ^{5}\
^{\shortmid }\mathring{g}^{5})\chi ^{5}]}{\int dp_{6}\ (\ _{3}^{\shortmid }%
\widehat{\Upsilon })\partial ^{6}(\zeta ^{5}\ ^{\shortmid }\mathring{g}^{5})}%
\right) }{\ _{1}n_{i_{2}}+16\ _{2}n_{i_{2}}[\int dp_{6}\frac{\left( \
\partial ^{6}[(\ \zeta ^{5}\ ^{\shortmid }\mathring{g}^{5})^{-1/4}]\right)
^{2}}{|\int dp_{6}\ (\ _{3}^{\shortmid }\widehat{\Upsilon })\partial
^{6}(\zeta ^{5}\ ^{\shortmid }\mathring{g}^{5})|}]};
\end{eqnarray*}%
\begin{equation*}
\zeta _{i_{2}6}=\frac{\partial _{i_{2}}\ \int dp_{6}(\ _{3}^{\shortmid }%
\widehat{\Upsilon })\ \partial ^{6}(\zeta ^{5})}{(\ ^{\shortmid }\mathring{N}%
_{i_{2}6})(\ _{3}^{\shortmid }\widehat{\Upsilon })\partial ^{6}(\zeta ^{5})}%
\mbox{ and }\chi _{i_{2}6}=\frac{\partial _{i_{2}}[\int dp_{6}(\
_{3}^{\shortmid }\widehat{\Upsilon })\partial ^{6}(\zeta ^{5}\ ^{\shortmid }%
\mathring{g}^{5})]}{\partial _{i_{2}}\ [\int dp_{6}(\ _{3}^{\shortmid }%
\widehat{\Upsilon })\partial ^{6}(\zeta ^{5})]}-\frac{\partial ^{6}(\zeta
^{5}\ ^{\shortmid }\mathring{g}^{5})}{\partial ^{6}(\zeta ^{5})},
\end{equation*}

$s=4$ (with generating functions, $\zeta ^{7},\chi ^{7};$ generating source
and cosmological constant$,\ _{4}^{\shortmid }\widehat{\Upsilon },\
_{4}^{\shortmid }\Lambda ;$ integration functions$,\ _{1}n_{k_{4}},\
_{2}n_{k_{4}};$ prescribed data for a prime s-metric, $(\ ^{\shortmid }%
\mathring{g}^{7},\ ^{\shortmid }\mathring{g}^{8};\ ^{\shortmid }\mathring{N}%
_{k_{3}7},\ ^{\shortmid }\mathring{N}_{i_{3}8}):$%
\begin{eqnarray*}
\zeta ^{8} &=&-\frac{4}{\ ^{\shortmid }\mathring{g}^{8}}\frac{[(|\ \zeta
^{7}\ ^{\shortmid }\mathring{g}^{7}|^{1/2})^{\ast }]^{2}}{|\int dE\{(\
_{4}^{\shortmid }\widehat{\Upsilon })[(\zeta ^{7}\ ^{\shortmid }\mathring{g}%
^{7})]^{\ast }\}|}\mbox{ and }\chi ^{8}=\frac{(\chi ^{7}|\ \zeta ^{7}\
^{\shortmid }\mathring{g}^{7}|^{1/2})^{\ast }}{4(|\ \zeta ^{7}\ ^{\shortmid }%
\mathring{g}^{7}|^{1/2})^{\ast }}-\frac{\int dE\{[(\ _{4}^{\shortmid }%
\widehat{\Upsilon })\ (\zeta ^{7}\ ^{\shortmid }\mathring{g}^{7})\chi
^{7}]^{\ast }\}}{\int dE\{(\ _{4}^{\shortmid }\widehat{\Upsilon })(\zeta
^{4}\ ^{\shortmid }\mathring{g}^{4})^{\ast }\}}, \\
\zeta _{i_{3}7} &=&\ (\ ^{\shortmid }\mathring{N}_{i_{3}7})^{-1}[\
_{1}n_{i_{3}}+16\ _{2}n_{i_{3}}[\int dE\{\frac{\left( [(\zeta ^{7}\
^{\shortmid }\mathring{g}^{7})^{-1/4}]^{\ast }\right) ^{2}}{|\int dE\ (\
_{4}^{\shortmid }\widehat{\Upsilon })(\zeta ^{7}\ ^{\shortmid }\mathring{g}%
^{7})^{\ast }|}]\mbox{ and } \\
\chi _{i_{3}7} &=&\ -\frac{16\ _{2}n_{i_{3}}\int dE\frac{\left( [(\ \zeta
^{7}\ ^{\shortmid }\mathring{g}^{7})^{-1/4}]^{\ast }\right) ^{2}}{|\int dE\
(\ _{4}^{\shortmid }\widehat{\Upsilon })(\ \zeta ^{7}\ ^{\shortmid }%
\mathring{g}^{7})^{\ast }|}\left( \frac{[(\ \zeta ^{7}\ ^{\shortmid }%
\mathring{g}^{7})^{-1/4}\chi ^{7})]^{\ast }}{2\ [(\ \zeta ^{7}\ ^{\shortmid }%
\mathring{g}^{7})^{-1/4}]^{\ast }}+\frac{\int dE\ (\ _{4}^{\shortmid }%
\widehat{\Upsilon })[(\zeta ^{7}\ ^{\shortmid }\mathring{g}^{7})\chi
^{7}]^{\ast }}{\int dE\ (\ _{4}^{\shortmid }\widehat{\Upsilon })(\zeta ^{7}\
^{\shortmid }\mathring{g}^{7})^{\ast }}\right) }{\ _{1}n_{i_{3}}+16\
_{2}n_{i_{3}}[\int dE\frac{\left( \ [(\ \zeta ^{7}\ ^{\shortmid }\mathring{g}%
^{7})^{-1/4}]^{\ast }\right) ^{2}}{|\int dE\ (\ _{4}^{\shortmid }\widehat{%
\Upsilon })(\zeta ^{7}\ ^{\shortmid }\mathring{g}^{7})^{\ast }|}]}, \\
\zeta _{i_{3}8} &=&\frac{\partial _{i_{3}}\ \int dE(\ _{4}^{\shortmid }%
\widehat{\Upsilon })\ (\zeta ^{7})^{\ast }}{(\ ^{\shortmid }\mathring{N}%
_{i_{3}8})(\ _{4}^{\shortmid }\widehat{\Upsilon })(\zeta ^{7})^{\ast }}%
\mbox{ and }\chi _{i_{3}7}=\frac{\partial _{i_{3}}[\int dE(\ _{4}^{\shortmid
}\widehat{\Upsilon })(\zeta ^{7}\ ^{\shortmid }\mathring{g}^{7})^{\ast }]}{%
\partial _{i_{3}}\ [\int dE(\ _{4}^{\shortmid }\widehat{\Upsilon })(\zeta
^{7})^{\ast }]}-\frac{(\zeta ^{7}\ ^{\shortmid }\mathring{g}^{7})^{\ast }}{%
(\zeta ^{7})^{\ast }}.
\end{eqnarray*}
\end{theorem}

\begin{proof}
See a sketch of proof in appendix \ref{assepsilon}. $\square $\vskip5pt
\end{proof}

Introducing the $\varepsilon $-coefficients from this theorem instead of $%
\eta $-coefficients of (\ref{offdiagpolf}) we prove

\begin{consequence}
\label{nqelgravpoleps}\textsf{[nonlinear quadratic elements for
gravitational polarizations with small parameter] } The quasi-stationary
phase configurations on cotangent Lorentz bundles described in terms of $%
\eta $-polarization functions by formulas (\ref{offdiagpolf}) can be written
in terms of $\varepsilon $-coefficients%
\begin{equation*}
ds^{2}=g_{\alpha _{s}\beta
_{s}}(x^{k},y^{3},p_{a_{3}},p_{a_{4}};g_{4},g^{5},g^{7},_{\shortmid
s}\Lambda ;\ _{s}^{\shortmid }\widehat{\Upsilon },\varepsilon )du^{\alpha
_{s}}du^{\beta _{s}}=e^{\psi _{0}(x^{k_{1}})}(1+\varepsilon \ ^{\psi }\chi
)[(dx^{1})^{2}+(dx^{2})^{2}]+
\end{equation*}%
\begin{eqnarray*}
&&\{-\frac{4}{\ ^{\shortmid }\mathring{g}_{3}}\frac{[(|\ \zeta _{4}\
^{\shortmid }\mathring{g}_{4}|^{1/2})^{\diamond }]^{2}}{|\int dy^{3}\{(\
_{2}^{\shortmid }\widehat{\Upsilon })(\ \zeta _{4}\ ^{\shortmid }\mathring{g}%
_{4})^{\diamond }\}|}+\varepsilon \left[ \frac{(\chi _{4}|\ \zeta _{4}\
^{\shortmid }\mathring{g}_{4}|^{1/2})^{\diamond }}{4(|\ \zeta _{4}\
^{\shortmid }\mathring{g}_{4}|^{1/2})^{\diamond }}-\frac{\int dy^{3}\{(\
_{2}^{\shortmid }\widehat{\Upsilon })[(\zeta _{4}\ ^{\shortmid }\mathring{g}%
_{4})\chi _{4}]^{\diamond }\}}{\int dy^{3}\{(\ _{2}^{\shortmid }\widehat{%
\Upsilon })(\ \zeta _{4}\ ^{\shortmid }\mathring{g}_{4})^{\diamond }\}}%
\right] \}\ ^{\shortmid }\mathring{g}_{3} \\
&&\{dy^{3}+[\frac{\partial _{i_{1}}\ \int dy^{3}(\ _{2}^{\shortmid }\widehat{%
\Upsilon })\ (\zeta _{4})^{\diamond }}{(\ ^{\shortmid }\mathring{N}%
_{i_{1}}^{3})(\ _{2}^{\shortmid }\widehat{\Upsilon })(\zeta _{4})^{\diamond }%
}+\varepsilon \left( \frac{\partial _{i_{1}}[\int dy^{3}(\ _{2}^{\shortmid }%
\widehat{\Upsilon })(\zeta _{4}\chi _{4})^{\diamond }]}{\partial _{i_{1}}\
[\int dy^{3}(\ _{2}^{\shortmid }\widehat{\Upsilon })(\zeta _{4})^{\diamond }]%
}-\frac{(\zeta _{4}\chi _{4})^{\diamond }}{(\zeta _{4})^{\diamond }}\right)
](\ ^{\shortmid }\mathring{N}_{i_{1}}^{3})dx^{i_{1}}\}^{2}+
\end{eqnarray*}%
\begin{eqnarray}
&&\zeta _{4}(1+\varepsilon \ \chi _{4})\ ^{\shortmid }\mathring{g}%
_{4}\{dt+[\ (\ ^{\shortmid }\mathring{N}_{k_{1}}^{4})^{-1}\left[ \
_{1}n_{k_{1}}+16\ _{2}n_{k_{1}}[\int dy^{3}\{\frac{\left( [(\ \zeta _{4}\
^{\shortmid }\mathring{g}_{4})^{-1/4}]^{\diamond }\right) ^{2}}{|\int
dy^{3}[(\ _{2}^{\shortmid }\widehat{\Upsilon })(\zeta _{4}\ ^{\shortmid }%
\mathring{g}_{4})]^{\diamond }|}\right] -  \label{offdncelepsilon} \\
&&\varepsilon \frac{16\ _{2}n_{k_{1}}\int dy^{3}\frac{\left( [(\ \zeta _{4}\
^{\shortmid }\mathring{g}_{4})^{-1/4}]^{\diamond }\right) ^{2}}{|\int
dy^{3}[(\ _{2}^{\shortmid }\widehat{\Upsilon })(\zeta _{4}\ ^{\shortmid }%
\mathring{g}_{4})]^{\diamond }|}\left( \frac{[(\ \zeta _{4}\ ^{\shortmid }%
\mathring{g}_{4})^{-1/4}\chi _{4})]^{\diamond }}{2[(\ \zeta _{4}\
^{\shortmid }\mathring{g}_{4})^{-1/4}]^{\diamond }}+\frac{\int dy^{3}[(\
_{2}^{\shortmid }\widehat{\Upsilon })(\zeta _{4}\chi _{4}\ ^{\shortmid }%
\mathring{g}_{4})]^{\diamond }}{\int dy^{3}[(\ _{2}^{\shortmid }\widehat{%
\Upsilon })(\zeta _{4}\ ^{\shortmid }\mathring{g}_{4})]^{\diamond }}\right)
}{\ _{1}n_{k_{1}}+16\ _{2}n_{k_{1}}[\int dy^{3}\frac{\left( [(\ \zeta _{4}\
^{\shortmid }\mathring{g}_{4})^{-1/4}]^{\diamond }\right) ^{2}}{|\int
dy^{3}[(\ _{2}^{\shortmid }\widehat{\Upsilon })(\zeta _{4}\ ^{\shortmid }%
\mathring{g}_{4})]^{\diamond }|}]}](\ ^{\shortmid }\mathring{N}%
_{k_{1}}^{4})dx^{\acute{k}_{1}}\}+  \notag
\end{eqnarray}%
\begin{eqnarray*}
&&\zeta ^{5}(1+\varepsilon \ \chi ^{5})\ ^{\shortmid }\mathring{g}%
^{5}\{dp_{5}+[\ (\ ^{\shortmid }\mathring{N}_{i_{2}5})^{-1}\left[ \
_{1}n_{i_{2}}+16\ _{2}n_{i_{2}}[\int dp_{6}\{\frac{\left( \partial
^{6}[(\zeta ^{5}\ ^{\shortmid }\mathring{g}^{5})^{-1/4}]\right) ^{2}}{|\int
dp_{6}\ \partial ^{6}[(\ _{3}^{\shortmid }\widehat{\Upsilon })(\zeta ^{5}\
^{\shortmid }\mathring{g}^{5})]|}\right] - \\
&&\varepsilon \frac{16\ _{2}n_{i_{2}}\int dp_{6}\frac{\left( \partial
^{6}[(\ \zeta ^{5}\ ^{\shortmid }\mathring{g}^{5})^{-1/4}]\right) ^{2}}{%
|\int dp_{6}(\ _{3}^{\shortmid }\widehat{\Upsilon })\ \partial ^{6}[(\ \zeta
^{5}\ ^{\shortmid }\mathring{g}^{5})]|}\left( \frac{\ \partial ^{6}[(\ \zeta
^{5}\ ^{\shortmid }\mathring{g}^{5})^{-1/4}\chi ^{5})]}{2\ \partial ^{6}[(\
\zeta ^{5}\ ^{\shortmid }\mathring{g}^{5})^{-1/4}]}+\frac{\int dp_{6}\ (\
_{3}^{\shortmid }\widehat{\Upsilon })\ \partial ^{6}[(\zeta ^{5}\
^{\shortmid }\mathring{g}^{5})\chi ^{5}]}{\int dp_{6}(\ _{3}^{\shortmid }%
\widehat{\Upsilon })\ \partial ^{6}[(\zeta ^{5}\ ^{\shortmid }\mathring{g}%
^{5})]}\right) }{\ _{1}n_{i_{2}}+16\ _{2}n_{i_{2}}[\int dp_{6}\frac{\left( \
\partial ^{6}[(\ \zeta ^{5}\ ^{\shortmid }\mathring{g}^{5})^{-1/4}]\right)
^{2}}{|\int dp_{6}\ (\ _{3}^{\shortmid }\widehat{\Upsilon })\ \partial
^{6}[(\zeta ^{5}\ ^{\shortmid }\mathring{g}^{5})]|}]}](\ ^{\shortmid }%
\mathring{N}_{i_{2}5})dx^{i_{2}}\}+
\end{eqnarray*}%
\begin{eqnarray*}
&&\{-\frac{4}{\ ^{\shortmid }\mathring{g}^{6}}\frac{[\partial ^{6}(|\ \zeta
^{5}\ ^{\shortmid }\mathring{g}^{5}|^{1/2})]^{2}}{|\int dp_{6}\{(\
_{3}^{\shortmid }\widehat{\Upsilon })\partial ^{6}[(\zeta ^{5}\ ^{\shortmid }%
\mathring{g}^{5})]\}|}+\varepsilon \left[ \frac{\partial _{i_{2}}[\int
dp_{6}(\ _{3}^{\shortmid }\widehat{\Upsilon })\partial ^{6}(\zeta ^{5}\
^{\shortmid }\mathring{g}^{5})]}{\partial _{i_{2}}\ [\int dp_{6}(\
_{3}^{\shortmid }\widehat{\Upsilon })\partial ^{6}(\zeta ^{5})]}-\frac{%
\partial ^{6}(\zeta ^{5}\ ^{\shortmid }\mathring{g}^{5})}{\partial
^{6}(\zeta ^{5})}\right] \}\ ^{\shortmid }\mathring{g}^{6} \\
&&\{dp_{6}+[\frac{\partial _{i_{2}}\ \int dp_{6}(\ _{3}^{\shortmid }\widehat{%
\Upsilon })\ \partial ^{6}(\zeta ^{5})}{(\ ^{\shortmid }\mathring{N}%
_{i_{2}6})(\ _{3}^{\shortmid }\widehat{\Upsilon })\partial ^{6}(\zeta ^{5})}%
+\varepsilon \left( \frac{\partial _{i_{2}}[\int dp_{6}(\ _{3}^{\shortmid }%
\widehat{\Upsilon })\partial ^{6}(\zeta ^{5}\ ^{\shortmid }\mathring{g}^{5})]%
}{\partial _{i_{2}}\ [\int dp_{6}(\ _{3}^{\shortmid }\widehat{\Upsilon }%
)\partial ^{6}(\zeta ^{5})]}-\frac{\partial ^{6}(\zeta ^{5}\ ^{\shortmid }%
\mathring{g}^{5})}{\partial ^{6}(\zeta ^{5})}\right) ](\ ^{\shortmid }%
\mathring{N}_{i_{2}6})dx^{i_{2}}\}+
\end{eqnarray*}%
\begin{eqnarray*}
&&\zeta ^{7}(1+\varepsilon \ \chi ^{7})\ ^{\shortmid }\mathring{g}%
^{7}\{dp_{7}+[\ (\ ^{\shortmid }\mathring{N}_{i_{3}7})^{-1}\left[ \
_{1}n_{i_{3}}+16\ _{2}n_{i_{3}}[\int dE\{\frac{\left( [(\zeta ^{7}\
^{\shortmid }\mathring{g}^{7})^{-1/4}]^{\ast }\right) ^{2}}{|\int dE\ (\
_{4}^{\shortmid }\widehat{\Upsilon })[(\zeta ^{7}\ ^{\shortmid }\mathring{g}%
^{7})]^{\ast }|}\right] - \\
&&\varepsilon \frac{16\ _{2}n_{i_{3}}\int dE\frac{\left( [(\ \zeta ^{7}\
^{\shortmid }\mathring{g}^{7})^{-1/4}]^{\ast }\right) ^{2}}{|\int dE\ (\
_{4}^{\shortmid }\widehat{\Upsilon })[(\ \zeta ^{7}\ ^{\shortmid }\mathring{g%
}^{7})]^{\ast }|}\left( \frac{[(\ \zeta ^{7}\ ^{\shortmid }\mathring{g}%
^{7})^{-1/4}\chi ^{7})]^{\ast }}{2\ [(\ \zeta ^{7}\ ^{\shortmid }\mathring{g}%
^{7})^{-1/4}]^{\ast }}+\frac{\int dE\ (\ _{4}^{\shortmid }\widehat{\Upsilon }%
)[(\zeta ^{7}\ ^{\shortmid }\mathring{g}^{7})\chi ^{7}]^{\ast }}{\int dE\ (\
_{4}^{\shortmid }\widehat{\Upsilon })[(\zeta ^{7}\ ^{\shortmid }\mathring{g}%
^{7})]^{\ast }}\right) }{\ _{1}n_{i_{3}}+16\ _{2}n_{i_{3}}[\int dE\frac{%
\left( \ [(\ \zeta ^{7}\ ^{\shortmid }\mathring{g}^{7})^{-1/4}]^{\ast
}\right) ^{2}}{|\int dE\ (\ _{4}^{\shortmid }\widehat{\Upsilon })[(\zeta
^{7}\ ^{\shortmid }\mathring{g}^{7})]|^{\ast }}]}]\ (^{\shortmid }\mathring{N%
}_{i_{3}7})dx^{i_{3}}\}+
\end{eqnarray*}%
\begin{eqnarray*}
&&\{-\frac{4}{\ ^{\shortmid }\mathring{g}^{8}}\frac{[(|\ \zeta ^{7}\
^{\shortmid }\mathring{g}^{7}|^{1/2})^{\ast }]^{2}}{|\int dE\{(\
_{4}^{\shortmid }\widehat{\Upsilon })[(\zeta ^{7}\ ^{\shortmid }\mathring{g}%
^{7})]^{\ast }\}|}+\varepsilon \left[ \frac{(\chi ^{7}|\ \zeta ^{7}\
^{\shortmid }\mathring{g}^{7}|^{1/2})^{\ast }}{4(|\ \zeta ^{7}\ ^{\shortmid }%
\mathring{g}^{7}|^{1/2})^{\ast }}-\frac{\int dE\{(\ _{4}^{\shortmid }%
\widehat{\Upsilon })[(\zeta ^{7}\ ^{\shortmid }\mathring{g}^{7})\chi
^{7}]^{\ast }\}}{\int dE\{(\ _{4}^{\shortmid }\widehat{\Upsilon })[(\zeta
^{4}\ ^{\shortmid }\mathring{g}^{4})]^{\ast }\}}\right] \}\ ^{\shortmid }%
\mathring{g}^{8} \\
&&\{dE+[\frac{\partial _{i_{3}}\ \int dE(\ _{4}^{\shortmid }\widehat{%
\Upsilon })\ (\zeta ^{7})^{\ast }}{(\ ^{\shortmid }\mathring{N}_{i_{3}8})(\
_{4}^{\shortmid }\widehat{\Upsilon })(\zeta ^{7})^{\ast }}+\varepsilon
\left( \frac{\partial _{i_{3}}[\int dE(\ _{4}^{\shortmid }\widehat{\Upsilon }%
)(\zeta ^{7}\ ^{\shortmid }\mathring{g}^{7})^{\ast }]}{\partial _{i_{3}}\
[\int dE(\ _{4}^{\shortmid }\widehat{\Upsilon })(\zeta ^{7})^{\ast }]}-\frac{%
(\zeta ^{7}\ ^{\shortmid }\mathring{g}^{7})^{\ast }}{(\zeta ^{7})^{\ast }}%
\right) ](\ ^{\shortmid }\mathring{N}_{i_{3}8})dx^{i_{3}}\}.
\end{eqnarray*}
\end{consequence}

In formulas (\ref{coeftargpol}), we considered that $\eta $-polarization
functions may depend, in principle, on all phase space coordinates. For $%
\varepsilon $--polarizations, it is convenient to adapt the constructions
when respective shell coefficients depend only on the same and lower shall
coordinates. We can consider, for instance, that the prime metric$\
_{s}^{\shortmid }\mathbf{\mathring{g}}$ defines a physical important
stationary solution in a 4-d Einstein or MGT with coefficients $[\
^{\shortmid }\mathring{g}_{\alpha _{s}},\ ^{\shortmid }\mathring{N}%
_{i_{s-1}}^{a_{s}}]$ depending only on spacetime variables conventionally
split on shells $1$ and $2.$ In such cases, \ the target solutions $\
_{s}^{\shortmid \varepsilon }\mathbf{g}$ (\ref{epstargsm}) define possible
generalizations, for instance, of black hole solutions on phase spaces.

\subsection{Diagonal phase space configurations}

We can chose some special conditions on generating functions and sources
which allow to construct exact solutions in N-adapted diagonal form. In
general, such conditions depend on the type of (effective) sources,
integration functions, and prime metrics we consider.

\subsubsection{Diagonal s-metrics with constrained $\protect\eta$-functions}

We can chose special data $(\ ^{\shortmid }\eta _{4},\ ^{\shortmid }\eta
^{5},\ ^{\shortmid }\eta ^{7})$ which generate diagonal configurations.

\begin{definition}
\label{diagonasm} \textsf{[diagonal s-metric coefficients] } A s-metric is
diagonal on a shell $s$ if there are satisfied the conditions%
\begin{equation}
\ ^{\shortmid }N_{i_{s-1}}^{a_{s}}=\ ^{\shortmid }\eta _{i_{s-1}a_{s}}\
^{\shortmid }\mathring{N}_{i_{s-1}a_{s}}=0.  \label{diagonalcond}
\end{equation}
\end{definition}

The solutions of equations (\ref{diagonalcond}) depend on the type of
generating and polarization functions used for explicit constructions. By
straightforward computations, we prove

\begin{corollary}
\label{conddiagsm}\textsf{[conditions of diagonal s-metrics] } There are
diagonal quasi-stationary phase space configurations determined by such
nonholonomic constraints on $\eta $-polarization functions for a s-metric (%
\ref{offdiagpolf}):%
\begin{eqnarray}
\ ^{\shortmid }\eta _{4} &=&(\ ^{\shortmid }\mathring{g}_{4})^{-1}\int
dy^{3}\ \ ^{\shortmid }\eta _{4}^{[1]}(y^{3})/(\ _{2}^{\shortmid }\widehat{%
\Upsilon })+\ \ ^{\shortmid }\eta _{4}^{[2]}(x^{k}),  \label{diagsmeta} \\
&& \mbox{ for integration functions }\ ^{\shortmid }\eta _{4}^{[1]}(y^{3})\ %
\mbox{ and }\ ^{\shortmid }\eta _{4}^{[2]}(x^{k_{1}}),\mbox{ where }\
_{1}n_{k_{1}} =\ _{2}n_{k_{1}}=0;  \notag \\
\ \ ^{\shortmid }\eta ^{5} &=&(\ ^{\shortmid }\mathring{g}^{5})^{-1}\int
dp_{6}\ \ ^{\shortmid }\eta _{\lbrack 1]}^{5}(p_{6})/(\ _{3}^{\shortmid }%
\widehat{\Upsilon })+\ \ ^{\shortmid }\eta _{\lbrack 2]}^{5}(x^{k_{2}}),
\notag \\
&& \mbox{ for integration functions }\ \ ^{\shortmid }\eta _{\lbrack
1]}^{5}(p_{6})\ \mbox{ and }\ \ ^{\shortmid }\eta _{\lbrack
2]}^{5}(x^{k_{2}}),\mbox{ where }\ _{1}n_{k_{2}} =\ _{2}n_{k_{2}}=0;  \notag
\\
\ ^{\shortmid }\eta ^{7} &=&(\ ^{\shortmid }\mathring{g}^{7})^{-1}\int dE\ \
^{\shortmid }\eta _{\lbrack 1]}^{7}(E)/(\ _{4}^{\shortmid }\widehat{\Upsilon
})+\ \ ^{\shortmid }\eta _{\lbrack 2]}^{7}(x^{k_{3}}),  \notag \\
&&\mbox{ for integration functions }\ \ ^{\shortmid }\eta _{\lbrack
1]}^{7}(E)\ \mbox{ and }\ \ ^{\shortmid }\eta _{\lbrack 2]}^{7}(x^{k_{3}}),%
\mbox{ where }\ _{1}n_{k_{3}} =\ _{2}n_{k_{3}}=0.  \notag
\end{eqnarray}
\end{corollary}

\begin{proof}
We can check by direct computations that if such conditions are satisfied
the N-connection coefficients (\ref{noffdiagpolf}) became zero. $\square $ %
\vskip5pt
\end{proof}

Diagonal configurations consist a special class of space nonlinear systems
when the generalized gravitational dynamics is defined by diagonal s-metrics
self-consistently embedded into diagonal phase spaces backgrounds. Such
solutions are important in modern cosmology and astrophysics.

\subsubsection{Diagonalization for small parametric $\protect\varepsilon $%
--decompositions}

Using formulas (\ref{epsilongenfdecomp}) from theorem \ref{epsilongeneration}%
, and related formulas (\ref{offdncelepsilon}) from consequence \ref%
{nqelgravpoleps}, in (\ref{diagsmeta}), we obtain

\begin{consequence}
\label{diagsmepsilo}\textsf{[conditions of diagonal s-metrics with} $%
\varepsilon $-\textsf{polarization} \textsf{] } For nonholonomic $%
\varepsilon $--deformations, quasi-stationary diagonal configurations are
generated by such $\varepsilon $-polarization functions for a s-metric (\ref%
{offdncelepsilon})
\begin{eqnarray}
\ \ ^{\shortmid }\eta _{4} &=&\zeta _{4}(1+\varepsilon \chi _{4}):\ \
^{\shortmid }\zeta _{4}=(\ ^{\shortmid }\mathring{g}_{4})^{-1}\int dy^{3}\ \
^{\shortmid }\zeta _{4}^{[1]}(y^{3})/(\ _{2}^{\shortmid }\widehat{\Upsilon }%
)+\ \ ^{\shortmid }\zeta _{4}^{[2]}(x^{k}),  \label{diagsmepsilon} \\
&& \mbox{ for integration functions }\ ^{\shortmid }\zeta _{4}^{[1]}(y^{3})\ %
\mbox{ and }\ ^{\shortmid }\zeta _{4}^{[2]}(x^{k_{1}}),  \notag
\end{eqnarray}
\begin{eqnarray}
\ ^{\shortmid }\chi _{4} &=&(\ ^{\shortmid }\mathring{g}_{4})^{-1}\int
dy^{3}\ \ ^{\shortmid }\chi _{4}^{[1]}(y^{3})/(\ _{2}^{\shortmid }\widehat{%
\Upsilon })+\ \ ^{\shortmid }\chi _{4}^{[2]}(x^{k}),  \notag \\
&&\mbox{ for integration functions }\ ^{\shortmid }\chi _{4}^{[1]}(y^{3})\ %
\mbox{ and }\ ^{\shortmid }\chi _{4}^{[2]}(x^{k_{1}}),\mbox{ where }\
_{1}n_{k_{1}}=\ _{2}n_{k_{1}}=0;  \notag
\end{eqnarray}
\begin{eqnarray}
\ ^{\shortmid }\eta ^{5} &=&\zeta ^{5}(1+\varepsilon \ \chi ^{5}):\ \
^{\shortmid }\zeta ^{5}=(\ ^{\shortmid }\mathring{g}^{5})^{-1}\int dp_{6}\ \
^{\shortmid }\zeta _{\lbrack 1]}^{5}(p_{6})/(\ _{3}^{\shortmid }\widehat{%
\Upsilon })+\ \ ^{\shortmid }\zeta _{\lbrack 2]}^{5}(x^{k_{2}}),  \notag \\
&&\mbox{ for integration functions }\ \ ^{\shortmid }\zeta _{\lbrack
1]}^{5}(p_{6})\ \mbox{ and }\ \ ^{\shortmid }\zeta _{\lbrack
2]}^{5}(x^{k_{2}}),  \notag \\
\ \ ^{\shortmid }\chi ^{5} &=&(\ ^{\shortmid }\mathring{g}^{5})^{-1}\int
dp_{6}\ \ ^{\shortmid }\chi _{\lbrack 1]}^{5}(p_{6})/(\ _{3}^{\shortmid }%
\widehat{\Upsilon })+\ \ ^{\shortmid }\chi _{\lbrack 2]}^{5}(x^{k_{2}}),
\notag \\
&&\mbox{ for integration functions }\ \ ^{\shortmid }\chi _{\lbrack
1]}^{5}(p_{6})\ \mbox{ and }\ \ ^{\shortmid }\chi _{\lbrack
2]}^{5}(x^{k_{2}}),\mbox{ where }\ _{1}n_{k_{2}}=\ _{2}n_{k_{2}}=0;  \notag
\end{eqnarray}%
\begin{eqnarray}
\ ^{\shortmid }\eta ^{7} &=&\zeta ^{7}(1+\varepsilon \ \chi ^{7}):\ \
^{\shortmid }\zeta ^{7}=(\ ^{\shortmid }\mathring{g}^{7})^{-1}\int dE\ \
^{\shortmid }\zeta _{\lbrack 1]}^{7}(E)/(\ _{4}^{\shortmid }\widehat{%
\Upsilon })+\ \ ^{\shortmid }\zeta _{\lbrack 2]}^{7}(x^{k_{3}}),  \notag \\
&&\mbox{ for integration functions }\ \ ^{\shortmid }\zeta _{\lbrack
1]}^{7}(E)\ \mbox{ and }\ \ ^{\shortmid }\zeta _{\lbrack 2]}^{7}(x^{k_{3}}),
\notag \\
\ ^{\shortmid }\chi ^{7} &=&(\ ^{\shortmid }\mathring{g}^{7})^{-1}\int dE\ \
^{\shortmid }\chi _{\lbrack 1]}^{7}(E)/(\ _{4}^{\shortmid }\widehat{\Upsilon
})+\ \ ^{\shortmid }\chi _{\lbrack 2]}^{7}(x^{k_{3}}),  \notag
\end{eqnarray}
\end{consequence}

We note that we can separate the formulas for $\zeta $- and $\chi $%
-coefficients only for linear configurations on $\varepsilon .$ In general,
such formulas are nonlinear and with mixed $\zeta $- and $\chi $%
-deformations. In such cases, we transform the deformations into certain not
exact but small parametric decompositions of solutions. Formulas with small
parameters are important for analysing, for instance, small deformations of
horizons and nonlinear polarizations of black hole solutions in GR extended
on cotangent Lorentz bundles.

\subsection{Vacuum quasi-stationary phase space configurations}

\label{ssvacuumfc}The configurations with zero (effective) source and/or
cosmological constants can not be generated as particular cases of (off--)
diagonal solutions constructed in previous sections. This is because for $\
\Upsilon ,\Lambda \rightarrow 0$ these limits are not smooth for the
coefficients of s-metrics.

For quasi-stationary ansatz (\ref{ansatz1}), we can construct solutions when
the nontrivial coefficients of the Ricci s--tensor (\ref{riccist2})--(\ref%
{riccist4a}) in Lemma (\ref{lemmaricci}) are zero, see also respective
formulas (\ref{e2a})--(\ref{eq4a}). The equation (\ref{eq1}) is a 2--d
Laplace equation, when a solution can be expressed in the form $%
g_{i}=e^{\psi (x^{k},\Upsilon =0)}=e^{\ ^{0}\psi (x^{k})}.$ For a 4-d vacuum
spacetime, there are three general classes of off--diagonal stationary
metrics which result in zero coefficients of (\ref{e2a}). Any such class can
be extended for 8-d phase spaces with zero (effective) source.

\subsubsection{Off-diagonal vacuum phase spaces of type 1}

\begin{definition}
\label{vacuumt1} \textsf{[vacuum phase spaces of type 1] } The type 1 vacuum
off-diagonal quasi-stationary phase space configurations are defined by the
conditions $g_{4}^{\diamond }=0$ but $g_{4}\neq 0,$ $g_{3}^{\diamond }\neq 0
$ and $g_{3}\neq 0;\partial ^{6}(\ ^{\shortmid }g^{5})=0$ but $\ ^{\shortmid
}g^{5}\neq 0,$ $\partial ^{6}(\ ^{\shortmid }g^{6})\neq 0$ and $\
^{\shortmid }g^{6}\neq 0;$ and $(\ ^{\shortmid }g^{7})^{\ast }=0$ but $\
^{\shortmid }g^{7}\neq 0,$ $\partial ^{6}(\ ^{\shortmid }g^{8})\neq 0$ and $%
\ ^{\shortmid }g^{8}\neq 0.$
\end{definition}

Using equations (\ref{riccist2})-(\ref{riccist4a}) and (\ref{e2a})-(\ref%
{eq4a}), we formulate

\begin{corollary}
\label{integrvarvac1}\textsf{[vacuum integral varieties of type 1] } The
off-diagonal stationary vacuum spacetime of type 1 are defined by target
metrics generated by such geometric data:%
\begin{eqnarray}
s=2: &&\mbox{ arbitrary generating functions }:\ g_{3}(x^{i_{1}},y^{3})%
\mbox{ and }w_{k_{1}}(x^{i_{1}},y^{3}),  \notag \\
&&\mbox{ arbitrary integration functions }:\ g_{4}(x^{i_{1}})\mbox{ and }\
_{1}n_{k_{1}}(x^{i_{1}}),\ _{2}n_{k_{1}}(x^{i_{1}}),  \notag \\
&&\qquad n_{k_{1}}=\ _{1}n_{k_{1}}+\ _{2}n_{k_{1}}\int dy^{3}/g_{3};
\label{nvacsol1}
\end{eqnarray}
\begin{eqnarray*}
s=3: &&\mbox{ arbitrary generating functions }:\ \ ^{\shortmid
}g^{6}(x^{i_{1}},y^{a_{2}},p_{6})\mbox{ and }\ \ ^{\shortmid
}w_{k_{2}}(x^{i_{1}},y^{a_{2}},p_{6}),  \notag \\
&&\mbox{ arbitrary integration functions }:\ \ ^{\shortmid }g^{5}(x^{i_{2}})%
\mbox{ and }\ _{1}n_{k_{2}}(x^{i_{2}}),\ _{2}n_{k_{2}}(x^{i_{2}}),%
\mbox{ for
}x^{i_{2}}=(x^{i_{1}},y^{a_{2}}),  \notag \\
&&\qquad \ ^{\shortmid }n_{k_{2}}=\ _{1}n_{k_{2}}+\ _{2}n_{k_{2}}\int
dp_{6}/(\ ^{\shortmid }g^{6});  \notag \\
s=4: &&\mbox{ arbitrary generating functions }:\ \ ^{\shortmid
}g^{8}(x^{i_{1}},y^{a_{2}},p_{a_{3}},E)\mbox{ and }\ \ ^{\shortmid
}w_{k_{2}}(x^{i_{1}},y^{a_{2}},p_{a_{3}},E),  \notag \\
&&\mbox{ arbitrary integration functions }:\ \ ^{\shortmid }g^{7}(x^{i_{3}})%
\mbox{ and }\ _{1}n_{k_{3}}(x^{i_{3}}),\ _{2}n_{k_{3}}(x^{i_{3}}),%
\mbox{ for
}x^{i_{3}}=(x^{i_{2}},p_{a_{3}}),  \notag \\
&&\qquad \ ^{\shortmid }n_{k_{3}}=\ _{1}n_{k_{3}}+\ _{2}n_{k_{3}}\int dE/(\
^{\shortmid }g^{8}).  \notag
\end{eqnarray*}
\end{corollary}

\begin{proof}
If $g_{4}^{\diamond }=0$ and $\ _{2}^{\shortmid }\widehat{\Upsilon },$ the
equations (\ref{e2a}) are solved for any generating date $%
g_{3}(x^{i_{1}},y^{3})$ and $w_{k_{1}}(x^{i_{1}},y^{3})$ and $%
g_{4}(x^{i_{1}})$ which can be considered as an integrating function. We
have to solve only the decoupled equations for the rest of N-coefficients,
which can be written $n_{k_{1}}^{\diamond \diamond }+n_{k_{1}}^{\diamond }\
(\ln |g_{3}|)^{\diamond }=0.$ Integrating this equation (two times on $y^{3}$%
), \ we obtain the solution for $n_{k_{1}}$ in (\ref{e2a}). Similarly, we
find the solutions for $s=3$ and 4. $\square $\vskip5pt
\end{proof}

For adapted coefficients in above vacuum phase space s-metrics, we obtain

\begin{consequence}
\label{nqelvacuum1}\textsf{[nonlinear quadratic elements for vacuum phase
s-space metric of type 1] } Nonlinear quadratic elements $d\ ^{\shortmid
}s_{v1,8d}^{2}$ for "pure" quasi-stationary vacuum s-metrics are constructed
as a sum of quadratic stationary vacuum space metrics $ds_{v1,s=2}^{2}$ and
additional shells $s=3$ and $s=4$ vacuum contributions, denoted
respectively, $ds_{v1,s=3}^{2}$ and $ds_{v1,s=4}^{2},$ all being of type 1,
\begin{equation*}
d\ ^{\shortmid
}s_{v1,8d}^{2}=ds_{v,s1}^{2}+ds_{v1,s2}^{2}+ds_{v1,s3}^{2}+ds_{v1,s4}^{2},
\end{equation*}%
\begin{eqnarray*}
\mbox{ where }ds_{v,s1}^{2} &=&e^{\ ^{0}\psi
(x^{k})}[(dx^{1})^{2}+(dx^{2})^{2}],\  \\
ds_{v1,s2}^{2}
&=&g_{3}(x^{i_{1}},y^{3})[dy^{3}+w_{k_{1}}(x^{i_{1}},y^{3})dx^{k_{1}}]^{2}+
g_{4}(x^{i_{1}})[dy^{4}+(\ _{1}n_{k_{1}}(x^{i_{1}})+\
_{2}n_{k_{1}}(x^{i_{1}})\int dy^{3}/g_{3})dx^{k_{1}}]^{2},
\end{eqnarray*}%
\begin{eqnarray*}
ds_{v1,s3}^{2} &=&\ \ ^{\shortmid }g^{5}(x^{i_{2}})[dp_{5}+(\
_{1}n_{k_{2}}(x^{i_{2}})+ \ _{2}n_{k_{2}}(x^{i_{2}})\int dp_{6}/\
^{\shortmid }g^{6} dx^{k_{2}}]^{2}+ \ ^{\shortmid
}g^{6}(x^{i_{2}},p_{6})[dp_{6}+\
^{\shortmid}w_{k_{2}}(x^{i_{2}},p_{6})dx^{k_{2}}]^{2}, \\
ds_{v1,s4}^{2} &=&\ \ ^{\shortmid }g^{7}(x^{i_{3}})[dp_{7}+(\
_{1}n_{k_{3}}(x^{i_{3}}) +\ _{2}n_{k_{3}}(x^{i_{3}})\int dE/\ ^{\shortmid
}g^{8})dx^{k_{3}}]^{2} + \ ^{\shortmid }g^{8}(x^{i_{3}},E)[dE+\
^{\shortmid}w_{k_{3}}(x^{i_{3}},E)dx^{k_{3}}]^{2}.
\end{eqnarray*}
\end{consequence}

In above formulas, for instance, the labels $v1,s3$ states that such a
quadratic element extension is for the shall 3 with vacuum solutions of type
one.

\subsubsection{Off-diagonal vacuum phase spaces of type 2}

\begin{definition}
\label{vacuumt2} \textsf{[vacuum phase spaces of type 2] } The type 2 vacuum
off-diagonal quasi-stationary phase space configurations are defined by the
conditions $g_{3}^{\diamond }\neq 0$ and $g_{4}^{\diamond }\neq 0$; $%
\partial ^{6}(\ ^{\shortmid }g^{5})\neq 0$ and $\ ^{\shortmid }g^{6}\neq 0$ $%
;$ and $(\ ^{\shortmid }g^{7})^{\ast }\neq 0$ \ and $\partial ^{6}(\
^{\shortmid }g^{8})\neq 0.$
\end{definition}

For such configurations, we formulate

\begin{corollary}
\label{integrvarvac2}\textsf{[vacuum integral varieties of type 2] } The
off-diagonal stationary vacuum spacetime of type 2 are defined by target
metrics generated by such geometric data:%
\begin{eqnarray*}
s=2: &&\mbox{ arbitrary generating functions }:\ g_{4}(x^{i_{1}},y^{3})%
\mbox{ and }w_{k_{1}}(x^{i_{1}},y^{3}), \\
&&\mbox{ arbitrary integration functions }:\ \ ^{0}g_{3}(x^{k_{1}})%
\mbox{
and }\ _{1}n_{k_{1}}(x^{i_{1}}),\ _{2}n_{k_{1}}(x^{i_{1}}), \\
&&\qquad g_{3}=\ ^{0}g_{3}(x^{k_{1}})[(\sqrt{|g_{4}|})^{\diamond }]^{2},\
n_{k_{1}}=\ _{1}n_{k_{1}}+\ _{2}n_{k_{2}}\int
dy^{3}[(|g_{4}|^{3/4})^{\diamond }]^{2};
\end{eqnarray*}%
\begin{eqnarray*}
s=3: &&\mbox{ arbitrary generating functions }:\ \ ^{\shortmid
}g^{5}(x^{i_{1}},y^{a_{2}},p_{6})\mbox{ and }\ \ ^{\shortmid
}w_{k_{2}}(x^{i_{1}},y^{a_{2}},p_{6}), \\
&&\mbox{ arbitrary integration functions }:\ \ ^{0}g^{6}(x^{i_{2}})%
\mbox{
and }\ _{1}n_{k_{2}}(x^{i_{2}}),\ _{2}n_{k_{2}}(x^{i_{2}}),\mbox{ for
}x^{i_{2}}=(x^{i_{1}},y^{a_{2}}), \\
&&\qquad \ ^{\shortmid }g^{6}=\ ^{0}g^{6}(x^{k_{2}})[(\sqrt{\partial ^{6}|\
^{\shortmid }g^{5}|})]^{2},\ \ ^{\shortmid }n_{k_{2}}=\ _{1}n_{k_{2}}+\
_{2}n_{k_{2}}\int dp_{6}[\partial ^{6}(|\ ^{\shortmid }g^{5}|^{3/4})]^{2};
\end{eqnarray*}%
\begin{eqnarray*}
s=4: &&\mbox{ arbitrary generating functions }:\ \ ^{\shortmid
}g^{7}(x^{i_{1}},y^{a_{2}},p_{a_{3}},E)\mbox{ and }\ \ ^{\shortmid
}w_{k_{3}}(x^{i_{1}},y^{a_{2}},p_{a_{3}},E), \\
&&\mbox{ arbitrary integration functions }:\ \ \ ^{0}g^{8}(x^{i_{3}})%
\mbox{
and }\ _{1}n_{k_{3}}(x^{i_{3}}),\ _{2}n_{k_{3}}(x^{i_{3}}),\mbox{ for
}x^{i_{3}}=(x^{i_{2}},p_{a_{3}}), \\
&&\qquad \ \ ^{\shortmid }g^{8}=\ ^{0}g^{8}(x^{k_{2}})[(\sqrt{|\ ^{\shortmid
}g^{7}|^{\ast }})]^{2},\ \ ^{\shortmid }n_{k_{2}}=\ _{1}n_{k_{2}}+\
_{2}n_{k_{2}}\int dE[(|\ ^{\shortmid }g^{7}|^{3/4})^{\ast }]^{2}.
\end{eqnarray*}
\end{corollary}

\begin{proof}
Let us consider the shell $s=2.$ Taking%
\begin{equation}
\ _{2}^{\shortmid }\varpi =\ln |g_{4}^{\diamond }/\sqrt{|g_{3}g_{4}}|=\
_{2}^{\shortmid }\varpi _{0}=const  \label{auxil2s}
\end{equation}%
for $g_{3}^{\diamond }\neq 0$ and $g_{4}^{\diamond }\neq 0,$ we solve the
equations (\ref{e2a}) and (\ref{e2b}) because $\ _{2}^{\shortmid }\alpha
_{i_{1}}=0$ and $\ _{2}^{\shortmid }\beta =0$ but
\begin{equation}
\ _{2}^{\shortmid }\gamma =(\ln |g_{4}|^{3/4}/|g_{3}|)^{\diamond }\neq 0.
\label{auxil2sa}
\end{equation}%
We can rewrite (\ref{auxil2s}) as $\ g_{3}=\ ^{0}g_{3}(x^{k_{1}})[(\sqrt{%
|g_{4}|})^{\diamond }]^{2}$ for an integration function $\
^{0}g_{3}(x^{k_{1}})$ and generating function $g_{4}(x^{k_{1}},y^{3}).$ For
zero coefficients of the linear algebraic system (\ref{e2b}) , we can
consider arbitrary generating functions $w_{k_{1}}(x^{i_{1}},y^{3}).$
Integrating (\ref{e2c}) for (\ref{auxil2sa}), we obtain
\begin{equation*}
n_{k_{1}}(x^{i_{1}},y^{3})=\ _{1}n_{k_{1}}(x^{i_{1}})+\
_{2}n_{k}(x^{i_{1}})\int dy^{3}[(|g_{4}|^{3/4})^{\diamond }]^{2}.
\end{equation*}

Proofs for shells $s=3,4$ can be performed in similar forms but with
corresponding redefinitions of symbols and indices. $\square $\vskip5pt
\end{proof}

For adapted coefficients in above vacuum phase space s-metrics, we obtain

\begin{consequence}
\label{nqelvacuum2}\textsf{[nonlinear quadratic elements for vacuum phase
s-space metric of type 2] } Nonlinear quadratic elements $d\ ^{\shortmid
}s_{v2,8d}^{2}$ \ for "pure" quasi-stationary vacuum s-metrics are
constructed as a sum of quadratic stationary vacuum space metrics $%
ds_{v2,s2}^{2}$ and additional shells $s=3$ and $s=4$ vacuum contributions,
denoted respectively, $ds_{v2,s=3}^{2}$ and $ds_{v2,s=4}^{2},$ all being of
type 2,
\begin{equation*}
d\ ^{\shortmid
}s_{v2,8d}^{2}=ds_{v,s1}^{2}+ds_{v2,s2}^{2}+ds_{v2,s3}^{2}+ds_{v2,s4}^{2},
\end{equation*}%
where
\begin{eqnarray*}
ds_{v2,s1}^{2} &=&e^{\ ^{0}\psi (x^{k})}[(dx^{1})^{2}+(dx^{2})^{2}],\  \\
ds_{v2,s2}^{2} &=&\ ^{0}g_{3}(x^{k_{1}})[(\sqrt{|g_{4}|})^{\diamond
}]^{2}[dy^{3}+w_{k_{1}}(x^{i_{1}},y^{3})dx^{k_{1}}]^{2} \\
&&+g_{4}(x^{i_{1}},y^{3})[dy^{4}+(\ _{1}n_{k_{1}}(x^{i_{1}})+\
_{2}n_{k_{1}}(x^{i_{1}})\int dy^{3}[(|g_{4}|^{3/4})^{\diamond
}]^{2})dx^{k_{1}}]^{2},
\end{eqnarray*}%
\begin{eqnarray*}
ds_{v2,s3}^{2} &=&\ \ ^{\shortmid
}g^{5}(x^{i_{1}},y^{a_{2}},p_{6})[dp_{5}+(\ _{1}n_{k_{2}}(x^{i_{2}})+\
_{2}n_{k_{2}}(x^{i_{2}})\int dp_{6}[\partial ^{6}(|\ ^{\shortmid
}g^{5}|^{3/4})]^{2})dx^{k_{2}}]^{2} \\
&&+\ ^{0}g^{6}(x^{k_{2}})[(\sqrt{\partial ^{6}|\ ^{\shortmid }g^{5}|}%
)]^{2}[dp_{6}+\ ^{\shortmid }w_{k_{2}}(x^{i_{2}},p_{6})dx^{k_{2}}]^{2}, \\
ds_{v2,s4}^{2} &=&\ \ \ ^{\shortmid
}g^{7}(x^{i_{1}},y^{a_{2}},p_{a_{3}},E)[dp_{7}+(\ _{1}n_{k_{3}}(x^{i_{3}})+\
_{2}n_{k_{3}}(x^{i_{3}})\int dE[(|\ ^{\shortmid }g^{7}|^{3/4})^{\ast
}]^{2})dx^{k_{3}}]^{2} \\
&&+\ ^{0}g^{8}(x^{k_{2}})[(\sqrt{|\ ^{\shortmid }g^{7}|^{\ast }})]^{2}[dE+\
^{\shortmid }w_{k_{3}}(x^{i_{3}},E)dx^{k_{3}}]^{2}.
\end{eqnarray*}
\end{consequence}

In above formulas, for instance, the labels $v2,s4$ states that we consider
extensions for the shall 4 for vacuum solutions of type two.

\subsubsection{Off-diagonal vacuum spacetimes and phase spaces of type 3}

\begin{definition}
\label{vacuumt3} \textsf{[vacuum phase spaces of type 3] } The type 3 vacuum
off-diagonal quasi-stationary phase space configurations are defined by the
conditions $g_{3}^{\diamond }=0$ and $g_{4}^{\diamond }\neq 0;$ $\
^{\shortmid }g^{6}=0$ $\ $and $\partial ^{6}(\ ^{\shortmid }g^{5})\neq 0;$
and $\partial ^{6}(\ ^{\shortmid }g^{8})\neq 0$ and $(\ ^{\shortmid
}g^{7})^{\ast }\neq 0.$
\end{definition}

For such configurations, we formulate

\begin{corollary}
\label{integrvarvac3}\textsf{[vacuum integral varieties of type 3] } The
off-diagonal stationary vacuum spacetime of type 3 are defined by such
coefficients of target s-metrics:%
\begin{eqnarray*}
s=2: &&\mbox{ arbitrary generating functions }:\ w_{k_{1}}(x^{i_{1}},y^{3}),
\\
&&\mbox{ arbitrary integration functions }:\ \ ^{0}g_{3}(x^{k_{1}}),\
_{1}n_{k_{1}}(x^{i_{1}}),\
_{2}n_{k_{1}}(x^{i_{1}}),c_{1}(x^{k_{1}}),c_{2}(x^{k_{1}}), \\
&&\qquad g_{4}=\left[ c_{1}(x^{k_{1}})+c_{2}(x^{k_{1}})y^{3}\right] ^{2},\
n_{k_{1}}=\ _{1}n_{i_{1}}(x^{k_{1}})+\ \frac{\ _{2}n_{i_{1}}(x^{k_{1}})}{%
\left[ c_{1}(x^{k_{1}})+c_{2}(x^{k_{1}})y^{3}\right] ^{2}}; \\
s=3: &&\mbox{ arbitrary generating functions }:\ ^{\shortmid
}w_{k_{2}}(x^{i_{1}},y^{a_{2}},p_{6}), \\
&&\mbox{ arbitrary integration functions }:\ \ _{0}^{\shortmid
}g^{6}(x^{k_{2}}),\ _{1}n_{k_{2}}(x^{i_{2}}),\ _{2}n_{k_{2}}(x^{i_{2}}),\ \
^{3}c_{1}(x^{k_{2}}),\ \ ^{3}c_{2}(x^{k_{2}}), \\
&&\qquad \ ^{\shortmid }g^{5}=\left[ \ \ ^{3}c_{1}(x^{k_{2}})+\ \
^{3}c_{2}(x^{k_{2}})p_{6}\right] ^{2},\ n_{i_{2}}=\ _{1}n_{i_{2}}(x^{k_{2}})+%
\frac{\ _{2}n_{i_{2}}(x^{k_{2}})}{\left[ \ \ ^{3}c_{1}(x^{k_{2}})+\ \
^{3}c_{2}(x^{k_{2}})p_{6}\right] ^{2}};
\end{eqnarray*}%
\begin{eqnarray*}
s=4: &&\mbox{ arbitrary generating functions }:\ \ ^{\shortmid
}w_{k_{3}}(x^{i_{1}},y^{a_{2}},p_{a_{3}},E), \\
&&\mbox{ arbitrary integration functions }:\ \ \ _{0}^{\shortmid
}g^{8}(x^{k_{3}}),\ _{1}n_{k_{3}}(x^{i_{3}}),\ _{2}n_{k_{3}}(x^{i_{3}}),\
^{4}c_{1}(x^{k_{3}}),\ \ ^{4}c_{2}(x^{k_{3}}), \\
&&\qquad \ \ \ ^{\shortmid }g^{7}=\left[ \ ^{4}c_{1}(x^{k_{3}})+\ \
^{4}c_{2}(x^{k_{3}})E\right] ^{2},\ n_{i_{3}}=\ _{1}n_{i_{3}}(x^{k_{3}})+%
\frac{\ _{2}n_{i_{3}}(x^{k_{3}})}{\left[ \ ^{4}c_{1}(x^{k_{3}})+\
^{4}c_{2}(x^{k_{3}})E\right] ^{2}}.
\end{eqnarray*}
\end{corollary}

\begin{proof}
Let us consider, for instance, the shell $s=2.$ For the conditions $%
g_{4}^{\diamond }\neq 0$ and $g_{3}^{\diamond }=0,$ the first equation in (%
\ref{riccist2}), and (\ref{e2a}), with zero source transforms into $%
g_{4}^{\diamond \diamond }-\frac{\left( g_{4}^{\diamond }\right) ^{2}}{2g_{4}%
}=0,$ which can be solved all together by some v-coefficients
\begin{equation*}
g_{3}=\ ^{0}g_{3}(x^{k_{1}})\mbox{ and }g_{4}(x^{k_{1}},y^{3})=\left[
c_{1}(x^{k_{1}})+c_{2}(x^{k_{1}})y^{3}\right] ^{2},
\end{equation*}%
with generating functions $c_{1}(x^{k_{1}}),c_{2}(x^{k_{2}})$, and $g_{3}=\
^{0}g_{3}(x^{k_{1}}).$ Fixing the condition $\ _{2}^{\shortmid }\varpi =\ln
|g_{4}^{\diamond }/\sqrt{|g_{3}g_{4}}|=\ _{2}^{\shortmid }\varpi _{0}=const,$
when $\ _{2}^{\shortmid }\alpha _{i_{1}}=0$ and $\ _{2}^{\shortmid }\beta =0$
for the linear algebraic system (\ref{e2b}), we can consider arbitrary
generating functions $w_{k_{1}}(x^{i_{1}},y^{3}).$ Integrating (\ref{e2c})
for (\ref{auxil2sa}), we obtain
\begin{equation*}
n_{i_{1}}(x^{k_{1}},y^{3})=\ _{1}n_{i_{1}}(x^{k_{1}})+\
_{2}n_{i_{1}}(x^{k_{1}})\int dy^{3}|g_{4}|^{-3/2}=\
_{1}n_{i_{1}}(x^{k_{1}})+\ _{2}n_{i_{1}}(x^{k_{1}})\left[
c_{1}(x^{k_{1}})+c_{2}(x^{k_{1}})y^{3}\right] ^{-2},
\end{equation*}%
with redefined integration constants. Proofs for shells $s=3,4$ can be
performed in similar forms but with corresponding redefinitions of symbols
and indices. $\square $\vskip5pt
\end{proof}

Using above coefficients for vacuum phase space s-metrics, we obtain

\begin{consequence}
\label{nqelvacuum3}\textsf{[nonlinear quadratic elements for vacuum phase
s-space metric of type 3] } Nonlinear quadratic elements $d\ ^{\shortmid
}s_{v3,8d}^{2}$ for "pure" quasi-stationary vacuum s-metrics of type 3 are
constructed as a sum of quadratic stationary vacuum space metrics $%
ds_{v3,s2}^{2}$ and additional shells $s=3$ and $s=4$ vacuum contributions,
denoted respectively, $ds_{v3,s3}^{2}$ and $ds_{v3,s4}^{2},$ all being of
type 3,
\begin{equation*}
d\ ^{\shortmid
}s_{v3,8d}^{2}=ds_{v,s1}^{2}+ds_{v3,s2}^{2}+ds_{v3,s3}^{2}+ds_{v3,s4}^{2},
\end{equation*}%
\begin{eqnarray*}
\mbox{ where }ds_{v3,s1}^{2} &=&e^{\ ^{0}\psi
(x^{k})}[(dx^{1})^{2}+(dx^{2})^{2}],\  \\
ds_{v3,s2}^{2} &=&\ \
^{0}g_{3}(x^{k_{1}})[dy^{3}+w_{k_{1}}(x^{i_{1}},y^{3})dx^{k_{1}}]^{2} \\
&&+\left[ c_{1}(x^{k_{1}})+c_{2}(x^{k_{1}})y^{3}\right] ^{2}[dy^{4}+(\
_{1}n_{i_{1}}(x^{k_{1}})+\ \frac{\ _{2}n_{i_{1}}(x^{k_{1}})}{\left[
c_{1}(x^{k_{1}})+c_{2}(x^{k_{1}})y^{3}\right] ^{2}})dx^{i_{1}}]^{2},
\end{eqnarray*}%
\begin{eqnarray*}
ds_{v2,s3}^{2} &=&\left[ \ \ ^{3}c_{1}(x^{k_{2}})+\ \
^{3}c_{2}(x^{k_{2}})p_{6}\right] ^{2}\ [dp_{5}+(\ _{1}n_{i_{2}}(x^{k_{2}})+%
\frac{\ _{2}n_{i_{2}}(x^{k_{2}})}{\left[ \ \ ^{3}c_{1}(x^{k_{2}})+\ \
^{3}c_{2}(x^{k_{2}})p_{6}\right] ^{2}})dx^{i_{2}}]^{2} \\
&&+\ \ _{0}^{\shortmid }g^{6}(x^{k_{2}})[dp_{6}+\ ^{\shortmid
}w_{k_{2}}(x^{i_{2}},p_{6})dx^{k_{2}}]^{2}, \\
ds_{v2,s4}^{2} &=&\ \left[ \ ^{4}c_{1}(x^{k_{3}})+\ \ ^{4}c_{2}(x^{k_{3}})E%
\right] ^{2}[dp_{7}+(\ _{1}n_{i_{3}}(x^{k_{3}})+\frac{\
_{2}n_{i_{3}}(x^{k_{3}})}{\left[ \ ^{4}c_{1}(x^{k_{3}})+\
^{4}c_{2}(x^{k_{3}})E\right] ^{2}})dx^{i_{3}}]^{2} \\
&&+\ \ \ _{0}^{\shortmid }g^{8}(x^{k_{3}})[dE+\ ^{\shortmid
}w_{k_{3}}(x^{i_{3}},E)dx^{k_{3}}]^{2}.
\end{eqnarray*}
\end{consequence}

In above formulas, for instance, the labels $v3,s4$ states that we consider
extensions for the shall 4 for vacuum solutions of third type.

\subsubsection{Phase spaces with mixed vacuum shells, polarization
functions, LC-conditions etc.}

Considering, for instance, a type 1 vacuum solution for shells 1 and 2, we
can construct vacuum solutions when the shell is described by solutions of
type 1, 2, or 3; or by solutions of type 1, 2, 3. Summarizing Conclusions %
\ref{nqelvacuum1}, \ref{nqelvacuum2}, and \ref{nqelvacuum3} for vacuum
solutions of type $v[\overline{a}],$ $\overline{a},\overline{b},\overline{c},%
\overline{d}=1,2,3$, we prove

\begin{theorem}
\label{vacuumsolutions}\textsf{[generating vacuum solutions on dyadic shells
of cotangent Lorentz bundles] } Using nonholonomic frame deformations, we
can generate $3\times 3\times 3=27$ types of quadratic nonlinear elements
for target vacuum s-metrics
\begin{equation*}
d\ ^{\shortmid }s_{vac,8d}^{2}=ds_{v,s1}^{2}+ds_{v[\overline{b}%
],s2}^{2}+ds_{v[\overline{c}],s3}^{2}+ds_{v[\overline{d}],s4}^{2}.
\end{equation*}
\end{theorem}

\begin{example}
\textsf{[target prime metrics with mixed \ vacuum configurations] } A
nonlinear quadratic line element for vacuum solutions determined by
s-metrics of respective types 1, 2, 3 on shells 2,3,4 is parameterized
\begin{equation*}
d\ ^{\shortmid
}s_{v[1,2,3],8d}^{2}=ds_{v,s1}^{2}+ds_{v1,s2}^{2}+ds_{v2,s3}^{2}+ds_{v3,s4}^{2},
\end{equation*}%
\begin{eqnarray*}
\mbox{ where }ds_{v,s1}^{2} &=&e^{\ ^{0}\psi
(x^{k})}[(dx^{1})^{2}+(dx^{2})^{2}],\  \\
ds_{v1,s2}^{2}
&=&g_{3}(x^{i_{1}},y^{3})[dy^{3}+w_{k_{1}}(x^{i_{1}},y^{3})dx^{k_{1}}]^{2}+g_{4}(x^{i_{1}})[dy^{4}+(\ _{1}n_{k_{1}}(x^{i_{1}})+\ _{2}n_{k_{1}}(x^{i_{1}})\int dy^{3}/g_{3})dx^{k_{1}}]^{2},
\end{eqnarray*}%
\begin{eqnarray*}
ds_{v2,s3}^{2} &=&\ \ ^{\shortmid
}g^{5}(x^{i_{1}},y^{a_{2}},p_{6})[dp_{5}+(\ _{1}n_{k_{2}}(x^{i_{2}})+\
_{2}n_{k_{2}}(x^{i_{2}})\int dp_{6}[\partial ^{6}(|\ ^{\shortmid
}g^{5}|^{3/4})]^{2})dx^{k_{2}}]^{2} \\
&&+\ ^{0}g^{6}(x^{k_{2}})[(\sqrt{\partial ^{6}|\ ^{\shortmid }g^{5}|}%
)]^{2}[dp_{6}+\ ^{\shortmid }w_{k_{2}}(x^{i_{2}},p_{6})dx^{k_{2}}]^{2},
\end{eqnarray*}%
\begin{eqnarray*}
ds_{v3,s4}^{2} &=&\ \left[ \ ^{4}c_{1}(x^{k_{3}})+\ \ ^{4}c_{2}(x^{k_{3}})E%
\right] ^{2}[dp_{7}+(\ _{1}n_{i_{3}}(x^{k_{3}})+\frac{\
_{2}n_{i_{3}}(x^{k_{3}})}{\left[ \ ^{4}c_{1}(x^{k_{3}})+\
^{4}c_{2}(x^{k_{3}})E\right] ^{2}})dx^{i_{3}}]^{2} \\
&&+\ \ \ _{0}^{\shortmid }g^{8}(x^{k_{3}})[dE+\ ^{\shortmid
}w_{k_{3}}(x^{i_{3}},E)dx^{k_{3}}]^{2}.
\end{eqnarray*}
\end{example}

Similarly, \ we can generate quadratic elements with another type mixing of
vacuum shell configurations. \ We conclude this section with two important
remarks:

\begin{remark}
\textsf{[off-diagonal vacuum solutions and zero torsion] } \ We can restrict
vacuum integral varieties in various forms which result in LC-configurations
like in subsection \ref{ssslconf}. For instance, we can generate LC-vacuum
configurations with nonlinear quadratic elements
\begin{equation*}
d\ ^{\shortmid
}s_{v[2]LC,8d}^{2}=ds_{vLC,s1}^{2}+ds_{v1LC,s2}^{2}+ds_{v2LC,s3}^{2}+ds_{v2LC,s4}^{2},
\end{equation*}%
\begin{eqnarray*}
\mbox{ where }ds_{vLC,s1}^{2} &=&e^{\ ^{0}\psi
(x^{k})}[(dx^{1})^{2}+(dx^{2})^{2}], \\
ds_{v1LC,s2}^{2} &=&g_{3}(x^{i_{1}},y^{3})[dy^{3}+\partial _{i_{1}}[\
_{2}^{\shortmid }\check{A}%
(x^{i_{1}},y^{3})]dx^{i_{1}}]^{2}+g_{4}(x^{i_{1}})[dy^{4}+\partial
_{i_{1}}[\ ^{2}n(x^{k_{1}})]dx^{i_{1}}]^{2},
\end{eqnarray*}%
\begin{eqnarray*}
ds_{v2LC,s3}^{2} &=&\ \ ^{\shortmid
}g^{5}(x^{i_{1}},y^{a_{2}},p_{6})[dp_{5}+\partial _{i_{2}}[\
^{3}n(x^{k_{2}})]dx^{i_{2}}]^{2} \\
&&+\ ^{0}g^{6}(x^{k_{2}})[(\sqrt{\partial ^{6}|\ ^{\shortmid }g^{5}|}%
)]^{2}[dp_{6}+\partial _{k_{2}}(\ _{3}^{\shortmid }\check{A}%
(x^{i_{2}},p_{6}))dx^{k_{2}}]^{2},
\end{eqnarray*}%
\begin{eqnarray*}
ds_{v2LC,s4}^{2} &=&\ \ \ ^{\shortmid
}g^{7}(x^{i_{1}},y^{a_{2}},p_{a_{3}},E)[dp_{7}+\partial _{i_{3}}[\
^{3}n(x^{k_{3}})]dx^{i_{3}}]^{2} \\
&&+\ ^{0}g^{8}(x^{k_{2}})[(\sqrt{|\ ^{\shortmid }g^{7}|^{\ast }}%
)]^{2}[dE+\partial _{k_{3}}(\ _{4}^{\shortmid }\check{A}%
(x^{i_{3}},E))dx^{k_{3}}]^{2}.
\end{eqnarray*}
\end{remark}

\begin{remark}
\textsf{[off-diagonal vacuum solutions and polarization and generation
functions] } We can parameterize s-metrics for vacuum phase configurations
in terms of conventional $\Psi $-generating, $\widetilde{\Psi }$-generating
and/or other type generating functions and/or using gravitationala $\eta $%
-polarization or $\varepsilon $-polarization functions considered in above
subsections.
\end{remark}

\subsection{Generating solutions for Einstein-Hamilton spaces and
Finsler-Lagrange geometry}

In principle, Finsler like nonholonomic variables can be always introduced
on (co) tangent Lorentz bundles and on base Lorentz manifolds. Any s-metric
structure can be rewritten as a Finsler like d-metric one, and inversely,
using corresponding frame transforms. For instance, Finsler-Lagrange and
Hamilton spaces are characterized by different N- and d-connection
structures which results in different types of modifications of the Einstein
equations. Nevertheless, using nonlinear symmetries between generating
functions and generating source and corresponding transforms (studied in
section \ref{ssnonlsym}) we can redefine various classes of off-diagonal
solutions in a MGT into a corresponding classes of generalized
Einstein-Hamilton gravity models.

On a $\mathbf{TV}$ and its dual $\mathbf{T}^{\ast }\mathbf{V,}$ we can
always consider respective canonical data for Lagrange and Hamilton spaces, $%
(\widetilde{L},\ \widetilde{\mathbf{N}};\widetilde{\mathbf{e}}_{\alpha },%
\widetilde{\mathbf{e}}^{\alpha })$ and/or $(\widetilde{H},\ ^{\shortmid }%
\widetilde{\mathbf{N}};\ ^{\shortmid }\widetilde{\mathbf{e}}_{\alpha },\
^{\shortmid }\widetilde{\mathbf{e}}^{\alpha }),$ determined, for instance,
by a Hamilton generating function $H(x,p)$ (\ref{nqed}), with defined
Legendre and inverse Legendre transforms, respectively, (\ref{legendre}) and
(\ref{invlegendre}), see details \cite{v18a} and Appendix \ref{appendixb}.

\begin{definition}
\label{defehvariab} \textsf{[Hamilton and Finsler-Lagrange variables on
cotangent Lorentz bundles with dyadic structure] } Let us consider an open
region $U\subset $ $\mathbf{T}^{\ast }\mathbf{V}\ $ enabled with a
nonholonomic dyadic structure $\ ^{\shortmid }\mathbf{e}_{\alpha _{s}}$ (\ref%
{nadapbds}) following conditions of Lemma \ref{lemadap}. Sets of respective
nonholonomic variables (vielbein structures), $\ \mathbf{e}_{\ \alpha
_{s}}^{\alpha }\ $\ and $\ ^{\shortmid }\mathbf{e}_{\ \alpha _{s}}^{\alpha
},\ \ $ determined by a $H(x,p)$ (\ref{nqed}) and $\widetilde{\mathbf{e}}%
_{\alpha }$ (\ref{ccnadap}) are stated by frame transforms
\begin{eqnarray*}
\mathbf{e}_{\alpha _{s}} &=&\mathbf{e}_{\ \alpha _{s}}^{\alpha }\
^{\shortmid }\widetilde{\mathbf{e}}_{\alpha },%
\mbox{ for Lagrange (Finsler)
variables} \\
\ ^{\shortmid }\mathbf{e}_{\alpha _{s}} &=&\ ^{\shortmid }\mathbf{e}_{\
\alpha _{s}}^{\alpha }\ ^{\shortmid }\widetilde{\mathbf{e}}_{\alpha },\ %
\mbox{ for Hamilton variables}.
\end{eqnarray*}%
Considering dual bases and respective inverse matrices, $\ \mathbf{e}%
_{\alpha ^{\prime }\ }^{\ \alpha _{s}}\ $\ and $\ ^{\shortmid }\mathbf{e}%
_{\alpha ^{\prime }\ }^{\ \alpha _{s}},$ we can always (inversely) transform
geometric data
\begin{equation*}
(\ ^{s}\mathbf{N};\mathbf{e}_{\alpha _{s}},\mathbf{e}^{\alpha
_{s}})\longleftrightarrow (\widetilde{L},\ \widetilde{\mathbf{N}};\widetilde{%
\mathbf{e}}_{\alpha },\widetilde{\mathbf{e}}^{\alpha })\mbox{ and/or }(\
_{s}^{\shortmid }\widetilde{\mathbf{N}};\ ^{\shortmid }\widetilde{\mathbf{e}}%
_{\alpha _{s}},\ ^{\shortmid }\widetilde{\mathbf{e}}^{\alpha
_{s}})\longleftrightarrow (\widetilde{H},\ ^{\shortmid }\widetilde{\mathbf{N}%
};\ ^{\shortmid }\widetilde{\mathbf{e}}_{\alpha },\ ^{\shortmid }\widetilde{%
\mathbf{e}}^{\alpha }).
\end{equation*}%
Using an atlas covering\ $_{s}\mathbf{T}^{\ast }\mathbf{V,}$ we can
introduce Lagrange-Hamilton variables on (co) tangent bundles with dyadic
structure.
\end{definition}

Hereafter, we shall work only with Hamilton like variables considering that
we can always introduce Lagrange (Finsler) variables using necessary type
Legendre transforms and homogeneity conditions. Probing particles and waves
in phase spaces enabled with Hamilton like variables are described directly
by a corresponding geometric mechanics encoding in explicitly form MDR (\ref%
{mdrg}). In another turn, dyadic configurations allow us to decouple
generalized Einstein equations.

\begin{corollary}
\textsf{[s-metrics in Hamilton variables] } Any s-metric (\ref{dm2and2}) and
respective off-diagonal metric structure (\ref{offd}) can be described
equivalently as Hamilton canonical d-metric (\ref{cdmds}).
\end{corollary}

\begin{proof}
Let us consider frame transforms $\ ^{\shortmid }\mathbf{g}_{\alpha
_{s}\beta _{s}}=\ ^{\shortmid }\mathbf{e}_{\ \alpha _{s}}^{\alpha }\
^{\shortmid }\mathbf{e}_{\ \beta _{s}}^{\beta }\ ^{\shortmid }\widetilde{%
\mathbf{g}}_{\alpha \beta }.$ For a prescribed $H(x,p),$ i.e. a symmetric$\
^{\shortmid }\widetilde{\mathbf{g}}_{\alpha \beta },$ with 20 components
determined by a respective Hessian, and $\ ^{\shortmid }\mathbf{g}_{\alpha
_{s}\beta _{s}}$ (with 12 independent components\footnote{%
in GR, there are 6 independent degrees of freedom of a pseudo-Riemannian
metric; because of Bianchi identities we can eliminate via coordinate
transforms 4 coefficients from 10 coefficients of a symmetric second rank
tensor; 6 other independent degrees are obtained for (co) fiber components}
which may define an exact solution in MGT with MDRs). We can consider the
frame transforms of a d-metric into a s-metric as algebraic equations for
some $\ ^{\shortmid }\mathbf{e}_{\ \alpha _{s}}^{\alpha }$ (totally, there
are $8\times 8$ such coefficients). So, we have 12 independent algebraic
quadratic equations relating 20 coefficients of a Hamiltonian d-metric, 12
coefficients of s-metric and 64 coefficients of $\ ^{\shortmid }\mathbf{e}%
_{\ \alpha _{s}}^{\alpha }.$ Such a system can be always solved up to
certain classes of frame/coordinate transforms. $\square $\vskip5pt
\end{proof}

Different Finsler like and Hamilton MGTs are characterized by corresponding
d-connection distortions of type (\ref{candistr}), see Theorem \ref{thdistr}.

\begin{lemma}
Canonical distortion relations $\ _{s}^{\shortmid }\widehat{\mathbf{D}}=\
^{\shortmid }\widetilde{\mathbf{D}}+$ $\ _{s}^{\shortmid }\widehat{\mathbf{Z}%
}$ result in effective sources of type
\begin{equation*}
\ _{\shortmid }^{e}\widehat{\Upsilon }_{\alpha _{s}\beta _{s}}:=\varkappa (\
\ _{\shortmid }^{e}\widehat{\mathbf{T}}_{\alpha _{s}\beta _{s}}-\frac{1}{2}\
^{\shortmid }\mathbf{g}_{\alpha _{s}\beta _{s}}\ \ \ _{\shortmid }^{e}%
\widehat{\mathbf{T}}),
\end{equation*}%
where $\varkappa \ \ \ _{\shortmid }^{e}\widehat{\mathbf{T}}_{\alpha
_{s}\beta _{s}}=-\ ^{\shortmid }\widehat{\mathbf{Z}}_{\alpha _{s}\beta
_{s}}[\ _{s}^{\shortmid }\mathbf{g[\ ^{\shortmid }\widetilde{\mathbf{g}}}%
_{\alpha \beta }\mathbf{]},\ \ ^{\shortmid }\widetilde{\mathbf{D}}].$
\end{lemma}

\begin{proof}
Let us introduce on $_{s}\mathbf{T}^{\ast }\mathbf{V}$ both dyadic and
Hamilton variables consider frame transforms $\ ^{\shortmid }\mathbf{g}%
_{\alpha _{s}\beta _{s}}=\ ^{\shortmid }\mathbf{e}_{\ \alpha _{s}}^{\alpha
}\ ^{\shortmid }\mathbf{e}_{\ \beta _{s}}^{\beta }\ ^{\shortmid }\widetilde{%
\mathbf{g}}_{\alpha \beta },$ which defined a functional relation $\
_{s}^{\shortmid }\mathbf{g[\ ^{\shortmid }\widetilde{\mathbf{g}}}_{\alpha
\beta }].$ Following the conditions of Theorem \ref{thcandist}, we
can work with distortions $\ _{s}^{\shortmid }\widehat{\mathbf{D}}=\nabla +\
_{s}^{\shortmid }\widehat{\mathbf{Z}}$ $=\ ^{\shortmid }\widetilde{\mathbf{D}%
}+\ _{s}^{\shortmid }\widetilde{\mathbf{Z}},$ which allows to compute $\
_{s}^{\shortmid }\widehat{\mathbf{D}}=\ ^{\shortmid }\widetilde{\mathbf{D}}+$
$\ _{s}^{\shortmid }\widehat{\mathbf{Z}}.$ Such distortions are computed for
$\ ^{\shortmid }\mathbf{g}=\ _{s}^{\shortmid }\mathbf{g=}\ ^{\shortmid }%
\mathbf{\tilde{g},}$ and allow us to express the distortions of the Ricci
d-tensor of type $\ ^{\shortmid }\mathbf{R}_{\alpha _{s}\beta _{s}}$ (\ref%
{driccisd}) for distortions
\begin{equation*}
\ ^{\shortmid }\widehat{\mathbf{R}}_{\alpha _{s}\beta _{s}}=\ ^{\shortmid }%
\widetilde{\mathbf{R}}_{\alpha _{s}\beta _{s}}[\ \ ^{\shortmid }\mathbf{%
\tilde{g}},\ \ ^{\shortmid }\widetilde{\mathbf{D}}]+\ ^{\shortmid }%
\widetilde{\mathbf{Z}}_{\alpha _{s}\beta _{s}}[\ \ ^{\shortmid }\mathbf{%
\tilde{g}},\ \ ^{\shortmid }\widetilde{\mathbf{D}}],\mbox{ where }\
_{\shortmid }^{e}\widehat{\Upsilon }_{\alpha _{s}\beta _{s}}=-\ ^{\shortmid }%
\widetilde{\mathbf{Z}}_{\alpha _{s}\beta _{s}},
\end{equation*}%
see also formulas $\ _{\shortmid }\Upsilon _{\alpha _{s}\beta _{s}}$ (\ref%
{totaldiadsourcd}). $\square $\vskip5pt
\end{proof}

Using above Lemma, we prove

\begin{theorem}
\label{thehshell} \textsf{[generating Finsler like and Einstein-Hamilton
solutions on dyadic shells of cotangent Lorentz bundles] } The solutions of
canonically distorted modified Einstein equations, $\ ^{\shortmid }\widehat{%
\mathbf{R}}_{\alpha _{s}\beta _{s}}[\ _{s}^{\shortmid }\widehat{\mathbf{D}}%
]=\ ^{\shortmid }\widehat{\Upsilon }_{\alpha _{s}\beta _{s}}$ (\ref%
{meinsteqtbcand}), model solutions of generalized Einstein-Hamilton
equations, $\ ^{\shortmid }\widetilde{\mathbf{R}}_{\alpha \beta }[\
^{\shortmid }\widetilde{\mathbf{D}}]=\ ^{\shortmid }\widetilde{\Upsilon }%
_{\alpha \beta }$ (\ref{modifeinsthameqdiad}), where the dyadic indices can
be omitted for $\ ^{\shortmid }\mathbf{g}_{\alpha _{s}\beta _{s}}=\
^{\shortmid }\mathbf{e}_{\ \alpha _{s}}^{\alpha }\ ^{\shortmid }\mathbf{e}%
_{\ \beta _{s}}^{\beta }\ ^{\shortmid }\widetilde{\mathbf{g}}_{\alpha \beta
},$ if the sources are correspondingly related by formulas
\begin{equation*}
\ ^{\shortmid }\widehat{\Upsilon }_{\alpha _{s}\beta _{s}}=\ ^{\shortmid }%
\widetilde{\Upsilon }_{\alpha _{s}\beta _{s}}+\ _{\shortmid }^{e}\widehat{%
\Upsilon }_{\alpha _{s}\beta _{s}}.
\end{equation*}
\end{theorem}

For generating exact solutions, it is preferred to work with the canonical
s-connection $\ _{s}^{\shortmid }\widehat{\mathbf{D}}$ which allow a very
general decoupling of gravitational field equations. Introducing
nonholonomic deformations to $\ ^{\shortmid }\widetilde{\mathbf{D}},$ we
constrain the possibilities for decoupling by can distinguish Finsler like
variables for Hamilton configurations encoding directly MDRs (\ref{mdrg}).

For quasi-stationary solutions on \ $_{s}\mathbf{T}^{\ast }\mathbf{V,}$ we
prove this

\begin{consequence}
\textsf{[quasi-stationary s-metrics with cotangent bundle Hamilton
variables] } Any quasi-stationary phase space s-metric (not depending on
time like variables on base spacetime Lorentz manifold) can be described in
Hamilton variables with respective nonlinear symmetries (\ref{nonltransf})
with effective sources and effective cosmological constants of type (\ref%
{sheelsplitcan}) associated to respective co-fiber degrees of freedom, $\ $
\begin{equation*}
\ _{\shortmid }\widehat{\Upsilon }_{\alpha _{s}\beta _{s}}=\ _{\shortmid
}^{\phi }\Upsilon _{\alpha _{s}\beta _{s}}+\ _{\shortmid }^{e}\widehat{%
\Upsilon }_{\alpha _{s}\beta _{s}};\ \mbox{ and/or }\ \ \ _{\shortmid
s}\Lambda =\ \ _{\shortmid s}^{\phi }\Lambda +\ \ _{\shortmid s}^{e}%
\widetilde{\Lambda },
\end{equation*}%
where tilde in $\ \ _{\shortmid s}^{e}\widetilde{\Lambda }$ points to
possible contributions to the cosmological constant \ resulting from
Hamilton like degrees of freedom.
\end{consequence}

Finally, we note that any type of vacuum and/or non-vacuum quasi-stationary
off-diagonal solutions constructed in dyadic variables admit equivalent
descriptions in Hamilton like variables.

\section{Summary of Main Results and Final Remarks}

\label{s6} In this work, we elaborated on a dyadic nonholonomic geometric
formulation of modified gravity theories, MGTs, with modified distortion
relations, MDRs, modelled on (co) tangent Lorentz bundles. This is a second
partner paper to \cite{v18a}, where an axiomatic approach to such MGTs
(which can geometrized as Finsler-Lagrange-Hamilton spaces) was formulated.
Our main goal was to complete the anholonomic frame deformation method,
AFDM, with a new geometric technique for constructing exact and parametric
quasi-stationary solutions for relativistic phase spaces and
Einstein-Hamilton gravity theories.\footnote{%
for reviews of former results and bibliography, see sections B.4.10 - B.4.12
and B.4.18 - B.4.20 in \cite{v18a}}

Extensions of GR to gravity theories with MDRs are characterized by
effective metrics $g_{\alpha \beta }(x^{k},p_{a})$ determined by a
non-quadratic form, see formula (\ref{nqed}) in Appendix B, on cotangent
Lorentz bundles. Corresponding geometric models allow descriptions in
certain canonical Hamilton and/or almost symplectic variables and involve
non-Riemannian geometries. In order to elaborate a self-consistent causal
approach, we have to consider pseudo-Euclidean local signatures both on base
a spacetime Lorentz manifold and on (in general, curved) typical cofiber.
MDRs determine naturally the fundamental geometric and physical objects such
as nonlinear quadratic elements, nonlinear connections and adapted frames
and certain classes of distinguished connections. In result, we can
elaborate on geometric and axiomatic principles and physical motivations for
constructing MGTs with MDRs and corresponding generalizations of the
Einstein equations for metrics and connections depending additionally on
velocity/momentum variables. Our nonholonomic geometric approach is
elaborated in a general form for any types of MDRs with indicators encoding
possible LIVs, noncommutative and nonassociative classical and quantum
deformations, string theory contributions etc.

To construct exact solutions with nontrivial dependence both on base
spacetime coordinates and fiber type velocity/ momentum type coordinates is
important to develop new geometric methods for generating solutions of
systems of nonlinear partial differential equations, PDEs, and ordinary
differential equations, ODEs. We can not apply certain diagonal ansatz for
metrics depending, for instance, on one radial type variable, $g_{\alpha
\beta }\sim diag[g_{\gamma }(r)]$, (like for constructing the Schwarzschild
solution) and allowing us to transform the Einstein equations into a
nonlinear system of ODEs which can be integrated in exact form. Even in the
simplest case of stationary cotangent bundle configurations $g_{\alpha \beta
}(r,p_{8}=E)$ we have to work with generic off-diagonal systems subjected to
certain additional nonholonomic constraints and when the (modified) Einstein
equations transforms into nonlinear systems of PDEs. The solutions of such
equations are defined not only by integration constants (for instance, like
for ODEs and black hole solutions) but determined by various classes of
generating functions, generating (effective) sources, integration functions
and integration constants. This is typical for nonlinear dynamical and
evolution systems with generic off-diagonal and locally anisotropic
interactions.

The main results of this article consist from a series of important theorems
on quasi-stationary solutions which (in certain adapted variables) are
parameterized by dyadic shell s-metrics not depending on the time like
coordinate at least for projections on the base spacetime Lorentz manifold:

\begin{itemize}
\item Theorem \ref{theordecoupl} \textsf{[decoupling property] } states the
possibilty to decouple generalized Einstein equations encoding MDRs in
general off-diagonal forms and with dependence both on spacetime and (co)
fiber coordinates.

\item Theorem \ref{gensoltorsion} \textsf{[general solutions with
nonholonomically induced torsion] } defines the quadratic nonlinear elements
for quasi-stationary off-diagonal solutions on cotangent Lorentz bundles.
Using additional assumptions, we can extract Levi-Civita configurations
(with zero torsion).

\item Theorem \ref{nsymgfs} \textsf{[nonlinear symmetries of generating
functions and sources] } defines a new type of nonlinear symmetries relating
generating functions, generating (effective) sources, and certain effective
cosmological constants, which are very important for constructing exact
solutions with off-diagonal metrics and generalized connections.

\item Theorem \ref{epsilongeneration} \textsf{[generating exact solutions
with small parameters for generating functions and sources] }shows how small
parametric solutions can be generated in a self-consistent form with
deformation of certain prime metrics into new types of target one.

\item Theorem \ref{vacuumsolutions} \textsf{[generating vacuum solutions on
dyadic shells of cotangent Lorentz bundles] }classifies 27 dyadic types of
vacuum phase spaces defined my MDRs on cotangent Lorentz bundles.

\item Theorem \ref{thehshell} \textsf{[generating Finsler like and
Einstein-Hamilton solutions on dyadic shells of cotangent Lorentz bundles] }
proves that exact solutions in the Einstein-Hamilton gravity can be
constructed using nonlinear symmetries (with re-definition of generating
functions and generating sources, in corresponding nonholonomic variables)
of analogous solutions form MGTs with MDRs on cotangent Lorentz bundles.
\end{itemize}

These results on geometric and analytic methods for constructing general
quasi-stationary solutions will be applied in our further works for research
of black hole and cosmological solutions in theories with MDRs and
Einstein-Hamilton MGTs. During the review process, the methods from the
preprint version of this work, arXiv: 1806.045600, have been applied in a
series of partner recently published articles:\ \cite{bvepjc18}, containing
an updated version of the axiomatic part of \cite{v18a}, but not the Appendix B of that work containing a historical review on geometry and physics of Finsler-Lagrange-Hamilton modified gravity theories;\ \cite{bvap19}, on explicit black
hole solutions in modified Finsler-Lagrange-Hamilton gravity theories and
the G. Perelman W-entropy considered additionally to (or instead of) the
Bekenstein-Hawking entropy; and \cite{vcqg18}, for locally anisotropic
cosmological solutions with quasi-periodic and spacetime quasi-crystal
structures and further developments for MDRs with explicit Hamilton like
momentum type dependencies.

\vskip5pt \textbf{Acknowledgments:}\ During 2012-2018, S. Vacaru's research
program on modified Finsler gravity theories was supported and related to a
project IDEI in Romania, PN-II-ID-PCE-2011-3-0256, a series of research
fellowships for CERN and DAAD and an adjunct position at Fresno State
University, California, the USA. He is grateful to D. Singleton, V. A.
Kostelecky, and P. Stavrinos for important discussions and support.

\appendix

\setcounter{equation}{0} \renewcommand{\theequation}
{A.\arabic{equation}} \setcounter{subsection}{0}
\renewcommand{\thesubsection}
{A.\arabic{subsection}}

\section{ N-adapted Dyadic Coefficient Formulas}

\label{appendixa} We provide some important N-adapted formulas and examples
of proofs, see similar methods in \cite{v18a,gvvepjc14} and references
therein.

\subsection{Proof of Lemma \protect\ref{lemadap} for canonical N--connection
coefficients}

\label{appendixa1}There are canonical N-elongated bases and dual bases with
nonholonomic (4+4) splitting on $\mathbf{TV}$ enabled with N-connection
structure (\ref{ncon}) are constructed by definition,
\begin{eqnarray}
\mathbf{e}_{\alpha } &=&(\mathbf{e}_{i}=\frac{\partial }{\partial x^{i}}%
-N_{i}^{a}(x,v)\frac{\partial }{\partial v^{a}},e_{b}=\frac{\partial }{%
\partial v^{b}}),\mbox{ on }\ T\mathbf{TV;}  \label{nadap} \\
\mathbf{e}^{\alpha } &=&(e^{i}=dx^{i},\mathbf{e}%
^{a}=dv^{a}+N_{i}^{a}(x,v)dx^{i}),\mbox{ on }T^{\ast }\mathbf{TV}  \notag
\end{eqnarray}%
Here we note that a local basis $\mathbf{e}_{\alpha }$ is nonholonomic if
the commutators
\begin{equation}
\mathbf{e}_{[\alpha }\mathbf{e}_{\beta ]}:=\mathbf{e}_{\alpha }\mathbf{e}%
_{\beta }-\mathbf{e}_{\beta }\mathbf{e}_{\alpha }=C_{\alpha \beta }^{\gamma
}(u)\mathbf{e}_{\gamma }  \label{anhrel}
\end{equation}%
contain nontrivial anholonomy coefficients $C_{\alpha \beta }^{\gamma
}=\{C_{ia}^{b}=\partial _{a}N_{i}^{b},C_{ji}^{a}=\mathbf{e}_{j}N_{i}^{a}-%
\mathbf{e}_{i}N_{j}^{a}\}.$

Similarly, we define N-elongated dual bases, cobase, for cotangent bundles,
\begin{eqnarray}
\ ^{\shortmid }\mathbf{e}_{\alpha } &=&(\ ^{\shortmid }\mathbf{e}_{i}=\frac{%
\partial }{\partial x^{i}}-\ ^{\shortmid }N_{ia}(x,p)\frac{\partial }{%
\partial p_{a}},\ ^{\shortmid }e^{b}=\frac{\partial }{\partial p_{b}}),%
\mbox{ on }\ T\mathbf{T}^{\ast }\mathbf{V;}  \label{nadapd} \\
\ \ ^{\shortmid }\mathbf{e}^{\alpha } &=&(\ ^{\shortmid }e^{i}=dx^{i},\
^{\shortmid }\mathbf{e}_{a}=dp_{a}+\ ^{\shortmid }N_{ia}(x,p)dx^{i}),%
\mbox{
on }T^{\ast }\mathbf{T}^{\ast }\mathbf{V.}  \notag
\end{eqnarray}

For "shell by shell" decompositions, we consider N-adapted bases
\begin{equation*}
\mathbf{e}_{\alpha _{s}}=(\ \mathbf{e}_{i_{s}}=\frac{\partial }{\partial
x^{i_{s}}}-N_{i_{s}}^{a_{s}}\frac{\partial }{\partial v^{a_{s}}},e_{b_{s}}=%
\frac{\partial }{\partial v^{b_{s}}})\mbox{ on }\ \ _{s}T\mathbf{TV},
\end{equation*}
where (respectively, for $s=1,2,3,4$)
\begin{eqnarray}
\mathbf{e}_{\alpha _{1}} &=&e_{i_{1}}=\frac{\partial }{\partial x^{i_{1}}},\
\mathbf{e}_{\alpha _{2}}=(\ \mathbf{e}_{i_{2}}=\frac{\partial }{\partial
x^{i_{2}}}-N_{i_{2}}^{a_{2}}\frac{\partial }{\partial x^{a_{2}}},\ e_{b_{2}}=%
\frac{\partial }{\partial x^{b_{2}}}),\mbox{ where }i=(i_{1},i_{2}),  \notag
\\
\ \mathbf{e}_{\alpha _{3}} &=&(\ \mathbf{e}_{i_{3}}=\frac{\partial }{%
\partial x^{i_{3}}}-N_{i_{3}}^{a_{3}}\frac{\partial }{\partial v^{a_{3}}},\
e_{b_{3}}=\frac{\partial }{\partial v^{b_{3}}}),\ \mathbf{e}_{\alpha
_{4}}=(\ \mathbf{e}_{i_{4}}=\frac{\partial }{\partial x^{i_{4}}}%
-N_{i_{4}}^{a_{4}}\frac{\partial }{\partial v^{a_{4}}},e_{b_{4}}=\frac{%
\partial }{\partial v^{b_{4}}}).  \label{nadapbs}
\end{eqnarray}%
Similarly, there are defined N-elongated bases
\begin{equation*}
\ ^{\shortmid }\mathbf{e}_{\alpha _{s}}=(\ \ ^{\shortmid }\mathbf{e}%
_{i_{s}}=\ \frac{\partial }{\partial x^{i_{s}}}-\ ^{\shortmid }N_{\
i_{s}a_{s}}\frac{\partial }{\partial p_{a_{s}}},\ \ ^{\shortmid }e^{b_{s}}=%
\frac{\partial }{\partial p_{b_{s}}})\mbox{ on }\ \ _{s}T\mathbf{T}^{\ast }%
\mathbf{V,}
\end{equation*}%
when
\begin{eqnarray}
\ ^{\shortmid }\mathbf{e}_{\alpha _{1}} &=&\ \ ^{\shortmid }e_{i_{1}}=\frac{%
\partial }{\partial x^{i_{1}}},\ ^{\shortmid }\mathbf{e}_{\alpha _{2}}=(\
^{\shortmid }\mathbf{e}_{i_{2}}=\frac{\partial }{\partial x^{i_{2}}}-\
^{\shortmid }N_{i_{2}}^{a_{2}}\frac{\partial }{\partial x^{a_{2}}},\ \
^{\shortmid }e_{b_{2}}=\frac{\partial }{\partial x^{b_{2}}}),
\label{nadapbds} \\
\ \ ^{\shortmid }\mathbf{e}_{\alpha _{3}} &=&(\ \ ^{\shortmid }\mathbf{e}%
_{i_{3}}=\ \frac{\partial }{\partial x^{i_{3}}}-\ ^{\shortmid }N_{i_{3}a_{3}}%
\frac{\partial }{\partial p_{a_{3}}},\ \ ^{\shortmid }e^{b_{3}}=\frac{%
\partial }{\partial p_{b_{3}}}),\ \ ^{\shortmid }\mathbf{e}_{\alpha _{4}}=(\
\ ^{\shortmid }\mathbf{e}_{i_{4}}=\ \frac{\partial }{\partial x^{i_{4}}}-\
^{\shortmid }N_{\ i_{4}a_{4}}\frac{\partial }{\partial p_{a_{4}}},\ \
^{\shortmid }e^{b_{4}}=\frac{\partial }{\partial p_{b_{4}}}).  \notag
\end{eqnarray}%
In dual form, we introduce (we omit explicit shell parameterizations)%
\begin{eqnarray}
\mathbf{e}^{\alpha _{s}} &=&(e^{i_{s}}=dx^{i_{s}},\mathbf{e}%
^{a_{s}}=dv^{a_{s}}+N_{j_{s}}^{a_{s}}dx^{j_{s}}),\mbox{ on }\ \ _{s}T^{\ast }%
\mathbf{TV;}  \label{nadapbdss} \\
\ ^{\shortmid }\mathbf{e}^{\alpha _{s}} &=&(\ ^{\shortmid }e^{i}=dx^{i},\
^{\shortmid }\mathbf{e}_{a}=dp_{a}+\ ^{\shortmid }N_{ia}(x,p)dx^{i})%
\mbox{
on }\ \ _{s}T^{\ast }\mathbf{T}^{\ast }\mathbf{V}  \notag
\end{eqnarray}

Finally, we note a N--connection on $\mathbf{TV,}$ or $\mathbf{T}^{\ast }%
\mathbf{V,}$ is characterized by such coefficients of N--connection
curvature (called also Neijenhuis tensors)
\begin{equation}
\Omega _{ij}^{a}=\frac{\partial N_{i}^{a}}{\partial x^{j}}-\frac{\partial
N_{j}^{a}}{\partial x^{i}}+N_{i}^{b}\frac{\partial N_{j}^{a}}{\partial y^{b}}%
-N_{j}^{b}\frac{\partial N_{i}^{a}}{\partial y^{b}},\mbox{\ or\ }\ \ \mathbf{%
\ ^{\shortmid }}\Omega _{ija}=\frac{\partial \mathbf{\ ^{\shortmid }}N_{ia}}{%
\partial x^{j}}-\frac{\partial \mathbf{\ ^{\shortmid }}N_{ja}}{\partial x^{i}%
}+\ \mathbf{^{\shortmid }}N_{ib}\frac{\partial \mathbf{\ ^{\shortmid }}N_{ja}%
}{\partial p_{b}}-\mathbf{\ ^{\shortmid }}N_{jb}\frac{\partial \mathbf{\
^{\shortmid }}N_{ia}}{\partial p_{b}}.  \label{neijtc}
\end{equation}%
Similar formulas can be written for dyadic decompositions considering shell
indices and respective coordinates.

\subsection{Off-diagonal local coefficients for d-metrics}

Introducing formulas of type (\ref{cnddapb}) and (\ref{ccnadap}),
respectively, into (\ref{dmt}), (\ref{dmct}) and (\ref{dmts}), (\ref{dmcts})
and regrouping with respect to local coordinate bases, one proves

\begin{corollary}
\textsf{[equivalent re-writing of d-metrics and s-metrics as off-diagonal
metrics]} With respect to local coordinate frames, any d--metric structures
on $\mathbf{TV,}$ $\mathbf{T}^{\ast }\mathbf{V}$ and s-metric on $\mathbf{T}%
^{\ast }\mathbf{V}$
\begin{eqnarray*}
\mathbf{g} &=&\mathbf{g}_{\alpha \beta }(x,y)\mathbf{e}^{\alpha }\mathbf{%
\otimes e}^{\beta }=g_{\underline{\alpha }\underline{\beta }}(x,y)du^{%
\underline{\alpha }}\mathbf{\otimes }du^{\underline{\beta }}\mbox{
and/or } \\
\ ^{\shortmid }\mathbf{g} &=&\ ^{\shortmid }\mathbf{g}_{\alpha \beta }(x,p)\
^{\shortmid }\mathbf{e}^{\alpha }\mathbf{\otimes \ ^{\shortmid }e}^{\beta
}=\ ^{\shortmid }g_{\underline{\alpha }\underline{\beta }}(x,p)d\
^{\shortmid }u^{\underline{\alpha }}\mathbf{\otimes }d\ ^{\shortmid }u^{%
\underline{\beta }},
\end{eqnarray*}%
can be parameterized via frame transforms, $\mathbf{g}_{\alpha \beta }=e_{\
\alpha }^{\underline{\alpha }}e_{\ \beta }^{\underline{\beta }}g_{\underline{%
\alpha }\underline{\beta }},$ $\ ^{\shortmid }\mathbf{g}_{\alpha \beta }=\
^{\shortmid }e_{\ \alpha }^{\underline{\alpha }}\ ^{\shortmid }e_{\ \beta }^{%
\underline{\beta }}\ ^{\shortmid }g_{\underline{\alpha }\underline{\beta }},$
in respective off-diagonal forms:
\begin{eqnarray}
g_{\underline{\alpha }\underline{\beta }} &=&\left[
\begin{array}{cc}
g_{ij}(x)+g_{ab}(x,y)N_{i}^{a}(x,y)N_{j}^{b}(x,y) & g_{ae}(x,y)N_{j}^{e}(x,y)
\\
g_{be}(x,y)N_{i}^{e}(x,y) & g_{ab}(x,y)%
\end{array}%
\right] ,  \notag \\
\ ^{\shortmid }g_{\underline{\alpha }\underline{\beta }} &=&\left[
\begin{array}{cc}
\ ^{\shortmid }g_{ij}(x)+\ ^{\shortmid }g^{ab}(x,p)\ ^{\shortmid
}N_{ia}(x,p)\ ^{\shortmid }N_{jb}(x,p) & \ ^{\shortmid }g^{ae}\ ^{\shortmid
}N_{je}(x,p) \\
\ ^{\shortmid }g^{be}\ ^{\shortmid }N_{ie}(x,p) & \ ^{\shortmid
}g^{ab}(x,p)\
\end{array}%
\right]  \label{offd}
\end{eqnarray}%
and, for nonholonomic (2+2)+(2+2) splitting on $\ _{s}\mathbf{T}^{\ast }%
\mathbf{V}$ determined by respective data (\ref{ncon2}) and with coordinate
parameterizations in footnote \ref{fnotnconcoef2}, using frame transforms $\
^{\shortmid }\mathbf{g}_{\alpha _{s}\beta _{s}}=\ ^{\shortmid }e_{\ \alpha
_{s}}^{\underline{\alpha }}\ ^{\shortmid }e_{\ \beta _{s}}^{\underline{\beta
}}\ ^{\shortmid }g_{\underline{\alpha }\underline{\beta }}$ the coefficients
of (\ref{dmcts}) are parameterized as follow:%
\begin{eqnarray}
\ _{s}^{\shortmid }\mathbf{g} &=&\ ^{\shortmid }\mathbf{g}_{\alpha _{s}\beta
_{s}}(x,y,\ _{s}p)\ ^{\shortmid }\mathbf{e}^{\alpha _{s}}\mathbf{\otimes \
^{\shortmid }e}^{\beta _{s}}=\ ^{\shortmid }\mathbf{g}%
_{i_{s}j_{s}}(x^{k_{s}})\ ^{\shortmid }\mathbf{e}^{i_{s}}\mathbf{\otimes \
^{\shortmid }e}^{j_{s}}+\ ^{\shortmid }\mathbf{g}%
^{a_{s}b_{s}}(x^{k_{s}},p_{c_{s}})\ ^{\shortmid }\mathbf{e}_{a_{s}}\mathbf{%
\otimes \ ^{\shortmid }e}_{b_{s}}  \notag \\
&=&\ ^{\shortmid }\mathbf{g}_{i_{1}j_{1}}(x^{k_{1}})\ dx^{i_{1}}\mathbf{%
\otimes \ }dx^{j_{1}}+\ ^{\shortmid }\mathbf{g}%
_{a_{2}b_{2}}(x^{k_{1}},x^{c_{2}})\ ^{\shortmid }\mathbf{e}^{a_{2}}\mathbf{%
\otimes \ ^{\shortmid }e}^{b_{2}}  \label{dm2and2} \\
&&+\ ^{\shortmid }\mathbf{g}^{a_{3}b_{3}}(x^{k_{1}},x^{c_{2}},p_{c_{3}})\
^{\shortmid }\mathbf{e}_{a_{3}}\mathbf{\otimes \ ^{\shortmid }e}_{b_{3}}+\
^{\shortmid }\mathbf{g}%
^{a_{4}b_{4}}(x^{k_{1}},x^{c_{2}},p_{c_{3}},p_{c_{4}})\ ^{\shortmid }\mathbf{%
e}_{a_{4}}\mathbf{\otimes \ ^{\shortmid }e}_{b_{4}},  \notag \\
&&\mbox{ where }  \notag \\
\ ^{\shortmid }\mathbf{e}^{i_{1}} &=&dx^{i_{1}},\mbox{ for }i_{1}=1,2  \notag
\\
\ ^{\shortmid }\mathbf{e}^{a_{2}} &=&dx^{a_{2}}+\ ^{\shortmid
}N_{j_{1}}^{a_{2}}(x^{k_{1}},x^{c_{2}})\ dx^{i_{1}},\mbox{ for }j_{1}=1,2%
\mbox{ and }a_{2},c_{2}=3,4;  \notag \\
\ ^{\shortmid }\mathbf{e}_{a_{3}} &=&dp_{a_{3}}+\ ^{\shortmid }N_{j_{2}\
a_{3}}(x^{k_{1}},x^{c_{2}},p_{c_{3}})\ dx^{e_{2}};\mbox{ for }%
j_{2}=1,2,3,4;e_{2}=3,4;a_{3},c_{3}=5,6;  \notag \\
\ ^{\shortmid }\mathbf{e}_{a_{4}} &=&dp_{a_{4}}+\ ^{\shortmid }N_{j_{3}\
a_{4}}(x^{k_{1}},x^{k_{2}},p_{c_{3}},p_{c_{4}})\ dp_{e_{3}};\mbox{
for }j_{3}=1,2,3,4,5,6;a_{4},c_{4}=7,8;  \notag
\end{eqnarray}%
and, in generic off-diagonal local coordinate form,%
\begin{equation*}
\ ^{\shortmid }\mathbf{g}=\ _{s}^{\shortmid }\mathbf{g}=\ _{s}^{\shortmid
}g_{\underline{\alpha }\underline{\beta }}(x,p)d\ ^{\shortmid }u^{\underline{%
\alpha }}\mathbf{\otimes }d\ ^{\shortmid }u^{\underline{\beta }},
\end{equation*}%
the coefficients can be expressed in matrix forms{\small
\begin{equation}
\ _{s}^{\shortmid }g_{\underline{\alpha }\underline{\beta }}=\left[
\begin{array}{cccc}
\begin{array}{c}
\lbrack \ ^{\shortmid }g_{i_{1}j_{1}}+ \\
\ ^{\shortmid }g_{a_{2}b_{2}}\ ^{\shortmid }N_{i_{1}}^{a_{2}}\ ^{\shortmid
}N_{j_{1}}^{a_{2}}+ \\
\ ^{\shortmid }g^{c_{3}f_{3}}\ ^{\shortmid }N_{i_{1}c_{3}}\ ^{\shortmid
}N_{j_{1}f_{3}}+ \\
\ ^{\shortmid }g^{c_{4}f_{4}}\ ^{\shortmid }N_{i_{1}c_{4}}\ ^{\shortmid
}N_{j_{1}f_{4}}]%
\end{array}
& \ ^{\shortmid }g_{a_{2}b_{2}}\ ^{\shortmid }N_{j_{1}}^{b_{2}} & \ \
^{\shortmid }g^{a_{3}e_{3}}\ ^{\shortmid }N_{j_{1}e_{3}} & \ ^{\shortmid
}g^{a_{4}e_{4}}\ ^{\shortmid }N_{j_{1}e_{4}} \\
\ ^{\shortmid }g_{a_{2}b_{2}}\ ^{\shortmid }N_{i_{1}}^{b_{2}} &
\begin{array}{c}
\lbrack \ ^{\shortmid }g_{a_{2}b_{2}}+ \\
\ ^{\shortmid }g^{c_{3}f_{3}}\ ^{\shortmid }N_{a_{2}c_{3}}\ ^{\shortmid
}N_{b_{2}f_{3}}+ \\
\ ^{\shortmid }g^{c_{4}f_{4}}\ ^{\shortmid }N_{a_{2}c_{4}}\ ^{\shortmid
}N_{b_{2}f_{4}}]%
\end{array}
& \ \ ^{\shortmid }g^{a_{3}e_{3}}\ ^{\shortmid }N_{a_{2}e_{3}} & \
^{\shortmid }g^{a_{4}e_{4}}\ ^{\shortmid }N_{a_{2}e_{4}} \\
\ ^{\shortmid }g^{b_{3}e_{3}}\ ^{\shortmid }N_{i_{1}e_{3}} & \ ^{\shortmid
}g^{a_{3}e_{3}}\ ^{\shortmid }N_{i_{2}e_{3}} & [\ ^{\shortmid
}g^{a_{3}b_{3}}+\ ^{\shortmid }g^{c_{4}f_{4}}\ ^{\shortmid }N_{\
c_{4}}^{a_{3}}\ ^{\shortmid }N_{\ f_{4}}^{b_{3}}] & \ ^{\shortmid
}g^{a_{4}e_{4}}\ ^{\shortmid }N_{\ e_{4}}^{a_{3}} \\
\ ^{\shortmid }g^{b_{4}e_{4}}\ ^{\shortmid }N_{i_{1}e_{4}} & \ ^{\shortmid
}g^{b_{4}e_{4}}\ ^{\shortmid }N_{i_{2}e_{4}} & \ ^{\shortmid
}g^{b_{4}e_{4}}\ ^{\shortmid }N_{\ e_{4}}^{a_{3}} & \ ^{\shortmid
}g^{a_{4}b_{4}}%
\end{array}%
\right] .  \label{offds}
\end{equation}%
}
\end{corollary}

Parameterizations of type (\ref{offd}) are considered, for instance, in the
Kaluza--Klein theory for associated vector bundles when gauge fields
interactions are modelled for extra dimension theories. These types of d-
and/or s-metrics are generic off-diagonal if the corresponding N-adapted
structure is not integrable. For MDR-generalizations of the Einstein
gravity, we can consider that the h-metrics $g_{ij}(x)=\ ^{\shortmid
}g_{ij}(x)$ are determined by a solution of standard Einstein equations but
the terms with $N$--coefficients are determined by solutions of certain
generalized gravitational field equations on nonholonomic phase spaces with
a corresponding shell by shell splitting. In general, such nonholonomic
dyadic solutions are not compactified on velocity/ momentum like
coordinates, $v^{a} $ / $p_{a}$ like in the standard Kaluza-Klein models.

Finally, we note that a s-metric written in a form (\ref{dm2and2}), describe
a 8-d phase space with curved both the base manifold and typical fiber
subspaces which is different from (\ref{lqed}) defined for extensions of a
base metric $\mathbf{g}_{ij}(x^{k})$ on a Lorentz manifold $\mathbf{V}$ in
GR to $\ ^{\shortmid }\mathbf{g}_{\alpha \beta }=[\mathbf{g}%
_{ij}(x^{k}),\eta ^{ab}]=[\mathbf{g}_{ij}(x^{k}),\eta ^{a_{3}b_{3}},\eta
^{a_{4}b_{4}}]$ on $T^{\ast }\mathbf{V}$ or $\ _{s}T^{\ast }\mathbf{V}$ with
a flat typical fiber characterized by $\eta ^{ab}=[\eta
^{a_{3}b_{3}}=diag(1,1),\eta ^{a_{4}b_{4}}=diag(1,-1)]$.

\subsection{Curvatures of d--connections with dyadic structure}

By explicit computations for $\mathbf{X}=\mathbf{e}_{\alpha },\mathbf{Y}=%
\mathbf{e}_{\beta },\mathbf{D}=\{\mathbf{\Gamma }_{\ \alpha \beta }^{\gamma
}\},$ $\ ^{\shortmid }\mathbf{X}=\ ^{\shortmid }\mathbf{e}_{\alpha },$ $\
^{\shortmid }\mathbf{Y}=\ ^{\shortmid }\mathbf{e}_{\beta },\ ^{\shortmid }%
\mathbf{D}=\{\ ^{\shortmid }\mathbf{\Gamma }_{\ \alpha \beta }^{\gamma }\},$
and $\ _{s}^{\shortmid }\mathbf{X}=\ ^{\shortmid }\mathbf{e}_{\alpha _{s}},$
$\ _{s}^{\shortmid }\mathbf{Y}=\ ^{\shortmid }\mathbf{e}_{\beta _{s}},\
_{s}^{\shortmid }\mathbf{D}=\{\ ^{\shortmid }\mathbf{\Gamma }_{\ \alpha
_{s}\beta _{s}}^{\gamma _{s}}\}$ introduced respectively in (\ref{dcurvabstr}%
) and (\ref{scurvabstr}), we prove

\begin{corollary}
\label{acorolcurv}For a d--connection $\mathbf{D}$ or $\ ^{\shortmid }%
\mathbf{D,}$ \ there are computed corresponding N--adapted coefficients:
\newline
d-curvature, $\mathcal{R}=\mathbf{\{R}_{\ \beta \gamma \delta }^{\alpha
}=(R_{\ hjk}^{i},R_{\ bjk}^{a},P_{\ hja}^{i},P_{\ bja}^{c},S_{\
hba}^{i},S_{\ bea}^{c})\},$ for
\begin{eqnarray}
R_{\ hjk}^{i} &=&\mathbf{e}_{k}L_{\ hj}^{i}-\mathbf{e}_{j}L_{\ hk}^{i}+L_{\
hj}^{m}L_{\ mk}^{i}-L_{\ hk}^{m}L_{\ mj}^{i}-C_{\ ha}^{i}\Omega _{\ kj}^{a},
\notag \\
R_{\ bjk}^{a} &=&\mathbf{e}_{k}\acute{L}_{\ bj}^{a}-\mathbf{e}_{j}\acute{L}%
_{\ bk}^{a}+\acute{L}_{\ bj}^{c}\acute{L}_{\ ck}^{a}-\acute{L}_{\ bk}^{c}%
\acute{L}_{\ cj}^{a}-C_{\ bc}^{a}\Omega _{\ kj}^{c},  \label{dcurv} \\
P_{\ jka}^{i} &=&e_{a}L_{\ jk}^{i}-D_{k}\acute{C}_{\ ja}^{i}+\acute{C}_{\
jb}^{i}T_{\ ka}^{b},\ P_{\ bka}^{c}=e_{a}\acute{L}_{\ bk}^{c}-D_{k}C_{\
ba}^{c}+C_{\ bd}^{c}T_{\ ka}^{c},  \notag \\
S_{\ jbc}^{i} &=&e_{c}\acute{C}_{\ jb}^{i}-e_{b}\acute{C}_{\ jc}^{i}+\acute{C%
}_{\ jb}^{h}\acute{C}_{\ hc}^{i}-\acute{C}_{\ jc}^{h}\acute{C}_{\ hb}^{i},%
\hspace{0in}\ S_{\ bcd}^{a}=e_{d}C_{\ bc}^{a}-e_{c}C_{\ bd}^{a}+C_{\
bc}^{e}C_{\ ed}^{a}-C_{\ bd}^{e}C_{\ ec}^{a},  \notag
\end{eqnarray}%
or $\ ^{\shortmid }\mathcal{R}=\mathbf{\{\ ^{\shortmid }R}_{\ \beta \gamma
\delta }^{\alpha }=(\ ^{\shortmid }R_{\ hjk}^{i},\ ^{\shortmid }R_{a\ jk}^{\
b},\ ^{\shortmid }P_{\ hj}^{i\ \ \ a},\ ^{\shortmid }P_{c\ j}^{\ b\ a},\
^{\shortmid }S_{\ hba}^{i},\ ^{\shortmid }S_{\ bea}^{c})\},$ for
\begin{eqnarray*}
\ ^{\shortmid }R_{\ hjk}^{i} &=&\ ^{\shortmid }\mathbf{e}_{k}\ ^{\shortmid
}L_{\ hj}^{i}-\ ^{\shortmid }\mathbf{e}_{j}\ ^{\shortmid }L_{\ hk}^{i}+\
^{\shortmid }L_{\ hj}^{m}\ ^{\shortmid }L_{\ mk}^{i}-\ ^{\shortmid }L_{\
hk}^{m}\ ^{\shortmid }L_{\ mj}^{i}-\ ^{\shortmid }C_{\ h}^{i\ a}\
^{\shortmid }\Omega _{akj}, \\
\ ^{\shortmid }R_{a\ jk}^{\ b} &=&\ ^{\shortmid }\mathbf{e}_{k}\ ^{\shortmid
}\acute{L}_{a\ j}^{\ b}-\ ^{\shortmid }\mathbf{e}_{j}\ ^{\shortmid }\acute{L}%
_{a\ k}^{\ b}+\ ^{\shortmid }\acute{L}_{c\ j}^{\ b}\ ^{\shortmid }\acute{L}%
_{a\ k}^{\ c}-\ ^{\shortmid }\acute{L}_{c\ k}^{\ b}\ ^{\shortmid }\acute{L}%
_{a\ j}^{\ c}-\ ^{\shortmid }C_{a\ }^{\ bc}\ ^{\shortmid }\Omega _{ckj}, \\
\ ^{\shortmid }P_{\ jk}^{i\ \ \ a} &=&\ ^{\shortmid }e^{a}\ ^{\shortmid
}L_{\ jk}^{i}-\ ^{\shortmid }D_{k}\ ^{\shortmid }\acute{C}_{\ j}^{i\ a}+\
^{\shortmid }\acute{C}_{\ j}^{i\ b}\ ^{\shortmid }T_{bk}^{\ \ \ a},\ \
^{\shortmid }P_{c\ k}^{\ b\ a}=\ ^{\shortmid }e^{a}\ ^{\shortmid }\acute{L}%
_{c\ k}^{\ b}-\ ^{\shortmid }D_{k}\ ^{\shortmid }C_{c\ }^{\ ba}+\
^{\shortmid }C_{\ bd}^{c}\ ^{\shortmid }T_{\ ka}^{c}, \\
\ ^{\shortmid }S_{\ j}^{i\ bc} &=&\ ^{\shortmid }e^{c}\ ^{\shortmid }\acute{C%
}_{\ j}^{i\ b}-\ ^{\shortmid }e^{b}\ ^{\shortmid }\acute{C}_{\ j}^{i\ c}+\
^{\shortmid }\acute{C}_{\ j}^{h\ b}\ ^{\shortmid }\acute{C}_{\ h}^{i\ c}-\
^{\shortmid }\acute{C}_{\ j}^{h\ c}\ ^{\shortmid }\acute{C}_{\ h}^{i\ b}, \\
\ ^{\shortmid }S_{a\ }^{\ bcd} &=&\ ^{\shortmid }e^{d}\ ^{\shortmid }C_{a\
}^{\ bc}-\ ^{\shortmid }e^{c}\ ^{\shortmid }C_{a\ }^{\ bd}+\ ^{\shortmid
}C_{a\ }^{\ bc}\ ^{\shortmid }C_{b\ }^{\ ed}-\ ^{\shortmid }C_{e}^{\ bd}\
^{\shortmid }C_{a}^{\ ec};
\end{eqnarray*}%
d-torsion, $\ \mathcal{T}=\{\mathbf{T}_{\ \alpha \beta }^{\gamma }=(T_{\
jk}^{i},T_{\ ja}^{i},T_{\ ji}^{a},T_{\ bi}^{a},T_{\ bc}^{a})\},$ for
\begin{equation}
T_{\ jk}^{i}=L_{jk}^{i}-L_{kj}^{i},T_{\ jb}^{i}=C_{jb}^{i},T_{\
ji}^{a}=-\Omega _{\ ji}^{a},\ T_{aj}^{c}=L_{aj}^{c}-e_{a}(N_{j}^{c}),T_{\
bc}^{a}=C_{bc}^{a}-C_{cb}^{a},  \label{dtors}
\end{equation}%
or $\ ^{\shortmid }\mathcal{T}=\{\ ^{\shortmid }\mathbf{T}_{\ \alpha \beta
}^{\gamma }=(\ ^{\shortmid }T_{\ jk}^{i},\ ^{\shortmid }T_{\ j}^{i\ a},\
^{\shortmid }T_{aji},\ ^{\shortmid }T_{a\ i}^{\ b},\ ^{\shortmid }T_{a\ }^{\
bc})\},$ for
\begin{equation*}
\ ^{\shortmid }T_{\ jk}^{i}=\ ^{\shortmid }L_{jk}^{i}-\ ^{\shortmid
}L_{kj}^{i},\ ^{\shortmid }T_{\ j}^{i\ a}=\ ^{\shortmid }C_{j}^{ia},\
^{\shortmid }T_{aji}=-\ ^{\shortmid }\Omega _{aji},\ \ ^{\shortmid }T_{c\
j}^{\ a}=\ ^{\shortmid }L_{c\ j}^{\ a}-\ ^{\shortmid }e^{a}(\ ^{\shortmid
}N_{cj}),\ ^{\shortmid }T_{a\ }^{\ bc}=\ ^{\shortmid }C_{a}^{\ bc}-\
^{\shortmid }C_{a}^{\ cb};
\end{equation*}%
d-nonmetricity, $\ \mathcal{Q}=\mathbf{\{Q}_{\gamma \alpha \beta }=\left(
Q_{kij},Q_{kab},Q_{cij},Q_{cab}\right) \},$ for
\begin{equation}
Q_{kij}=D_{k}g_{ij},Q_{kab}=D_{k}g_{ab},Q_{cij}=D_{c}g_{ij},Q_{cab}=D_{c}g_{ab}
\label{dnonm}
\end{equation}%
or $\ ^{\shortmid }\mathcal{Q}=\mathbf{\{\ ^{\shortmid }Q}_{\gamma \alpha
\beta }=\left( \ ^{\shortmid }Q_{kij},\ ^{\shortmid }Q_{kab},\ ^{\shortmid
}Q_{cij},\ ^{\shortmid }Q_{cab}\right) \},$ for
\begin{equation*}
\ ^{\shortmid }Q_{kij}=\ ^{\shortmid }D_{k}\ ^{\shortmid }g_{ij},\
^{\shortmid }Q_{k}^{\ ab}=\ ^{\shortmid }D_{k}\ ^{\shortmid }g^{ab},\
^{\shortmid }Q_{\ ij}^{c}=\ ^{\shortmid }D^{c}\ ^{\shortmid }g_{ij},\
^{\shortmid }Q^{cab}=\ ^{\shortmid }D^{c}\ ^{\shortmid }g^{ab}.
\end{equation*}%
N-adapted formulas for $\ _{s}\mathbf{T}^{\ast }\mathbf{V}$ are written
respectively: \newline
$\ _{s}^{\shortmid }\mathcal{R}=\mathbf{\{\ ^{\shortmid }R}_{\ \beta
_{s}\gamma _{s}\delta _{s}}^{\alpha _{s}}=(\ ^{\shortmid }R_{\
h_{s}j_{s}k_{s}}^{i_{s}},\ ^{\shortmid }R_{a_{s}\ j_{s}k_{s}}^{\ b_{s}},\
^{\shortmid }P_{\ h_{s}j_{s}}^{i_{s}\ \ \ a_{s}},\ ^{\shortmid }P_{c_{s}\
j_{s}}^{\ b_{s}\ a_{s}},\ ^{\shortmid }S_{\ h_{s}b_{s}a_{s}}^{i_{s}},\
^{\shortmid }S_{\ b_{s}e_{s}a_{s}}^{c_{s}})\},$ for
\begin{eqnarray*}
\ ^{\shortmid }R_{\ h_{s}j_{s}k_{s}}^{i_{s}} &=&\ ^{\shortmid }\mathbf{e}%
_{k_{s}}\ ^{\shortmid }L_{\ h_{s}j_{s}}^{i_{s}}-\ ^{\shortmid }\mathbf{e}%
_{j_{s}}\ ^{\shortmid }L_{\ h_{s}k_{s}}^{i_{s}}+\ ^{\shortmid }L_{\
h_{s}j_{s}}^{m_{s}}\ ^{\shortmid }L_{\ m_{s}k_{s}}^{i_{s}}-\ ^{\shortmid
}L_{\ h_{s}k_{s}}^{m_{s}}\ ^{\shortmid }L_{\ m_{s}j_{s}}^{i_{s}}-\
^{\shortmid }C_{\ h_{s}}^{i_{s}\ a_{s}}\ ^{\shortmid }\Omega
_{a_{s}k_{s}j_{s}}, \\
\ ^{\shortmid }R_{a_{s}\ j_{s}k_{s}}^{\ b_{s}} &=&\ ^{\shortmid }\mathbf{e}%
_{k_{s}}\ ^{\shortmid }\acute{L}_{a_{s}\ j_{s}}^{\ b_{s}}-\ ^{\shortmid }%
\mathbf{e}_{j_{s}}\ ^{\shortmid }\acute{L}_{a_{s}\ k_{s}}^{\ b_{s}}+\
^{\shortmid }\acute{L}_{c_{s}\ j_{s}}^{\ b_{s}}\ ^{\shortmid }\acute{L}%
_{a_{s}\ k_{s}}^{\ c_{s}}-\ ^{\shortmid }\acute{L}_{c_{s}\ k_{s}}^{\ b_{s}}\
^{\shortmid }\acute{L}_{a_{s}\ j_{s}}^{\ c_{s}}-\ ^{\shortmid }C_{a_{s}\
}^{\ b_{s}c_{s}}\ ^{\shortmid }\Omega _{c_{s}k_{s}j_{s}}, \\
\ ^{\shortmid }P_{\ j_{s}k_{s}}^{i_{s}\ \ \ a_{s}} &=&\ ^{\shortmid
}e^{a_{s}}\ ^{\shortmid }L_{\ j_{s}k_{s}}^{i_{s}}-\ ^{\shortmid }D_{k_{s}}\
^{\shortmid }\acute{C}_{\ j_{s}}^{i_{s}\ a_{s}}+\ ^{\shortmid }\acute{C}_{\
j_{s}}^{i_{s}\ b_{s}}\ ^{\shortmid }T_{b_{s}k_{s}}^{\ \ \ a_{s}}, \\
\ \ ^{\shortmid }P_{c_{s}\ k_{s}}^{\ b_{s}\ a_{s}} &=&\ ^{\shortmid
}e^{a_{s}}\ ^{\shortmid }\acute{L}_{c_{s}\ k_{s}}^{\ b_{s}}-\ ^{\shortmid
}D_{k_{s}}\ ^{\shortmid }C_{c_{s}\ }^{\ b_{s}a_{s}}+\ ^{\shortmid }C_{\
b_{s}d_{s}}^{c_{s}}\ ^{\shortmid }T_{\ k_{s}a_{s}}^{c_{s}}, \\
\ ^{\shortmid }S_{\ j_{s}}^{i_{s}\ b_{s}c_{s}} &=&\ ^{\shortmid }e^{c_{s}}\
^{\shortmid }\acute{C}_{\ j_{s}}^{i_{s}\ b_{s}}-\ ^{\shortmid }e^{b_{s}}\
^{\shortmid }\acute{C}_{\ j_{s}}^{i_{s}\ c_{s}}+\ ^{\shortmid }\acute{C}_{\
j_{s}}^{h_{s}\ b_{s}}\ ^{\shortmid }\acute{C}_{\ h_{s}}^{i_{s}\ c_{s}}-\
^{\shortmid }\acute{C}_{\ j_{s}}^{h_{s}\ c_{s}}\ ^{\shortmid }\acute{C}_{\
h_{s}}^{i_{s}\ b_{s}}, \\
\ ^{\shortmid }S_{a_{s}\ }^{\ b_{s}c_{s}d_{s}} &=&\ ^{\shortmid }e^{d_{s}}\
^{\shortmid }C_{a_{s}\ }^{\ b_{s}c_{s}}-\ ^{\shortmid }e^{c_{s}}\
^{\shortmid }C_{a_{s}\ }^{\ b_{s}d_{s}}+\ ^{\shortmid }C_{a_{s}\ }^{\
b_{s}c_{s}}\ ^{\shortmid }C_{b_{s}\ }^{\ e_{s}d_{s}}-\ ^{\shortmid
}C_{e_{s}}^{\ b_{s}d_{s}}\ ^{\shortmid }C_{a_{s}}^{\ e_{s}c_{s}};
\end{eqnarray*}%
s-torsion $\ _{s}^{\shortmid }\mathcal{T}=\{\ ^{\shortmid }\mathbf{T}_{\
\alpha _{s}\beta _{s}}^{\gamma _{s}}=(\ ^{\shortmid }T_{\
j_{s}k_{s}}^{i_{s}},\ ^{\shortmid }T_{\ j_{s}}^{i_{s}\ a_{s}},\ ^{\shortmid
}T_{a_{s}j_{s}i_{s}},\ ^{\shortmid }T_{a_{s}\ i_{s}}^{\ b_{s}},\ ^{\shortmid
}T_{a_{s}\ }^{\ b_{s}c_{s}})\},$ for
\begin{eqnarray*}
\ ^{\shortmid }T_{\ j_{s}k_{s}}^{i_{s}} &=&\ ^{\shortmid
}L_{j_{s}k_{s}}^{i_{s}}-\ ^{\shortmid }L_{k_{s}j_{s}}^{i_{s}},\ ^{\shortmid
}T_{\ j_{s}}^{i_{s}\ a_{s}}=\ ^{\shortmid }C_{j_{s}}^{i_{s}a_{s}},\
^{\shortmid }T_{a_{s}j_{s}i_{s}}=-\ ^{\shortmid }\Omega _{a_{s}j_{s}i_{s}},
\\
\ \ ^{\shortmid }T_{c_{s}\ j_{s}}^{\ a_{s}} &=&\ ^{\shortmid }L_{c_{s}\
j_{s}}^{\ a_{s}}-\ ^{\shortmid }e^{a_{s}}(\ ^{\shortmid }N_{c_{s}j_{s}}),\
^{\shortmid }T_{a_{s}\ }^{\ b_{s}c_{s}}=\ ^{\shortmid }C_{a_{s}}^{\
b_{s}c_{s}}-\ ^{\shortmid }C_{a_{s}}^{\ c_{s}b_{s}};
\end{eqnarray*}%
s-nonmetricity $\ _{s}^{\shortmid }\mathcal{Q}=\mathbf{\{\ ^{\shortmid }Q}%
_{\gamma _{s}\alpha _{s}\beta _{s}}=\left( \ ^{\shortmid
}Q_{k_{s}i_{s}j_{s}},\ ^{\shortmid }Q_{k_{s}a_{s}b_{s}},\ ^{\shortmid
}Q_{c_{s}i_{s}j_{s}},\ ^{\shortmid }Q_{c_{s}a_{s}b_{s}}\right) \},$ for
\begin{equation*}
\ ^{\shortmid }Q_{k_{s}i_{s}j_{s}}=\ ^{\shortmid }D_{k_{s}}\ ^{\shortmid
}g_{i_{s}j_{s}},\ ^{\shortmid }Q_{k_{s}}^{\ a_{s}b_{s}}=\ ^{\shortmid
}D_{k_{s}}\ ^{\shortmid }g^{a_{s}b_{s}},\ ^{\shortmid }Q_{\
i_{s}j_{s}}^{c_{s}}=\ ^{\shortmid }D^{c_{s}}\ ^{\shortmid }g_{i_{s}j_{s}},\
^{\shortmid }Q^{c_{s}a_{s}b_{s}}=\ ^{\shortmid }D^{c_{s}}\ ^{\shortmid
}g^{a_{s}b_{s}}.
\end{equation*}
\end{corollary}

Similar formulas as in this Remarks can be proven for $\ _{s}\mathbf{TV}$
(the coefficients are written without label "$\ ^{\shortmid }$" and $%
a,b,c... $ symbols go up, or down, comparing with corresponding low, or up,
ones, and inversely).

\subsection{The coefficients of canonical d-connections and dyadic splitting}

Such d-connections are very important for elaborating MGTs on (co) tangent
bundles because they allow a very general decoupling and integration of
generalized Einstein and matter field equations. By explicit computations in
N-adapted frames, we can prove that necessary conditions for defining and
constructing, respectively, $\widehat{\mathbf{D}}$ (\ref{canondcl}) and $\
^{\shortmid }\widehat{\mathbf{D}}$ (\ref{canondch}), are satisfied following

\begin{corollary}
\label{acoroldand}The N-adapted coefficients of canonical Lagrange and
Hamilton d-connections are computed respectively:
\begin{eqnarray}
\mbox{ on }T\mathbf{TV},\ \widehat{\mathbf{D}} &=&\{\widehat{\mathbf{\Gamma }%
}_{\ \alpha \beta }^{\gamma }=(\widehat{L}_{jk}^{i},\widehat{L}_{bk}^{a},%
\widehat{C}_{jc}^{i},\widehat{C}_{bc}^{a})\},\mbox{ for }\lbrack \mathbf{g}%
_{\alpha \beta }=(g_{jr},g_{ab})\mathbf{,N=\{}N_{i}^{a}\mathbf{\}]},  \notag
\\
\widehat{L}_{jk}^{i} &=&\frac{1}{2}g^{ir}\left( \mathbf{e}_{k}g_{jr}+\mathbf{%
e}_{j}g_{kr}-\mathbf{e}_{r}g_{jk}\right) ,\ \widehat{L}%
_{bk}^{a}=e_{b}(N_{k}^{a})+\frac{1}{2}g^{ac}(e_{k}g_{bc}-g_{dc}\
e_{b}N_{k}^{d}-g_{db}\ e_{c}N_{k}^{d}),  \notag \\
\widehat{C}_{jc}^{i} &=&\frac{1}{2}g^{ik}e_{c}g_{jk},\ \widehat{C}_{bc}^{a}=%
\frac{1}{2}g^{ad}\left( e_{c}g_{bd}+e_{b}g_{cd}-e_{d}g_{bc}\right)
\label{canlc}
\end{eqnarray}%
\begin{eqnarray}
\mbox{ and, on }T\mathbf{T}^{\ast }\mathbf{V},\ \ ^{\shortmid }\widehat{%
\mathbf{D}} &=&\{\ ^{\shortmid }\widehat{\mathbf{\Gamma }}_{\ \alpha \beta
}^{\gamma }=(\ ^{\shortmid }\widehat{L}_{jk}^{i},\ ^{\shortmid }\widehat{L}%
_{a\ k}^{\ b},\ ^{\shortmid }\widehat{C}_{\ j}^{i\ c},\ ^{\shortmid }%
\widehat{C}_{\ j}^{i\ c})\},\mbox{ for }\lbrack \ ^{\shortmid }\mathbf{g}%
_{\alpha \beta }=(\ ^{\shortmid }g_{jr},\ ^{\shortmid }g^{ab})\mathbf{,\
^{\shortmid }N=\{}\ ^{\shortmid }N_{ai}\mathbf{\}]},  \notag \\
\ ^{\shortmid }\widehat{L}_{jk}^{i} &=&\frac{1}{2}\ ^{\shortmid }g^{ir}(\
^{\shortmid }\mathbf{e}_{k}\ ^{\shortmid }g_{jr}+\ ^{\shortmid }\mathbf{e}%
_{j}\ ^{\shortmid }g_{kr}-\ ^{\shortmid }\mathbf{e}_{r}\ ^{\shortmid
}g_{jk}),\   \notag \\
\ ^{\shortmid }\widehat{L}_{a\ k}^{\ b} &=&\ ^{\shortmid }e^{b}(\
^{\shortmid }N_{ak})+\frac{1}{2}\ ^{\shortmid }g_{ac}(\ ^{\shortmid }e_{k}\
^{\shortmid }g^{bc}-\ ^{\shortmid }g^{dc}\ \ ^{\shortmid }e^{b}\ ^{\shortmid
}N_{dk}-\ ^{\shortmid }g^{db}\ \ ^{\shortmid }e^{c}\ ^{\shortmid }N_{dk}),
\notag \\
\ ^{\shortmid }\widehat{C}_{\ j}^{i\ c} &=&\frac{1}{2}\ ^{\shortmid }g^{ik}\
^{\shortmid }e^{c}\ ^{\shortmid }g_{jk},\ \ ^{\shortmid }\widehat{C}_{\
a}^{b\ c}=\frac{1}{2}\ ^{\shortmid }g_{ad}(\ ^{\shortmid }e^{c}\ ^{\shortmid
}g^{bd}+\ ^{\shortmid }e^{b}\ ^{\shortmid }g^{cd}-\ ^{\shortmid }e^{d}\
^{\shortmid }g^{bc}).  \label{canhc}
\end{eqnarray}%
We use formulas (\ref{canlc}) for the shells $s=1,2$ (with $%
i_{1},j_{1},...=1,2$ and $a_{2},b_{2},...=3,4)$ of the cotangent Lorentz
bundle,
\begin{eqnarray}
\mbox{ on }\ \ _{2}T\mathbf{T}^{\ast }\mathbf{V},\ \ _{2}^{\shortmid }%
\widehat{\mathbf{D}} &=&\{\widehat{\mathbf{\Gamma }}_{\ \alpha _{2}\beta
_{2}}^{\gamma _{2}}=(\widehat{L}_{j_{1}k_{1}}^{i_{1}},\widehat{L}%
_{b_{2}k_{1}}^{a_{2}},\widehat{C}_{j_{1}c_{2}}^{i_{1}},\widehat{C}%
_{b_{2}c_{2}}^{a_{2}})\},\mbox{ for }\lbrack \mathbf{g}_{\alpha _{1}\beta
_{1}}=(g_{j_{1}r_{1}},g_{a_{2}b_{2}})\mathbf{,}\ \ _{2}^{\shortmid }\mathbf{%
N=\{}N_{i_{1}}^{a_{2}}\mathbf{\}]},  \notag \\
\widehat{L}_{j_{1}k_{1}}^{i_{1}} &=&\frac{1}{2}g^{i_{1}r_{1}}\left( \mathbf{e%
}_{k_{1}}g_{j_{1}r_{1}}+\mathbf{e}_{j_{1}}g_{k_{1}r_{1}}-\mathbf{e}%
_{r_{1}}g_{j_{1}k_{1}}\right) ,\   \notag \\
\widehat{L}_{b_{2}k_{1}}^{a_{2}} &=&e_{b_{2}}(N_{k_{1}}^{a_{2}})+\frac{1}{2}%
g^{a_{2}c_{2}}(e_{k_{1}}g_{b_{2}c_{2}}-g_{d_{2}c_{2}}\
e_{b_{2}}N_{k_{1}}^{d_{2}}-g_{d_{2}b_{2}}\ e_{c_{2}}N_{k_{1}}^{d_{2}}),
\notag \\
\widehat{C}_{j_{1}c_{2}}^{i_{1}} &=&\frac{1}{2}%
g^{i_{1}k_{1}}e_{c_{2}}g_{j_{1}k_{1}},\ \widehat{C}_{b_{2}c_{2}}^{a_{2}}=%
\frac{1}{2}g^{a_{2}d_{2}}\left(
e_{c_{2}}g_{b_{2}d_{2}}+e_{b_{2}}g_{c_{2}d_{2}}-e_{d_{2}}g_{b_{2}c_{2}}%
\right) ,  \label{candcons12}
\end{eqnarray}%
and formulas (\ref{canhc}) for shells $s=3,4$ ($i_{3},j_{3},...=1,2,3,4$ and
$a_{3},b_{3}=5,6;i_{4},j_{4},...=1,2,3,4,5,6$ and $a_{4},b_{4}=7,8)$ and%
\begin{eqnarray}
\mbox{on }\ _{s}T\mathbf{T}^{\ast }\mathbf{V},\ \ _{s}^{\shortmid }\widehat{%
\mathbf{D}} &=&\{\ ^{\shortmid }\widehat{\mathbf{\Gamma }}_{\ \alpha
_{s}\beta _{s}}^{\gamma _{s}}=(\ ^{\shortmid }\widehat{L}%
_{j_{s}k_{s}}^{i_{s}},\ ^{\shortmid }\widehat{L}_{a_{s}\ k_{s}}^{\ b_{s}},\
^{\shortmid }\widehat{C}_{\ j_{s}}^{i_{s}\ c_{s}},\ ^{\shortmid }\widehat{C}%
_{\ j_{s}}^{i_{s}\ c_{s}})\},\mbox{ where, }  \notag \\
&&\mbox{ for }\lbrack \ ^{\shortmid }\mathbf{g}_{\alpha _{s}\beta _{s}}=(\
^{\shortmid }g_{j_{s}r_{s}},\ ^{\shortmid }g^{a_{s}b_{s}})\mathbf{,\
_{s}^{\shortmid }N=\{}\ ^{\shortmid }N_{a_{s}i_{s}}\mathbf{\}]}  \notag \\
\ ^{\shortmid }\widehat{L}_{j_{s}k_{s}}^{i_{s}} &=&\frac{1}{2}\ ^{\shortmid
}g^{i_{s}r_{s}}(\ ^{\shortmid }\mathbf{e}_{k_{s}}\ ^{\shortmid
}g_{j_{s}r_{s}}+\ ^{\shortmid }\mathbf{e}_{j_{s}}\ ^{\shortmid
}g_{k_{s}r_{s}}-\ ^{\shortmid }\mathbf{e}_{r_{s}}\ ^{\shortmid
}g_{j_{s}k_{s}}),\   \notag \\
\ ^{\shortmid }\widehat{L}_{a_{s}\ k_{s}}^{\ b_{s}} &=&\ ^{\shortmid
}e^{b_{s}}(\ ^{\shortmid }N_{a_{s}k_{s}})+\frac{1}{2}\ ^{\shortmid
}g_{a_{s}c_{s}}(\ ^{\shortmid }e_{k_{s}}\ ^{\shortmid }g^{b_{s}c_{s}}-\
^{\shortmid }g^{d_{s}c_{s}}\ \ ^{\shortmid }e^{b_{s}}\ ^{\shortmid
}N_{d_{s}k_{s}}-\ ^{\shortmid }g^{d_{s}b_{s}}\ \ ^{\shortmid }e^{c_{s}}\
^{\shortmid }N_{d_{s}k_{s}}),  \notag \\
\ ^{\shortmid }\widehat{C}_{\ j_{s}}^{i_{s}\ c_{s}} &=&\frac{1}{2}\
^{\shortmid }g^{i_{s}k_{s}}\ ^{\shortmid }e^{c_{s}}\ ^{\shortmid
}g_{j_{s}k_{s}},\ \ ^{\shortmid }\widehat{C}_{\ a_{s}}^{b_{s}\ c_{s}}=\frac{1%
}{2}\ ^{\shortmid }g_{a_{s}d_{s}}(\ ^{\shortmid }e^{c_{s}}\ ^{\shortmid
}g^{b_{s}d_{s}}+\ ^{\shortmid }e^{b_{s}}\ ^{\shortmid }g^{c_{s}d_{s}}-\
^{\shortmid }e^{d_{s}}\ ^{\shortmid }g^{b_{s}c_{s}}).  \label{candcons34}
\end{eqnarray}
\end{corollary}

In a similar form, we can prove that all N-adapted coefficient formulas
necessary formulating and finding solutions of physically important field
and evolution equations in theories with MDRs and LIVs.

\subsection{Proof of Theorem \protect\ref{gensoltorsion}}

\label{assgensolt}The coefficients $g_{i_{1}}=e^{\psi (x^{k_{1}})}$ for the
first dyadic shell $s=1$ are defined by solutions of the corresponding 2-d
Poisson equation (\ref{eq1}) for any given source $\ _{1}^{\shortmid }%
\widehat{\Upsilon }(x^{k_{1}})$.

The system (\ref{e2a})-(\ref{e2c}) for the second dyadic shell $s=2$ can be
solved following the same procedure following formulas (49) - (58) in
section 2.3.6 of \cite{vacaruepjc17} (similar results were published in \cite%
{vijtp10} and \cite{gvvepjc14}). We have to re-define the coordinates and
letters for the d-metric and s-connection coefficients following conventions
of this paper. We obtain, respectively, from formula for $\ _{2}^{\shortmid
}\gamma $ and equations (\ref{e2a})-(\ref{e2c}), this nonlinear system
\begin{eqnarray}
(\ _{2}^{\shortmid }\Psi )^{\diamond }g_{4}^{\diamond } &=&2g_{3}g_{4}(\
_{2}^{\shortmid }\widehat{\Upsilon })(\ _{2}^{\shortmid }\Psi ),
\label{auxa1} \\
\sqrt{|g_{3}g_{4}|}\ _{2}^{\shortmid }\Psi &=&g_{4}^{\diamond }
\label{auxa2} \\
\ (\ _{2}^{\shortmid }\Psi )^{\diamond }w_{i_{1}}-\partial _{i_{1}}(\
_{2}^{\shortmid }\Psi ) &=&\ 0  \label{aux1ab} \\
\ n_{i_{1}}^{\diamond \diamond }+\left( \ln \frac{|\ g_{4}|^{3/2}}{|g_{3}|}%
\right) ^{\diamond }n_{i_{1}}^{\diamond } &=&0,\   \label{aux1ac}
\end{eqnarray}%
for $g_{4}^{\diamond }=\partial _{3}g_{4}=\partial g_{4}/\partial
y^{3}=\partial g_{4}/\partial \varphi .$ Prescribing generating function and
source, $\ _{2}^{\shortmid }\Psi $ and $\ _{2}^{\shortmid }\widehat{\Upsilon
},$ we can integrate in general form this system with decoupling of
equations. Let us prove this in new variables which are different from those
used in section 2.3.6 of \cite{vacaruepjc17}. Let introduce
\begin{equation}
\rho ^{2}:=-g_{3}g_{4}  \label{rho}
\end{equation}%
which allows us to write (\ref{auxa1}) and (\ref{auxa2}), respectively, in
the form%
\begin{equation}
(\ _{2}^{\shortmid }\Psi )^{\diamond }g_{4}^{\diamond }=-2\rho ^{2}(\
_{2}^{\shortmid }\widehat{\Upsilon })(\ _{2}^{\shortmid }\Psi )\mbox{ and }%
g_{4}^{\diamond }=\rho \ _{2}^{\shortmid }\Psi .  \label{auxa3a}
\end{equation}%
Substituting in this line the value of $g_{4}^{\diamond }$ from the second
equation into the first equation, we get
\begin{equation}
\rho =-(\ _{2}^{\shortmid }\Psi )^{\diamond }/2(\ _{2}^{\shortmid }\widehat{%
\Upsilon }).  \label{rho1}
\end{equation}%
Introducing this $\rho $ into the second equation in (\ref{auxa3a}) and
integrating on $y^{3},$ we obtain
\begin{equation}
\ g_{4}=g_{4}^{[0]}(x^{k_{1}})-\int dy^{3}[(\ _{2}^{\shortmid }\Psi
)^{2}]^{\diamond }/4(\ _{2}^{\shortmid }\widehat{\Upsilon }).  \label{g4}
\end{equation}%
This formula can be used in (\ref{rho}) and (\ref{rho1}), which allows us to
compute%
\begin{equation}
g_{3}=-\frac{1}{4g_{4}}\left( \frac{(\ _{2}^{\shortmid }\Psi )^{\diamond }}{%
\ _{2}^{\shortmid }\widehat{\Upsilon }}\right) ^{2}=-\left( \frac{(\
_{2}^{\shortmid }\Psi )^{\diamond }}{2\ _{2}^{\shortmid }\widehat{\Upsilon }}%
\right) ^{2}\left( g_{4}^{[0]}(x^{k_{1}})-\int dy^{3}\frac{[(\
_{2}^{\shortmid }\Psi )^{2}]^{\diamond }}{4(\ _{2}^{\shortmid }\widehat{%
\Upsilon })}\right) ^{-1}.  \label{g3}
\end{equation}%
Having computed $g_{3}$ and $g_{4},$ we can integrate two times on $y^{3}$
the equation (\ref{aux1ac}), $\ n_{i_{1}}^{\diamond \diamond }+\left( \ln
\frac{|g_{4}|^{3/2}}{|\ g_{3}|}\right) ^{\diamond }n_{i_{1}}^{\diamond }=0,$
when
\begin{eqnarray}
n_{k_{1}}(x^{k_{1}},y^{3}) &=&\ _{1}n_{k_{1}}+\ _{2}n_{k_{1}}\int dy^{3}\
\frac{g_{3}}{|\ g_{4}|^{3/2}}=\ _{1}n_{k_{1}}+\ _{2}n_{k_{1}}\int
dy^{3}\left( \frac{(\ _{2}^{\shortmid }\Psi )^{\diamond }}{2\
_{2}^{\shortmid }\widehat{\Upsilon }}\right) ^{2}|\ g_{4}|^{-5/2}  \notag \\
&=&\ _{1}n_{k_{1}}+\ _{2}n_{k_{1}}\int dy^{3}\left( \frac{(\ _{2}^{\shortmid
}\Psi )^{\diamond }}{2\ _{2}^{\shortmid }\widehat{\Upsilon }}\right)
^{2}\left\vert g_{4}^{[0]}(x^{k_{1}})-\int dy^{3}[(\ _{2}^{\shortmid }\Psi
)^{2}]^{\diamond }/4(\ _{2}^{\shortmid }\widehat{\Upsilon })\right\vert
^{-5/2},  \label{gn}
\end{eqnarray}%
with two integration functions $\ _{1}n_{k_{1}}(x^{i_{1}})$ and a redefined $%
\ _{2}\widetilde{n}_{k_{1}}(x^{i_{1}})=\ _{2}\widetilde{n}%
_{k_{1}}(x^{i_{1}}).$ Finally, for the shell $s=2,$ we can solve the
algebraic system (\ref{aux1ab}) and find%
\begin{equation}
w_{i_{1}}=\partial _{i_{1}}\ (\ _{2}^{\shortmid }\Psi )/(\ _{2}^{\shortmid
}\Psi )^{\diamond }.  \label{gw}
\end{equation}

At the next step, we provide a proof for the shell $s=4$ (for constructions
for $s=3$ are similar but with less fiber momentum coordinates). Such
constructions have not yet considered for cotangent bundles or Hamilton
geometries. Introducing the values of coefficients $\alpha _{i_{3}},\
_{4}^{\shortmid }\beta ,\ _{4}^{\shortmid }\gamma $ into (\ref{eq4a})-(\ref%
{eq4c}), we obtain a system of nonlinear PDEs with a decoupling property
\begin{eqnarray}
\ \ \ _{4}^{\shortmid }\Psi ^{\ast }\ (g^{7})^{\ast } &=&2\ g^{7}g^{8}\ (\
_{4}^{\shortmid }\widehat{\Upsilon })\ _{4}^{\shortmid }\Psi ,  \label{auxd1}
\\
\sqrt{|g^{7}g^{8}|}\ _{4}^{\shortmid }\Psi &=&(g^{7})^{\ast },  \notag \\
\ n_{i_{3}}^{\ast \ast }+\left( \ln \frac{|g^{8}|}{|\ g^{7}|^{3/2}}\right)
^{\ast }n_{i_{3}}^{\ast } &=&0,\   \notag \\
\ \ _{4}^{\shortmid }\Psi ^{\ast }w_{i_{3}}-\partial _{i_{3}}(\
_{4}^{\shortmid }\Psi ) &=&\ 0,\   \notag
\end{eqnarray}%
for $(g^{7})^{\ast }=\partial ^{8}g^{7}=\partial g^{7}/\partial
p_{8}=\partial g^{7}/\partial E.$ This system can be integrated in very
general forms by prescribing $\ _{4}^{\shortmid }\Psi
(x^{k_{1}},y^{3},t,p_{a_{3}},E)$ and $\ \ \ _{4}^{\shortmid }\widehat{%
\Upsilon }(x^{k_{1}},y^{3},t,p_{a_{3}},E),$ where $%
k_{1}=1,2;a_{3}=5,6;i_{3}=1,2,...6;a_{4}=7,8;$ and $x^{4}=y^{4}=t$ and $%
p_{8}=E.$ Introducing the function
\begin{equation}
\ (\ _{4}\rho )^{2}:=-g^{7}g^{8},  \label{aux2ba}
\end{equation}%
(the sign - is motivated by the pseudo-Euclidean signature, this sign is +
for $s=3$), we express\ the first two equations in above system in the form
\begin{equation}
\ _{4}^{\shortmid }\Psi ^{\ast }\ (g^{7})^{\ast }=-2\ (\ _{4}\rho )^{2}\ (\
_{4}^{\shortmid }\widehat{\Upsilon })\ _{4}^{\shortmid }\Psi \mbox{ and }\ \
(g^{7})^{\ast }=\ _{4}\rho (\ _{4}^{\shortmid }\Psi ).  \label{aux1da}
\end{equation}%
Then we substitute in the first equation of (\ref{auxd1}) the value of $%
(g^{7})^{\ast }$ given by the the second equation (\ref{aux1da}) and obtain
\begin{equation}
\ _{4}\rho =-\ _{4}^{\shortmid }\Psi ^{\ast }/2\ (\ _{4}^{\shortmid }%
\widehat{\Upsilon }).  \label{qf1d}
\end{equation}%
Using this value and the second equation of (\ref{aux1da}) and integrating
on $E,$ we obtain
\begin{equation}
\ g^{7}=g_{[0]}^{7}(x^{k_{1}},y^{3},t,p_{a_{3}})-\int dE[(\ _{4}^{\shortmid
}\Psi )^{2}]^{\ast }/4(\ _{4}^{\shortmid }\widehat{\Upsilon }),  \label{g7d}
\end{equation}%
where $g_{[0]}^{7}(x^{k_{3}})$ is an integration function. At the next step,
considering formulas (\ref{aux2ba}), (\ref{qf1d}) (\ref{g7d}), we compute%
\begin{equation}
\ g^{8}=-\frac{1}{4\ g^{7}}(\frac{\ \ _{4}^{\shortmid }\Psi ^{\ast }}{\
_{4}^{\shortmid }\widehat{\Upsilon }})^{2}=-\frac{1}{4}(\frac{\ \
_{4}^{\shortmid }\Psi ^{\ast }}{\ _{4}^{\shortmid }\widehat{\Upsilon }}%
)^{2}/\left( g_{[0]}^{7}-\frac{1}{4}\int dE\frac{[(\ _{4}^{\shortmid }\Psi
)^{2}]^{\ast }}{\ _{4}^{\shortmid }\widehat{\Upsilon }}\right) .  \label{g8d}
\end{equation}

Now, can perform an integration of the third equation in (\ref{auxd1}). The
first subset of the N--connection coefficients, $n_{i_{3}}^{\ast },$ are
found by integrating two times on $E$ in that equations written in the form
\begin{equation*}
\ n_{i_{3}}^{\ast \ast }=(n_{i_{3}}^{\ast })^{\ast }=-\ n_{i_{3}}^{\ast
}(\ln |g^{7}|^{3/2}/|g^{8}|)^{\ast }
\end{equation*}%
for the coefficient $\ _{4}^{\shortmid }\gamma $ defined in for the equation
(\ref{eq4b}) of the theorem. \ Using explicit values (\ref{g7d}) and (\ref%
{g8d}), we find
\begin{eqnarray*}
n_{k_{3}}(x^{k_{1}},y^{3},t,p_{a_{3}},E) &=&\ _{1}n_{k_{3}}+\
_{2}n_{k_{3}}\int dE\ \frac{|\ g^{8}|}{|g^{7}|^{3/2}}=\ _{1}n_{k_{3}}+\ _{2}%
\widetilde{n}_{k_{3}}\int dE\ \left( \frac{\ \ _{4}^{\shortmid }\Psi ^{\ast }%
}{\ _{4}^{\shortmid }\widehat{\Upsilon }}\right) ^{2}|\ g_{7}|^{-5/2} \\
&=&\ _{1}n_{k_{1}}+\ _{2}n_{k_{1}}\int dE\left( \frac{\ \ _{4}^{\shortmid
}\Psi ^{\ast }}{\ _{4}^{\shortmid }\widehat{\Upsilon }}\right)
^{2}\left\vert g_{[0]}^{7}(x^{k_{1}},y^{3},t,p_{a_{3}})-\int dE[(\
_{4}^{\shortmid }\Psi )^{2}]^{\ast }/4(\ _{4}^{\shortmid }\widehat{\Upsilon }%
)\right\vert ^{-5/2},
\end{eqnarray*}%
where two integration functions are parameterized $\
_{1}n_{k_{3}}(x^{i_{3}})=$ $\ _{1}n_{k_{3}}(x^{k_{1}},y^{3},t,p_{a_{3}})$
and a redefined $\ _{2}\widetilde{n}_{k_{3}}(x^{i_{3}})=\ _{2}\widetilde{n}%
_{k_{3}}(x^{k_{1}},y^{3},t,p_{a_{3}}).$

The second subset of N-connection coefficients, $w_{i_{3}},$ can be easily
found as a solution of the linear algebraic equations in (\ref{auxd1}),
  $w_{i_{3}}=\partial _{i_{3}}\ (\ _{4}^{\shortmid }\Psi )/(\ _{4}^{\shortmid
}\Psi )^{\ast }$.

Putting together above formulas for the s-metric and N-connection
coefficients on the shell $s=4,$ we construct a general solution of the
system (\ref{eq4a})-(\ref{eq4c}),
\begin{eqnarray}
\ \ g^{7} &=&g_{[0]}^{7}-\int dE[(\ _{4}^{\shortmid }\Psi )^{2}]^{\ast }/4(\
_{4}^{\shortmid }\widehat{\Upsilon });\   \label{gnw78} \\
\ g^{8} &=&-\frac{1}{4\ g^{7}}(\frac{\ \ _{4}^{\shortmid }\Psi ^{\ast }}{\
_{4}^{\shortmid }\widehat{\Upsilon }})^{2}=-\frac{1}{4}(\frac{\ \
_{4}^{\shortmid }\Psi ^{\ast }}{\ _{4}^{\shortmid }\widehat{\Upsilon }}%
)^{2}/\left( g_{[0]}^{7}-\frac{1}{4}\int dE\frac{[(\ _{4}^{\shortmid }\Psi
)^{2}]^{\ast }}{\ _{4}^{\shortmid }\widehat{\Upsilon }}\right) ;  \notag \\
n_{k_{3}} &=&\ _{1}n_{k_{1}}+\ _{2}n_{k_{1}}\int dE\left( \frac{\ \
_{4}^{\shortmid }\Psi ^{\ast }}{\ _{4}^{\shortmid }\widehat{\Upsilon }}%
\right) ^{2}\left\vert g_{[0]}^{7}(x^{k_{1}},y^{3},t,p_{a_{3}})-\int dE[(\
_{4}^{\shortmid }\Psi )^{2}]^{\ast }/4(\ _{4}^{\shortmid }\widehat{\Upsilon }%
)\right\vert ^{5/2};  \notag \\
w_{i_{3}} &=&\partial _{i_{3}}\ (\ _{4}^{\shortmid }\Psi )/(\
_{4}^{\shortmid }\Psi )^{\ast }.  \notag
\end{eqnarray}

Finally, introducing all s-metric and N-connection elements for all shell,
we can construct the quadratic line element (\ref{qeltors}).

\subsection{Proof of Theorem \protect\ref{epsilongeneration}}

\label{assepsilon} We provide some details of geometric computations in
order to show how such nonholonomic parametric deformations can be
constructed in explicit form. There used the formulas (\ref{offdiagpolf}), (%
\ref{noffdiagpolf}), and Definition \ref{defsmalpnd}.

Let us consider the shell $s=1.$ Writing
\begin{equation*}
\ ^{\shortmid }g_{i_{1}}=\zeta _{i_{1}}(1+\varepsilon \chi _{i_{1}})\
^{\shortmid }\mathring{g}_{i_{1}}=e^{\psi (x^{k_{1}})}\approx e^{\psi
_{0}(x^{k_{1}})(1+\varepsilon \ ^{\psi }\chi (x^{k_{1}}))}\approx e^{\psi
_{0}(x^{k_{1}})}(1+\varepsilon \ ^{\psi }\chi ),
\end{equation*}
we obtain $\zeta _{i_{1}}\mathring{g}_{i_{1}}=e^{\psi _{0}(x^{k_{1}})}$ and $%
\chi _{i_{1}}\ ^{\shortmid }\mathring{g}_{i_{1}}=\ ^{\psi }\chi .$

For $s=2,$ when $\ ^{\shortmid }\eta _{\alpha _{2}}=\zeta _{\alpha
_{2}}(1+\varepsilon \chi _{\alpha _{2}})$ and$\ ^{\shortmid }\eta
_{i_{1}}^{a_{2}}=\zeta _{i_{1}}^{a_{2}}(1+\varepsilon \chi
_{i_{1}}^{a_{2}}), $ i.e. for $\ ^{\shortmid }\eta _{3}=\zeta
_{3}(1+\varepsilon \chi _{3})$ and $\ ^{\shortmid }\eta _{4}=\zeta
_{4}(1+\varepsilon \chi _{4}),$ we can express all $\varepsilon $%
--polarizations as functionals on generating data $\zeta _{4}$ and $\chi
_{4}.$ We perform such computations for respective coefficients of a
d-metric and N-connections:%
\begin{eqnarray*}
\ ^{\shortmid }\eta _{3}\ ^{\shortmid }\mathring{g}_{3} &=&-\frac{[(\
^{\shortmid }\eta _{4}\ ^{\shortmid }\mathring{g}_{4})^{\diamond }]^{2}}{%
|\int dy^{3}(\ _{2}^{\shortmid }\widehat{\Upsilon })(\ ^{\shortmid }\eta
_{4}\ ^{\shortmid }\mathring{g}_{4})^{\diamond }|\ (\ ^{\shortmid }\eta
_{4}\ ^{\shortmid }\mathring{g}_{4})}=-4\frac{[(|\ ^{\shortmid }\eta _{4}\
^{\shortmid }\mathring{g}_{4}|^{1/2})^{\diamond }]^{2}}{|\int dy^{3}(\
_{2}^{\shortmid }\widehat{\Upsilon })(\ ^{\shortmid }\eta _{4}\ ^{\shortmid }%
\mathring{g}_{4})^{\diamond }|\ }\simeq \\
\zeta _{3}(1+\varepsilon \chi _{3})\ ^{\shortmid }\mathring{g}_{3} &=&-4%
\frac{[(|\ \zeta _{4}\ ^{\shortmid }\mathring{g}_{4}(1+\varepsilon \chi
_{4})|^{1/2})^{\diamond }]^{2}}{|\int dy^{3}(\ _{2}^{\shortmid }\widehat{%
\Upsilon })[(\ \zeta _{4}\ ^{\shortmid }\mathring{g}_{4})(1+\varepsilon \chi
_{4})]^{\diamond }|}
\end{eqnarray*}%
\begin{equation*}
=-4\frac{[(|\ \zeta _{4}\ ^{\shortmid }\mathring{g}_{4}|^{1/2}\
|1+\varepsilon \chi _{4}\ |^{1/2})^{\diamond }]^{2}}{|\int dy^{3}(\
_{2}^{\shortmid }\widehat{\Upsilon })[(\ \zeta _{4}\ ^{\shortmid }\mathring{g%
}_{4})+\varepsilon \ (\zeta _{4}\chi _{4}\ ^{\shortmid }\mathring{g}%
_{4})]^{\diamond }|}=-4\frac{[(|\ \zeta _{4}\ ^{\shortmid }\mathring{g}%
_{4}|^{1/2}|1+\frac{\varepsilon }{2}\chi _{4}|)^{\diamond }]^{2}}{|\int
dy^{3}(\ _{2}^{\shortmid }\widehat{\Upsilon })\{(\ \zeta _{4}\ ^{\shortmid }%
\mathring{g}_{4})^{\diamond }+\varepsilon \lbrack (\zeta _{4}\chi _{4})\
^{\shortmid }\mathring{g}_{4})]^{\diamond }\}|}
\end{equation*}%
\begin{eqnarray*}
&=&-4\frac{[(|\ \zeta _{4}\ ^{\shortmid }\mathring{g}_{4}|^{1/2}+\frac{%
\varepsilon }{2}\chi _{4}|\ \zeta _{4}\ ^{\shortmid }\mathring{g}%
_{4}|^{1/2})^{\diamond }]^{2}}{|\int dy^{3}[(\ _{2}^{\shortmid }\widehat{%
\Upsilon })(\ \zeta _{4}\ ^{\shortmid }\mathring{g}_{4})^{\diamond
}]+\varepsilon \int dy^{3}\{(\ _{2}^{\shortmid }\widehat{\Upsilon })[(\zeta
_{4}\chi _{4})\ ^{\shortmid }\mathring{g}_{4})]^{\diamond }\}|} \\
&=&-4\frac{[(|\ \zeta _{4}\ ^{\shortmid }\mathring{g}_{4}|^{1/2})^{\diamond
}+\frac{\varepsilon }{2}(\chi _{4}|\ \zeta _{4}\ ^{\shortmid }\mathring{g}%
_{4}|^{1/2})^{\diamond }]^{2}}{|\int dy^{3}[(\ _{2}^{\shortmid }\widehat{%
\Upsilon })(\ \zeta _{4}\ ^{\shortmid }\mathring{g}_{4})^{\diamond
}\}|(1+\varepsilon \frac{\int dy^{3}\{(\ _{2}^{\shortmid }\widehat{\Upsilon }%
)[\ (\zeta _{4}\chi _{4})\ ^{\shortmid }\mathring{g}_{4})]^{\diamond }\}}{%
\int dy^{3}\{(\ _{2}^{\shortmid }\widehat{\Upsilon })(\ \zeta _{4}\
^{\shortmid }\mathring{g}_{4})^{\diamond }\}}\}|}
\end{eqnarray*}%
\begin{equation*}
=-4\frac{[(|\ \zeta _{4}\ ^{\shortmid }\mathring{g}_{4}|^{1/2})^{\diamond }+%
\frac{\varepsilon }{2}(\chi _{4}|\ \zeta _{4}\ ^{\shortmid }\mathring{g}%
_{4}|^{1/2})^{\diamond }]^{2}}{|\int dy^{3}\{(\ _{2}^{\shortmid }\widehat{%
\Upsilon })(\ \zeta _{4}\ ^{\shortmid }\mathring{g}_{4})^{\diamond
}\}|(1+\varepsilon \frac{\int dy^{3}\{(\ _{2}^{\shortmid }\widehat{\Upsilon }%
)[(\zeta _{4}\chi _{4})\ ^{\shortmid }\mathring{g}_{4})]^{\diamond }\}}{\int
dy^{3}[(\ _{2}^{\shortmid }\widehat{\Upsilon })(\ \zeta _{4}\ ^{\shortmid }%
\mathring{g}_{4})^{\diamond }]})}
\end{equation*}%
\begin{eqnarray*}
&=&-4\frac{[(|\ \zeta _{4}\ ^{\shortmid }\mathring{g}_{4}|^{1/2})^{\diamond
}]^{2}[1+\frac{\varepsilon }{2}\frac{(\chi _{4}|\ \zeta _{4}\ ^{\shortmid }%
\mathring{g}_{4}|^{1/2})^{\diamond }}{(|\ \zeta _{4}\ ^{\shortmid }\mathring{%
g}_{4}|^{1/2})^{\diamond }}]^{2}(1-\varepsilon \frac{\int dy^{3}\{(\
_{2}^{\shortmid }\widehat{\Upsilon })[\ (\zeta _{4}\chi _{4})\ ^{\shortmid }%
\mathring{g}_{4})]^{\diamond }\}}{\int dy^{3}[(\ _{2}^{\shortmid }\widehat{%
\Upsilon })(\ \zeta _{4}\ ^{\shortmid }\mathring{g}_{4})^{\diamond }]})}{%
|\int dy^{3}[(\ _{2}^{\shortmid }\widehat{\Upsilon })(\ \zeta _{4}\
^{\shortmid }\mathring{g}_{4})^{\diamond }]|} \\
&=&-4\frac{[(|\ \zeta _{4}\ ^{\shortmid }\mathring{g}_{4}|^{1/2})^{\diamond
}]^{2}}{|\int dy^{3}[(\ _{2}^{\shortmid }\widehat{\Upsilon })(\ \zeta _{4}\
^{\shortmid }\mathring{g}_{4})^{\diamond }]|}\left[ 1+\varepsilon \left(
\frac{(\chi _{4}|\ \zeta _{4}\ ^{\shortmid }\mathring{g}_{4}|^{1/2})^{%
\diamond }}{(|\ \zeta _{4}\ ^{\shortmid }\mathring{g}_{4}|^{1/2})^{\diamond }%
}-\frac{\int dy^{3}\{(\ _{2}^{\shortmid }\widehat{\Upsilon })[(\zeta
_{4}\chi _{4})\ ^{\shortmid }\mathring{g}_{4})]^{\diamond }\}}{\int
dy^{3}[(\ _{2}^{\shortmid }\widehat{\Upsilon })(\ \zeta _{4}\ ^{\shortmid }%
\mathring{g}_{4})^{\diamond }]}\right) \right]
\end{eqnarray*}

From these formulas, we express
\begin{eqnarray*}
&&\zeta _{3}\ ^{\shortmid }\mathring{g}_{3}=-4\frac{[(|\ \zeta _{4}\
^{\shortmid }\mathring{g}_{4}|^{1/2})^{\diamond }]^{2}}{|\int dy^{3}[(\
_{2}^{\shortmid }\widehat{\Upsilon })(\ \zeta _{4}\ ^{\shortmid }\mathring{g}%
_{4})^{\diamond }]|}\mbox{ and } \\
\mbox{ proportional to }\varepsilon &:&\chi _{3}=\frac{(\chi _{4}|\ \zeta
_{4}\ ^{\shortmid }\mathring{g}_{4}|^{1/2})^{\diamond }}{4(|\ \zeta _{4}\
^{\shortmid }\mathring{g}_{4}|^{1/2})^{\diamond }}-\frac{\int dy^{3}\{(\
_{2}^{\shortmid }\widehat{\Upsilon })[(\zeta _{4}\ ^{\shortmid }\mathring{g}%
_{4})\chi _{4}]^{\diamond }\}}{\int dy^{3}[(\ _{2}^{\shortmid }\widehat{%
\Upsilon })(\ \zeta _{4}\ ^{\shortmid }\mathring{g}_{4})^{\diamond }]}.
\end{eqnarray*}%
For the N--connection coefficients, we have%
\begin{eqnarray*}
\ ^{\shortmid }\eta _{i_{1}}^{3}\ ^{\shortmid }\mathring{N}_{i_{1}}^{3} &=&%
\frac{\partial _{i_{1}}\ \int dy^{3}(\ _{2}^{\shortmid }\widehat{\Upsilon }%
)\ (\ ^{\shortmid }\eta _{4}\ ^{\shortmid }\mathring{g}_{4})^{\diamond }}{(\
_{2}^{\shortmid }\widehat{\Upsilon })\ (\ ^{\shortmid }\eta _{4}\
^{\shortmid }\mathring{g}_{4})^{\diamond }}\simeq \\
\zeta _{i_{1}}^{3}(1+\varepsilon \chi _{i_{1}}^{3})\ ^{\shortmid }\mathring{N%
}_{i_{1}}^{3} &=&\frac{\partial _{i_{1}}\ \int dy^{3}(\ _{2}^{\shortmid }%
\widehat{\Upsilon })\ [\zeta _{4}(1+\varepsilon \chi _{4})]^{\diamond }}{(\
_{2}^{\shortmid }\widehat{\Upsilon })\ [\zeta _{4}(1+\varepsilon \chi
_{4})]^{\diamond }}=\frac{\partial _{i_{1}}\ \int dy^{3}(\ _{2}^{\shortmid }%
\widehat{\Upsilon })\ \{(\zeta _{4})^{\diamond }+\varepsilon \lbrack (\zeta
_{4}\chi _{4})]^{\diamond }\}}{(\ _{2}^{\shortmid }\widehat{\Upsilon })\
\{(\zeta _{4})^{\diamond }+\varepsilon \lbrack (\zeta _{4}\chi
_{4})]^{\diamond }\}}
\end{eqnarray*}%
\begin{equation*}
=\frac{\partial _{i_{1}}\ \{\int dy^{3}(\ _{2}^{\shortmid }\widehat{\Upsilon
})(\zeta _{4})^{\diamond }\ [1+\varepsilon \frac{(\zeta _{4}\chi
_{4})^{\diamond }}{(\zeta _{4})^{\diamond }}]\}}{(\ _{2}^{\shortmid }%
\widehat{\Upsilon })\ (\zeta _{4})^{\diamond }[1+\varepsilon \frac{(\zeta
_{4}\chi _{4})^{\diamond }}{(\zeta _{4})^{\diamond }}]}=\frac{\partial
_{i_{1}}\ [\int dy^{3}(\ _{2}^{\shortmid }\widehat{\Upsilon })(\zeta
_{4})^{\diamond }]\ +\varepsilon \partial _{i_{1}}[\int dy^{3}(\
_{2}^{\shortmid }\widehat{\Upsilon })(\zeta _{4}\chi _{4})^{\diamond }]}{(\
_{2}^{\shortmid }\widehat{\Upsilon })\ (\zeta _{4})^{\diamond }}\times
\lbrack 1-\varepsilon \frac{(\zeta _{4}\chi _{4})^{\diamond }}{(\zeta
_{4})^{\diamond }}]
\end{equation*}%
\begin{eqnarray*}
&=&\frac{\{\partial _{i_{1}}\ [\int dy^{3}(\ _{2}^{\shortmid }\widehat{%
\Upsilon })(\zeta _{4})^{\diamond }]\}\{1+\varepsilon \frac{\partial
_{i_{1}}[\int dy^{3}(\ _{2}^{\shortmid }\widehat{\Upsilon })(\zeta _{4}\chi
_{4})^{\diamond }]}{\partial _{i_{1}}\ [\int dy^{3}(\ _{2}^{\shortmid }%
\widehat{\Upsilon })(\zeta _{4})^{\diamond }]}\}}{(\ _{2}^{\shortmid }%
\widehat{\Upsilon })\ (\zeta _{4})^{\diamond }}\{1-\varepsilon \frac{(\zeta
_{4}\chi _{4})^{\diamond }}{(\zeta _{4})^{\diamond }}\} \\
&=&\frac{\partial _{i_{1}}\ [\int dy^{3}(\ _{2}^{\shortmid }\widehat{%
\Upsilon })(\zeta _{4})^{\diamond }]}{(\ _{2}^{\shortmid }\widehat{\Upsilon }%
)\ (\zeta _{4})^{\diamond }}\{1+\varepsilon \left( \frac{\partial
_{i_{1}}[\int dy^{3}(\ _{2}^{\shortmid }\widehat{\Upsilon })(\zeta _{4}\chi
_{4})^{\diamond }]}{\partial _{i_{1}}\ [\int dy^{3}(\ _{2}^{\shortmid }%
\widehat{\Upsilon })(\zeta _{4})^{\diamond }]}-\frac{(\zeta _{4}\chi
_{4})^{\diamond }}{(\zeta _{4})^{\diamond }}\right) \}.
\end{eqnarray*}%
It should be noted that there is not summation on repeating indices because
they are not arraged on the rule "up-low". \ From these formulas, we obtain
\begin{equation*}
\zeta _{i_{1}}^{3}\ ^{\shortmid }\mathring{N}_{i_{1}}^{3}=\frac{\partial
_{i_{1}}[\int dy^{3}(\ _{2}^{\shortmid }\widehat{\Upsilon })\ (\zeta
_{4})^{\diamond }]}{(\ _{2}^{\shortmid }\widehat{\Upsilon })(\zeta
_{4})^{\diamond }}\mbox{ and }\varepsilon :\chi _{i_{1}}^{3}=\frac{\partial
_{i_{1}}[\int dy^{3}(\ _{2}^{\shortmid }\widehat{\Upsilon })(\zeta _{4}\chi
_{4})^{\diamond }]}{\partial _{i_{1}}\ [\int dy^{3}(\ _{2}^{\shortmid }%
\widehat{\Upsilon })(\zeta _{4})^{\diamond }]}-\frac{(\zeta _{4}\chi
_{4})^{\diamond }}{(\zeta _{4})^{\diamond }}.
\end{equation*}%
Finally (for this shell), we compute$\ $%
\begin{eqnarray*}
\ ^{\shortmid }\eta _{k_{1}}^{4}\ ^{\shortmid }\mathring{N}_{k_{1}}^{4} &=&\
_{1}n_{k_{1}}+\ _{2}n_{k_{1}}\int dy^{3}\frac{[(\ ^{\shortmid }\eta _{4}\
^{\shortmid }\mathring{g}_{4})^{\diamond }]^{2}}{|\int dy^{3}(\
_{2}^{\shortmid }\widehat{\Upsilon })(\ ^{\shortmid }\eta _{4}\ ^{\shortmid }%
\mathring{g}_{4})^{\diamond }|\ (\ ^{\shortmid }\eta _{4}\ ^{\shortmid }%
\mathring{g}_{4})^{5/2}} \\
&=&\ _{1}n_{k_{1}}+16\ \ _{2}n_{k_{1}}\int dy^{3}\frac{\left( [(\
^{\shortmid }\eta _{4}\ ^{\shortmid }\mathring{g}_{4})^{-1/4}]^{\diamond
}\right) ^{2}}{|\int dy^{3}(\ _{2}^{\shortmid }\widehat{\Upsilon })(\
^{\shortmid }\eta _{4}\ ^{\shortmid }\mathring{g}_{4})^{\diamond }|\ }\simeq
\end{eqnarray*}%
\begin{eqnarray*}
\zeta _{k_{1}}^{4}(1+\varepsilon \chi _{k_{1}}^{4})\ ^{\shortmid }\mathring{N%
}_{k_{1}}^{4} &=&\ _{1}n_{k_{1}}+16\ _{2}n_{k_{1}}\left[ \int dy^{3}\frac{%
\left( \{[(\ \zeta _{4}\ ^{\shortmid }\mathring{g}_{4})(1+\varepsilon \chi
_{4})]^{-1/4}\}^{\diamond }\right) ^{2}}{|\int dy^{3}(\ _{2}^{\shortmid }%
\widehat{\Upsilon })[(\zeta _{4}(1+\varepsilon \chi _{4})\ ^{\shortmid }%
\mathring{g}_{4})]^{\diamond }|\ }\right] \\
&=&\ _{1}n_{k_{1}}+16\ _{2}n_{k_{1}}\left[ \int dy^{3}\frac{\{[(\ \zeta
_{4}\ ^{\shortmid }\mathring{g}_{4})^{-1/4}(1-\frac{\varepsilon }{4}\chi
_{4})]^{\diamond }\}^{2}}{|\int dy^{3}\{(\ _{2}^{\shortmid }\widehat{%
\Upsilon })(\zeta _{4}\ ^{\shortmid }\mathring{g}_{4})^{\diamond
}+\varepsilon (\ _{2}^{\shortmid }\widehat{\Upsilon })(\zeta _{4}\chi _{4}\
^{\shortmid }\mathring{g}_{4})^{\diamond }\}|\ }\right]
\end{eqnarray*}%
\begin{eqnarray*}
&=&\ _{1}n_{k_{1}}+16\ _{2}n_{k_{1}}\left[ \int dy^{3}\frac{\{[(\ \zeta
_{4}\ ^{\shortmid }\mathring{g}_{4})^{-1/4}-\frac{\varepsilon }{4}(\ \zeta
_{4}\ ^{\shortmid }\mathring{g}_{4})^{-1/4}\chi _{4})]^{\diamond }\}^{2}}{%
|\int dy^{3}(\ _{2}^{\shortmid }\widehat{\Upsilon })(\zeta _{4}\ ^{\shortmid
}\mathring{g}_{4})^{\diamond }|\{1+\varepsilon \frac{\int dy^{3}(\
_{2}^{\shortmid }\widehat{\Upsilon })(\zeta _{4}\chi _{4}\ ^{\shortmid }%
\mathring{g}_{4})^{\diamond }}{\int dy^{3}(\ _{2}^{\shortmid }\widehat{%
\Upsilon })(\zeta _{4}\ ^{\shortmid }\mathring{g}_{4})^{\diamond }}\}}\right]
\\
&=&\ _{1}n_{k_{1}}+16\ _{2}n_{k_{1}}\left[ \int dy^{3}\frac{\{[(\ \zeta
_{4}\ ^{\shortmid }\mathring{g}_{4})^{-1/4}]^{\diamond }-\frac{\varepsilon }{%
4}[(\ \zeta _{4}\ ^{\shortmid }\mathring{g}_{4})^{-1/4}\chi _{4})]^{\diamond
}\}^{2}}{|\int dy^{3}(\ _{2}^{\shortmid }\widehat{\Upsilon })(\zeta _{4}\
^{\shortmid }\mathring{g}_{4})^{\diamond }|\{1+\varepsilon \frac{\int
dy^{3}(\ _{2}^{\shortmid }\widehat{\Upsilon })(\zeta _{4}\chi _{4}\
^{\shortmid }\mathring{g}_{4})^{\diamond }}{\int dy^{3}(\ _{2}^{\shortmid }%
\widehat{\Upsilon })(\zeta _{4}\ ^{\shortmid }\mathring{g}_{4})^{\diamond }}%
\}}\right]
\end{eqnarray*}%
\begin{eqnarray*}
&=&\ _{1}n_{k_{1}}+16\ _{2}n_{k_{1}}\left[ \int dy^{3}\frac{\{[(\ \zeta
_{4}\ ^{\shortmid }\mathring{g}_{4})^{-1/4}]^{\diamond }\left( 1-\frac{%
\varepsilon }{4}\frac{[(\ \zeta _{4}\ ^{\shortmid }\mathring{g}%
_{4})^{-1/4}\chi _{4})]^{\diamond }}{[(\ \zeta _{4}\ ^{\shortmid }\mathring{g%
}_{4})^{-1/4}]^{\diamond }}\right) \}^{2}}{|\int dy^{3}(\ _{2}^{\shortmid }%
\widehat{\Upsilon })(\zeta _{4}\ ^{\shortmid }\mathring{g}_{4})^{\diamond
}|\{1+\varepsilon \frac{\int dy^{3}(\ _{2}^{\shortmid }\widehat{\Upsilon }%
)(\zeta _{4}\chi _{4}\ ^{\shortmid }\mathring{g}_{4})^{\diamond }}{\int
dy^{3}(\ _{2}^{\shortmid }\widehat{\Upsilon })(\zeta _{4}\ ^{\shortmid }%
\mathring{g}_{4})^{\diamond }}\}}\right] \\
&=&\ _{1}n_{k_{1}}+16\ _{2}n_{k_{1}}\left[ \int dy^{3}\frac{\left( [(\ \zeta
_{4}\ ^{\shortmid }\mathring{g}_{4})^{-1/4}]^{\diamond }\right) ^{2}\left( 1-%
\frac{\varepsilon }{2}\frac{[(\ \zeta _{4}\ ^{\shortmid }\mathring{g}%
_{4})^{-1/4}\chi _{4})]^{\diamond }}{[(\ \zeta _{4}\ ^{\shortmid }\mathring{g%
}_{4})^{-1/4}]^{\diamond }}\right) }{|\int dy^{3}(\ _{2}^{\shortmid }%
\widehat{\Upsilon })(\zeta _{4}\ ^{\shortmid }\mathring{g}_{4})^{\diamond
}|\{1+\varepsilon \frac{\int dy^{3}(\ _{2}^{\shortmid }\widehat{\Upsilon }%
)(\zeta _{4}\chi _{4}\ ^{\shortmid }\mathring{g}_{4})^{\diamond }}{\int
dy^{3}(\ _{2}^{\shortmid }\widehat{\Upsilon })(\zeta _{4}\ ^{\shortmid }%
\mathring{g}_{4})^{\diamond }}\}}\right]
\end{eqnarray*}%
\begin{eqnarray*}
&=&\ _{1}n_{k_{1}}+16\ _{2}n_{k_{1}}\{\int dy^{3}\frac{\left( [(\ \zeta
_{4}\ ^{\shortmid }\mathring{g}_{4})^{-1/4}]^{\diamond }\right) ^{2}}{|\int
dy^{3}(\ _{2}^{\shortmid }\widehat{\Upsilon })(\zeta _{4}\ ^{\shortmid }%
\mathring{g}_{4})^{\diamond }|}[1-\varepsilon (\frac{[(\ \zeta _{4}\
^{\shortmid }\mathring{g}_{4})^{-1/4}\chi _{4})]^{\diamond }}{2[(\ \zeta
_{4}\ ^{\shortmid }\mathring{g}_{4})^{-1/4}]^{\diamond }}+\frac{\int
dy^{3}(\ _{2}^{\shortmid }\widehat{\Upsilon })(\zeta _{4}\chi _{4}\
^{\shortmid }\mathring{g}_{4})^{\diamond }}{\int dy^{3}(\ _{2}^{\shortmid }%
\widehat{\Upsilon })(\zeta _{4}\ ^{\shortmid }\mathring{g}_{4})^{\diamond }}%
)]\} \\
&=&\ _{1}n_{k_{1}}+16\ _{2}n_{k_{1}}[\int dy^{3}\frac{\left( [(\ \zeta _{4}\
^{\shortmid }\mathring{g}_{4})^{-1/4}]^{\diamond }\right) ^{2}}{|\int
dy^{3}(\ _{2}^{\shortmid }\widehat{\Upsilon })(\zeta _{4}\ ^{\shortmid }%
\mathring{g}_{4})^{\diamond }|}]- \\
&&\varepsilon 16\ _{2}n_{k_{1}}\int dy^{3}\frac{\left( [(\ \zeta _{4}\
^{\shortmid }\mathring{g}_{4})^{-1/4}]^{\diamond }\right) ^{2}}{|\int
dy^{3}(\ _{2}^{\shortmid }\widehat{\Upsilon })(\zeta _{4}\ ^{\shortmid }%
\mathring{g}_{4})^{\diamond }|}\left( \frac{[(\ \zeta _{4}\ ^{\shortmid }%
\mathring{g}_{4})^{-1/4}\chi _{4})]^{\diamond }}{2[(\ \zeta _{4}\
^{\shortmid }\mathring{g}_{4})^{-1/4}]^{\diamond }}+\frac{\int dy^{3}(\
_{2}^{\shortmid }\widehat{\Upsilon })(\zeta _{4}\chi _{4}\ ^{\shortmid }%
\mathring{g}_{4})^{\diamond }}{\int dy^{3}(\ _{2}^{\shortmid }\widehat{%
\Upsilon })(\zeta _{4}\ ^{\shortmid }\mathring{g}_{4})^{\diamond }}\right)
\end{eqnarray*}%
\begin{eqnarray*}
&=&\{\ _{1}n_{k_{1}}+16\ _{2}n_{k_{1}}[\int dy^{3}\frac{\left( [(\ \zeta
_{4}\ ^{\shortmid }\mathring{g}_{4})^{-1/4}]^{\diamond }\right) ^{2}}{|\int
dy^{3}(\ _{2}^{\shortmid }\widehat{\Upsilon })(\zeta _{4}\ ^{\shortmid }%
\mathring{g}_{4})^{\diamond }|}]\}\times \\
&&\left[ 1-\varepsilon \frac{16\ _{2}n_{k_{1}}\int dy^{3}\frac{\left( [(\
\zeta _{4}\ ^{\shortmid }\mathring{g}_{4})^{-1/4}]^{\diamond }\right) ^{2}}{%
|\int dy^{3}(\ _{2}^{\shortmid }\widehat{\Upsilon })(\zeta _{4}\ ^{\shortmid
}\mathring{g}_{4})^{\diamond }|}\left( \frac{[(\ \zeta _{4}\ ^{\shortmid }%
\mathring{g}_{4})^{-1/4}\chi _{4})]^{\diamond }}{2[(\ \zeta _{4}\
^{\shortmid }\mathring{g}_{4})^{-1/4}]^{\diamond }}+\frac{\int dy^{3}(\
_{2}^{\shortmid }\widehat{\Upsilon })(\zeta _{4}\chi _{4}\ ^{\shortmid }%
\mathring{g}_{4})^{\diamond }}{\int dy^{3}(\ _{2}^{\shortmid }\widehat{%
\Upsilon })(\zeta _{4}\ ^{\shortmid }\mathring{g}_{4})^{\diamond }}\right) }{%
\ _{1}n_{k_{1}}+16\ _{2}n_{k_{1}}[\int dy^{3}\frac{\left( [(\ \zeta _{4}\
^{\shortmid }\mathring{g}_{4})^{-1/4}]^{\diamond }\right) ^{2}}{|\int
dy^{3}(\ _{2}^{\shortmid }\widehat{\Upsilon })(\zeta _{4}\ ^{\shortmid }%
\mathring{g}_{4})^{\diamond }|}]}\right]
\end{eqnarray*}%
\begin{eqnarray*}
\mbox{ where } &&\zeta _{k_{1}}^{4}\ ^{\shortmid }\mathring{N}_{k_{1}}^{4}=\
_{1}n_{k_{1}}+16\ _{2}n_{k_{1}}[\int dy^{3}\{\frac{\left( [(\ \zeta _{4}\
^{\shortmid }\mathring{g}_{4})^{-1/4}]^{\diamond }\right) ^{2}}{|\int
dy^{3}(\ _{2}^{\shortmid }\widehat{\Upsilon })(\zeta _{4}\ ^{\shortmid }%
\mathring{g}_{4})^{\diamond }|}\mbox{ and } \\
\varepsilon &:&\chi _{k_{1}}^{4}=\ -\frac{16\ _{2}n_{k_{1}}\int dy^{3}\frac{%
\left( [(\ \zeta _{4}\ ^{\shortmid }\mathring{g}_{4})^{-1/4}]^{\diamond
}\right) ^{2}}{|\int dy^{3}(\ _{2}^{\shortmid }\widehat{\Upsilon })(\zeta
_{4}\ ^{\shortmid }\mathring{g}_{4})^{\diamond }|}\left( \frac{[(\ \zeta
_{4}\ ^{\shortmid }\mathring{g}_{4})^{-1/4}\chi _{4})]^{\diamond }}{2[(\
\zeta _{4}\ ^{\shortmid }\mathring{g}_{4})^{-1/4}]^{\diamond }}+\frac{\int
dy^{3}(\ _{2}^{\shortmid }\widehat{\Upsilon })[(\zeta _{4}\chi _{4}\
^{\shortmid }\mathring{g}_{4})]^{\diamond }}{\int dy^{3}(\ _{2}^{\shortmid }%
\widehat{\Upsilon })(\zeta _{4}\ ^{\shortmid }\mathring{g}_{4})^{\diamond }}%
\right) }{\ _{1}n_{k_{1}}+16\ _{2}n_{k_{1}}[\int dy^{3}\frac{\left( [(\
\zeta _{4}\ ^{\shortmid }\mathring{g}_{4})^{-1/4}]^{\diamond }\right) ^{2}}{%
|\int dy^{3}(\ _{2}^{\shortmid }\widehat{\Upsilon })(\zeta _{4}\ ^{\shortmid
}\mathring{g}_{4})^{\diamond }|}]}.
\end{eqnarray*}

Computations for shells $s=3$ and $s=4$ are similar but with re-definition
of corresponding symbols for shall coordinates and indices.

\setcounter{equation}{0} \renewcommand{\theequation}
{B.\arabic{equation}} \setcounter{subsection}{0}
\renewcommand{\thesubsection}
{B.\arabic{subsection}}

\section{Relativistic models of Finsler-Lagrange-Hamilton Phase Spaces}

\label{appendixb}In this Appendix, we summarize necessary results on
Finsler-Lagrange-Hamilton geometries modelled on (co) tangent bundles $TV$
and $T^{\ast }V$ on 4-d Lorentz manifolds, see detail with proofs in \cite%
{v18a} and references therein.

\subsection{Generating functions determined by MDR-indicators}

For any MDR of type (\ref{mdrg}), we can construct an effective Hamiltonian
\begin{equation}
H(p):=E=\pm (c^{2}\overrightarrow{\mathbf{p}}^{2}+c^{4}m^{2}-\varpi (E,%
\overrightarrow{\mathbf{p}},m;\ell _{P}))^{1/2}.  \label{hamfp}
\end{equation}%
This Hamiltonian describes relativistic point particles propagating in a
typical co-fiber of a $T^{\ast }V$ over a point $x=\{x^{i}\}\in V.$
Globalizing the constructions for a Lorentz manifold basis, with an
effective phase space endowed with local coordinates $(x^{i},p_{a}),$ we
obtain indicators $\varpi (x^{i},E,\overrightarrow{\mathbf{p}},m;\ell _{P})$
depending both on spacetime and phase space coordinates. Considering general
systems of frames/coordinates and their transforms on total dual bundle $%
T^{\ast }V,$ the Hamiltonian (\ref{hamfp}) can be written in the form $%
H(x,p).$

Similarly to relativistic mechanics, we can define (inverse) Legendre
transforms and the concept of $L$-duality for certain Lagrange and Hamilton
densities. The Legendre transforms, $L\rightarrow H,$ are constructed
\begin{equation}
H(x,p):=p_{a}v^{a}-L(x,v)  \label{legendre}
\end{equation}%
and $v^{a}$ determining solutions of the equations $p_{a}=\partial
L(x,v)/\partial v^{a}.$ In a similar manner, the inverse Legendre
transforms, $H\rightarrow L,$ are constructed
\begin{equation}
L(x,v):=p_{a}v^{a}-H(x,p)  \label{invlegendre}
\end{equation}%
and $p_{a}$ determining solutions of the equations $y^{a}=\partial
H(x,p)/\partial p_{a}.$

Our main assumption is that MGTs with MDRs are described by basic Lorentzian
and non-Riemannian total phase space geometries determined by nonlinear
quadratic line elements
\begin{eqnarray}
ds_{L}^{2} &=&L(x,y),\mbox{ for models on  }TV;  \label{nqe} \\
d\ ^{\shortmid }s_{H}^{2} &=&H(x,p),\mbox{ for models on  }T^{\ast }V.
\label{nqed}
\end{eqnarray}
The functions (\ref{nqe}) and (\ref{nqed}) are corresponding called Lagrange
and Hamilton fundamental (equivalent, generating) functions. For localized
zero indicators in (\ref{mdrg}), $\varpi =0,$ (\ref{nqe}) and (\ref{nqed})
transform correspondingly into linear quadratic elements (\ref{lqe}) and (%
\ref{lqed}) and in (pseudo) Riemannian geometry extended on tangent bundles.

A relativistic 4-d model of Lagrange space $L^{3,1}=(TV,L(x,y))$ is defined
by a fundamental/generating Lagrange function, $TV\ni (x,y)\rightarrow
L(x,y)\in \mathbb{R},$ which is a real valued and differentiable function on
$\widetilde{TV}:=TV/\{0\},$ for $\{0\}$ being the null section of $TV,$ and
continuous on the null section of $\pi :TV\rightarrow V.$ We say that such a
model is regular if the vertical metric (v-metric, Hessian)
\begin{equation}
\widetilde{g}_{ab}(x,y):=\frac{1}{2}\frac{\partial ^{2}L}{\partial
y^{a}\partial y^{b}}  \label{hessls}
\end{equation}%
is non-degenerate, i.e. $\det |\widetilde{g}_{ab}|\neq 0,$ and of constant
signature.

In a similar form, we can introduce a 4-d relativistic model of Hamilton
space $H^{3,1}=(T^{\ast }V,H(x,p))$ which is determined by a
fundamental/generating Hamilton function on a Lorentz manifold $V.$ Such a
real valued function is constructed as $T^{\ast }V\ni (x,p)\rightarrow
H(x,p)\in \mathbb{R}$ subjected to the conditions that it is differentiable
on $\widetilde{T^{\ast }V}:=T^{\ast }V/\{0^{\ast }\},$ for $\{0^{\ast }\}$
being the null section of $T^{\ast }V,$ and continuous on the null section
of $\pi ^{\ast }:\ T^{\ast }V\rightarrow V.$ Such a model is regular if the
co-vertical metric (cv-metric, Hessian),
\begin{equation}
\ \ ^{\shortmid }\widetilde{g}^{ab}(x,p):=\frac{1}{2}\frac{\partial ^{2}H}{%
\partial p_{a}\partial p_{b}}  \label{hesshs}
\end{equation}%
is non-degenerate, i.e. $\det |\ ^{\shortmid }\widetilde{g}^{ab}|\neq 0,$
and of constant signature.

We shall use tilde "\symbol{126}", for instance, for values $\widetilde{g}%
_{ab}$ and $\ ^{\shortmid }\widetilde{g}^{ab},$ in order to emphasize that
certain geometric objects are defined canonically by respective Lagrange and
Hamilton generating functions. In their turn, such fundamental functions may
encode various types of MDRs and LIVs terms etc. For general frame/
coordinate transforms on $TV$ and/or $T^{\ast }V,$ we can express any
"tilde" value in a "non-tilde" form. In such cases, we shall write $%
g_{ab}(x,v),$ for a v-metric, and $\ ^{\shortmid }g^{ab}(x,p),$ for a
cv-metric. Inversely, prescribing any v-metric or cv-metric structure, we
can consider such (co) frame /coordinate systems, when the geometric values
can be transformed into certain canonical ones with "tilde". Here we note
that, in general, a v-metric $g_{ab}$ is different from the inverse of $\
^{\shortmid }g^{ab}$ and from $\ ^{\shortmid }g_{ab}.$

It is important to emphasize that a relativistic 4-d model of Finsler space
is an example of relativistic Lagrange space when a regular $L=F^{2}$ is
defined by a fundamental (generating) Finsler function subjected to certain
additional conditions:

\begin{itemize}
\item the fundamental/generating Finsler function $F$ is a real positive
valued one which is differential on $\widetilde{TV}$ and continuous on the
null section of the projection $\pi :TV\rightarrow V;$

\item $F$ satisfies the homogeneity condition $F(x,\lambda v)=|\lambda |$ $%
F(x,v),$ for a nonzero real value $\lambda ;$ and

\item Hessian (\ref{hessls}) is defined by $F^{2}$ in such a form that in
any point $(x_{(0)},v_{(0)})$ the v-metric is of signature $(+++-).$
\end{itemize}

In a similar form, there are defined relativistic 4-d Cartan spaces $%
C^{3,1}=(V,C(x,p)),$ when $H=C^{2}(x,p)$ is 1-homogeneous on co-fiber
coordinates $p_{a}.$ In principle, we can always introduce on a $TV$ and/or $%
T^{\ast }V$ a subclass of Finsler and/or Cartan variables using respective
classes for frame/coordinate transforms. For simplicity, we shall prefer to
use in this work terms like Lagrange and/or Hamilton geometry, or
Lagrange-Hamilton geometry, considering that constructions with Finsler
and/or Cartan structures as certain particular examples.

\subsection{Canonical N--connections and adapted metrics}

There are possible coordinate-free and N-adapted, or local, coefficient
forms for defining geometric objects in Lagrange-Hamilton geometry. Certain
formulas and results will be written in coefficient forms, with respect to
N--adapted frames, which is important for constructing in explicit form
exact and parametric solutions following the AFDM.

\subsubsection{N-connections for Lagrange-Hamilton spaces}

Any MDR (\ref{mdrg}) defines $L$--dual, i.e. related via Legendre
transforms, canonical relativistic models of Hamilton space $\widetilde{H}%
^{3,1}=(T^{\ast }V,\widetilde{H}(x,p))$ and Lagrange space $\widetilde{L}%
^{3,1}=(TV,\widetilde{L}(x,y)).$

Following standard geometric and variational calculus, we can prove such
results:\ The dynamics of a probing point particle in $L$-dual effective
phase spaces $\widetilde{H}^{3,1}$ and $\widetilde{L}^{3,1}$ is described by
fundamental generating functions $\widetilde{H}$ and $\widetilde{L}$ \
subjected to solve respective via Hamilton-Jacobi equations
\begin{equation*}
\frac{dx^{i}}{d\tau }=\frac{\partial \widetilde{H}}{\partial p_{i}}%
\mbox{
and }\frac{dp_{i}}{d\tau }=-\frac{\partial \widetilde{H}}{\partial x^{i}},
\end{equation*}
\begin{equation*}
\mbox{ or Euler-Lagrange equations, }\frac{d}{d\tau }\frac{\partial
\widetilde{L}}{\partial y^{i}}-\frac{\partial \widetilde{L}}{\partial x^{i}}%
=0.
\end{equation*}%
In their turn, such relativistic effective mechanics equations are
equivalent to the \textbf{nonlinear geodesic (semi-spray) equations}
\begin{equation}
\frac{d^{2}x^{i}}{d\tau ^{2}}+2\widetilde{G}^{i}(x,y)=0,\mbox{ for }%
\widetilde{G}^{i}=\frac{1}{2}\widetilde{g}^{ij}(\frac{\partial ^{2}%
\widetilde{L}}{\partial y^{i}}y^{k}-\frac{\partial \widetilde{L}}{\partial
x^{i}}),  \label{ngeqf}
\end{equation}%
$\,\ $ with $\widetilde{g}^{ij}$ being inverse to $\widetilde{g}_{ij}$ (\ref%
{hessls}).

The equations (\ref{ngeqf}) prove that point like probing particles in a
relativistic effective phase space do not move along usual geodesics as on
Lorentz manifolds but follow some nonlinear geodesic equations determined
canonically by MDRs.

By explicit constructions on open sets covering $V,TV$ and $T^{\ast }V,$ it
was proven that there are canonical N--connections determined by MDRs in $L$%
--dual form when
\begin{equation*}
\ \widetilde{\mathbf{N}}=\left\{ \widetilde{N}_{i}^{a}:=\frac{\partial
\widetilde{G}}{\partial y^{i}}\right\} \mbox{ and }\ ^{\shortmid }\widetilde{%
\mathbf{N}}=\left\{ \ ^{\shortmid }\widetilde{N}_{ij}:=\frac{1}{2}\left[ \{\
\ ^{\shortmid }\widetilde{g}_{ij},\widetilde{H}\}-\frac{\partial ^{2}%
\widetilde{H}}{\partial p_{k}\partial x^{i}}\ ^{\shortmid }\widetilde{g}%
_{jk}-\frac{\partial ^{2}\widetilde{H}}{\partial p_{k}\partial x^{j}}\
^{\shortmid }\widetilde{g}_{ik}\right] \right\}
\end{equation*}%
where $\ \ ^{\shortmid }\widetilde{g}_{ij}$ is inverse to $\ \ ^{\shortmid }%
\widetilde{g}^{ab}$ (\ref{hesshs}).

The canonical N--connections s $\widetilde{\mathbf{N}}$ and $\ ^{\shortmid }%
\widetilde{\mathbf{N}}$ define respective canonical systems of N--adapted
(co) frames
\begin{eqnarray}
\widetilde{\mathbf{e}}_{\alpha } &=&(\widetilde{\mathbf{e}}_{i}=\frac{%
\partial }{\partial x^{i}}-\widetilde{N}_{i}^{a}(x,y)\frac{\partial }{%
\partial y^{a}},e_{b}=\frac{\partial }{\partial y^{b}}),\mbox{ on }TV;
\label{cnddapb} \\
\widetilde{\mathbf{e}}^{\alpha } &=&(\widetilde{e}^{i}=dx^{i},\widetilde{%
\mathbf{e}}^{a}=dy^{a}+\widetilde{N}_{i}^{a}(x,y)dx^{i}),\mbox{ on }%
(TV)^{\ast };\mbox{and \ }  \notag \\
\ ^{\shortmid }\widetilde{\mathbf{e}}_{\alpha } &=&(\ ^{\shortmid }%
\widetilde{\mathbf{e}}_{i}=\frac{\partial }{\partial x^{i}}-\ ^{\shortmid }%
\widetilde{N}_{ia}(x,p)\frac{\partial }{\partial p_{a}},\ ^{\shortmid }e^{b}=%
\frac{\partial }{\partial p_{b}}),\mbox{ on }T^{\ast }V;  \label{ccnadap} \\
\ \ ^{\shortmid }\widetilde{\mathbf{e}}^{\alpha } &=&(\ ^{\shortmid
}e^{i}=dx^{i},\ ^{\shortmid }\mathbf{e}_{a}=dp_{a}+\ ^{\shortmid }\widetilde{%
N}_{ia}(x,p)dx^{i})\mbox{ on }(T^{\ast }V)^{\ast }.  \notag
\end{eqnarray}

We conclude that the nonholonomic structure of a Lorentz manifold and
respective (co) tangent bundles nonholonomically deformed by MDR (\ref{mdrg}%
) can be described in equivalent forms using canonical data $(\widetilde{L}%
,\ \widetilde{\mathbf{N}};\widetilde{\mathbf{e}}_{\alpha },\widetilde{%
\mathbf{e}}^{\alpha })$ and/or $(\widetilde{H},\ ^{\shortmid }\widetilde{%
\mathbf{N}};\ ^{\shortmid }\widetilde{\mathbf{e}}_{\alpha },\ ^{\shortmid }%
\widetilde{\mathbf{e}}^{\alpha }).$ For general frame and coordinate
transforms, considering a general N-splitting without effective
Lagrangians/Hamiltonians, the nonholonomic geometry is described in terms of
geometric data $(\mathbf{N};\mathbf{e}_{\alpha },\mathbf{e}^{\alpha })$
and/or $(\ ^{\shortmid }\mathbf{N};\ ^{\shortmid }\mathbf{e}_{\alpha },\
^{\shortmid }\mathbf{e}^{\alpha }).$

Vector fields on \ $T\mathbf{V}$ and $T^{\ast }\mathbf{V}$ are called
d--vectors if they are written in a form adapted to a prescribed
N--connection structure. For instance, we consider tilde and non-tilde
decompositions,
\begin{eqnarray}
\mathbf{X} &=&\widetilde{\mathbf{X}}^{\alpha }\widetilde{\mathbf{e}}_{\alpha
}=\widetilde{\mathbf{X}}^{i}\widetilde{\mathbf{e}}_{i}+X^{b}e_{b}=\mathbf{X}%
^{\alpha }\mathbf{e}_{\alpha }=\mathbf{X}^{i}\mathbf{e}_{i}+X^{b}e_{b}\in T%
\mathbf{TV},  \label{candvector} \\
\ ^{\shortmid }\mathbf{X} &=&\ ^{\shortmid }\widetilde{\mathbf{X}}^{\alpha }%
\widetilde{\mathbf{e}}_{\alpha }=\ ^{\shortmid }\widetilde{\mathbf{X}}^{i}\
^{\shortmid }\widetilde{\mathbf{e}}_{i}+\ ^{\shortmid }X_{b}\ ^{\shortmid
}e^{b}=\ ^{\shortmid }\mathbf{X}^{\alpha }\ ^{\shortmid }\mathbf{e}_{\alpha
}=\ ^{\shortmid }\mathbf{X}^{i}\ ^{\shortmid }\mathbf{e}_{i}+\ ^{\shortmid
}X_{b}\ ^{\shortmid }e^{b}\in T\mathbf{T}^{\ast }\mathbf{V.}  \notag
\end{eqnarray}%
In brief, we can writte conventional h-v and/or h-cv decompositions, $%
\mathbf{X}^{\alpha }=\widetilde{\mathbf{X}}^{\alpha }=(\widetilde{\mathbf{X}}%
^{i},X^{b})=(\mathbf{X}^{i},X^{b}),\ ^{\shortmid }\mathbf{X}^{\alpha }=\
^{\shortmid }\widetilde{\mathbf{X}}^{\alpha }=(\ ^{\shortmid }\widetilde{%
\mathbf{X}}^{i},\ ^{\shortmid }X_{b})=(\ ^{\shortmid }\mathbf{X}^{i},\
^{\shortmid }X_{b}).$ It is possible to consider $\mathbf{X}$ and $\
^{\shortmid }\mathbf{X}$ as 1-forms,
\begin{eqnarray*}
\mathbf{X} &=&\widetilde{\mathbf{X}}_{\alpha }\ \mathbf{e}^{\alpha }=X_{i}\
e^{i}+\widetilde{\mathbf{X}}^{a}\widetilde{\mathbf{e}}_{a}=\widetilde{%
\mathbf{X}}_{\alpha }\mathbf{e}^{\alpha }=X_{i}e^{i}+\mathbf{X}^{a}\mathbf{e}%
_{a}\ \in T^{\ast }\mathbf{TV} \\
\ ^{\shortmid }\mathbf{X} &=&\ ^{\shortmid }\widetilde{\mathbf{X}}_{\alpha
}\ ^{\shortmid }\mathbf{e}^{\alpha }=\ ^{\shortmid }X_{i}\ ^{\shortmid
}e^{i}+\ ^{\shortmid }\widetilde{\mathbf{X}}^{a}\ ^{\shortmid }\widetilde{%
\mathbf{e}}_{a}=\ ^{\shortmid }\widetilde{\mathbf{X}}_{\alpha }\ ^{\shortmid
}\mathbf{e}^{\alpha }=\ ^{\shortmid }X_{i}\ ^{\shortmid }e^{i}+\ ^{\shortmid
}\mathbf{X}^{a}\ ^{\shortmid }\mathbf{e}_{a}\ \in T^{\ast }\mathbf{T}^{\ast }%
\mathbf{V,}
\end{eqnarray*}%
or, in brief, $\mathbf{X}_{\alpha }=\widetilde{\mathbf{X}}_{\alpha }=(X_{i},%
\widetilde{\mathbf{X}}^{a})=(X_{i},\mathbf{X}^{a}),\ ^{\shortmid }\mathbf{%
X_{\alpha }=}\ ^{\shortmid }\widetilde{\mathbf{X}}_{\alpha }=(\ ^{\shortmid
}X_{i},\ ^{\shortmid }\widetilde{\mathbf{X}}^{a})=(\ ^{\shortmid }X_{i},\
^{\shortmid }\mathbf{X}^{a}).$

Using tensor products of N-adapted (co) frames, we can parameterize in
various N-adapted forms arbitrary tensors fields, called as d-tensors
(similarly, for d-connections, d-tensors etc.).

\subsubsection{Canonical d-metrics and d-connections for Lagrange-Hamilton
spaces}

There are canonical d-metric structures $\widetilde{\mathbf{g}}$ and $\
^{\shortmid }\widetilde{\mathbf{g}}$ completely determined by a MDR (\ref%
{mdrg}) and respective data $(\widetilde{L},\ \ \widetilde{\mathbf{N}};%
\widetilde{\mathbf{e}}_{\alpha },\widetilde{\mathbf{e}}^{\alpha };\widetilde{%
g}_{jk},\widetilde{g}^{jk})$ and $(\widetilde{H},\ ^{\shortmid }\widetilde{%
\mathbf{N}};\ ^{\shortmid }\widetilde{\mathbf{e}}_{\alpha },\ ^{\shortmid }%
\widetilde{\mathbf{e}}^{\alpha };\ \ ^{\shortmid }\widetilde{g}^{ab},\ \
^{\shortmid }\widetilde{g}_{ab}),$
\begin{eqnarray}
\widetilde{\mathbf{g}} &=&\widetilde{\mathbf{g}}_{\alpha \beta }(x,y)%
\widetilde{\mathbf{e}}^{\alpha }\mathbf{\otimes }\widetilde{\mathbf{e}}%
^{\beta }=\widetilde{g}_{ij}(x,y)e^{i}\otimes e^{j}+\widetilde{g}_{ab}(x,y)%
\widetilde{\mathbf{e}}^{a}\otimes \widetilde{\mathbf{e}}^{a}\mbox{
and/or }  \label{cdms} \\
\ ^{\shortmid }\widetilde{\mathbf{g}} &=&\ ^{\shortmid }\widetilde{\mathbf{g}%
}_{\alpha \beta }(x,p)\ ^{\shortmid }\widetilde{\mathbf{e}}^{\alpha }\mathbf{%
\otimes \ ^{\shortmid }}\widetilde{\mathbf{e}}^{\beta }=\ \ ^{\shortmid }%
\widetilde{g}_{ij}(x,p)e^{i}\otimes e^{j}+\ ^{\shortmid }\widetilde{g}%
^{ab}(x,p)\ ^{\shortmid }\widetilde{\mathbf{e}}_{a}\otimes \ ^{\shortmid }%
\widetilde{\mathbf{e}}_{b}.  \label{cdmds}
\end{eqnarray}

By straightforward N-adapted calculus with $\widetilde{\mathbf{e}}_{\alpha
}=(\widetilde{\mathbf{e}}_{i},e_{b})$ (\ref{cnddapb}) and$\ ^{\shortmid }%
\widetilde{\mathbf{e}}_{\alpha }=(\ ^{\shortmid }\widetilde{\mathbf{e}}%
_{i},\ ^{\shortmid }e^{b})$ (\ref{ccnadap}) and (\ref{neijtc}), we prove
nontrivial indicators for MDRs and respective canonical N-connection
structures induce canonical nonholonomic frame structures on $\mathbf{TV}$
and/or $\mathbf{T}^{\ast }\mathbf{V.}$ Such nonholonomic frame bases are
characterized by corresponding anholonomy relations
\begin{equation*}
\lbrack \widetilde{\mathbf{e}}_{\alpha },\widetilde{\mathbf{e}}_{\beta }]=%
\widetilde{\mathbf{e}}_{\alpha }\widetilde{\mathbf{e}}_{\beta }-\widetilde{%
\mathbf{e}}_{\beta }\widetilde{\mathbf{e}}_{\alpha }=\widetilde{W}_{\alpha
\beta }^{\gamma }\widetilde{\mathbf{e}}_{\gamma },
\end{equation*}%
with (antisymmetric) anholonomy coefficients $\widetilde{W}%
_{ia}^{b}=\partial _{a}\widetilde{N}_{i}^{b}$ and $\widetilde{W}_{ji}^{a}=%
\widetilde{\Omega }_{ij}^{a},$ and
\begin{equation*}
\lbrack \ ^{\shortmid }\widetilde{\mathbf{e}}_{\alpha },\ ^{\shortmid }%
\widetilde{\mathbf{e}}_{\beta }]=\ ^{\shortmid }\widetilde{\mathbf{e}}%
_{\alpha }\ ^{\shortmid }\widetilde{\mathbf{e}}_{\beta }-\ ^{\shortmid }%
\widetilde{\mathbf{e}}_{\beta }\ ^{\shortmid }\widetilde{\mathbf{e}}_{\alpha
}=\ ^{\shortmid }\widetilde{W}_{\alpha \beta }^{\gamma }\ ^{\shortmid }%
\widetilde{\mathbf{e}}_{\gamma },
\end{equation*}%
with anholonomy coefficients $\ ^{\shortmid }\widetilde{W}_{ib}^{a}=\partial
\ ^{\shortmid }\widetilde{N}_{ib}/\partial p_{a}$ and $\ ^{\shortmid }%
\widetilde{W}_{jia}=\ \mathbf{\ ^{\shortmid }}\widetilde{\Omega }_{ija}.$
Explicit definitions and formulas for the Neijenhuis tensors $\widetilde{%
\Omega }_{ij}^{a}$ and $\ \mathbf{\ ^{\shortmid }}\widetilde{\Omega }_{ija}$
can be found, for instance, in \cite{v18a}, see also (\ref{neijtc}).

There are so-called Cartan-Lagrange and Cartan-Hamilton d-connections
induced directly by an indicator of MDR and determined respectively by
coefficients of Lagrange and Hamilton d-metrics (\ref{cdms}) and (\ref{cdmds}%
), when all coefficients are generated by "tilde" objects with
identifications of d-metric coefficients with corresponding base and (co)
fiber indices:
\begin{eqnarray}
\mbox{ on }T\mathbf{TV},\ \widetilde{\mathbf{D}} &=&\{\widetilde{\mathbf{%
\Gamma }}_{\ \alpha \beta }^{\gamma }=(\widetilde{L}_{jk}^{i},\widetilde{L}%
_{bk}^{a},\widetilde{C}_{jc}^{i},\widetilde{C}_{bc}^{a})\},\mbox{ for }%
\mathbf{[}\widetilde{\mathbf{g}}_{\alpha \beta }=(\widetilde{g}_{jr},%
\widetilde{g}_{ab}),\widetilde{\mathbf{N}}_{i}^{a}=\widetilde{N}_{i}^{a}],
\notag \\
\widetilde{L}_{jk}^{i} &=&\frac{1}{2}\widetilde{g}^{ir}\left( \widetilde{%
\mathbf{e}}_{k}\widetilde{g}_{jr}+\widetilde{\mathbf{e}}_{j}\widetilde{g}%
_{kr}-\widetilde{\mathbf{e}}_{r}\widetilde{g}_{jk}\right) ,\ \widetilde{L}%
_{bk}^{a}\mbox{ as }\widetilde{L}_{jk}^{i},  \notag \\
\ \widetilde{C}_{bc}^{a} &=&\frac{1}{2}\widetilde{g}^{ad}\left( e_{c}%
\widetilde{g}_{bd}+e_{b}\widetilde{g}_{cd}-e_{d}\widetilde{g}_{bc}\right) %
\mbox{ being similar to }\widetilde{C}_{jc}^{i};  \label{carlc}
\end{eqnarray}%
\begin{eqnarray}
\mbox{ and, on }T\mathbf{T}^{\ast }\mathbf{V},\ \ ^{\shortmid }\widetilde{%
\mathbf{D}} &=&\{\ ^{\shortmid }\widetilde{\mathbf{\Gamma }}_{\ \alpha \beta
}^{\gamma }=(\ ^{\shortmid }\widetilde{L}_{jk}^{i},\ ^{\shortmid }\widetilde{%
L}_{a\ k}^{\ b},\ ^{\shortmid }\widetilde{C}_{\ j}^{i\ c},\ ^{\shortmid }%
\widetilde{C}_{\ j}^{i\ c})\},\mbox{ for }\mathbf{[}\ ^{\shortmid }%
\widetilde{\mathbf{g}}_{\alpha \beta }=(\ ^{\shortmid }\widetilde{g}_{jr},\
^{\shortmid }\widetilde{g}^{ab}),\ ^{\shortmid }\widetilde{\mathbf{N}}%
_{ai}=\ ^{\shortmid }\widetilde{N}_{ai}],  \notag \\
\ ^{\shortmid }\widetilde{L}_{jk}^{i} &=&\frac{1}{2}\ ^{\shortmid }%
\widetilde{g}^{ir}(\ ^{\shortmid }\widetilde{\mathbf{e}}_{k}\ ^{\shortmid }%
\widetilde{g}_{jr}+\ ^{\shortmid }\widetilde{\mathbf{e}}_{j}\ ^{\shortmid }%
\widetilde{g}_{kr}-\ ^{\shortmid }\widetilde{\mathbf{e}}_{r}\ ^{\shortmid }%
\widetilde{g}_{jk}),\ \mbox{ with similar }\ ^{\shortmid }\widetilde{L}_{a\
k}^{\ b},  \notag \\
\ \ \ ^{\shortmid }\widetilde{C}_{\ a}^{b\ c} &=&\frac{1}{2}\ ^{\shortmid }%
\widetilde{g}_{ad}(\ ^{\shortmid }e^{c}\ ^{\shortmid }\widetilde{g}^{bd}+\
^{\shortmid }e^{b}\ ^{\shortmid }\widetilde{g}^{cd}-\ ^{\shortmid }e^{d}\
^{\shortmid }\widetilde{g}^{bc})\mbox{ being similar to }^{\shortmid }%
\widetilde{C}_{\ j}^{i\ c}.  \label{carhc}
\end{eqnarray}%
These d-connections are similar respectively to the canonical d-connections (%
\ref{canlc}) and (\ref{canhc}) from Corollary \ref{acoroldand} but possess
an important property that they are also canonical almost symplectic
connections. It is difficult to find rich classes of exact physically
important solutions working directly with $\widetilde{\mathbf{D}}$ (\ref%
{carlc}) \ or $\ ^{\shortmid }\widetilde{\mathbf{D}}$ (\ref{carhc}).

\end{document}